\documentclass[a4paper,UKenglish,cleveref, autoref, thm-restate]{lipics-v2021}

\usepackage{amsfonts}
\usepackage{amsmath,xspace}
\usepackage{mathtools}
\usepackage{bussproofs}
\usepackage{stmaryrd}
\usepackage{mdframed}
\usepackage{multicol}
\mdfdefinestyle{alttight}{innertopmargin=-2pt,innerleftmargin=2pt,innerrightmargin=2pt,innerbottommargin=2pt}

\usepackage[dvipsnames]{xcolor}

\usepackage{enumerate,color}

\newcommand{\spi}{\ensuremath{\mathsf{s}\pi}\xspace}
 
\newcommand{\lamrfail}{\ensuremath{\lambda^{\lightning}_{\oplus}}\xspace}
\newcommand{\lamrfailunres}{\ensuremath{\mathit{u}\lambda^{ \lightning}_{\oplus}}\xspace}
\newcommand{\lamrsharfailunres}{\ensuremath{u\widehat{\lambda}^{\lightning}_{\oplus}}\xspace}

\newcommand{\linvar}[1]{ #1 ^{\ell}}
\newcommand{\banged}[1]{ #1 ^{!}  }
\newcommand{\unvar}[1]{{#1}}
\newcommand{\bagsep}{\star}

\newcommand{\sep}{\ | \ } 
\newcommand{\bag}[1]{\lbag #1 \rbag} 
\newcommand{\oneb}{\mathtt{1}}
\newcommand{\perm}[1]{\ensuremath{\mathsf{PER}(#1)}}
\newcommand{\dom}[1]{\mathtt{dom}(#1)}
\newcommand{\arrt}[2]{\ensuremath{#1 \rightarrow #2}}
\newcommand{\fail}{\mathtt{fail}}

\newcommand{\concat}{\ensuremath{\diamond}}
\newcommand{\relunbag}{\ensuremath{ \varpropto}}

\newcommand{\size}[1]{\mathsf{size}(#1)} 
\newcommand{\headf}[1]{\ensuremath{\mathsf{head}(#1)}}
\newcommand{\headsum}[1]{\ensuremath{\mathsf{head}_{\sum}}(#1)}

\newcommand{\dash}{\text{-}}

\def\subst#1#2{\{ \raisebox{.5ex}{\small$#1$}\! / \mbox{\small$#2$}\}} 
\def\linsub#1#2{\langle \raisebox{.5ex}{\small$#1$}\! / \mbox{\small$#2$}\rangle} 
\def\esubst#1#2{\langle\!\langle \raisebox{.5ex}{\small$#1$}\! / \mbox{\small$#2$}\rangle\!\rangle} 

\newcommand{\headlin}[1]{ { \{\!|} #1 { |\!\} }} 
\newcommand{\linexsub}[1]{{\langle \! |} #1 {| \! \rangle} }
\newcommand{\unexsub}[1]{\llfloor #1 \rrceil } 

\newcommand{\expr}[1]{\ensuremath{\mathbb{#1}}}
\newcommand{\lfv}[1]{\mathsf{fv}(#1)}
\newcommand{\llfv}[1]{\mathsf{lfv}(#1)}
\newcommand{\mfv}[1]{\mathsf{mfv}(#1)}
\newcommand{\mlfv}[1]{\mathsf{mlfv}(#1)}
\newcommand{\unit}{\mathbf{unit}}
\newcommand{\head}[1]{\mathsf{head}(#1)} 

\newcommand{\red}{\longrightarrow}
\newcommand{\tred}{\stackrel{*}{\red}}
\newcommand{\redlab}[1]{\ensuremath{\mathtt{[#1]} }}

\newcommand{\pequiv}{\equiv_\lambda}

\newcommand{\secref}[1]{$\S$\,\ref{#1}\xspace}
\newcommand{\figref}[1]{Fig.\,\ref{#1}\xspace}
\newcommand{\defref}[1]{Def.\,\ref{#1}\xspace}
\newcommand{\appref}[1]{App.\,\ref{#1}\xspace}
\newcommand{\thmref}[1]{Theorem~\ref{#1}\xspace}

\newcommand{\wfdash}{\models}

\newcommand{\outact}[2]{\overline{#1}(#2)}
\newcommand{\outsev}[2]{\overline{#1}?(#2)}

\newcommand{\some}{\mathtt{some}}
\newcommand{\none}{\mathtt{none}}

\newcommand{\case}[2]{{#1}.\mathtt{l}_{#2}}
\newcommand{\choice}[5]{{#1}.\mathtt{case}_{#2 \in #3} \{ \mathtt{l}_{#4} : #5 \}}

\newcommand{\close}{\mathtt{close}}
\newcommand{\sclose}{\ensuremath{[\,]}}

\newcommand{\para}{\mathord{\;\mathbf{|}\;}}
\newcommand{\zero}{{\bf 0}}

\newcommand{\fn}[1]{\mathit{fn}(#1)}

\newcommand{\ampy}{\mathbin{\bindnasrepma}}
\newcommand{\with}{{\binampersand}}
\newcommand{\onef}{\mathbf{1}}
\newcommand{\dual}[1]{\overline{#1}}

\newcommand{\colorone}[1]{\textcolor{RedOrange}{#1}}
\newcommand{\colortwo}[1]{\textcolor{RoyalBlue}{#1}}

\newcommand{\encod}[2]{\llbracket#1\rrbracket_{#2}} 
\newcommand{\recencod}[1]{\{ \!\! \{ #1 \} \!\! \} } 

\newcommand{\recencodf}[1]{\colorone{\llparenthesis}  #1 \colorone{\rrparenthesis^{\bullet}}} 

\newcommand{\recencodopenf}[1]{\colorone{\llparenthesis}  #1 \colorone{\rrparenthesis^{\circ}}}

\newcommand{\piencodf}[1]{\colortwo{\llbracket}  #1 \colortwo{\rrbracket}}

\newcommand{\succp}[2]{\ensuremath{#1 \Downarrow_{#2}}} 

\newcommand{\joehide}[1]{}
\newcommand{\alerthide}[1]{}

\newcommand{\colthree}[1]{\textcolor{blue}{ #1 }}

\usepackage{tikz}
\usetikzlibrary{matrix}
\usetikzlibrary{decorations.pathreplacing,calligraphy}
\usetikzlibrary{shapes.geometric, arrows}

\newtheorem{notation}{Notation}

\newif\iftypes
\typesfalse  

\title{Types and Terms Translated:  Unrestricted Resources in Encoding  Functions as Processes}

\titlerunning{Unrestricted Resources in Encoding  Functions as Processes}

\author{Joseph W. N. Paulus}
{University of Groningen, The Netherlands \and \url{https://www.rug.nl/staff/j.w.n.paulus/?lang=en} }
{j.w.n.paulus@rug.nl}
{https://orcid.org/0000-0002-1711-9361?lang=en}{}

\author{Daniele Nantes-Sobrinho}
{University of Bras\'ilia, Bras\'ilia, Brazil \and \url{https://www.mat.unb.br/~dnantes/Welcome.html}}
{daniele.nantes@gmail.com}
{https://orcid.org/0000-0002-1959-8730}{}

\author{Jorge A. P\'{e}rez}
{University of Groningen, The Netherlands \and \url{https://www.jperez.nl}}
{j.a.perez@rug.nl}
{https://orcid.org/0000-0002-1452-6180}{}

\authorrunning{J.\,W. N. Paulus,  D.\,Nantes-Sobrinho, and J.\,A. P\'{e}rez} 

\Copyright{Joseph W. N.  Paulus, Daniele Nantes-Sobrinho and Jorge A. P\'{e}rez} 

\ccsdesc{Theory of computation~Type structures}
\ccsdesc{Theory of computation~Process calculi}

\keywords{Resource $\lambda$-calculus, intersection types, session types, process calculi.} 

\category{} 

\iftypes
\relatedversion{Online appendix with omitted proofs and further examples:}
\relatedversiondetails[cite=DBLP:journals/corr/abs-2111]{Full Version} {http://arxiv.org/abs/2112.01593}
\else
\fi

 \acknowledgements{We are grateful to the anonymous reviewers for their constructive feedback.}
\funding{Research partially supported by the Dutch Research Council (NWO) under project No. 016.Vidi.189.046 (Unifying Correctness for Communicating Software).}

\EventEditors{Henning Basold, Jesper Cockx, and Silvia Ghilezan}
\EventNoEds{3}
\EventLongTitle{27th International Conference on Types for Proofs and Programs}
\EventShortTitle{TYPES 2021}
\EventAcronym{TYPES}
\EventYear{2021}
\EventDate{June 14--18, 2021}
\EventLocation{ University Leiden,Netherlands}
\EventLogo{}
\SeriesVolume{39}
\ArticleNo{10}

\hideLIPIcs
 \nolinenumbers

\begin{document}

\maketitle

\begin{abstract}
Type-preserving translations are effective rigorous tools in the study of core programming calculi. 
In this paper, we develop a new typed translation that connects sequential and concurrent calculi; it is governed by type systems that control \emph{resource consumption}. 
Our main contribution is the source language, a  new resource $\lambda$-calculus with non-determinism and failures, dubbed  \lamrfailunres. 
In \lamrfailunres, resources are split into linear and unrestricted; failures are explicit and arise from this distinction. 
We define a type system  based on  intersection types
to control resources and fail-prone computation.
The target language is \spi, an existing session-typed $\pi$-calculus that results from a Curry-Howard correspondence between linear logic and session types. 
Our typed translation subsumes our prior work; interestingly, it treats  unrestricted resources in \lamrfailunres as client-server session behaviours in \spi. 
\end{abstract}

\section{Introduction}







\subparagraph{Context}

\emph{Type-preserving translations} are effective rigorous tools in the study of core programming calculi. 
They can be seen as an abstract counterpart to the type-preserving compilers that enable key optimisations in the implementation of programming languages.
The goal of this paper is to develop a new typed translation that connects sequential and concurrent calculi, and is governed by  type systems that control \emph{resource consumption}.

A central idea in the {resource} $\lambda$-calculus is to consider that in an application $M\, N$ the argument $N$ is a \emph{resource} of possibly limited availability. 
This generalisation of the $\lambda$-calculus triggers many fascinating questions, such as typability, solvability, expressiveness power, etc.,  which have been studied in different settings (see, e.g.,~\cite{DBLP:conf/concur/Boudol93,DBLP:conf/birthday/BoudolL00,PaganiR10,DBLP:journals/corr/abs-1211-4097}). 
In established resource $\lambda$-calculi, such as those by Boudol~\cite{DBLP:conf/concur/Boudol93} and by Pagani and Ronchi della Rocca~\cite{PaganiR10}, a more general  form of application is considered:  a term can be applied to a bag of resources $B=\bag{N_1}\cdot \ldots \cdot \bag{N_k}$, where  $N_1, \ldots, N_k$ denote terms; then, an application $M\ B$ must take into account that each $N_i$ may be reusable or not.  Thus, non-determinism is natural in resource $\lambda$-calculi, because a  term has now  multiple ways of consuming resources from the bag. This bears a strong resemblance with  process calculi such as the $\pi$-calculus~\cite{DBLP:journals/iandc/MilnerPW92a}, in which concurrent interactions are intrinsically non-deterministic. 

{There are different flavors of non-determinism. Over two decades ago, Boudol and Laneve~\cite{BoudolL96,DBLP:conf/birthday/BoudolL00} explored connections between a resource $\lambda$-calculus and the $\pi$-calculus. In their setting, an application $M\ B$ would branch, i.e., $M$ could consume a resource $N_j$ in $B$ (with $j \in \{1, \ldots k\}$) and discard the other $k-1$ resources in a non-confluent manner; this is what we call a {\em collapsing} approach to non-determinism.
On a different direction, Pagani and Ronchi della Rocca~\cite{PaganiR10} proposed $\lambda^r$, a resource $\lambda$-calculus that implements  \emph{non-collapsing}  non-determinism, whereby all the possible alternatives for resource consumption are retained together in a sum, ensuring confluence. They investigated typability and characterisations of solvability in $\lambda^r$, but no connection with the $\pi$-calculus was established. In an attempt to address this gap, our previous work~\cite{DBLP:conf/fscd/PaulusN021} identified $\lamrfail$, a resource $\lambda$-calculus with non-collapsing non-determinism, explicit failure, and \emph{linear} resources  (to be used exactly once), and developed a correct typed translation into a session typed $\pi$-calculus~\cite{CairesP17}. } The calculus~$\lamrfail$, however, does not include \emph{unrestricted} resources (to be used zero or many times).

\subparagraph{This Paper}
Here we introduce a new $\lambda$-calculus, dubbed $\lamrfailunres$, its intersection type system, and its translation into session-typed processes.
Our motivation is twofold: to elucidate the status of unrestricted resources in a functional setting with non-collapsing non-determinism, and to characterise unrestricted resources within a translation of functions into processes.
Unlike its predecessors, \lamrfailunres distinguishes between  {linear} and {unrestricted} resources. This distinction determines the semantics of terms and especially the deadlocks (\emph{failures}) that arise due to mismatches in resources. This way,  $\lamrfailunres$ subsumes   $\lamrfail$, which is  purely linear  and cannot express failures related to unrestricted resources.

Distinguishing  linear and unrestricted resources is not a new insight. This idea goes back to Boudol's $\lambda$-calculus with multiplicities~\cite{DBLP:conf/concur/Boudol93}, where arguments can be tagged as unrestricted. 
What is new about $\lamrfailunres$ is that the distinction between linear and unrestricted resources leads to two main differences. 
First,  occurrences of a variable can be linear or unrestricted, depending on the kind of resources they should be substituted with. This way, e.g., a linear occurrence of variable must be substituted with a linear resource. 
In $\lamrfailunres$, a variable can have linear and unrestricted occurrences in the same term.
(Notice that we use the adjective `{linear}' in connection to resources used exactly once, and not to the number of occurrences of a variable in a term.)
Second, failures depend on the nature of the involved resource(s). In $\lamrfailunres$, a linear failure arises from a mismatch between required and available (linear) resources; an unrestricted failure arises when a specific (unrestricted) resource is not available. 

Accordingly, the syntax of $\lamrfailunres$ incorporates linear and unrestricted resources, enabling their consistent separation, within non-collapsing non-determinism. 
The calculus allows for linear and unrestricted occurrences of variables, as just discussed;  
bags comprise two separate zones, linear and unrestricted; 
and the \emph{failure term} $\fail^{x_1, \cdots, x_n}$ explicitly mentions the linear variables $x_1, \ldots, x_n$. 
The (lazy) reduction semantics of $\lamrfailunres$ includes two different rules for ``fetching'' terms from bags, and for consistently handling the failure term. 
\alerthide{reviewer here wants non-deterministic sums instead, I am unsure what you both think would be best here}

 We equip $\lamrfailunres$ with non-idempotent intersection types, extending the approach in~\cite{DBLP:conf/fscd/PaulusN021}:   in $\lamrfailunres$, intersection types  account for more than resource multiplicity, since the elements of the unrestricted bag can have different types. Using intersection types,  we define a class of \emph{well-formed} \lamrfailunres expressions, which includes terms that correctly consume resources but also terms that may reduce to the failure term. Well-formed expressions thus subsume the \emph{well-typed} expressions that can be defined in a sub-language of \lamrfailunres without the failure term.  

The calculus \lamrfailunres can express terms whose dynamic behaviour is not captured by prior works. This way, e.g., the  identity function ${\bf I}$ admits two formulations, depending on whether the variable occurrence  is linear or unrestricted. 
One can have $\lambda x. x$, as usual, but also the unrestricted variant $\lambda x. x[i]$, where `$[i]$' is an index annotation (similar to a qualifier or a tag), which indicates that $x$ should be replaced by the $i$-th element of the unrestricted zone of the bag. The behaviour of these functions will depend on the bags that are provided as their arguments.
Similarly,  we can express variants of $\Delta=\lambda x.xx$ and $\Omega=\Delta\,\Delta$ whose  behaviours again depend on linear or unrestricted occurrences of variables and bags. 
{Consider the term $\Delta_7=\lambda x. (x[1](\oneb \bagsep \bag{x[1]}^!\concat \bag{x[2]}^!))$, 
where we use `$\bagsep$' to separate linear and unrestricted resources in the bag, and `$\concat$' denotes concatenation of unrestricted resources. 
Term $\Delta_7$ is an abstraction on $x$ of an application of an unrestricted occurrence of $x$, which aims to consume the first component of an unrestricted bag, to a bag with an empty linear zone (denoted $\oneb$) and an unrestricted zone with resources $\bag{x[1]}^!$ and $\bag{x[2]}^!$. The self-application $\Delta_7\Delta_7$ produces a non-terminating behaviour and yet $ \Delta_7 $ itself is well-formed}  (see Example~\ref{ex:delta7_wf1}).

  Both \lamrfailunres and \lamrfail are  \emph{logically motivated} resource $\lambda$-calculi, in the following sense: their design has been strongly influenced by $\spi$, a typed $\pi$-calculus resulting from the Curry-Howard correspondence between linear logic and session types in~\cite{CairesP17}, where proofs correspond to processes and cut elimination to process communication. 
As demonstrated in~\cite{CairesP17}, providing primitive support for explicit failures is key to expressing many useful programming idioms (such as exceptions); this insight is a leading motivation in our design for \lamrfailunres.


To attest to the logical underpinnings of \lamrfailunres, we develop a typed translation (or \emph{encoding}) of $\lamrfailunres$ into $\spi$ and establish its correctness with respect to well-established criteria \cite{DBLP:journals/iandc/Gorla10,DBLP:journals/iandc/KouzapasPY19}. 
As in~\cite{DBLP:conf/fscd/PaulusN021}, we encode \lamrfail into \spi by relying on an intermediate language with \emph{sharing} constructs~\cite{DBLP:conf/lics/GundersenHP13,GhilezanILL11,DBLP:journals/iandc/KesnerL07}. 
A key idea in encoding \lamrfailunres is to codify the behaviour of unrestricted occurrences of a variable and their corresponding resources in the bag as \emph{client-server connections}, leveraging the copying semantics for the exponential ``$!A$'' induced by the Curry-Howard correspondence. 
This typed encoding into \spi  justifies the semantics of \lamrfailunres in terms of precise session protocols (i.e., linear logic propositions, because of the correspondence).

 In summary, the \textbf{main contributions} of this paper are:
 (1)~The resource calculus \lamrfailunres of linear and unrestricted resources, and its associated intersection type system. 
    (2)~A typed encoding of \lamrfailunres into \spi,  which connects well-formed expressions (disciplined by intersection types) and well-typed concurrent processes (disciplined by session types,  under the Curry-Howard correspondence with linear logic), subsuming the results in~\cite{DBLP:conf/fscd/PaulusN021}.


\iftypes
\else
\subparagraph{Additional Material}
The appendices contain omitted material. 
\appref{appA} collects  technical details on \lamrfailunres. 
\appref{appB} details the proof of subject reduction for well-formed \lamrfailunres expressions. \appref{appC}--\appref{app:encodingtwo} collect omitted definitions and proofs for our encoding of \lamrfailunres into \spi. 
\fi 

\section[Unrestricted Resources, Non-Determinism, and Failure]{\lamrfailunres: Unrestricted Resources, Non-Determinism, and Failure}
\label{s:lambda}


\subparagraph{Syntax.} We shall use $x, y, \ldots$ to range over \emph{variables}, and $i,j\ldots$, as positive integers, to range over  {\em indices}.
Variable occurrences will be \emph{annotated} to distinguish the kind of resource  they should be substituted with (linear or unrestricted). 
With a slight abuse of terminology, we may write `linear variable' and `unrestricted variable' to refer to linear and unrestricted {occurrences} of a variable. 
As we will see, a variable's annotation will be inconsequential for binding purposes.
We write  $ {\widetilde{x}}$ to abbreviate  $ {x}_1, \ldots,  {x}_n$, for $n\geq 1$ and each $x_i$ distinct.

%

\begin{definition}[$\lamrsharfailunres$]\label{def:rsyntaxfailunres}
We define \emph{terms} ($M,N$), \emph{bags} ($A,B$), and \emph{expressions} ($\expr{M}, \expr{N}$) as:
\begin{align*}
&\mbox{(Annotations)}&[*]&::= [i] \sep [{\ell}] \qquad i\in \mathbb{N} \\
& \mbox{(Terms)} &M,N &::=   x[*] \sep \lambda x . M \sep (M\ B) \sep  M \esubst{B}{x}  \sep   \fail^{ {\widetilde{x}}}   \\
&\mbox{(Linear Bags)} &C, D &::= \oneb \sep \bag{M}  \cdot\, C 
\\
& \mbox{(Unrestricted Bags)} & U, V &::= \banged{\oneb} \sep \banged{\bag{M}} \sep U \concat V \\
& \mbox{(Bags)} &A, B&::=  C \bagsep U \\
&\mbox{(Expressions)} & \expr{M}, \expr{N} &::=  M \sep \expr{M}+\expr{N}
\end{align*}
To lighten up notation, we shall omit the annotation for linear variables. This way, e.g., we write $(\lambda x. x)B$
rather than $(\lambda x. x[{\ell}])B$.
\end{definition}

\Cref{def:rsyntaxfailunres} introduces three syntactic categories: \emph{terms} (in functional position); \emph{bags} (multisets of resources, in argument position), and \emph{expressions}, which are finite formal sums that denote possible results of a computation. Below we describe each category in details. 
\begin{itemize} 
\item Terms (unary expressions):
\begin{itemize}
\item Variables: We write $x[{\ell}]$ to denote a \emph{linear} occurrence of $x$, i.e,  an occurrence that can only be substituted for linear resources. Similarly, $x[i]$ denotes an \emph{unrestricted} occurrence of $x$, i.e., an occurrence that can only be substituted for a resource located at the $i$-th position of an unrestricted bag. 
\item Abstractions  $\lambda x. M$ of a variable $x$ in a term $M$, which may have contain linear or unrestricted occurrences of $x$. 
This way, e.g.,   $\lambda x.x$ and $\lambda x. x[i]$ are linear and unrestricted versions of the identity function. 
Notice that the scope of $x$ is $M$, as usual, and that $\lambda x. (\cdot)$ binds both linear and unrestricted occurrences of $x$. 
\item Applications  of a term $M$ to a bag $B$ (written $M\ B$) and the explicit substitution of a bag $B$ for a variable $x$ (written $\esubst{B}{x}$) are as expected (cf.~\cite{DBLP:conf/concur/Boudol93,DBLP:conf/birthday/BoudolL00}). 
Notice that in $M\esubst{B}{x}$ the occurrences of $x$ in $M$, linear and unrestricted, are bound.
Some conditions apply to $B$: this will be evident later on, after we define our operational semantics (cf. \figref{fig:reductions_lamrfailunres}).
\item The failure term $\fail^{ {\widetilde{x}}}$ denotes a term that will  result from a reduction in which there is a lack or excess of resources, where  $ {\widetilde{x}}$ denotes a multiset of free linear variables that are encapsulated within failure. 
\end{itemize}

\item A bag $B$ is defined as $C\bagsep U$: the concatenation of a  bag of linear resources $C$ with a bag (actually, a list) of unrestricted resources $U$. We write $\bag{M}$  to denote the linear bag that encloses term $M$, and use $\bag{M}^!$ in the unrestricted case.
\begin{itemize}
\item Linear bags ($C, D, \ldots$) are multisets of terms. The empty linear bag is denoted~$\oneb$. We write $C_1 \cdot C_2$ to denote the concatenation of   $C_1$ and $C_2$; this is a commutative and associative operation, where $\oneb$ is the identity. 
\item Unrestricted bags ($U, V, \ldots$) are ordered lists of terms. 
The empty  unrestricted bag is denoted as $\banged{\oneb}$.
 The concatenation of $U_1$ and $U_2$ is denoted by $U_1\concat U_2$; this operation is  associative but not commutative. 
 Given $i \geq 1$, we write $U_i$ to denote the $i$-th element of the unrestricted (ordered) bag $U$.
 \end{itemize}
 \item  Expressions are sums of terms, denoted as $\sum_{i}^{n} N_i$, where $n > 0$. Sums are associative and commutative; reordering of the terms in a sum is performed silently. 
\end{itemize}
\begin{example}\label{ex:var_lin_unr}
Consider the term  
 $M:=\lambda x.( x[1] \bag{x}\bagsep \bag{y[1]}^!)$,
 which has linear and unrestricted occurrences of the same variable.
  This  is an abstraction of an application that contains two bound occurrences of $x$ (one unrestricted with index $1$, and one linear) and one free unrestricted occurrence of $y[1]$, occurring in an unrestricted bag. 
 As we will see, in $M\ (C\bagsep U)$, the unrestricted occurrence `$x[1]$' should be replaced by the first element of $U$. 
 \end{example}


 The salient features of \lamrfailunres ---the explicit construct for failure, the index annotations on unrestricted variables,  the ordering of unrestricted bags---are \emph{design choices} that  will be responsible for interesting behaviours, as   the following examples  illustrate.

\begin{example}\label{ex:id_term}
As already mentioned, 
    $\lamrfailunres$ admits different  variants of the usual $\lambda$-term ${\bf I}=\lambda x. x$.  We could have one in which $x$ is a linear  variable (i.e., $\lambda x. x$), but also several possibilities if $x$ is unrestricted (i.e., $\lambda x. x[i]$, for some positive integer $i$).
    Interestingly, because \lamrfailunres supports {failures}, {non-determinism}, and the {consumption} of arbitrary terms of the unrestricted bag, these two variants of ${\bf I}$ can have  behaviours that may differ from the usual interpretation of ${\bf I}$. In Example~\ref{ex:id_sem} we will show that the six terms below give different behaviours:
    
\begin{minipage}{5cm}
  \begin{itemize}
        \item $M_1= (\lambda x. x )(\bag{N}\bagsep U)$
        \item $M_2= (\lambda x. x )(\bag{N_1}\cdot \bag{N_2}\bagsep U)$
        \item $M_3=(\lambda x. x[1] )(\bag{N}\bagsep \oneb^!)$
    \end{itemize}
\end{minipage} 
\hspace*{0.5cm}
\begin{minipage}{5cm}
  \begin{itemize}
        \item $M_4=(\lambda x. x[1] )( \oneb \bagsep \bag{N}^! \concat U)$
        \item $M_5=(\lambda x. x[1] )( \oneb \bagsep \oneb^! \concat U)$
        \item $M_6=(\lambda x. x[i] )( C\concat U)$
    \end{itemize}
\end{minipage}

\noindent We will see that $M_1$, $M_4$, $M_6$ reduce without failures, whereas $M_2$, $M_3$, $M_5$ reduce to failure. 
\end{example}


\begin{example}\label{ex:delta}
Similarly, \lamrfailunres allows for several forms of the standard $\lambda$-terms such as $\Delta:=\lambda x. xx$ and $\Omega:=\Delta \Delta$, depending on whether the variable $x$ is linear or unrestricted:
\begin{enumerate}
 \item  $\Delta_1:= \lambda x. (x(\bag{x}\bagsep \oneb^!))$ consists of an abstraction of a linear occurrence of $x$ applied to a linear bag containing another linear occurrence of $x$. There are two forms of self-applications of $\Delta_1$, namely: $ \Delta_1(\bag{\Delta_1}\bagsep \oneb^!)$ and $ \Delta_1(1\bagsep \bag{\Delta_1}^!)$.
\item $\Delta_4:= \lambda x. (x[1](\bag{x}\bagsep \oneb^!))$ consists of an unrestricted occurrence of $x$ applied to a linear bag (containing a linear occurrence of $x$) that is composed with an empty unrestricted bag. Similarly, there are two self-applications of $\Delta_4$, namely: $ \Delta_4(\bag{\Delta_4}\bagsep \oneb^!)$ and $ \Delta_4(\oneb\bagsep \bag{\Delta_4}^!)$.

\item We show applications of  an unrestricted variable occurrence ($x[2]$ or $x[1]$) applied to an empty linear bag composed with a non-empty  unrestricted bag (of size two):

\begin{minipage}{6cm}
    \begin{itemize}
        \item $\Delta_3= \lambda x.(x[1](\oneb\bagsep \bag{x[1]}^!\concat \bag{x[1]}^!))$
        \item $\Delta_5:= \lambda x. (x[2](\oneb\bagsep \bag{x[1]}^!\concat \bag{x[2]}^!))$ 
    \end{itemize}
\end{minipage} 
\hspace*{0.5cm}
\begin{minipage}{6cm}
    \begin{itemize}
        \item $\Delta_6:= \lambda x.( x[1](\oneb\bagsep \bag{x[1]}^!\concat \bag{x[2]}^!))$ 
        \item  $\Delta_7:= \lambda x.( x[2](\oneb\bagsep \bag{x[1]}^!\concat \bag{x[1]}^!))$
    \end{itemize}
\end{minipage}

Applications between these terms express behaviour, similar to a lazy evaluation of  $\Omega$: 

\begin{minipage}{6cm}
    \begin{itemize}
        \item $\Omega_5:=\Delta_5(\oneb\bagsep \bag{\Delta_5}^!\concat \bag{\Delta_5}^!)$
        \item $\Omega_{5,6}:=\Delta_5(\oneb\bagsep \bag{\Delta_5}^!\concat \bag{\Delta_6}^!)$
    \end{itemize}
\end{minipage} 
\hspace*{0.5cm}
\begin{minipage}{6cm}
    \begin{itemize}
        \item $\Omega_{6,5}:=\Delta_6(\oneb\bagsep \bag{\Delta_5}^!\concat \bag{\Delta_6}^!)$
        \item $\Omega_7:=\Delta_7(\oneb\bagsep \bag{\Delta_7}^!\concat \bag{\Delta_7}^!)$
    \end{itemize}
\end{minipage}

\end{enumerate}
The behaviour of these terms will be made explicit later on (see Examples~\ref{ex:deltasem} and \ref{ex:deltasemii}).
\end{example}



\subparagraph{Semantics.}
The semantics of \lamrfailunres captures that linear resources can be used only once, and that unrestricted resources can be used {\em ad libitum}. Thus, the evaluation of a function applied to a multiset of linear resources produces 
 different possible behaviours, depending on the way these resources are substituted for the linear variables. This induces non-determinism, which we formalise using a \emph{non-collapsing} approach, in which  expressions keep all the different possibilities open, and do not commit to one of them. This is in contrast to \emph{collapsing} non-determinism, in which selecting one alternative discards the rest.
 
We define a reduction relation $\red$,   which operates lazily on expressions. Informally, a $\beta$-reduction induces an explicit substitution of a bag $B=C\bagsep U$ for a variable $x$, denoted $\esubst{B}{x}$, in a term $M$. 
This explicit substitution is then expanded 
depending on whether the head of $M$
 has a linear or an unrestricted variable. 
 Accordingly, in \lamrfailunres there are \emph{two sources of failure}: one concerns mismatches on linear resources (required vs available resources); the other concerns the unavailability of a required unrestricted resource (an empty bag $\banged{\oneb}$).
 
To formalise reduction, we require a few auxiliary notions.

\begin{definition}\label{d:fvars}
The multiset of free linear variables of  $\mathbb{M}$, denoted $\mlfv{\mathbb{M}}$,  is defined below.
We denote by $[ {x}]$ the multiset containing the linear variable $x$ and $[x_1,\ldots, x_n]$ denotes the multiset containing $x_1,\ldots, x_n$. We write $\widetilde{x}\uplus \widetilde{y}$ to denote the multiset union of $\widetilde{x}$, and $\widetilde{y}$ and $\widetilde{x} \setminus y$ to express that every occurrence of $y$ is removed from $\widetilde{x}$.
{\normalsize
\begin{align*}
\mlfv{ {x}}  &= [  {x} ] &\mlfv{{x}[i]} &= \mlfv{ {\oneb}}  = \emptyset\\
  \mlfv{C \bagsep U}  &= \mlfv{C} &  \mlfv{M\ B}  &=  \mlfv{M} \uplus \mlfv{B}   \\
     \mlfv{ {\bag{M}}}  &= \mlfv{M} & \mlfv{\lambda x . M} & = \mlfv{M}\!\setminus\! \{  {x} \}\\ 
    \mlfv{M \esubst{B}{x}}  &= (\mlfv{M}\setminus \{  {x} \}) \uplus \mlfv{B} & \mlfv{ {\bag{M}}  \cdot C} &= \mlfv{ {M}} \uplus \mlfv{ C} \\
    \mlfv{\expr{M}+\expr{N}}  &= \mlfv{\expr{M}} \uplus \mlfv{\expr{N}} & \mlfv{\fail^{ {x}_1, \cdots ,  {x}_n}}  &= [  {x}_1, \ldots ,  {x}_n ]
\end{align*}
}
\normalsize
A term $M$ (resp. expression $\expr{M}$) is called  \emph{linearly closed} if $\mlfv{M} = \emptyset$ (resp. $\mlfv{\expr{M}} = \emptyset$).
\end{definition}

\begin{notation}
We shall use the following notations. 
\begin{itemize}
\item 
$N \in \expr{M}$ means that 
$N$ occurs in the sum  $\expr{M}$. 
Also, we write $N_i \in C$ to denote that $N_i$ occurs in the linear bag $C$, and $C \setminus N_i$ to denote the linear bag   obtained by removing one occurrence of  $N_i$ from $C$.
\item  $\#( {x}, M)$ denotes the number of (free) linear occurrences of $x$ in $M$. 
Also, $\#(x,\widetilde{y}) $ denotes the number of occurrences of $x$ in the multiset $\widetilde{y}$. 
\item 
   $\perm{C}$  is the set of all permutations of a linear  bag $C$ and $C_i(n)$ denotes the $n$-th term in the (permuted)  $C_i$.
\item $\size{C}$ denotes the number of terms in a linear bag $C$. 
That is, $\size{\oneb} = 0$
and 
$\size{\bag{M}  \cdot\, C} = 1 + \size{C}$. Given a bag $B = C \bagsep U$, we define  $\size{B}$ as $\size{C}$.
\end{itemize}
\end{notation}

\begin{definition}[Head]
\label{def:headfailure}
Given a term $M$, we define $\headf{M}$ inductively as:
{\small \[  
    \begin{array}{l}
        \begin{aligned}
            \headf{x}  &= x & \headf{M\ B}  &= \headf{M}& \headf{\lambda x.M}  &= \lambda x.M \\
            \headf{x[i]}  &= x[i]  & \headf{\fail^{\widetilde{x}}}  &= \fail^{\widetilde{x}}&  \headf{M \esubst{ B }{x}} &=
        \begin{cases}
            \headf{M} & \text{if $\#(x,M) = \size{B}$}\\
            \fail^{\emptyset} & \text{otherwise}
        \end{cases}
        \end{aligned} 
    \end{array}
\]
}
\end{definition}

\begin{definition}[Head Substitution]
\label{def:linsubfail}
Let $M$ be a term such that $\headf{M}=x$. 
The \emph{head substitution} of a term $N$ for $x$ in $M$, denoted  $M\headlin{ N/x }$,  is inductively  defined as follows (where $ x \not = y$):
\[
\begin{aligned}
&x \headlin{ N / x}   = N 
\hspace{2mm} (M\ B)\headlin{ N/x}  = (M \headlin{ N/x })\ B \quad  (M\ \esubst{B}{y})\headlin{ N/x }  = (M\headlin{ N/x })\ \esubst{B}{y} 
\\[1mm]
\end{aligned}
\]
\end{definition}
\noindent When $\headf{M}=x[i]$, the head substitution   $M\headlin{N/x[i]}$ works as expected: $x[i]\headlin{N/x[i]}=N$ as the base case of the definition. 
Finally, we define contexts for terms and expressions:

\begin{definition}[Evaluation Contexts]\label{def:context_lamrfail}
Contexts for terms (CTerm) and expressions (CExpr) are defined by the following grammar:
\[
\begin{array}{l@{\hspace{1cm}}l}
 \text{(CTerm)}\quad  C[\cdot] ,  C'[\cdot] ::=  ([\cdot])B \mid ([\cdot])\esubst{B}{x} &
 \text{(CExpr)}  \quad   D[\cdot] , D'[\cdot] ::= M + [\cdot] 
\end{array}
\]
\end{definition}

\begin{figure*}[!t]

\begin{mdframed}[style=alttight]
	    \centering
	    \smallskip
	    \small
\begin{prooftree}
        \AxiomC{\ }
        \noLine
        \UnaryInfC{\ }
        \LeftLabel{\redlab{R:Beta}}
        \UnaryInfC{\((\lambda x. M) B \red M\esubst{B}{x}\)}
\end{prooftree}
\begin{prooftree}
\AxiomC{$\headf{M} =  {x}$}
    \AxiomC{$C = {\bag{N_1}}\cdot \dots \cdot {\bag{N_k}} \ , \ k\geq 1 $}
    \AxiomC{$ \#( {x},M) = k $}
    \LeftLabel{\redlab{R:Fetch^{\ell}}}
    \TrinaryInfC{\(
    M\esubst{ C \bagsep U  }{x } \red M \headlin{ N_{1}/ {x} } \esubst{ (C \setminus N_1)\bagsep U}{ x }  + \cdots + M \headlin{ N_{k}/ {x} } \esubst{ (C \setminus N_k)\bagsep U}{x}
    \)}
\end{prooftree}
\begin{prooftree}
    \AxiomC{$\headf{M} = {x}[i]
 \quad \#( {x},M) = \size{C}$}
    \AxiomC{$ U_i = \banged{\bag{N}}$}
    \LeftLabel{\redlab{R:Fetch^!}}
    \BinaryInfC{\(
    M\ \esubst{ C \bagsep U  }{x } \red M \headlin{ N/{x}[i] } \esubst{ C \bagsep U}{ x } 
    \)}
    \end{prooftree}

    \begin{prooftree}
\AxiomC{$\#( {x},M) \neq \size{C} \quad \widetilde{y} = (\mlfv{M} \!\setminus x) \uplus \mlfv{C} $}
    \LeftLabel{\redlab{R:Fail^{ \ell }}}
    \UnaryInfC{\(  M\esubst{C \bagsep U}{x } \red {}\!\!\!\!  \displaystyle\sum_{\perm{C}}\!\!\! \fail^{\widetilde{y}} \)}
\end{prooftree}
\vspace{-0.7cm}

\begin{prooftree}
    \AxiomC{$\#( {x},M) = \size{C}$}
    \AxiomC{$U_i = \banged{\oneb} \quad \headf{M} = {x}[i]  $}
    \LeftLabel{\redlab{R:Fail^!}}
    \BinaryInfC{\(  M \esubst{C \bagsep U}{x } \!\red  \! M \headlin{ \fail^{\emptyset} /{x}[i] } \esubst{ C \bagsep U}{ x }\)}
\end{prooftree}

\begin{prooftree}
   \AxiomC{$\widetilde{y} = \mlfv{C} $}
    \LeftLabel{$\redlab{R:Cons_1}$}
    \UnaryInfC{\(  (\fail^{\widetilde{x}})\ C \bagsep U \red {} \!\!\!\!\! \displaystyle\sum_{\perm{C}} \fail^{\widetilde{x} \uplus \widetilde{y}} \)}
\DisplayProof
\hfill
    \AxiomC{$ \#(z , \widetilde{x}) =  \size{C} \  \widetilde{y} = \mlfv{C} $}
    \LeftLabel{$\redlab{R:Cons_2}$}
    \UnaryInfC{$\fail^{\widetilde{x}}\ \esubst{C \bagsep U}{z}  \red {}\!\!\!\!\!\displaystyle\sum_{\perm{C}} \fail^{(\widetilde{x} \setminus z) \uplus\widetilde{y}}$}
\end{prooftree}
\vspace{-0.7cm}
\begin{prooftree}
        \AxiomC{$ \expr{M}  \red \expr{M}'  $}
        \LeftLabel{\redlab{R:ECont}}
        \UnaryInfC{$D[\expr{M}]  \red D[\expr{M}']  $}
        \DisplayProof
        \qquad 
    \AxiomC{$   M \red   \sum_{i=1}^k M'_{i}  $}
        \LeftLabel{\redlab{R:TCont}}
        \UnaryInfC{$ C[M] \red  \sum_{i=1}^k C[M'_{i}] $}
\end{prooftree}
\end{mdframed}
  \vspace{-5mm}
    \caption{Reduction rules for $\lamrfailunres$.}
    \label{fig:reductions_lamrfailunres}
        \vspace{-3mm}
    \end{figure*}

\smallskip
\noindent
Reduction  is defined by the rules in \figref{fig:reductions_lamrfailunres}.
 Rule~$\redlab{R:Beta}$ induces explicit substitutions.
Resource consumption is implemented by two fetch rules, which open up explicit substitutions:
\begin{itemize}
	\item 
Rule~$\redlab{R:Fetch^{\ell}}$, the \emph{linear fetch}, ensures that the number of required resources matches the size of the linear bag $C$.
It induces a sum of terms with head substitutions, each denoting the partial evaluation of an element from $C$. 
Thus, the size of $C$ determines the summands in the resulting expression.
\item 
 Rule~$\redlab{R:Fetch^!}$, the \emph{unrestricted fetch}, consumes a resource occurring in a specific position of the unrestricted bag $U$ via a linear head substitution of an unrestricted variable occurring in the head of the term. In this case, reduction results in an explicit substitution with $U$ kept unaltered. Note that we check for the size of the linear bag $C$: in the case $\#(x,M)\neq \size{C}$, the term evolves to a linear failure via Rule~$\redlab{R:\fail^\ell}$ (see Example~\ref{ex:sizeC}). This is another design choice: linear failure is prioritised in $\lamrfailunres$.
 \end{itemize}

Four rules show reduction to failure terms, and accumulate free variables involved in failed reductions.
Rules~$\redlab{R:Fail^{\ell}}$ and $\redlab{R:Fail^!}$ formalise the failure to evaluate an explicit substitution $M\esubst{C\bagsep U}{x}$.
The former rule targets a linear failure, which  occurs when the size of $C$ does not match the number of occurrences of ${x}$. The multiset $\widetilde{y}$ preserves all free linear variables in $M$ and $C$. 
The latter rule targets an \emph{unrestricted failure}, which occurs when the head of the term is $x[i]$ and $U_i$ (i.e., the $i$-th element of $U$) is empty. 
In this case, failure preserves the free linear variables in $M$ and $C$ excluding the head unrestricted occurrence $x[i]$ which is replaced by $\fail^\emptyset$.

Rules~$\redlab{R:Cons_1}$ and~$\redlab{R:Cons_2}$ describe reductions that lazily consume the failure term, when a term has $\fail^{\widetilde{x}}$ at its head position. 
The former rule consumes bags attached to it whilst preserving all its free linear variables; the latter rule consumes explicit substitution attached to it whilst also preserving all its free linear variables. The side condition $\#(z,\widetilde{x})=\size{C}$ is necessary in Rule~$\redlab{R:Cons_2}$ to avoid a clash with the premise of Rule~$\redlab{R:Fail^{\ell}}$.
Finally, 
Rules $\redlab{R:ECont}$ and $\redlab{R:TCont}$ state  closure by  the $C$ and $D$ contexts (cf. \defref{def:context_lamrfail}).


    Notice that the left-hand sides of the reduction rules in $\lamrfailunres$  do not interfere with each other. 
As a result, reduction in \lamrfailunres satisfies a \emph{diamond property}: for all $\expr{M}\in \lamrfailunres$, if there exist $\expr{M}_1,\expr{M}_2\in \lamrfailunres$ such that $\expr{M}\red \expr{M}_1$ and  $\expr{M}\red \expr{M}_2$, then there exists $\expr{N}\in \lamrfailunres$ such that $\expr{M}_1\red \expr{N}\longleftarrow \expr{M}_2$
\iftypes 
(see \cite{DBLP:journals/corr/abs-2111} for more details).
\else 
(see \appref{appA}).
\fi 

\begin{notation}
As usual,  
  $\red^*$ denotes  the reflexive-transitive closure of $\red$. 
We write $\expr{N}\red_{\redlab{R}} \expr{M}$ to denote that $\redlab{R}$ is the last (non-contextual) rule used in the step from $\expr{N}$ to $\expr{M}$.
\end{notation}


\begin{example}[Cont. Example~\ref{ex:id_term}] 
\label{ex:id_sem}
We illustrate different reductions for $\lambda x. x$ and $\lambda x. x[i]$.  
\begin{enumerate}
    \item $M_1= (\lambda x. x )(\bag{N}\bagsep U)$ concerns a linear variable $x$ with an linear bag containing one element. This is similar to the usual {meaning} of applying an identity function to a term:
    
\(
\begin{aligned} 
(\lambda x. x )(\bag{N}\bagsep U)&\red_{\redlab{R:Beta}} x\esubst{\bag{N}\bagsep U}{x} \red_{\redlab{R:Fetch^{\ell}}} x\headlin{N/x}\esubst{\oneb\bagsep U}{x}=N\esubst{\oneb\bagsep U}{x},    
 \end{aligned}
\)

with a ``garbage collector'' that collects unused unrestricted resources.
    \item $M_2= (\lambda x. x )(\bag{N_1}\cdot \bag{N_2}\bagsep U)$ concerns the case in which a linear variable $x$ has a single occurrence but the linear bag has size two. Term $M_2$ reduces to a sum of failure terms: 
     \(\begin{aligned} 
    (\lambda x. x )(\bag{N_1}\cdot \bag{N_2}\bagsep U)&\red_{\redlab{R:Beta}} x\esubst{\bag{N_1}\cdot \bag{N_2}\bagsep U}{x}\red_{\redlab{R:Fail^{\ell}}} \sum_{\perm{C}}\fail^{\widetilde{y}}
    \end{aligned}\)
    
    for $C=\bag{N_1}\cdot \bag{N_2}$ and $\widetilde{y}=\mlfv{C}$.
    
    \item $M_3=(\lambda x. x[1] )(\bag{N}\bagsep \oneb^!)$ represents  an abstraction of an unrestricted variable, which aims to consume the first element of the unrestricted bag. Because this bag is empty, $M_3$ reduces to failure:
       
        \(\begin{aligned}
    (\lambda x. x[1])(\bag{N}\bagsep \oneb^!)&\red_{\redlab{R:Beta}}x[1]\esubst{\bag{N}\bagsep \oneb^!}{x}\red_{\redlab{R:\fail^{\ell}}}\fail^{\widetilde{y}},
    \end{aligned}\)
    
        for   $\widetilde{y}=\mlfv{N}$. Notice that $0= \#(x,x[1])\neq \size{\bag{N}}=1$, since there are no linear occurrences of  $x$ in $x[1]$.

\end{enumerate}

\end{example}

\begin{example}
\label{ex:sizeC}
To illustrate the need to check `$\size{C}$' in $\redlab{R:Fail^!}$, consider the term $ x[1] \esubst{\bag{M} \bagsep \oneb^!}{x}$, which features both a mismatch of linear bags for the linear variables to be substituted and an empty unrestricted bag with the need for the first element to be substituted. 
We check the size of the linear bag because we wish to prioritise the reduction of Rule~$\redlab{R:Fail^{\ell}}$. Hence, in case of a mismatch of linear resources  we wish not to  perform a reduction via Rule~$\redlab{R:Fail^!}$.
 This is a design choice: our semantics collapses linear failure at the earliest moment it arises.
\end{example}



    
    \begin{example}[Cont. Example~\ref{ex:delta}]
    \label{ex:deltasem}
    Self-applications of $\Delta_1$ do not behave as an expected variation of a lazy reduction from $\Omega$. Both
    $ \Delta_1(\bag{\Delta_1}\bagsep \oneb^!)$ and  $ \Delta_1(\oneb \bagsep \bag{\Delta_1}^!)$ reduce to failure  since the number of linear occurrences of $x$ does not match the number of  resources in the linear bag:
    \(
    \begin{aligned}
     \Delta_1(\bag{\Delta_1}\bagsep \oneb^!) 
     \red(x(\bag{x}\bagsep \oneb))\esubst{\bag{\Delta_1}\bagsep \oneb^!}{x} \red\fail^\emptyset.
     \end{aligned}
     \)
     
     The term $\Delta_4(\oneb \bagsep \bag{\Delta_4}^!)$ also fails: the linear bag is empty and there is one linear occurrence of $x$ in $\Delta_4$.  Note that $\Delta_4(\bag{\Delta_4}\bagsep \bag{\Delta_4}^!)$ reduces to another application of $\Delta_4$ before failing:
     {
     \[
     \begin{aligned}
     \Delta_4(\bag{\Delta_4}\bagsep \bag{\Delta_4}^!) &= (\lambda x. (x[1] (\bag{x}\bagsep \oneb^!)))(\bag{\Delta_4}\bagsep \bag{\Delta_4}^!)\\
     &\red_{\redlab{R{:}Beta}} (x[1] (\bag{x}\bagsep \oneb^!))\esubst{\bag{\Delta_4}\bagsep \bag{\Delta_4}^!}{x}\\
     &\red_{\redlab{R{:}Fetch^!}} (\Delta_4 (\bag{x}\bagsep \oneb^!))\esubst{\bag{\Delta_4}\bagsep \bag{\Delta_4}^!}{x}\\
     &\red^*    \fail^\emptyset \esubst{\bag{x}\bagsep \oneb^!}{y}\esubst{\bag{\Delta_4}\bagsep \bag{\Delta_4}^!}{x}
     \end{aligned}
     \]
     }
     \end{example} 
\noindent Differently from~\cite{DBLP:conf/fscd/PaulusN021}, there are terms in $\lamrfailunres$ that when applied to each other behave similarly to $\Omega$, namely $\Omega_{5,6}$, $\Omega_{6,5}$, and $\Omega_7$ (Example~\ref{ex:delta}).

\begin{example}[Cont. Example~\ref{ex:delta}]  

\label{ex:deltasemii}
The following reductions illustrate different behaviours  provided that subtle changes are made within $\lamrfailunres$-terms:
\begin{itemize}

\item An interesting behaviour of $\lamrfailunres$ is that variations of $\Delta$ can be applied to each other and appear alternately (highlighted in \textcolor{blue}{blue}) in the functional position throughout the computation---this behaviour is illustrated in \figref{fig:diag_delta56}:

  \(
  \begin{aligned}
  \Omega_{5,6} &= \Delta_5 ( \oneb\bagsep \banged{\bag{\Delta_5}} \concat \banged{\bag{\Delta_6}} )\\
 & =(\lambda x. ( x[2] ( 1\bagsep \banged{\bag{x[1]}} \concat \banged{\bag{x[2]}})))
  ( 1\bagsep \banged{\bag{\Delta_5}} \concat \banged{\bag{\Delta_6}})\\ 
&  \red_\redlab{R{:}Beta} (x[2] ( \oneb\bagsep \banged{\bag{x[1]}} \concat \banged{\bag{x[2]}})) \esubst{\oneb\bagsep \banged{\bag{\Delta_5}} \concat \banged{\bag{\Delta_6}} }{ x}  \\
    &\red_\redlab{R:Fetch^!} (\colthree{\Delta_6} ( {\oneb\bagsep \banged{\bag{x[1]}} \concat \banged{\bag{x[2]}}})) \esubst{ \oneb\bagsep \banged{\bag{\Delta_5}} \concat \banged{\bag{\Delta_6}} }{ x} \\ 
     &\red_\redlab{R{:}Beta} (y[1](\oneb\bagsep \bag{y[1]}^!\concat \bag{y[2]}^!) \esubst{(\oneb\bagsep \banged{\bag{x[1]}} \concat \banged{\bag{x[2]}})}{y} \esubst{\oneb\bagsep \banged{\bag{\Delta_5}} \concat \banged{\bag{\Delta_6}} }{ x }\\ 
    &\red_\redlab{R{:}Fetch^!} (x[1](1\bagsep \bag{y[1]}^!\concat \bag{y[2]}^!) \esubst{(\oneb\bagsep \banged{\bag{x[1]}} \concat \banged{\bag{x[2]}})}{y} \esubst{\oneb\bagsep \banged{\bag{\Delta_5}} \concat \banged{\bag{\Delta_6}}}{ x}   \\
     &\red_\redlab{R{:}Fetch^!} (\colthree{\Delta_5}(1\bagsep \bag{y[1]}^!\concat \bag{y[2]}^!) \esubst{(\oneb\bagsep \banged{\bag{x[1]}} \concat \banged{\bag{x[2]}})}{y} \esubst{\oneb\bagsep \banged{\bag{\Delta_5}} \concat \banged{\bag{\Delta_6}}}{ x}   \\
&\red \ldots
        \end{aligned}
        \)
 \begin{figure}[!t]
\begin{mdframed}

\begin{center}
\begin{tikzpicture}[scale=0.8]
\draw[->] (2,1.3) --node[anchor= south,xshift=1.75mm,yshift=-0.5mm] {*}  (4.5,2.7); 
\draw[->] (4.75,2.7) --node[anchor= south,xshift=1.75mm,yshift=-1.75mm] {*} (7.75,1.3); 
\node[right] at (-1,1){$\Delta_{5}( \oneb\bagsep \banged{\bag{\Delta_5}} \concat \banged{\bag{\Delta_6}})$};
\node at (7.0,1){$\Delta_{5}\ldots \esubst{\oneb\bagsep \banged{\bag{\Delta_5}} \concat \banged{\bag{\Delta_6}}}{x}$};
\draw[->]  (8.4,1.3) --node[anchor= south,xshift=1.75mm,yshift=-0.5mm] {*} (11.5,2.7); 
\node at (4,3){$\Delta_{6}\ldots  \esubst{\oneb\bagsep \banged{\bag{\Delta_5}} \concat \banged{\bag{\Delta_6}}}{x}$};
\node at (11,3){$\Delta_{6}\ldots \esubst{\oneb\bagsep \banged{\bag{\Delta_5}} \concat \banged{\bag{\Delta_6}}}{x}$};
\draw[->]  (11.7,2.7)--node[anchor= south,xshift=1.75mm,yshift=-1.75mm] {*} (15,1.3); 
\node at (15,1) {\ldots };
\end{tikzpicture}
\end{center}
\end{mdframed}
\caption{An $\Omega$-like behaviour in $\lamrfailunres$ (cf. Example~\ref{ex:deltasemii}).}\label{fig:diag_delta56}
\end{figure}       
  \item Applications of $\Delta_7$ into two unrestricted copies of $\Delta_7$
  behave as $\Omega$ producing a non-terminating behaviour.
  Letting $B = 1\bagsep \bag{x[1]}^!\concat \bag{x[1]}^!$, we have:
  
\(
\begin{aligned} 
&\Omega_7=(\lambda x. (x[2](1\bagsep \bag{x[1]}^!\concat \bag{x[1]}^!)))(1\bagsep \bag{\Delta_7}^!\concat \bag{\Delta_7}^!)\\
&\red_{\redlab{R{:}Beta}} (x[2](1\bagsep \bag{x[1]}^!\concat \bag{x[1]}^!))\esubst{1\bagsep \bag{\Delta_7}^!\concat \bag{\Delta_7}^!}{x}\\
&\red_{\redlab{R{:}Fetch^!}} (\Delta_7({1\bagsep \bag{x[1]}^!\concat \bag{x[1]}^!}))\esubst{1\bagsep \bag{\Delta_7}^!\concat \bag{\Delta_7}^!}{x}\\
&\red_{\redlab{R{:}Beta}} (y[2](1\bagsep \bag{y[1]}^!\concat \bag{y[1]}^!))\esubst{B}{y}\esubst{1\bagsep \bag{\Delta_7}^!\concat \bag{\Delta_7}^!}{x}\\
&\red_{\redlab{R{:}Fetch^!}} (x[1](1\bagsep \bag{y[1]}^!\concat \bag{y[1]}^!))\esubst{B}{y}\esubst{1\bagsep \bag{\Delta_7}^!\concat \bag{\Delta_7}^!}{x}\\
&\red_{\redlab{R{:}Fetch^!}} (\Delta_7(1\bagsep \bag{y[1]}^!\concat \bag{y[1]}^!))\esubst{B}{y}\esubst{1\bagsep \bag{\Delta_7}^!\concat \bag{\Delta_7}^!}{x}  \\
&\red \ldots
\end{aligned}
\) 

Later on we will show that this term is well-formed (see Example~\ref{ex:delta7_wf1}) with respect to the intersection type system introduced in \secref{sec:lam_types}.
\end{itemize}
\end{example}


\section[Intersection Types]{Well-Formed Expressions via Intersection Types}\label{sec:lam_types}

We define  \emph{well-formed} $\lamrfailunres$-expressions by relying on a non-idempotent intersection type system,  based on the system by Bucciarelli et al.~\cite{BucciarelliKV17}.
Our system for well-formed expressions  subsumes the one in \cite{DBLP:conf/fscd/PaulusN021}: it uses {\em strict} and {\em multiset} types to check  linear bags;  moreover, it uses  {\em list} and {\em tuple} types to check unrestricted bags. 
As in \cite{DBLP:conf/fscd/PaulusN021}, we write
``well-formedness'' (of terms, bags, and expressions) to stress that, unlike usual type systems, our system can account for terms that may reduce to the failure term (cf. Remark~\ref{r:types}).

\begin{definition}[Types for \lamrfailunres]
\label{d:typeslamrfail}
We define {\em strict}, {\em multiset}, \emph{list}, and \emph{tuple types}.
\[
\begin{array}{l@{\hspace{.8cm}}l}
\begin{array}{l@{\hspace{.2cm}}rcl}
  \text{(Strict)} & \quad  \sigma, \tau, \delta &::=& \unit \sep \arrt{( \pi , \eta)}{\sigma}   
  \\
  \text{(Multiset)} & \quad \pi, \zeta &::=& \bigwedge_{i \in I} \sigma_i \sep \omega
\end{array}
&
\begin{array}{l@{\hspace{.2cm}}rcl}
  \text{(List)}& \quad \eta, \epsilon  &::=&  \sigma  \sep \epsilon \concat \eta 
  \\
  \text{(Tuple)}&\quad   ( \pi , \eta)
\end{array}
\end{array}
\]
\end{definition}

A strict type can be the   $\unit$ type  or a functional type  \arrt{( \pi , \eta)}{\sigma}, where 
$( \pi , \eta)$ is a tuple type and $\sigma$ is  a strict type. Multiset types can be 
either the empty type $\omega$ or 
an intersection  of strict types $\bigwedge_{i \in I} \sigma_i$, with $I$ non-empty. The operator $\wedge $ is commutative, associative, non-idempotent, that is, $\sigma \wedge \sigma \neq \sigma$, with identity $\omega$. The intersection type $\bigwedge_{i\in I} \sigma_i$ is the type of a linear bag; the cardinality of $I$ corresponds to its size. 

A list type can be either an strict type $\sigma$ or the composition $\epsilon\concat \eta$ of two list types $\epsilon$ and $\eta$. We use the list type  $\epsilon\concat \eta$ to type the concatenation of two unrestricted bags.
A tuple type $(\pi,\eta)$ types the concatenation of a linear bag of type $\pi$ with an unrestricted bag of type $\eta$.
Notice that a list type $\epsilon \concat \eta$ can be recursively unfolded into a finite composition of strict types $\sigma_1\concat   \ldots \concat \sigma_n$, for some $n\geq 1$. In this case the length of $\epsilon\concat \eta$ is $n$ and that $\sigma_i$ is its $i$-th strict type, for $1\leq i \leq n$.



\begin{notation}
\label{not:lamrrepeat}
 Given $k \geq 0$, we  write $\sigma^k$ to stand for $\sigma \wedge \cdots \wedge \sigma$
($k$ times, if $k>0$) or for $\omega$ (if $k=0$).
Similarly,  $\hat{x}:\sigma^k$
stands for $x:\sigma , \cdots , x:\sigma$ ($k$ times, if $k>0$)
 or for $x:\omega$  (if $k=0$).  Given $k \geq 1$, we  write $\banged{x}:\eta   $ to stand for $x[1]:\eta_1 , \cdots , x[k]:\eta_k$.
\end{notation}

\begin{notation}[$\eta \relunbag \epsilon$]
\label{not:ltypes}
 Let $\epsilon$ and $\eta$ be two list types, with the length of $\epsilon$ greater or equal to that of $\eta$. 
Let us write $\epsilon_i$ and $\eta_i$ to denote the $i$-th strict type in $\epsilon$ and $\eta$, respectively.
We write $\eta \relunbag \epsilon$ meaning the initial sublist, whenever there exist $  \epsilon'$ and $ \epsilon''$   such that: i) $ \epsilon = \epsilon' \concat \epsilon''$; ii)  the size of $\epsilon' $ is that of $\eta$; iii)  for all $i$, $\epsilon'_i = \eta_i $.
\end{notation}



\noindent {\em Linear contexts} range over  $\Gamma , \Delta, \ldots $ and {\em  unrestricted contexts} range over $\Theta, \Upsilon, \ldots$. They are defined by the following grammar:
$$
\begin{array}{c@{\hspace{.9cm}}c}
   \Gamma, \Delta ::= \dash ~|~  x:\sigma ~\mid ~\Gamma, x:\sigma & \Theta, \Upsilon ::= \dash ~|~  x^!:\eta ~|~ \Theta,x^!:\eta
\end{array}
$$

The empty linear/unrestricted type assignment is  denoted `$\dash$'.
Linear variables can occur more than once in a linear context; they are assigned only strict types. 
For instance,  $x:(\tau , \sigma) \rightarrow \tau , x: \tau$  is a valid context: it means that $x$ can be of both type $(\tau , \sigma) \rightarrow \tau$ and $\tau$. 
 In contrast, unrestricted variables can occur at most once  in unrestricted contexts; they are assigned only list types.
The multiset of linear variables in $\Gamma$ is denoted as $\dom{\Gamma}$; similarly, $\dom{\Theta}$ denotes the set of unrestricted variables in $\Theta$. 

Judgements are of the form $\Theta;\Gamma \wfdash \expr{M}:\sigma$, where the left-hand side contexts are separated by ``;'' and  $\expr{M}:\sigma$ means that    $\expr{M}$ has type $\sigma$. We write $\wfdash \expr{M}:\sigma$ to denote $\dash \,;\, \dash \wfdash \expr{M}:\sigma$.

\begin{figure*}[!t]
\small
\begin{mdframed}[style=alttight]
	    \centering    
    \begin{prooftree}
        \AxiomC{}
        \LeftLabel{\redlab{F{:}var^{ \ell}}}
        \UnaryInfC{\( \Theta ;  {x}: \sigma \wfdash  {x} : \sigma\)}
        \DisplayProof
\hfill
        \AxiomC{\( \Theta , \banged{x}: \eta; {x}: \eta_i , \Delta \wfdash  {x} : \sigma\)}
        \LeftLabel{\redlab{F{:}var^!}}
        \UnaryInfC{\( \Theta , \banged{x}: \eta; \Delta \wfdash {x}[i] : \sigma\)}
        \DisplayProof
        \hfill
        \AxiomC{\(  \)}
        \LeftLabel{\redlab{F{:}\oneb^{\ell}}}
        \UnaryInfC{\( \Theta ; \dash \wfdash \oneb : \omega \)}
    \end{prooftree}
  \vspace{-5mm}
  \begin{prooftree}
        \AxiomC{\(  \)}
        \LeftLabel{\redlab{F{:}\oneb^!}}
        \UnaryInfC{\( \Theta ;  \dash  \wfdash \banged{\oneb} : \sigma \)}
          \DisplayProof
    \hfill
  \AxiomC{\(\Theta , \banged{z} : \eta ; \Gamma ,  {\hat{z}}: \sigma^{k} \wfdash M : \tau \quad z\notin \dom{\Gamma} \)}
        \LeftLabel{\redlab{F{:}abs}}
        \UnaryInfC{\( \Theta ;\Gamma \wfdash \lambda z . M : (\sigma^{k} , \eta )   \rightarrow \tau \)}
\end{prooftree}
  \vspace{-6mm}
\begin{prooftree}
     \AxiomC{\quad \( \quad \Theta ;\Gamma \wfdash M : (\sigma^{j} , \eta ) \rightarrow \tau \)}
\noLine
         \UnaryInfC{\(  \Theta ;\Delta \wfdash B : (\sigma^{k} , \epsilon )  \)}
      \AxiomC{\( \eta \relunbag \epsilon \)}
            \LeftLabel{\redlab{F{:}app}}
        \BinaryInfC{\( \Theta ; \Gamma, \Delta \wfdash M\ B : \tau\)}
 \DisplayProof
 \hfill
 \AxiomC{\qquad \( ~\quad \Theta ,\banged{x} : \eta ; \Gamma ,  {\hat{x}}: \sigma^{j} \wfdash M : \tau  \)}
    \noLine
   \UnaryInfC{\( \Theta ; \Delta \wfdash B : (\sigma^{k} , \epsilon ) \)}
     \AxiomC{\( \eta \relunbag \epsilon \)}
        \LeftLabel{\redlab{F{:}ex \dash sub}}  
        \BinaryInfC{\( \Theta ; \Gamma, \Delta \wfdash M \esubst{ B }{ x } : \tau \)}
   \end{prooftree}
  \vspace{-5mm}
    \begin{prooftree}
  \AxiomC{\( \Theta ; \Gamma\wfdash C : \sigma^k\)}
    \AxiomC{\(  \Theta ;\dash \wfdash  U : \eta \)}
        \LeftLabel{\redlab{F{:}bag}}
        \BinaryInfC{\( \Theta ; \Gamma \wfdash C \bagsep U : (\sigma^{k} , \eta  ) \)}
\DisplayProof
\hfill
   \AxiomC{\( \Theta ; \Gamma \wfdash M : \sigma\)}
        \AxiomC{\( \Theta ; \Delta \wfdash C : \sigma^k\)}
        \LeftLabel{\redlab{F{:}bag^{ \ell}}}
        \BinaryInfC{\( \Theta ; \Gamma, \Delta \wfdash \bag{M}\cdot C:\sigma^{k+1}\)}
    \end{prooftree}
  \vspace{-6mm}
    \begin{prooftree}
        \AxiomC{\( \Theta ; \dash \wfdash M : \sigma\)}
        \LeftLabel{\redlab{F{:}bag^!}}
        \UnaryInfC{\( \Theta ; \dash  \wfdash \banged{\bag{M}}:\sigma \)}
   \DisplayProof
   \hfill
        \AxiomC{\( \Theta ; \dash \wfdash U : \epsilon\)}
        \AxiomC{\( \Theta ; \dash \wfdash V : \eta\)}
        \LeftLabel{\redlab{F{:}\concat bag^!}}
        \BinaryInfC{\( \Theta ; \dash  \wfdash U \concat V :\epsilon \concat \eta \)}
    \DisplayProof
    \hfill
  \AxiomC{\( \dom{\Gamma} = \widetilde{x} \)}
        \LeftLabel{\redlab{F{:}fail}}
        \UnaryInfC{\( \Theta ; \Gamma \wfdash  \fail^{\widetilde{x}} : \tau  \)}
    \end{prooftree}
     \vspace{-6mm}
    \begin{prooftree}
         \AxiomC{$ \Theta ;\Gamma \wfdash \expr{M} : \sigma \qquad \Theta ;\Gamma \wfdash \expr{N} : \sigma$}
        \LeftLabel{\redlab{F{:}sum}}
        \UnaryInfC{$ \Theta ;\Gamma \wfdash \expr{M}+\expr{N}: \sigma$}
          \DisplayProof
        \qquad 
                \AxiomC{$ \Theta; \Gamma \wfdash M: \sigma \quad  x \not \in \dom{\Gamma}$}
        \LeftLabel{\redlab{F{:}weak}}
        \UnaryInfC{$ \Theta ; \Gamma, x:\omega \wfdash M: \sigma $}
    \end{prooftree}
\end{mdframed}
\vspace{-5mm}
    \caption{Well-formedness rules for \lamrfailunres (cf. Def.~\ref{d:wellf_unres}). In Rules~\redlab{F{:}app} and \redlab{F{:}ex\dash sub}: $k, j \geq 0$.}\label{fig:wf_rules_unres}
\vspace{-2mm}
\end{figure*}

\begin{definition}[Well-formed \lamrfailunres expressions]
\label{d:wellf_unres}
An expression $ \expr{M}$ is well-formed (wf, for short) if  there exist  $\Gamma$, $\Theta$ and  $\tau$ such that  $ \Theta ; \Gamma \wfdash  \expr{M} : \tau  $ is entailed via the rules in \figref{fig:wf_rules_unres}.
\end{definition}


\noindent We describe the well-formedness rules in \figref{fig:wf_rules_unres}. 
\begin{itemize}
\item Rules~$\redlab{F:var^{ \ell}}$ and $\redlab{F:var^!}$ assign types to linear and unrestricted variables, respectively. 
\item Rule~$\redlab{F:var^!}$ resembles the {\em copy} rule~\cite{CairesP10} where we use a linear copy of an unrestricted variable $x[i]$ of type $\sigma$, typed with $x^!:\eta$, and type the linear copy with the corresponding strict type $\eta_i$ which in this case the linear copy $x$ would have type equal to $\sigma$. 
\item Rules~$\redlab{F:\oneb^{ \ell}}$ and $\redlab{F:\oneb^!}$ assign types to the empty linear/unrestricted bag: $\oneb$ has type $\omega$, whereas  $\oneb^!$ has an arbitrary strict type $\sigma$. Arbitrariness is allowed since  the substitution of an unrestricted variable for $1^!$ leads to a \(\fail\) term (Rule \redlab{R:Fail^!}), which has an arbitrary strict type.

\item Rule $\redlab{F:abs}$ assigns type $(\sigma^k,\eta)\to \tau $ to an abstraction $\lambda z. M$, provided that the unrestricted occurrences of $z$ may be typed by the unrestricted context containing $z^!:\eta$, the linear  occurrences of $z$ are typed with the linear context containing $\hat{z}: \sigma^k$, for some $k\geq 0$, and there are no other linear occurrences of $z$ in the linear context $\Gamma$.
\item Rules $\redlab{F{:}app}$ and  $\redlab{F{:}ex \dash sub}$ (for application  and explicit substitution, resp.) use the condition $\eta \relunbag \epsilon $ (cf. Notation~\ref{not:ltypes}), which 
captures the portion of the unrestricted bag that is effectively used in a term: it 
ensures that $\epsilon$ can be decomposed into some $\epsilon'$ and $\epsilon''$, such that each  type component $\epsilon_i'$ matches with $\eta_i$. If this requirement is satisfied,  Rule~$\redlab{F{:}app}$ types an application $M\ B$  given that $M$ has a functional type in which the left of the arrow is  a  tuple type $ (\sigma^{j} , \eta ) $ whereas the bag $ B$ is typed with tuple  $ (\sigma^{k} , \epsilon ) $. Similarly,  Rule~$\redlab{F{:}ex \dash sub}$ types the term $ M \esubst{ B }{ x }$ provided that  $ B$ has the tuple type $ (\sigma^{k} , \epsilon ) $ and $M$ is typed with the variable $x$ having linear type assignment $ \sigma^{j} $ and unrestricted type assignment $ \eta $.
 
\end{itemize}

\begin{remark}
 Differently from  intersection type systems~\cite{BucciarelliKV17, PaganiR10},  in Rules  $\redlab{F{:}app}$ and   $\redlab{F{:}ex \dash sub}$ there is no equality requirement  between $j$ and $k$,  as we would like to capture terms that  fail due to a mismatch in resources: we only require that the linear part of the tuples are composed of the same strict type, say  $\sigma$. 
As a term can take an unrestricted bag with arbitrary size we only require that the elements of the unrestricted bag that are used have a ``consistent'' type, i.e., the type of the unrestricted bag satisfies the relation $\relunbag$ with the unrestricted fragment of the corresponding tuple type.
\end{remark}

\noindent There are four rules for bags:
\begin{itemize}
\item Rule~$\redlab{F:bag^!}$ types an unrestricted bag $\bag{M}^!$ with  the  type $\sigma$ of $M$. Note that $\bag{x}^!$, an unrestricted bag containing a linear variable $x$, is not well-formed, whereas $\bag{x[i]}^!$ is well-formed. 
\item Rule $\redlab{F:bag}$ assigns the tuple type $(\sigma^k,\eta)$ to the concatenation of a linear bag of type $\sigma^k$ with an unrestricted bag of type $\eta$. \item Rules $\redlab{F:bag^{ \ell}}$
and $\redlab{F{:}\concat bag^!}$ type the concatenation of linear and unrestricted bags. 
\item Rule~$ \redlab{F{:}\oneb^!} $ allows an empty unrestricted bag to have an arbitrary $\sigma$ type since it may be referred to by a variable for substitution: we must be able to compare its type with the type of unrestricted variables that may  consume the empty bag (this reduction would inevitably lead to failure).
\end{itemize}
As in~\cite{DBLP:conf/fscd/PaulusN021}, Rule  $\redlab{F{:}\fail}$ handles the failure term, and is the main difference with respect to standard type systems. Rules for sums and weakening ($\redlab{F:sum}$ and $\redlab{F:weak}$) are standard.

\begin{example}[Cont. Example~\ref{ex:deltasemii}]\label{ex:delta7_wf1} Term $\Delta_7:=\lambda x. x[2]( \oneb \bagsep \banged{\bag{x[1]}} \concat \banged{\bag{x[1]}} ) $ is well-formed, as ensured by the judgement $ \Theta ; \dash \wfdash \Delta_7: (\omega ,  \sigma' \concat (\sigma^{j} , \sigma' \concat   \sigma' ) \rightarrow \tau )    \rightarrow \tau  $, whose derivation is given below:
        
\begin{itemize}
\item $\Pi_3$ is the derivation of 
\(  \Theta , \banged{x} : \eta ;\dash \wfdash  \banged{\bag{x[1]}} : \sigma',\)  for $ \eta = \sigma' \concat (\sigma^{j} , \sigma' \concat   \sigma' ) \rightarrow \tau $.
\item $\Pi_4$ is the derivation:
\( \Theta , \banged{x}: \eta ; \dash \wfdash {x}[2] : (\sigma^{j} , \sigma' \concat   \sigma' ) \rightarrow \tau \)
\item $\Pi_5$ is the derivation: \( \Theta , \banged{x} :  \eta   ; x : \omega \wfdash ( \oneb \bagsep \banged{\bag{x[1]}} \concat \banged{\bag{x[1]}} )  : (\omega , \sigma' \concat   \sigma' ) \)
\end{itemize}
Therefore, 
 \begin{prooftree}
    \AxiomC{\( \Pi_5\)}
    \AxiomC{\( \Pi_4 \)}
       \AxiomC{\( \sigma' \concat   \sigma' \relunbag \sigma' \concat   \sigma' \)}
    \LeftLabel{\redlab{F{:}app}}
            \TrinaryInfC{\( \Theta , \banged{x} :  \eta  ;  x : \omega \wfdash x[2]( \oneb \bagsep \banged{\bag{x[1]}} \concat \banged{\bag{x[1]}} ) : \tau \)}
           \LeftLabel{\redlab{F{:}abs}}
            \UnaryInfC{\( \Theta ; \dash \wfdash \underbrace{\lambda x. (x[2]( \oneb \bagsep \banged{\bag{x[1]}} \concat \banged{\bag{x[1]}} ))}_{\Delta_7}: ( \omega ,  \eta )    \rightarrow \tau  \)}
        \end{prooftree}
\end{example}

Well-formed expressions satisfy subject reduction (SR);  
\iftypes 
see \cite{DBLP:journals/corr/abs-2111} for a full proof.
\else 
see \appref{appB} for a proof.
\fi 

 \begin{theorem}[SR in \lamrfailunres]
\label{t:lamrfailsr_unres}
If $\Theta ; \Gamma \wfdash \expr{M}:\tau$ and $\expr{M} \red \expr{M}'$ then $\Theta ;\Gamma \wfdash \expr{M}' :\tau$.
\end{theorem}
\begin{proof}By structural induction on the reduction rules. 
We proceed by analysing the rule applied in $\expr{M}$. 
An interesting case occurs when the rule is $\redlab{F:Fetch^!}$:  Then $ \expr{M} = M\ \esubst{ C \bagsep U }{x }$, where $U = \banged{\bag{N_1}} \concat \cdots \concat \banged{\bag{N_l}} $ and $\headf{M} =  {x}[i]$. The reduction is as follows:

    \begin{prooftree}
        \AxiomC{$\headf{M} =  {x}[i]$}
        \AxiomC{$ U_i = \banged{\bag{N_i}} $}
        \LeftLabel{\redlab{R:Fetch^!}}
        \BinaryInfC{\(
        M\ \esubst{ C \bagsep U  }{x } \red M \headlin{ N_i/ {x}[i] } \esubst{ C \bagsep U}{ x } 
        \)}
    \end{prooftree}
    
     By hypothesis, one has the derivation:
     \begin{prooftree}
            \AxiomC{\( \Theta , \banged{x} : \eta ; \Gamma' , \hat{x}: \sigma^{j} \wfdash M : \tau \)}
               \AxiomC{$\Pi$}
               \noLine
               \UnaryInfC{\(  \Theta ; \cdot\wfdash U : \epsilon\)}
                \AxiomC{\( \Theta ; \Delta \wfdash  {C} : \sigma^k\)}
            \LeftLabel{\redlab{F{:}bag}}
            \BinaryInfC{\( \Theta ; \Delta \wfdash C \bagsep U : (\sigma^{k} , \epsilon ) \)}
            \AxiomC{\(  \eta \relunbag \epsilon  \)}
        \LeftLabel{\redlab{F{:}ex \dash sub}}  
        \TrinaryInfC{\( \Theta ; \Gamma', \Delta \wfdash M \esubst{ C \bagsep U }{ x } : \tau \)}
    \end{prooftree}
    where $\Pi$ has the form
    \begin{prooftree}
            \AxiomC{\( \Theta ; \cdot \wfdash N_1 : \epsilon_1\)}
            \LeftLabel{\redlab{F{:}bag^!}}
            \UnaryInfC{\( \Theta ; \cdot \wfdash \banged{\bag{N_1}}  : \epsilon_1\)}
            \AxiomC{\( \cdots \)}
            \AxiomC{\( \Theta ; \cdot \wfdash N_l : \epsilon_l \)}
            \LeftLabel{\redlab{F{:}bag^!}}
            \UnaryInfC{\( \Theta ; \cdot \wfdash \banged{\bag{N_l}} : \epsilon_l \)}
        \LeftLabel{\redlab{F{:}\concat bag^!}}
        \TrinaryInfC{\( \Theta ; \cdot  \wfdash \banged{\bag{N_1}} \concat \cdots \concat \banged{\bag{N_l}}  :\epsilon \)}
    \end{prooftree}
    with $\Gamma = \Gamma' , \Delta $. Notice that if $\epsilon_i = \delta$ and $\eta \relunbag \epsilon$ then $\eta_i = \delta$.
    \iftypes 
It can be shown that there 
\else 
     By Lemma~\ref{lem:subt_lem_failunres_un}, there 
\fi 
exists a derivation
 $\Pi_1$ of  $
    \Theta , \banged{x} : \eta ; \Gamma' , \hat{x}: \sigma^{j} \wfdash   M \headlin{ N_{i}/  {x}[i] } : \tau $. Therefore, we have: 
         \begin{prooftree}
            \AxiomC{\( \Theta , \banged{x} : \eta ; \Gamma' , \hat{x}: \sigma^{j} \wfdash M \headlin{ N_{1}/  {x}[i] } : \tau \)}
               \AxiomC{\(  \Theta ; \cdot\wfdash U : \epsilon\)}
                \AxiomC{\( \Theta ; \Delta \wfdash  {C} : \sigma^k\)}
            \LeftLabel{\redlab{F{:}bag}}
            \BinaryInfC{\( \Theta ; \Delta \wfdash C \bagsep U : (\sigma^{k} , \epsilon ) \)}
            \AxiomC{\(  \eta \relunbag \epsilon  \)}
        \LeftLabel{\redlab{F{:}ex \dash sub}}  
        \TrinaryInfC{\( \Theta ; \Gamma', \Delta \wfdash M \headlin{ N_{i}/ {x[i]} } \esubst{  C \bagsep U }{ x } : \tau \)}
    \end{prooftree}

\end{proof}

\begin{remark}[Well-Formed vs Well-Typed Expressions]
\label{r:types}
Our type system (and Theorem \ref{t:lamrfailsr_unres}) can be specialised to the case of \emph{well-typed} expressions that do not contain (and never reduce to) the failure term.
In particular, Rules  $\redlab{F{:}app}$ and   $\redlab{F{:}ex \dash sub}$ would need to check that $\sigma^k = \sigma^j$, as failure can be caused due to a mismatch of linear resources. {A difference between well typed and well formed expressions is that the former satisfy subject expansion, but the latter do not: expressions that lead to failure can be ill-typed yet failure itself is well-formed.}

\end{remark}



\section[A Translation into Processes]{A Typed Encoding of \lamrfailunres into Concurrent Processes}
\label{s:encoding}
We encode \lamrfailunres into \spi,
a session $\pi$-calculus that stands on a Curry-Howard correspondence between linear logic and session types (\secref{s:pi}). We 
extend the two-step approach  that we devised in~\cite{DBLP:conf/fscd/PaulusN021} for the sub-calculus \lamrfail  (with linear resources only) (cf.~\figref{fig:two-step-enc}).
    First, in \secref{ssec:first_enc}, we define an encoding $\recencodopenf{\cdot}$ from well-formed expressions in \lamrfailunres to well-formed expressions in  a variant of \lamrfailunres with \emph{sharing}, dubbed \lamrsharfailunres (\secref{ssec:lamshar}).
    Then, in \secref{ssec:second_enc}, we define an encoding $\piencodf{\cdot}_u$ (for a name $u$) from well-formed expressions in \lamrsharfailunres to well-typed processes in \spi.

We prove that  $\recencodopenf{\cdot}$ and $\piencodf{\cdot}_u$ satisfy well-established correctness criteria~\cite{DBLP:journals/iandc/Gorla10,DBLP:journals/iandc/KouzapasPY19}: \emph{type preservation}, \emph{operational  completeness},  \emph{operational soundness}, and \emph{success sensitiveness} 
\iftypes 
(cf.~\cite{DBLP:journals/corr/abs-2111}).
\else 
(cf.~\appref{ss:criteria}).
\fi 
Because \lamrfailunres includes unrestricted resources, the results given here strictly  generalise those in~\cite{DBLP:conf/fscd/PaulusN021}.

\subsection[Session-Typed Calculus]{\spi: A Session-Typed $\pi$-Calculus}
\label{s:pi}
\spi is a $\pi$-calculus with \emph{session types}~\cite{DBLP:conf/concur/Honda93,DBLP:conf/esop/HondaVK98}, which  ensure that the endpoints of a channel perform matching actions.
We consider the full process  framework in~\cite{CairesP17}, including constructs for specifying labelled choices and client/server connections;   they will be useful to codify   unrestricted resources and variables in \lamrfailunres.
Following~\cite{CairesP10,DBLP:conf/icfp/Wadler12},
\spi stands on a Curry-Howard correspondence between session types and a linear logic with dual modalities/types ($\with A$ and $\oplus A$),  which define \emph{non-deterministic} session behaviour.
As in~\cite{CairesP10,DBLP:conf/icfp/Wadler12},
in \spi, cut elimination corresponds to communication, proofs to processes, and propositions to session types.

 \tikzstyle{mynode1} = [rectangle, rounded corners, minimum width=1.75cm, minimum height=0.8cm,text centered, draw=black, fill=gray!30]
\tikzstyle{arrow} = [thick,->,>=stealth]

\tikzstyle{mynode2} = [rectangle, rounded corners, minimum width=1.75cm, minimum height=0.8cm,text centered, draw=black, fill=RedOrange!20]

\tikzstyle{mynode3} = [rectangle, rounded corners, minimum width=1.75cm, minimum height=0.8cm,text centered, draw=black, fill=RoyalBlue!20]
\tikzstyle{arrow} = [thick,->,>=stealth]

\begin{figure}[!t]
\begin{center}
\begin{tikzpicture}[node distance=2cm, scale=.7pt]
\node (source) [mynode1] {$\lamrfailunres$};
\node (interm) [mynode2, right of=source,  xshift=3cm] {$\lamrsharfailunres$};
\node (target) [mynode3, right of=interm,  xshift=3cm] {$\spi$};
\draw[arrow] (source) --  node[anchor=south] {$ \recencodopenf{\cdot}$} (interm);
\node (enc1) [right of= source, xshift=.4cm, yshift=-.35cm] {\secref{ssec:first_enc}};
\node (enc2) [right of= interm, xshift=.4cm, yshift=-.35cm] {\secref{ssec:second_enc}};
\draw[arrow] (interm) -- node[anchor=south] {$ \piencodf{\cdot }_u $ } (target);
\end{tikzpicture}
\end{center}
\vspace{-3mm}
\caption{Our two-step approach to encode  $\lamrfailunres$ into $\spi$.}\label{fig:two-step-enc}
\end{figure}

\subparagraph{Syntax.}
\emph{Names}  $x, y,z, w \ldots$ denote the endpoints of protocols specified by session types. We write $P\subst{x}{y}$ for the capture-avoiding substitution of $x$ for $y$ in process $P$.

\begin{definition}[Processes]\label{d:spi}
The syntax of \spi processes is given by the grammar below. 
	\begin{align*}
P,Q ::=  ~&  \zero \sep \overline{x}(y).P \sep  x(y).P \sep \case{x}{i};P \sep \choice{x}{i}{I}{i}{P_i} \sep  x.\overline{\close} \sep x.\close;P \\
 \sep &  (P \para Q)
  \sep  [x \leftrightarrow y]  \sep (\nu x)P \sep ~!x(y).P \sep  \outsev{x}{y}.P 
  \\
   \sep &  x.\overline{\some};P \sep x.\overline{\none}
\sep      x.\some_{(w_1, \cdots, w_n)};P \sep  (P \oplus Q)
\end{align*}
\end{definition}
\noindent
Process $\zero$ denotes inaction. 
Process $\overline{x}(y).P$ sends a fresh name $y$ along  $x$ and then continues as $P$. Process 
     $x(y).P$ receives a name $z$ along $x$ and then continues as  $P\subst{z}{y}$. Process $\choice{x}{i}{I}{i}{P_i}$ is a branching construct, with labelled alternatives indexed by the finite set $I$: it awaits a choice on $x$ with continuation $P_j$ for each  $j \in I$. Process $\case{x}{i};P$ selects on $x$ the alternative indexed by $i$ before continuing as $P$. Processes $x.\overline{\close}$ and $x.\close;P$ are complementary actions for  closing session $x$. 
     We sometimes use the  shorthand notations $\overline{y}\sclose$ and $y\sclose;P$ to stand for ${y}.\overline{\close}$ and $y.\close;P$, respectively.
      Process $P \para Q$ is the parallel execution of $P$ and $Q$. 
The forwarder process $[x \leftrightarrow y]$ denotes a bi-directional link between sessions $x$ and $y$. 
Process $(\nu x)P$ denotes the process $P$ in which name $x$  is kept private (local) to $P$. 
Process $!x(y).P$ defines a server that spawns copies of $P$ upon requests on $x$.
Process $\outsev{x}{y}.P$ denotes a client that connects to a server by sending the fresh name $y$ on $x$. 

The remaining constructs  come from~\cite{CairesP17} and introduce non-determi\-nis\-tic sessions which  \emph{may} provide a session protocol  \emph{or} fail.
    Process $x.\overline{\some};P$ confirms that the session  on $x$ will execute and  continues as $P$. Process $x.\overline{\none}$ signals the failure of implementing the session on $x$.
    Process $x.\some_{(w_1, \cdots,w_n)};P$ specifies a dependency on a non-deterministic session $x$. 
    This process can  either (i)~synchronise with an action $x.\overline{\some}$ and continue as $P$, or (ii)~synchronise with an action $x.\overline{\none}$, discard $P$, and propagate the failure on $x$ to $(w_1, \cdots, w_n)$, which are sessions implemented in $P$.
    When $x$ is the only session implemented in $P$, there is no tuple of dependencies $(w_1, \cdots,w_n)$ and so we write simply $x.\some;P$.
        Finally, process $P \oplus Q$ denotes a non-deterministic choice between $P$ and $Q$. We shall often write $\bigoplus_{i \in \{1 , \cdots , n \} } P_i$ to stand for $P_1 \oplus \cdots \oplus P_n$.
In  $(\nu y)P$ and $x(y).P$ the   occurrence of name $y$ is binding, with scope $P$.
The set of free names of $P$ is denoted by $\fn{P}$.

\subparagraph{Semantics.}
The   \emph{reduction relation} of \spi specifies the computations that a process performs on its own (cf.~\figref{fig:redspi}). 
It is closed by  \emph{structural congruence}, denoted $\equiv$, which expresses basic identities for processes and the non-collapsing nature of non-determinism  
\iftypes 
(cf.~\cite{DBLP:journals/corr/abs-2111}).
\else 
(cf. \appref{appC}).
\fi


\begin{figure}[!t]
\begin{mdframed}[style=alttight]
\vspace{-2mm}
\small
\centering
\begin{align*}
 & \overline{x}{(y)}.Q \para x(y).P  \red (\nu y) (Q \para P) &
 & x.\overline{\some};P \para x.\some_{(w_1, \cdots, w_n)};Q  \red P \para Q\\
 & Q \red Q' \Rightarrow P \oplus Q   \red P \oplus Q' &
 & x.\overline{\close} \para x.\close;P  \red P \\ 
 & \case{x}{i};Q \para \choice{x}{i}{I}{i}{P_i}   \red  Q \para P_i &   
 & !x(y).Q \para  \outsev{x}{y}.P   \red (\nu x)( !x (y).Q \para  (\nu y)( Q \para P ) ) \\
 & (\nu x)( [x \leftrightarrow y] \para P) \red P \subst{y}{x}  \quad (x \neq y) &
 & P\equiv P'\wedge P' \red Q' \wedge Q'\equiv Q \Rightarrow P   \red Q\\
 & Q \red Q' \Rightarrow P \para Q   \red P \para Q'&
 & P \red Q  \Rightarrow (\nu y)P  \red (\nu y)Q \\
 & 
\end{align*}
\vspace{-1.2cm}
\[x.\overline{\none} \para x.\some_{(w_1, \cdots, w_n)};{Q} \red w_1.\overline{\none} \para \cdots \para w_n.\overline{\none}\]
\end{mdframed}
\vspace{-5mm}
    \caption{Reduction for \spi}
    \label{fig:redspi}
        \vspace{-3mm}
\end{figure}

The first reduction rule formalises communication, which concerns bound names only (internal mobility), as $y$ is bound in $\overline{x}{(y)}.Q$ and $x(y).P$.
 Reduction  for the forwarder process leads to a substitution.
The reduction rule for closing a session is self-explanatory, as is the rule in which prefix $x.\overline{\some}$ confirms the availability of a non-deterministic session. When a non-deterministic session is not available,   $x.\overline{\none}$ triggers this failure to all dependent sessions $w_1, \ldots, w_n$; this may in turn trigger further failures (i.e., on sessions that depend on $w_1, \ldots, w_n$).
The remaining rules define contextual reduction with respect to restriction,  composition, and non-deterministic choice.




\subparagraph{Type System}
Session types govern the behaviour of the names of a process.
An assignment $x:A$ enforces the use of name $x$ according to the  protocol specified by $A$.

\begin{definition}[Session Types]
Session types are given by 
\[
\begin{array}{rl}
  A,B ::= & \bot \sep   \onef \sep 
A \otimes B  \sep A \ampy B  
\sep  \oplus_{i \in I} \{ \mathtt{l}_i : A_i  \}  
 \sep \with_{i \in I} \{ \mathtt{l}_i : A_i \}  \sep ! A \sep   ? A 
 \sep  \with A \sep \oplus A  
\end{array}
\]
\end{definition}

\noindent
The multiplicative units  $\bot$ and  $\onef$ are used to type closed session endpoints.
We use $A \otimes B$ to type a name that first outputs a name of type $A$ before proceeding as specified by $B$.
Similarly, $A \ampy B $ types a name that first inputs a name of type $A$ before proceeding as specified by $B$.
Then, $! A$ types a name that repeatedly provides a service specified by $A$.
Dually, $ ? A $ is the type of a name that can connect to a server offering $A$.
Types 
$\oplus_{i \in I} \{ \mathtt{l}_i : A_i  \}$ and $\with_{i \in I} \{ \mathtt{l}_i : A_i \}$
are assigned to names that can select and offer a labelled choice, respectively.
Then we have the two modalities introduced in~\cite{CairesP17}.
We use $\with A$ as the type of a (non-deterministic) session that \emph{may  produce} a behaviour of type $A$.
Dually, $\oplus A$ denotes the type of a session that \emph{may consume} a behaviour of type $A$.

The two endpoints of a  session should be \emph{dual} to ensure  absence of communication errors. The dual of a type $A$ is denoted $\dual{A}$. 
Duality corresponds to negation $(\cdot)^\bot$ in  linear logic~\cite{CairesP17}. 

\begin{definition}[Duality]
\label{def:duality}
Duality on types is given by:
\[
\small
\begin{array}{l}
\begin{array}{c@{\hspace{.5cm}}c@{\hspace{.5cm}}c@{\hspace{.5cm}}c@{\hspace{.5cm}}c}
\dual{\onef}  =  \bot 
&
\dual{\bot}    =  \onef
&
\dual{A\otimes B}   = \dual{A} \ampy \dual{B}
&
\dual{\oplus_{i \in I} \{ \mathtt{l}_i : A_i  \}}   = \with_{i \in I} \{ \mathtt{l}_i : \dual{A_i} \}
&
\dual{\oplus A}    =    \with \dual{A}
\\
\dual{!A}  =  ?A 
 &
 \dual{?A}    =  !A
 &
\dual{A \ampy B}   = \dual{A} \otimes \dual{B} 
&
\dual{\with_{i \in I} \{ \mathtt{l}_i : A_i \}}   = \oplus_{i \in I} \{ \mathtt{l}_i : \dual{A_i}  \}
&
 \dual{\with A}   =   \oplus \overline {A}
\end{array}
\end{array}
\]
\end{definition}

\noindent
Judgements are of the form $P \vdash \Delta; \Theta$, where $P$ is a process, $\Delta$ is the linear context, and $\Theta$ is the unrestricted context.
Both $\Delta$ and $\Theta$ contain assignments of types to names, but satisfy different substructural principles: while $\Theta$ satisfies weakening, contraction and exchange, $\Delta$ only satisfies exchange. The empty context is denoted `$\cdot$'. 
We write $\with \Delta$ to denote that all assignments in $\Delta$ have a non-deterministic type, i.e., $\Delta = w_1{:}\with A_1, \ldots, w_n{:}\with A_n$, for some $A_1, \ldots, A_n$. 
The typing judgement $P \vdash \Delta$ corresponds to the logical sequent for classical linear logic, which can be recovered by erasing processes and name assignments. 

Typing rules for processes in \figref{fig:trulespi} correspond to proof rules in linear logic;
we discuss some of them.
Rule~$\redlab{Tid}$ interprets the identity axiom using the forwarder process. 
Rules~$\redlab{T \onef}$ and $\redlab{T \bot}$ type the process constructs for session termination.
 Rules~$\redlab{T\otimes}$ and $\redlab{T \ampy }$ type output and input of a name along a session, resp. 
The last four rules are used to type process constructs related to non-de\-ter\-mi\-nism and failure. 
Rules~$ \redlab{T \with_d^x}$ and $ \redlab{T \with^x}$ introduce a session of type $\with A$, which may produce a behaviour of type $A$: while the former rule covers the case in which $x:A$ is indeed available, the latter rule formalises the case in which $x:A$ is not available (i.e., a failure).
Given a sequence of names $\widetilde{w} = w_1, \ldots, w_n$,  Rule~$\redlab{T \oplus^x_{\widetilde{w}}}$ accounts for the possibility of not being able 
to consume the session $x:A$  by considering sessions different from $x$ as potentially not available. 
Rule~$\redlab{T \oplus }$ expresses non-deterministic choice of processes $P$ and $Q$ that implement non-deterministic behaviours only.
Finally, Rule~$\redlab{T\oplus_i} $ and $ \redlab{T\with} $ correspond, resp., to selection and branching: the former provides a selection of behaviours along $x$ as long as  $P$ is guarded with the $i$-{th} behaviour; the latter offers a labelled choice where each behaviour $A_i$ is matched to a corresponding $P_i$.

\begin{figure*}[!t]
\small
\begin{mdframed}[style=alttight]
\centering
\begin{prooftree}
\AxiomC{\mbox{\ }}
    \LeftLabel{\redlab{Tid}}
    \UnaryInfC{$[x \leftrightarrow y] \vdash x{:}A, y{:}\dual{A}; \,\Theta$}
    \DisplayProof
\hfill
    \AxiomC{\mbox{\ }}
    \LeftLabel{\redlab{T\onef}}
    \UnaryInfC{$x.\dual{\close} \vdash x: \onef; \,\Theta$}
    \DisplayProof
\hfill
    \AxiomC{$P\vdash \Delta ; \Theta $}
    \LeftLabel{\redlab{T\bot}}
    \UnaryInfC{$x.\close;P \vdash x{:}\bot, \Delta; \,\Theta$}
\end{prooftree}
  \vspace{-5mm}
  \begin{prooftree}
    \AxiomC{$P \vdash  \Delta, y:{A}; \,\Theta \quad Q \vdash \Delta', x:B; \,\Theta $}
    \LeftLabel{\redlab{T\otimes}}
    \UnaryInfC{$\dual{x}(y). (P \para Q) \vdash  \Delta, \Delta', x: A\otimes B; \,\Theta$}
\DisplayProof
\hfill 
    \AxiomC{$P \vdash \Delta, y:C, x:D; \,\Theta$}
    \LeftLabel{\redlab{T\ampy}}
    \UnaryInfC{$x(y).P \vdash \Delta, x: C\ampy D; \,\Theta $}
    \end{prooftree}
  \vspace{-6mm}
\begin{prooftree}
       \AxiomC{$P \;{ \vdash} \widetilde{w}:\with\Delta, x:A; \,\Theta$}
    \LeftLabel{\redlab{T\oplus^x_{\widetilde{w}}}}
    \UnaryInfC{$x.\some_{\widetilde{w}};P \vdash \widetilde{w}{:}\with\Delta, x{:}\oplus A; \,\Theta$}
    \DisplayProof
    \hfill
    \AxiomC{$P \vdash \Delta, x:A; \,\Theta$}
    \LeftLabel{\redlab{T\with_d^x}}
    \UnaryInfC{$x.\dual{\some};P \vdash \Delta, x :\with A; \,\Theta$}
    \end{prooftree}
      \vspace{-6mm}
    \begin{prooftree}
    \AxiomC{}
    \LeftLabel{\redlab{T\with^x}}
    \UnaryInfC{$x.\dual{\none} \vdash x :\with A; \,\Theta$}
    \DisplayProof
    \hfill 
    \AxiomC{$P \vdash \with\Delta; \,\Theta \qquad Q  \;{\vdash} \with\Delta; \,\Theta$}
    \LeftLabel{\redlab{T\oplus}}
    \UnaryInfC{$P\oplus Q \vdash \with\Delta; \,\Theta$}
 \end{prooftree}
   \vspace{-6mm}
   \begin{prooftree}
    \AxiomC{$ P  \vdash  \Delta, x: A_i ; \Theta $}
    \LeftLabel{\redlab{T\oplus_i}}
    \UnaryInfC{$\case{x}{i};P \vdash  \Delta,  x: \oplus_{i \in I} \{ \mathtt{l}_i : A_i  \}; \Theta $}
\DisplayProof
\hfill
    \AxiomC{$ P_i \vdash \Delta , x: A_i ; \Theta \quad (\forall i \in I)$}
    \LeftLabel{\redlab{T\with}}
    \UnaryInfC{$ \choice{x}{i}{I}{i}{P_i} \vdash \Delta , x: \with_{i \in I} \{ \mathtt{l}_i : A_i \} ; \Theta $}
\end{prooftree}
  \vspace{-6mm}
\begin{prooftree}
    \AxiomC{$P \vdash \Delta ; x : A , \Theta$}
    \LeftLabel{\redlab{T?}}
    \UnaryInfC{$P \vdash \Delta ,x :? A ; \Theta$}
\DisplayProof
\hfill
\AxiomC{$ P \vdash y: A; \Theta$}
    \LeftLabel{\redlab{T!}}
    \UnaryInfC{$!x(y).P \vdash x: !A; \Theta $}
\DisplayProof
\hfill
       \AxiomC{$P \;{ \vdash} \Delta , y: A ; x: A , \Theta $}
    \LeftLabel{\redlab{Tcopy}}
    \UnaryInfC{$ \outsev{x}{y}.P \vdash \Delta ; x: A , \Theta $}
\end{prooftree}
\end{mdframed}
  \vspace{-5mm}
\caption{Typing rules for \spi.}
\label{fig:trulespi}
\end{figure*}

The type system enjoys type preservation, a result that
follows from the cut elimination property in linear logic; it ensures that the observable interface of a system is invariant under reduction.
The type system also ensures other properties for well-typed processes (e.g. global progress, strong normalisation,  and confluence); see~\cite{CairesP17} for details.

\begin{theorem}[Type Preservation~\cite{CairesP17}]
If $P \vdash \Delta; \,\Theta$ and $P \red Q$ then $Q \vdash \Delta; \,\Theta$.
\end{theorem}

\subsection[An Auxiliary Calculus With Sharing]{\lamrsharfailunres: An Auxiliary Calculus With Sharing}\label{ssec:lamshar}

To facilitate the encoding of \lamrfailunres into \spi, we define \lamrsharfailunres: an auxiliary calculus whose constructs are inspired by the work of Gundersen et al.~\cite{DBLP:conf/lics/GundersenHP13}, Ghilezan et al.~\cite{GhilezanILL11}, and Kesner and Lengrand~\cite{DBLP:journals/iandc/KesnerL07}. 
The syntax of \lamrsharfailunres only modifies the syntax of terms, which is defined by the grammar below; variables ${x}[*]$, bags $B$, and expressions $\expr{M}$ are  as in \Cref{def:rsyntaxfailunres}.
\begin{align*}
\mbox{(Terms)} \quad  M,N, L ::= &
~~  {x}[*] 
\sep \lambda x . (M[ {\widetilde{x}} \leftarrow  {x}]) 
\sep (M\ B) \sep M \linexsub{N /  {x}} 
\sep M \unexsub{U / x}
\\
    \sep &~~
\fail^{\widetilde{x}}\sep 
M [  {\widetilde{x}} \leftarrow  {x} ] 
\sep (M[ {\widetilde{x}} \leftarrow  {x}])\esubst{ B }{ x } 
\end{align*}
\noindent
We consider the {\it sharing  construct} $M[\widetilde{x}\leftarrow x]$ and two   kinds of explicit substitutions: the {\it explicit linear substitution}, written $M\linexsub{N/ {x}}$, and the {\it explicit unrestricted substitution}, written $M\unexsub{U/\unvar{x}}$. 
The term $M[ {\widetilde{x}}\leftarrow  {x}]$ defines the sharing of variables $ {\widetilde{x}}$ occurring in $M$ using the linear variable $ {x}$. 
We shall refer to $ {x}$ as \emph{sharing variable} and to $ {\widetilde{x}}$ as \emph{shared variables}. A linear variable is only allowed to appear once in a term. 
Notice that $ {\widetilde{x}}$ can be empty: $M[\leftarrow  {x}]$ expresses that $ {x}$ does not share any variables in $M$. 
As in $\lamrfailunres$, the term $\fail^{ {\widetilde{x}}}$ explicitly accounts for failed attempts at substituting the variables in $ {\widetilde{x}}$.

We summarise some requirements. 
In  $M [ \widetilde{x} \leftarrow x ]$, we require: (i)~every $x_i \in \widetilde{x}$ occurs exactly once in $M$ and that (ii)~$x_i$ is not a sharing variable.
The occurrence of $x_i$ can appear within the fail term $\fail^{\widetilde{y}}$, if $x_i \in \widetilde{y}$.
In the explicit linear substitution $M \linexsub{ N /  {x}}$, we require: the variable $x$ has to occur in $M$; $x$ cannot be a sharing variable; and $x$ cannot be in an explicit linear substitution occurring in $M$;  {all free {\em linear} occurrences of $x$ in $M$ are bound}. In the explicit unrestricted substitution $M \unexsub{ U / \unvar{x}}$, we require: all free {\em unrestricted} occurrences of $x$  in $M$ are bound; 
$\unvar{x}$ cannot be in an explicit unrestricted substitution occurring in $M$. 
This way, e.g.,  $M'\linexsub{L/ {x}}\linexsub{N/ {x}}$ 
and $M'\linexsub{U'/\unvar{x}}\linexsub{U/\unvar{x}}$ 
are not valid terms in $\lamrsharfailunres$.

The following congruence will be important when proving encoding correctness.

\begin{definition}\label{def:rsPrecongruencefailure}
The congruence $\pequiv$ for \lamrsharfailunres on terms and expressions  is given by the identities below.
{\small
\[
\begin{array}{rcl}
    M \unexsub{U / \unvar{x}}  &\pequiv& M, \ x \not \in M\\
    
    (MB) \linexsub{N/x} & \pequiv& (M\linexsub{N/x})B, \ x \not \in \lfv{B} \\
    
    (MB) \unexsub{U / \unvar{x}} & \pequiv & (M\unexsub{U / \unvar{x}})B, \ x \not \in \lfv{B} \\
    
    (MA)[\widetilde{x} \leftarrow x]\esubst{B}{x}  &\pequiv& (M[\widetilde{x} \leftarrow x]\esubst{B}{x})A, \   x_i \in \widetilde{x} \Rightarrow x_i \not \in \lfv{A}\\
    
    M[\widetilde{y} \leftarrow y]\esubst{A}{y}[\widetilde{x} \leftarrow x]\esubst{B}{x} & \pequiv & (M[\widetilde{x} \leftarrow x]\esubst{B}{x})[\widetilde{y} \leftarrow y]\esubst{A}{y},\  
    \begin{aligned}
    x_i \in \widetilde{x} \Rightarrow x_i \not \in \lfv{A}, \\
       y_i \in \widetilde{y} \Rightarrow y_i \not \in \lfv{B}
    \end{aligned}
     \\[1mm]
    
    M \linexsub{N_2/y}\linexsub{N_1/x} &\pequiv & M\linexsub{N_1/x}\linexsub{N_2/y}, x \not \in \lfv{N_2}, y\notin \lfv{N_1} \\[1mm]
    
    M \unexsub{U_2 / \unvar{y}}\unexsub{U_1 / \unvar{x}} & \pequiv & M\unexsub{U_1 / \unvar{x}}\unexsub{U_2 / \unvar{y}}, x \not \in \lfv{U_2}, y\notin \lfv{U_1} \\[1mm]
    
    C[M]   &\pequiv& C[M'], \ \text{with } M \pequiv M' \\
    
    D[\expr{M}]  &\pequiv& D[\expr{M}'], \ \text{with } \expr{M} \pequiv \expr{M}'
\end{array}
\]
}
\end{definition}
The first rule states that we may remove unneeded unrestricted substitutions when the variable in concern does not appear within the term. The next three identities enforce that bags can always be moved in and out of all forms of explicit substitution, which are  useful manipulate expressions and to form a redex for Rule \redlab{R:Beta}. The other rules deal with permutation of explicit substitutions and contextual closure.

Well-formedness for \lamrsharfailunres, based on intersection types, is defined as in \secref{sec:lam_types}; 
\iftypes
see \cite{DBLP:journals/corr/abs-2111}.
\else 
see  \appref{app:ssec:lamshar}.
\fi 


\subsection[First Step]{Encoding \lamrfailunres into \lamrsharfailunres} \label{ssec:first_enc}
We define an encoding  $\recencodopenf{\cdot}$ from well-formed terms in $\lamrfailunres$ into $\lamrsharfailunres$. This encoding relies on an intermediate encoding $\recencodf{\cdot}$ on $\lamrfailunres$-terms. 
 
 \begin{notation}
Given a term $M$ such that  $\#(x , M) = k$
and a sequence of pairwise distinct fresh variables $\widetilde{x} = x_1, \ldots, x_k$ we write $M \linsub{\widetilde{x}}{x}$ or $M\linsub{x_1,\cdots, x_k}{x}$ to stand for 
$M\linsub{x_1}{x}\cdots \linsub{x_k}{x}$, i.e.,  a simultaneous linear substitution whereby each distinct linear occurrence of $x$ in $M$ is replaced by a distinct $x_i \in \widetilde{x}$. 
Notice that each $x_i$ has the same type as $x$.
We use (simultaneous) linear substitutions to force all bound linear variables in \lamrfailunres to become shared variables in \lamrsharfailunres.
\end{notation}

\begin{definition}[From $\lamrfailunres$ to $\lamrsharfailunres$]
\label{def:enctolamrsharfailunres}
    Let $M \in \lamrfailunres$.
    Suppose $\Theta ; \Gamma \wfdash {M} : \tau$, with
    $\dom{\Gamma} = \llfv{M}=\{ {x}_1,\cdots, {x}_k\}$ and  $\#( {x}_i,M)=j_i$.  
    We define $\recencodopenf{M}$ as
    \begin{equation*}
     \recencodopenf{M} = 
    \recencodf{M\linsub{ {\widetilde{x}_{1}}}{ {x}_1}\cdots \linsub{ {\widetilde{x}_k}}{ {x}_k}}[ {\widetilde{x}_1}\leftarrow  {x}_1]\cdots [ {\widetilde{x}_k}\leftarrow  {x}_k] 
     \end{equation*}
       where  $ {\widetilde{x}_i}= {x}_{i_1},\cdots,  {x}_{{j_i}}$
       and the encoding $\recencodf{\cdot}: \lamrfailunres \to \lamrsharfailunres$  is defined in~\figref{fig:auxencfailunres} on $\lamrfailunres$-terms. 
       The encoding $\recencodopenf{\cdot}$ extends homomorphically to expressions.
    \end{definition}

The encoding $\recencodopenf{\cdot}$  converts $n$ occurrences of   $x$ in a term into $n$ distinct variables $x_1, \ldots, x_n$.
The sharing construct coordinates them by constraining each to occur exactly once within a term. 
We proceed in two stages. 
First, we share all linear free linear variables using $\recencodopenf{\cdot}$: this ensures that free variables are replaced by shared variables which are then bound by the sharing construct. 
Second, we apply the encoding $\recencodf{\cdot}$ on the corresponding term. 
The encoding is presented in \figref{fig:auxencfailunres}:
$ \recencodf{\cdot}$ maintains $x[i]$ unaltered, and acts homomorphically over  concatenation  of bags   and explicit substitutions. {The encoding renames bound variables with bound shared variables.}
As we will see, this will enable a tight operational correspondence result with $\spi$.
\iftypes 
In \cite{DBLP:journals/corr/abs-2111}
\else 
In \appref{app:encodingone} 
\fi 
we  establish the correctness of   $\recencodopenf{\cdot}$.

\begin{figure*}[t]
\begin{mdframed}[style=alttight]
\small
\[
\begin{array}{c@{\hspace{.5cm}}c@{\hspace{.5cm}}c@{\hspace{.5cm}}c}
   \recencodf{  {x}  }  =   {x}  &
    \recencodf{ {x}[i]  }  =  {x}[i]  & \recencodf{  \oneb  } =  \oneb   &\\
     \recencodf{\banged{\oneb}} = \banged{\oneb}&  \recencodf{  \fail^{\widetilde{x}} } = \fail^{\widetilde{x}} & \recencodf{  M\ B }  =  \recencodf{M}\ \recencodf{B} 
& \\
    \recencodf{\banged{\bag{M}}} = \banged{\bag{M}}&  \recencodf{  \bag{M}\cdot C}  =  \bag{ \recencodf{M}} \cdot \recencodf{C}& \recencodf{  C \bagsep U  }  = \recencodf{C} \bagsep \recencodf{ U }&\\
  \recencodf{U \concat V} = U \concat V &  \recencodf{  M \linexsub{N /  {x}} }  =  \recencodf{M} \linexsub{ \recencodf{N} /  {x}}
     &
     \recencodf{ M \unexsub{U / \unvar{x}} }  =  \recencodf{M} \unexsub{ \recencodf{U} / \unvar{x}}&
\end{array}
\]
\vspace{-5mm}
\begin{align*}
	    \recencodf{  \lambda x . M  }  &=   \lambda x . (\recencodf{M\langle  {x}_1 , \cdots ,  {x}_n /  {x}  \rangle} [ {x}_1 , \cdots ,  {x}_n \leftarrow  {x}])
  \quad  \text{ $\#( {x},M) = n$, each $ {x}_i$ is fresh}
\\
   \recencodf{  M \esubst{ C \bagsep U }{ x }  } &= 
   \begin{cases}
   \displaystyle\sum_{C_i \in \perm{\recencodf{ C  } }}\hspace{-5mm}\recencodf{ M \langle  {\widetilde{x}}/  {x}  \rangle } \linexsub{C_i(1)/ {x}_1} \cdots \linexsub{C_i(k)/ {x}_k}\unexsub{ U/ \unvar{x}}, ~\text{if $\#( {x},M) = \size{C} = k$} \\
    \recencodf{M\langle  {x}_1, \cdots ,  {x}_k /  {x}  \rangle} [  {x}_1, \cdots ,  {x}_k \leftarrow  {x}] \esubst{ \recencodf{ C \bagsep U } }{ x }, ~ \text{if }\#( {x},M) = k\geq 0
     \end{cases} 
\end{align*}
\end{mdframed}
\vspace{-5mm}
    \caption{Auxiliary Encoding: \lamrfailunres into \lamrsharfailunres}
    \label{fig:auxencfailunres}
\end{figure*}


\begin{example}\label{ex:id_enc} We apply the encoding $\recencodf{\cdot}$ in some of the $\lamrfailunres$-terms  from Example~\ref{ex:id_term}: for simplicity, we assume that $N$ and $U$ have no free variables.
 \[{\small 
    \begin{aligned}
    \recencodf{(\lambda x.x)\ \bag{N}\bagsep U}&=\recencodf{\lambda x. x}\recencodf{\bag{N}\bagsep U}=\lambda x. x_1[x_1\leftarrow x]\bag{\recencodf{N}}\bagsep \recencodf{U}\\
   \recencodf{(\lambda x. x[1] ) \oneb \bagsep \bag{N}^! \concat U}&=\recencodf{(\lambda x. x[1]}\recencodf{\oneb \bagsep \bag{N}^! \concat U}=(\lambda x. x[1] [ \leftarrow x] ) \oneb \bagsep \bag{\recencodf{N}}^! \concat \recencodf{U}
    \end{aligned}
    }
\]

\end{example}




\subsection[Second Step]{Encoding $\lamrsharfailunres$ into $\spi$}\label{ssec:second_enc}



We now define our encoding of $\lamrsharfailunres$ into $\spi$, and establish its correctness. 

\begin{notation}
To help illustrate the behaviour of the encoding, 
 we use the names $x$, $\linvar{x}$, and $\banged{x}$ to denote three distinct channel names: while $\linvar{x}$ is the channel that performs the linear substitution behaviour of the encoded term, channel $\banged{x}$ performs the unrestricted behaviour. 
\end{notation}

\begin{definition}[From $\lamrsharfailunres$ into $\spi$: Expressions]
\label{def:enc_lamrsharpifailunres}
Let $u$ be a name.
The encoding $\piencodf{\cdot }_u: \lamrsharfailunres \rightarrow \spi$ is defined in \figref{fig:encfailunres}.
\end{definition}

Every (free) variable $x$  in an $\lamrsharfailunres$ expression becomes a name $x$ in its corresponding \spi process.
As customary in encodings of $\lambda$ into $\pi$, we use a name $u$ to provide the behaviour of the encoded expression. 
In our case, $u$ is a non-deterministic session: the encoded expression can be effectively available or not; this is signalled by prefixes $u.\overline{\some}$
and 
$u.\overline{\none}$, respectively. 




We discuss the most salient aspects of the encoding in \figref{fig:encfailunres}.
   \begin{itemize}
   \item 
    While linear variables are encoded as in \cite{DBLP:conf/fscd/PaulusN021}, the encoding of an unrestricted variable $x[j]$, not treated in \cite{DBLP:conf/fscd/PaulusN021}, is much more interesting: it first connects to a server along channel $x$ via a request $ \outsev{\banged{x}}{{x_i}}$ followed by a selection on $ {x}_i.l_{j}$, which takes the $j$-{th} branch.
   
    \item The encoding of   $\lambda x. M[\widetilde{x} \leftarrow x] $ confirms its behaviour first followed by the receiving of a channel $x$. The channel $x$ provides a linear channel $\linvar{x}$ and an unrestricted channel $\banged{x}$ for dedicated substitutions of the linear and unrestricted bag components.

    \item We encode $M\, (C \bagsep U)$ as  a non-deterministic sum: an application involves a choice in the order in which the elements of $C$ are substituted. 
    \item 
    The encoding of $C \bagsep U$ synchronises with the encoding of ${\lambda x . M[\widetilde{x} \leftarrow x]}$. The channel $ \linvar{x}$ provides the linear behaviour of the bag $C$ while $\banged{x}$ provides the behaviour of   $U$; this is done by guarding the encoding of $U$ with a server connection such that every time a channel synchronises with $!\banged{x} (x_i)$ a fresh copy of $U$ is spawned.
    
    \item The encoding of ${\bag{M} \cdot\, C}$ synchronises with the encoding of ${M[\widetilde{x} \leftarrow x]}$, just discussed. The name $y_i$ is used to trigger a failure in the computation if there is a lack of elements in the encoding of the bag.
    \item
    The encoding of ${M[\widetilde{x} \leftarrow x]}$ first confirms the availability of the linear behaviour along $\linvar{x}$. Then it sends a name $y_i$,  which is used to collapse the process in the case of a failed reduction. 
    Subsequently, for each shared variable, the encoding receives a name, which will act as an occurrence of the shared variable. 
    At the end, a failure prefix on $x$ is used to signal that there is no further  information to send over.
    \item 
    The encoding of $U$ synchronises with the last half encoding of ${x[j]}$; the name $x_i$ selects the $j$-th term in the unrestricted bag.
    \item 
    The encoding of ${ M \linexsub{N / x}} $ is the composition of the encodings of $M$ and $N$, where we await a confirmation of a behaviour along the variable that is being substituted.
    \item 
     ${ M \unexsub{U / \unvar{x}}} $ is encoded as the composition of the encoding of $M$ and a server guarding the encoding of $U$: in order for $\piencodf{M}_u$ to gain access to $\piencodf{U}_{x_i}$ it must first synchronise with the server channel $\banged{x}$ to spawn a fresh copy of $U$.
    \item The encoding of ${\expr{M}+\expr{N} }$ homomorphically preserves non-determinism. 
    Finally, the encoding of $\fail^{x_1 , \cdots , x_k}$ simply triggers failure on $u$ and on each of $x_1 , \cdots , x_k$.
\end{itemize}


\begin{figure*}[t!]
\begin{mdframed}[style=alttight]
\vspace{-2mm}
\small
\[
	\begin{array}{rl}
	   \piencodf{ {x}}_u &\hspace{-2mm} =  
	   {x}.\overline{\some} ; [ {x} \leftrightarrow u]  
	   \\[1mm]
   \piencodf{{x}[j]}_u & \hspace{-2mm}= 
   \outsev{\banged{x}}{{x_i}}. {x}_i.l_{j}; [{x_i} \leftrightarrow u] 
   \\[1mm]
    \piencodf{\lambda x.M[ {\widetilde{x}} \leftarrow  {x}]}_u & \hspace{-2mm}= 
    u.\overline{\some}; u(x). x.\overline{\some}; x(\linvar{x}). x(\banged{x}). x. \close ; \piencodf{M[ {\widetilde{x}} \leftarrow  {x}]}_u
    \\[1mm]
          \piencodf{ M[ {\widetilde{x}} \leftarrow  {x}] \esubst{ C \bagsep U }{ x} }_u &\hspace{-2mm} =\!\! \displaystyle  
      \bigoplus_{C_i \in \perm{C}}\!\! (\nu x)( x.\overline{\some}; x(\linvar{x}). x(\banged{x}). x. \close ;\piencodf{ M[ {\widetilde{x}} \leftarrow  {x}]}_u \para \piencodf{ C_i \bagsep U}_x )  
      \\[1mm]
  \piencodf{M (C \bagsep U)}_u &\hspace{-2mm} =  \displaystyle
  \bigoplus_{C_i \in \perm{C}} (\nu v)(\piencodf{M}_v \para v.\some_{u , \llfv{C}} ; \outact{v}{x} . ([v \leftrightarrow u] \para \piencodf{C_i \bagsep U}_x ) )  
  \\[1mm]    
  \piencodf{ C \bagsep U }_x & \hspace{-2mm}=  
  x.\some_{\llfv{C}} ; \outact{x}{\linvar{x}}.\big( \piencodf{ C }_{\linvar{x}} \para \outact{x}{\banged{x}} .( !\banged{x} (x_i). \piencodf{ U }_{x_i} \para x.\overline{\close} ) \big)
    \\[1mm]
        \piencodf{{\bag{M}} \cdot {C}}_{\linvar{x}} & \hspace{-2mm}=
       \linvar{x}.\some_{\llfv{\bag{M} \cdot C} } ; \linvar{x}(y_i). \linvar{x}.\some_{y_i, \llfv{\bag{M} \cdot C}};\linvar{x}.\overline{\some} ; \outact{\linvar{x}}{x_i}
       .\\[1mm]
       & \qquad (x_i.\some_{\llfv{M}} ; \piencodf{M}_{x_i} \para \piencodf{C}_{\linvar{x}} \para y_i. \overline{\none}) 
 \\[1mm]
    \piencodf{{\oneb}}_{\linvar{x}} &\hspace{-2mm} = 
   \linvar{x}.\some_{\emptyset} ; \linvar{x}(y_n). ( y_n.\overline{\some};y_n . \overline{\close} \para \linvar{x}.\some_{\emptyset} ; \linvar{x}. \overline{\none}) 
   \\[1mm]
   \piencodf{\banged{\oneb}}_{x} &\hspace{-2mm} =  x.\overline{\none} 
   \\[1mm]
   \piencodf{\banged{\bag{N}}}_{x}  &\hspace{-2mm} =  \piencodf{N}_{x}
   \\[1mm]
       \piencodf{ U }_{x}  &\hspace{-2mm}=  \choice{x}{U_i}{U}{i}{\piencodf{U_i}_{x}} 
       \\[1mm] 
      \piencodf{ M \linexsub{N /  {x}}  }_u    &\hspace{-2mm} =   
      (\nu  {x}) ( \piencodf{ M }_u \para    {x}.\some_{\llfv{N}};\piencodf{ N }_{ {x}}  )
      \\[1mm] 
      \piencodf{ M \unexsub{U / \unvar{x}}  }_u   &\hspace{-2mm} =   
      (\nu \banged{x}) ( \piencodf{ M }_u \para   ~!\banged{x}(x_i).\piencodf{ U }_{x_i} ) 
      \\[1mm]
       \piencodf{M[  \leftarrow  {x}]}_u &\hspace{-2mm} =
  \linvar{x}. \overline{\some}. \outact{\linvar{x}}{y_i} . ( y_i . \some_{u,\llfv{M}} ;y_{i}.\close; \piencodf{M}_u \para \linvar{x}. \overline{\none})
  \\[1mm]
          \piencodf{M[ {x}_1, \cdots ,  {x}_n \leftarrow  {x}]}_u &\hspace{-2mm} =
      \linvar{x}.\overline{\some}. \outact{\linvar{x}}{y_1}. \big(y_1 . \some_{\emptyset} ;y_{1}.\close;\zero   \para 
      \\[1mm]
      & \qquad \linvar{x}.\overline{\some};\linvar{x}.\some_{u, (\llfv{M} \setminus  {x}_1 , \cdots ,  {x}_n )};\linvar{x}( {x}_1) . \piencodf{M[ {x}_2,\cdots ,  {x}_n \leftarrow  {x}]}_u \big)
   \\[1mm] 
     \piencodf{\expr{M}+\expr{N} }_u    &\hspace{-2mm} =  \piencodf{ \expr{M} }_u \oplus \piencodf{ \expr{N} }_u
     \\[1mm]
   \piencodf{\fail^{x_1 , \cdots , x_k}}_u  &\hspace{-2mm} = 
   u.\overline{\none} \para x_1.\overline{\none} \para \cdots \para x_k.\overline{\none} 
\end{array}
\]
\end{mdframed}
\vspace{-2.0ex}
    \caption{Encoding \lamrsharfailunres into \spi (cf. Def.~\ref{def:enc_lamrsharpifailunres}).}
    \label{fig:encfailunres}
\end{figure*}

\begin{example}\label{ex:id_pienc}[Cont. \cref{ex:id_term}]
We illustrate the encoding $\piencodf{\cdot }$ on the $\lamrsharfailunres$-terms/bags occurring in $ M_1=\lambda x.x_1[x_1\leftarrow x](\bag{\recencodf{N}}\bagsep \recencodf{U})$ as below:


\(
\begin{aligned}
\piencodf{\lambda x.x_1[x_1\leftarrow x]}_v&=v.\overline{\some}; v(x). x.\overline{\some}; x(\linvar{x}).x(\banged{x}). x\sclose ;\piencodf{x_1[x_1\leftarrow x]}_v
\end{aligned}
\)

 \(
\begin{aligned}
\piencodf{\bag{\recencodf{N}}\bagsep \recencodf{U}}_x& =x.\some_{\llfv{\bag{\recencodf{N}}}}; \outact{x}{\linvar{x}} .( \piencodf{ \recencodf{N}}_{\linvar{x}} \para\outact{x}{\banged{x}} .( !\banged{x} (x_i). \piencodf{ \recencodf{U} }_{x_i} \para \overline{x}\sclose ) )
\end{aligned}
\)



{
\(
\begin{aligned}
  \piencodf{\recencodf{M_1}}_{u}&= 
\piencodf{\lambda x. x_1[x_1\leftarrow x]\bag{\recencodf{N}}\bagsep \recencodf{U}}_{u}\\
 &=(\nu v)(\piencodf{\lambda x.x_1[x_1\leftarrow x]}_v\para v.\some_{u,\llfv{\recencodf{N}}};\outact{v}{x}.([v\leftrightarrow u]\para \piencodf{\bag{\recencodf{N}}\bagsep \recencodf{U}}_x))            
 \\
&=(\nu v)(  v.\overline{\some}; v(x). x.\overline{\some}; x(\linvar{x}). x(\banged{x}). x\sclose ;  \linvar{x}.\overline{\some}. \outact{\linvar{x}}{y_1}. (y_1 . \some_{\emptyset} ;y_{1}\sclose;\zero \para
\\
& \quad  \linvar{x}.\overline{\some};\linvar{x}.\some_{u}; \linvar{x}( {x}_1) . \linvar{x}. \overline{\some}. \outact{\linvar{x}}{y_2} .  ( y_2 . \some_{u, x_1 } ;y_{2}\sclose; \piencodf{\colthree{x_1}}_v \para \linvar{x}. \overline{\none}) )
            \para 
            \\ 
& \quad  v.\some_{u , \llfv{\recencodf{N}}};  \outact{v}{x} . ([v \leftrightarrow u] \para \\
& \quad 
 x.\some_{\llfv{\recencodf{N}}} ;  \outact{x}{\linvar{x}} .( \linvar{x}.\some_{\llfv{\recencodf{N}} } ; \linvar{x}.\some_{y_1, \llfv{ \bag{ \recencodf{N}}}};
\\
            & \quad  \linvar{x}. \overline{\some} ; \outact{\linvar{x}}{x_1} 
            .  (x_1.\some_{\llfv{\recencodf{N}}} ;\piencodf{\recencodf{N}}_{x_1} \para y_1. \overline{\none} \para  \linvar{x}.\some_{\emptyset} ;\linvar{x}(y_2).
            \\
            & \quad   ( y_2.\overline{\some};  \overline{y_2} \sclose \para \linvar{x}.\some_{\emptyset} ; \linvar{x}. \overline{\none}) )  \para  \outact{x}{\banged{x}} .( !\banged{x} (x_i). \piencodf{ \recencodf{U} }_{x_i} \para \overline{x}\sclose ) ) )
            )   
            \end{aligned}
            \)
}
\end{example}

We now encode intersection types (for  \lamrsharfailunres) into session types (for \spi).

\begin{figure}[!t]
\begin{mdframed}[style=alttight]
\vspace{-2mm}
	{\small 
\begin{align*}
 \piencodf{\unit} &= \with \onef 
 \\ 
  \piencodf{ \eta } &= \&_{\eta_i \in \eta} \{ \mathtt{l}_i ; \piencodf{\eta_i} \}  
  \\
  \piencodf{(\sigma^{k} , \eta )   \rightarrow \tau} &= \with( \dual{\piencodf{ (\sigma^{k} , \eta  )  }_{(\sigma, i)}} \ampy \piencodf{\tau}) 
  \\
  \piencodf{ (\sigma^{k} , \eta  )  }_{(\sigma, i)} &= \oplus( (\piencodf{\sigma^{k} }_{(\sigma, i)}) \otimes ((!\piencodf{\eta}) \otimes (\onef))  )
  \\[1mm]
 \piencodf{ \sigma \wedge \pi }_{(\sigma, i)} &= \overline{   \with(( \oplus \bot) \otimes ( \with  \oplus (( \with  \overline{\piencodf{ \sigma }} )  \ampy (\overline{\piencodf{\pi}_{(\sigma, i)}})))) } \\
     & = \oplus(( \with \onef) \ampy ( \oplus  \with (( \oplus \piencodf{\sigma} ) \otimes (\piencodf{\pi}_{(\sigma, i)}))))
     \\
 \piencodf{\omega}_{(\sigma, i)} & =  \quad \begin{cases}
     \overline{\with(( \oplus \bot )\otimes ( \with \oplus \bot )))} &  \text{if $i = 0$}
     \\
\overline{   \with(( \oplus \bot) \otimes ( \with  \oplus (( \with  \overline{\piencodf{ \sigma }} )  \ampy (\overline{\piencodf{\omega}_{(\sigma, i - 1)}})))) } & \text{if $i > 0$}
\end{cases}
\end{align*}
}
\vspace{-2mm}
\end{mdframed}
  \vspace{-5mm}
    \caption{Encoding of intersection types into session types  (cf. Def.~\ref{def:enc_sestypfailunres})}
    \label{fig:enc_types}
        \vspace{-5mm}
\end{figure}

\begin{definition}[From $\lamrsharfailunres$ into $\spi$: Types]
\label{def:enc_sestypfailunres}
The translation  $\piencodf{\cdot}$  in Figure~\ref{fig:enc_types} extends as follows to a  context 
$\Gamma =  {x}_1{:} \sigma_1, \cdots,  {x}_m {:} \sigma_m,  {v}_1{:} \pi_1 , \cdots ,  {v}_n{:} \pi_n$
and a context 
$\Theta =\banged{x}_1 {:} \eta_1 , \cdots , \banged{x}_k {:} \eta_k$:
\[
\begin{aligned}
\piencodf{\Gamma} &= {x}_1 : \with \overline{\piencodf{\sigma_1}} , \cdots ,   {x}_m : \with \overline{\piencodf{\sigma_m}} ,
  {v}_1:  \overline{\piencodf{\pi_1}_{(\sigma, i_1)}}, \cdots ,  {v}_n: \overline{\piencodf{\pi_n}_{(\sigma, i_n)}}\\
  \piencodf{\Theta}&=\banged{x}_1 : \dual{\piencodf{\eta_1}} , \cdots , \banged{x}_k : \dual{\piencodf{\eta_k}} 
\end{aligned}
\]

\end{definition}

\noindent
This encoding formally expresses how non-de\-termi\-nis\-tic session protocols (typed with `$\with$') capture linear and unrestricted resource consumption in \lamrsharfailunres. 
Notice that the encoding of the multiset type $\pi$ depends on two arguments (a strict type $\sigma$ and a number $i \geq 0$) which are left unspecified above.
This is crucial to represent failures in \lamrsharfailunres as typable processes in \spi. 
For instance, given $(\sigma^{j} , \eta ) \rightarrow \tau$ and $ ( \sigma^{k} , \eta)$, the well-formedness rule for application admits a mismatch 
\iftypes
($j \neq k$, see \cite{DBLP:journals/corr/abs-2111}).
\else 
($j \neq k$, cf. Rule \redlab{FS{:}app} in \figref{app_fig:wfsh_rulesunres}, \appref{app:ssec:lamshar}).
\fi 
 In our proof of type preservation, the two arguments of the encoding are instantiated appropriately. 
Notice also how the client-server behaviour of unrestricted resources appears as `$!\piencodf{\eta}$' in the encoding of the tuple type
$(\sigma^{k} , \eta)$.
With our encodings of expressions and types in place, we can now define our encoding of judgements:

\begin{definition}
If $\expr{M}$ is an \lamrsharfailunres expression such that 
$\Theta ; \Gamma \wfdash \expr{M} : \tau$
then we define the encoding of the judgement to be: 
$\piencodf{\expr{M}}_u \vdash 
\piencodf{\Gamma}, 
u : \piencodf{\tau} ; \piencodf{\Theta} $.

\end{definition}



 The correctness  of  our encoding $\piencodf{\cdot}_u:\lamrfailunres\to \spi$, stated in Theorem~\ref{thm:op_correct} 
 \iftypes
 (and detailed in~\cite{DBLP:journals/corr/abs-2111}), 
\else
 (and detailed in~\appref{app:encodingtwo}), 
 \fi 
 relies on a notion of {\it success}  for both $\lamrfailunres $  and $\spi$, given by the $\checkmark$ construct:




\begin{definition}
\label{def:Suc3}We extend the syntax of terms for $\lamrsharfailunres$  and processes for \spi with $\checkmark$:
\begin{itemize}
    \item {\bf (In \lamrsharfailunres)} \succp{\expr{M}}{\checkmark} iff
there exist  $M_1, \cdots, M_k$ such that 
$\expr{M} \red^*  M_1 + \cdots + M_k$ and
$\headf{M_j'} = \checkmark$, for some  $j \in \{1, \ldots, k\}$ and term $M_j'$ such that $M_j\pequiv  M_j'$.
\item {\bf (In \spi)} $\succp{P}{\checkmark}$ holds whenever there exists a $P'$
such that 
$P \red^* P'$
and $P'$ contains an unguarded occurrence of $\checkmark$
(i.e., an occurrence that does not occur behind a prefix).
\end{itemize}

\end{definition}


\noindent
We now state operational correctness.
 \figref{fig:my_label} illustrates the relation between completeness and soundness that the encoding satisfies: solid arrows denote reductions assumed, dashed arrows denote the application of $\piencodf{\cdot }_u$, and dotted arrows denote the existing reductions that can be implied from the results.
 
We remark that since \lamrsharfailunres satisfies the diamond property, it suffices to consider  completeness based on a single reduction ($ \expr{N}\red \expr{M}$). Soundness uses the congruence $\pequiv$ in \defref{def:rsPrecongruencefailure}.
We write $N \red_{\pequiv} N'$ iff $N \pequiv N_1 \red N_2 \pequiv N' $, for some $N_1, N_2$. Then, $\red_{\pequiv}^*$ is the reflexive, transitive closure of $\red_{\pequiv}$. For success sensitivity,  
we decree $\piencodf{\checkmark}_u = \checkmark$. We  have:

\begin{figure}
\begin{center}
	\begin{tikzpicture}[scale=.9pt]
\draw[rounded corners, color=black,fill=RedOrange!20] (0,5.5) rectangle (14.5,6.5);
\node (opcom) at (4.0,6.8){Operational Completeness};
\draw[rounded corners, color=black, fill=RoyalBlue!20] (0,2.5) rectangle (14.5,3.5);
\node (lamrfail) at (0.7,6) {$\lamrsharfailunres$:};
\node (expr1) [right of=lamrfail, xshift=.3cm] {$\mathbb{N}$};
\node (expr2) [right of=expr1, xshift=2cm] {$\mathbb{M}$};
\draw[arrow] (expr1) --  (expr2);
\node (lamrsharfail) at (0.7,3.0) {$\spi$:};
\node (transl1) [right of=lamrsharfail, xshift=.3cm] {$\piencodf{\mathbb{N}}_u$};
\node (transl2) [right of=transl1, xshift=2cm] {$Q \equiv \piencodf{\mathbb{M}}_u$};
\draw[arrow, dotted] (transl1) -- node[anchor= south,yshift=-1.25mm] {*}(transl2);
\node (enc1) at (1.65,4.5) {$\piencodf{\cdot }_u$};
\node  at (6.0,4.5) {$\piencodf{\cdot }_u$};
\node (refcomp) [right of=enc1, xshift=.9cm]{Thm~\ref{thm:op_correct} (b)};
\draw[arrow, dashed] (expr1) -- (transl1);
\draw[arrow, dashed] (expr2) -- (transl2);
\node (opcom) at (10.65,6.8){Operational Soundness};
\node (expr1shar) [right of=expr2, xshift=1.8cm] {$\mathbb{N}$};
\node (expr2shar) [right of=expr1shar, xshift=2.8cm] {$\mathbb{N'}$};
\draw[arrow, dotted] (expr1shar) -- node[anchor= south, yshift=-1.25mm] {*} (expr2shar);
\node (transl1shar) [right of=transl2, xshift=1.8cm] {$\piencodf{\mathbb{N}}_u$};
\node (transl2shar) [right of=transl1shar, xshift=1cm] {$Q$};
\node (expr3shar) [right of=transl2shar, xshift=.8cm]{$Q'\equiv \piencodf{\mathbb{N'}}_u$};
\draw[arrow,dashed] (expr2shar) --  (expr3shar);
\draw[arrow,dashed] (expr1shar) -- (transl1shar);
\node (enc2) at (8.1,4.5) {$\piencodf{\cdot }_u$};
\node  at (13.35,4.5) {$\piencodf{\cdot }_u$};
\node (refsound) [right of=enc2, xshift=1.5cm]{Thm~\ref{thm:op_correct}(c)};
\draw[arrow] (transl1shar) -- node[anchor= south, yshift=-1.25mm] {*} (transl2shar);
\draw[arrow,dotted] (transl2shar) -- node[anchor= south, yshift=-1.25mm] {*}(expr3shar);
\node at (12.45,5.7) {\small $\pequiv$ };
\end{tikzpicture}
\end{center}
    \caption{An overview  of operational soundness and completeness for $\piencodf{\cdot }_u$.}
    \label{fig:my_label}
\end{figure} 

 \begin{theorem}[Operational Correctness] \label{thm:op_correct} Let $\expr{N} $ and $ \expr{M} $ be well-formed $\lamrsharfailunres $ closed expressions.
 \begin{enumerate}[(a)]
 \item (Type Preservation) Let $B$  be a bag. We have:
         \begin{enumerate}[(i)]
        \item If $\Theta ; \Gamma \wfdash B :(\sigma^{k} , \eta )$
        then 
        $\piencodf{B}_u \wfdash  \piencodf{\Gamma}, u : \piencodf{(\sigma^{k} , \eta )}_{(\sigma, i)} ; \piencodf{\Theta}$.
        
        \item If $\Theta ; \Gamma \wfdash \expr{M} : \tau$
        then 
        $\piencodf{\expr{M}}_u \wfdash  \piencodf{\Gamma}, u :\piencodf{\tau} ; \piencodf{\Theta}$.
        \end{enumerate}
  \item (Completeness)  If $ \expr{N}\red \expr{M}$ then there exists $Q$ such that $\piencodf{\expr{N}}_u  \red^* Q \pequiv \piencodf{\expr{M}}_u$.
 \item \label{thm:op_correct_sound} (Soundness) If $ \piencodf{\expr{N}}_u \red^* Q$
then 
$Q \red^* Q'$, $\expr{N}  \red^*_{\pequiv} \expr{N}'$ 
and 
$\piencodf{\expr{N}'}_u \equiv Q'$, 
for some $Q', \expr{N}'$.
\item (Success Sensitivity)  $\succp{\expr{M}}{\checkmark}$ if, and only if, $\succp{\piencodf{\expr{M}}_u}{\checkmark}$.
 \end{enumerate}
 \end{theorem}
 
  \iftypes
 \else
  \begin{proof}
  Below we illustrate the most interesting case of the proof of soundness. Detailed proof can be found in \appref{app:encodingtwo}.
  \end{proof}
  \fi 

\begin{proof}
 All items are proven by structural induction; a detailed proof can be  
 \iftypes 
found in~\cite{DBLP:journals/corr/abs-2111}.
\else 
found in~\appref{app:encodingtwo}.
\fi 
 Below we present the most interesting case in the proof of {\em soundness}: the case when $\expr{N} =  M (C \bagsep U) $. Then, 
  $$ \piencodf{\expr{N}}_u=\piencodf{M (C \bagsep U)}_u = \bigoplus_{C_i \in \perm{C}} (\nu v)(\piencodf{M}_v \para v.\some_{u , \llfv{C}} ; \outact{v}{x} . ([v \leftrightarrow u] \para \piencodf{C_i \bagsep U}_x ) ).$$ 
  
  The proof then proceeds by induction on the number of reduction steps $k$ that can be taken from $\piencodf{\expr{N}}_u$, i.e, $\piencodf{\expr{N}}_u \red^k Q$. We will consider the case when $k \geq 1$, where for some process $R$ and non-negative integers $n, m$ such that $k = n+m$, we have the following:
            \[
            \begin{aligned}
               \piencodf{\expr{N}}_u & \red^m  \bigoplus_{C_i \in \perm{C}} (\nu v)(R \para v.\some_{u , \llfv{C}} ; \outact{v}{x} . ([v \leftrightarrow u] \para \piencodf{C_i \bagsep U}_x ) ) \red^n  Q\\
            \end{aligned}
            \]
There are several cases to analyse depending on the values of $m$ and $n$, and the shape of $M$. We consider  $m = 0$, $n \geq 1$ and $M=(\lambda x . (M'[ {\widetilde{x}} \leftarrow  {x}])) \linexsub{N_1 / y_1} \cdots \linexsub{N_p / y_p} \unexsub{U_1 / \unvar{z_1}} \cdots \unexsub{U_q / \unvar{z_q}}$, where  $p, q \geq 0$.  Then,  $\piencodf{\expr{N}}_u$ can perform the following reduction:
\[
 \begin{aligned}
 \piencodf{\expr{N}}_u & \red^* \bigoplus_{C_i \in \perm{C}} (\nu  \widetilde{y},\widetilde{z}, x)( x.\overline{\some}; x(\linvar{x}). x(\banged{x}). x \sclose ; \piencodf{M'[ {\widetilde{x}} \leftarrow  {x}]}_u \para Q''\para  \piencodf{C_i \bagsep U}_x ) ~ (:= Q_3) \\
\end{aligned}
     \]
where $Q''$ defines the encoding of explicit substitutions within the encoded subterm $M$.
Notice that:
 \[
                    \begin{aligned}
                        \expr{N} &=(\lambda x . (M'[ {\widetilde{x}} \leftarrow  {x}])) \linexsub{N_1 / y_1} \cdots \linexsub{N_p / y_p} \unexsub{U_1 / \unvar{z_1}} \cdots \unexsub{U_q / \unvar{z_q}}(C\bagsep U)\\
                        &\pequiv (\lambda x . (M'[ {\widetilde{x}} \leftarrow  {x}]) (C \bagsep U)) \linexsub{N_1 / y_1} \cdots \linexsub{N_p / y_p} \unexsub{U_1 / \unvar{z_1}} \cdots \unexsub{U_q / \unvar{z_q}} \\
                      & \red   M'[ {\widetilde{x}} \leftarrow  {x}] \esubst{(C \bagsep U)}{x} \linexsub{N_1 / y_1} \cdots \linexsub{N_p / y_p} \unexsub{U_1 / \unvar{z_1}} \cdots \unexsub{U_q / \unvar{z_q}} = \expr{M}
                    \end{aligned}
                    \]
where the  congruence holds assuming the necessary $\alpha$-renaming of variables. 
  Finally, one can verify that  $\piencodf{\expr{M}}_u = Q_3$, and the result follows. 
    \end{proof}


\begin{example}\label{ex:opcorr}
Recall again term $M_1$  from Example~\ref{ex:id_term}.
It can be shown that $\recencodf{M_1}\red^* \recencodf{N}\unexsub{\recencodf{U}/x^!}$. 
To illustrate operational completeness, we can verify preservation of  reduction, via $\piencodf{\cdot}$: reductions below use the  rules for $\spi$ in Figure~\ref{fig:redspi}---see Figure~\ref{fig:opcorr}.
          \end{example}

\begin{figure}[t!]
\begin{mdframed}[style=alttight]
{\small 	
   \[
         \begin{aligned}
            & \piencodf{\recencodf{M_1}}  = 
            \\
            & (\nu v)(  \colthree{ v.\overline{\some}; v(x).} x.\overline{\some}; x(\linvar{x}). x(\banged{x}). x\sclose ; \linvar{x}.\overline{\some}. \outact{\linvar{x}}{y_1}.  (y_1 . \some_{\emptyset} ; y_{1}\sclose;\zero \para\\
            &\qquad  \linvar{x}.\overline{\some}; \linvar{x}.\some_{u}; \linvar{x}( {x}_1) . \linvar{x}. \overline{\some}. \outact{\linvar{x}}{y_2} . ( y_2 . \some_{u, x_1 } ;y_{2}.\sclose; \piencodf{x_1}_v \para \linvar{x}. \overline{\none}) )
            \para \\ 
            &\qquad\colthree{v.\some_{u , \llfv{\recencodf{N}}} ;\outact{v}{x}.} (\colthree{[v \leftrightarrow u] }\para   x.\some_{\llfv{\recencodf{N}}} ; \outact{x}{\linvar{x}} .(\linvar{x}.\some_{\llfv{\recencodf{N}}} ;\linvar{x}(y_1). \\
            & \qquad  \linvar{x}. \some_{y_1, \llfv{ \bag{ \recencodf{N}}}}; \linvar{x}. \overline{\some} ; \outact{\linvar{x}}{x_1}.  (x_1.\some_{\llfv{\recencodf{N}}}; \piencodf{\recencodf{N}}_{x_1} \para y_1. \overline{\none} \para  \linvar{x}.\some_{\emptyset} ;\\ 
            & \qquad     \linvar{x}(y_2). ( y_2.\overline{\some};\overline{y_2}  \sclose \para \linvar{x}.\some_{\emptyset} ; \linvar{x}. \overline{\none}) )  \para \outact{x}{\banged{x}} .( !\banged{x} (x_i). \piencodf{ \recencodf{U} }_{x_i} \para \overline{x}\sclose ) ) )
            )
            \\[1mm]
            & \red^{3}
            (\nu x)( \colthree{ x.\overline{\some}; x(\linvar{x}).} x(\banged{x}). x \sclose ; \linvar{x}.\overline{\some}. \outact{\linvar{x}}{y_1}.  (y_1 . \some_{\emptyset} ; y_{1}\sclose;\zero \para \linvar{x}.\overline{\some};\linvar{x}.\some_{u}; \\ 
            &\qquad  \linvar{x}( {x}_1). \linvar{x}. \overline{\some}. \outact{\linvar{x}}{y_2} . ( y_2 . \some_{u, x_1 } ;y_{2}\sclose; \piencodf{x_1}_u \para \linvar{x}. \overline{\none}) )
            \para  (\colthree{x.\some_{\llfv{\recencodf{N}}} ; \outact{x}{\linvar{x}}}. \\
            & \qquad  (\linvar{x}.\some_{\llfv{\bag{\recencodf{N}}}};\linvar{x}(y_1). \linvar{x}.  \some_{y_1, \llfv{ \bag{ \recencodf{N}}}}; \linvar{x}. \overline{\some} ; \outact{\linvar{x}}{x_1}. (x_1.\some_{\llfv{\recencodf{N}}} ; \piencodf{\recencodf{N}}_{x_1} \para \\
           & \qquad  y_1. \overline{\none} \para  \linvar{x}.  \some_{\emptyset} ;\linvar{x}(y_2). ( y_2.\overline{\some};\overline{y_2} \sclose \para \linvar{x}.\some_{\emptyset} ; \linvar{x}. \overline{\none}) )  \para \outact{x}{\banged{x}}.( !\banged{x} (x_i). \piencodf{ \recencodf{U} }_{x_i} \para \overline{x}\sclose ) ) )
            )
            \\[1mm]
                    & \red^{2}
            (\nu x, \linvar{x})( \colthree{x(\banged{x}).} x. \sclose ; \linvar{x}.\overline{\some}. \outact{\linvar{x}}{y_1}.  (y_1 . \some_{\emptyset} ; y_{1}\sclose;\zero  \para \linvar{x}.\overline{\some}; \linvar{x}.\some_{u}; \linvar{x}( {x}_1) .  \\ 
            &\qquad  \linvar{x}. \overline{\some}. \outact{\linvar{x}}{y_2} . ( y_2 . \some_{u, x_1 } ;y_{2}\sclose; \piencodf{x_1}_u \para \linvar{x}. \overline{\none}) )
            \para   (\linvar{x}.\some_{\llfv{\recencodf{N}}}  ;\linvar{x}(y_1).  \\ 
            & \qquad \linvar{x}.  \some_{y_1, \llfv{  \recencodf{N}}}; \linvar{x}.\overline{\some} ;\outact{\linvar{x}}{x_1}.    (x_1.\some_{\llfv{\recencodf{N}}} ;\piencodf{\recencodf{N}}_{x_1} \para y_1. \overline{\none} \para  \linvar{x}.\some_{\emptyset} ;  \linvar{x}(y_2).  \\
            & \qquad  ( y_2.\overline{\some};\overline{y_2}\sclose  \para \linvar{x}.\some_{\emptyset} ; \linvar{x}. \overline{\none}) )\para \colthree{\outact{x}{\banged{x}}.}( !\banged{x} (x_i). \piencodf{ \recencodf{U} }_{x_i} \para \overline{x}\sclose ) ) )
            )
            \\[1mm]
            & \red
            (\nu x, \linvar{x}, x^!)( \colthree{x \sclose }; \linvar{x}.\overline{\some}. \outact{\linvar{x}}{y_1}.  (y_1 . \some_{\emptyset} ; y_{1}\sclose;\zero  \para \linvar{x}.\overline{\some}; \linvar{x}.\some_{u}; \linvar{x}( {x}_1) .  \linvar{x}. \overline{\some}. \\ 
            &\qquad \outact{\linvar{x}}{y_2} .( y_2 . \some_{u, x_1 } ;y_{2}\sclose; \piencodf{x_1}_u \para \linvar{x}. \overline{\none}) )
            \para   (\linvar{x}.\some_{\llfv{\recencodf{N}}}  ;\linvar{x}(y_1). \linvar{x}. \some_{y_1, \llfv{  \recencodf{N}}}; \\ 
            & \qquad \linvar{x}. \overline{\some} ;   \outact{\linvar{x}}{x_1}. (x_1.\some_{\llfv{\recencodf{N}}} ; \piencodf{\recencodf{N}}_{x_1} \para y_1. \overline{\none} \para  \\
            &
             \qquad  \linvar{x}.\some_{\emptyset} ; \linvar{x}(y_2). ( y_2.\overline{\some};\overline{y_2}  \sclose \para 
             \linvar{x}.\some_{\emptyset} ; \linvar{x}. \overline{\none}) )  \para !\banged{x} (x_i). \piencodf{ \recencodf{U} }_{x_i} \para \colthree{\overline{x}\sclose}  ) )
            )
           \\[1mm]
             & \red
            (\nu  \linvar{x}, x^!)( \colthree{\linvar{x}.\overline{\some}. \outact{\linvar{x}}{y_1}. } (y_1 . \some_{\emptyset} ; y_{1}.\sclose;\zero  \para \linvar{x}.\overline{\some}; \linvar{x}.\some_{u}; \linvar{x}( {x}_1) . \\ 
            &\qquad  \linvar{x}. \overline{\some}. \outact{\linvar{x}}{y_2} . ( y_2 . \some_{u, x_1 } ;y_{2}.\sclose; \piencodf{x_1}_u \para \linvar{x}. \overline{\none}) )
            \para   (\colthree{\linvar{x}.\some_{\llfv{\recencodf{N}}}  ;\linvar{x}(y_1).}\\
            & \qquad \linvar{x}.  \some_{y_1, \llfv{  \recencodf{N}}}; \linvar{x}. \overline{\some} ; \outact{\linvar{x}}{x_1}. (x_1.\some_{\llfv{\recencodf{N}}} ; \piencodf{\recencodf{N}}_{x_1} \para y_1. \overline{\none} \para \\ 
            & \qquad  \linvar{x}.\some_{\emptyset} ;  \linvar{x}(y_2). ( y_2.\overline{\some};\overline{y_2}  \sclose \para \linvar{x}.\some_{\emptyset} ; \linvar{x}. \overline{\none}) )  \para  !\banged{x} (x_i). \piencodf{ \recencodf{U} }_{x_i}  ) )
            ) 
           \\[1mm]
            & \red
            (\nu  \linvar{x}, y_1,  x^!)( y_1 . \some_{\emptyset};
            y_{1}\sclose;\zero \para \linvar{x}.\overline{\some}; \linvar{x}.\some_{u}; \linvar{x}( {x}_1) .  \linvar{x}. \overline{\some}. \outact{\linvar{x}}{y_2}.   \\ 
            & \qquad ( y_2 . \some_{u, x_1 } ; y_{2}\sclose; \piencodf{x_1}_u \para \linvar{x}. \overline{\none})\para   (             \linvar{x}. \some_{y_1, \llfv{  \recencodf{N}}}; \linvar{x}. \overline{\some} ;  \outact{\linvar{x}}{x_1}. (x_1.\some_{\llfv{\recencodf{N}}} ; \\ 
            & \qquad  \piencodf{\recencodf{N}}_{x_1} \para y_1. \overline{\none} \para  \linvar{x}.\some_{\emptyset} ; \linvar{x}(y_2). ( y_2.\overline{\some};\overline{y_2}\sclose \para \linvar{x}.\some_{\emptyset} ; \linvar{x}. \overline{\none}) )  \para   !\banged{x} (x_i). \piencodf{ \recencodf{U} }_{x_i}  ) )
            \\[1mm]
        & \red^*  (\nu  x_1,  x^!)({x_1}.\overline{\some} ; [ {x_1} \leftrightarrow u]
            \para  x_1.\some_{\llfv{\recencodf{N}}} ; \piencodf{\recencodf{N}}_{x_1} \para    !\banged{x} (x_i). \piencodf{ \recencodf{U} }_{x_i})
           \\[1mm]
            & \red^* (\nu x^!)(
             \piencodf{\recencodf{N}}_{u} \para    !\banged{x} (x_i). \piencodf{ \recencodf{U} }_{x_i})
             \\
&              =\piencodf{\recencodf{N}\unexsub{\recencodf{U}/x^!}}_u
         \end{aligned}
         \]
         }
         \end{mdframed}
         \caption{Illustrating operational correspondence, following Example~\ref{ex:opcorr}.\label{fig:opcorr}}
         \end{figure}

 \section{Concluding Remarks}
\subparagraph{Summary} 
We have extended the line of work we developed in \cite{DBLP:conf/fscd/PaulusN021}, on resource $\lambda$-calculi with firm logical foundations via typed concurrent processes. We presented \lamrfailunres, a resource calculus with
 non-determinism and explicit failures, with dedicated treatment for linear and unrestricted resources. By means of examples, we illustrated the expressivity, (lazy) semantics, and design decisions underpinning   \lamrfailunres, and introduced a class of well-formed expressions based on intersection types, which includes fail-prone expressions. To bear witness to the logical foundations of \lamrfailunres, we defined and proved correct a typed encoding into the  concurrent calculus \spi, which subsumes the one in \cite{DBLP:conf/fscd/PaulusN021}.
 We plan to study key properties for \lamrfailunres (such as solvability and normalisation) by leveraging our  typed encoding into \spi.
 
\subparagraph{Related Work}
With respect to previous resource calculi, a distinctive feature of 
 \lamrfailunres  is its support of explicit failures, which may arise depending on the interplay between (i)~linear and unrestricted occurrences of variables in a term and (ii)~associated resources in the bag. This feature allows \lamrfailunres to express variants of usual $\lambda$-terms ($\mathbf{I}$, $\Delta$, $\Omega$) not expressible in other resource calculi. 


Related to \lamrfailunres is Boudol's work on a $\lambda$-calculus  in which  multiplicities can be infinite~\cite{DBLP:conf/concur/Boudol93,DBLP:conf/birthday/BoudolL00}. 
An intersection type system is used to prove {\em adequacy} with respect to a testing semantics. However, failing behaviours as well as typability are not explored. Multiplicities can be expressed in $\lamrfailunres$:  a linear resource is available $m$ times when the linear bag contains $m$ copies of it; the term fails if the corresponding number of linear variables is different from $m$.

Also related is the resource $\lambda$-calculus by Pagani and Ronchi della Rocca~\cite{PaganiR10}, which includes linear and reusable resources; the latter are available in multisets, also called bags. In their setting, $M[N^!]$ denotes an application of a term $M$ to a resource $N$ that can be  used {\em ad libitum}.  Standard terms such as {\bf I}, $\Delta$ and $\Omega$ are expressed as $\lambda x.x$, $\Delta:=\lambda x. x[x^!]$,  and $\Omega:=\Delta[\Delta^!]$, respectively;  different variants are possible but cannot express the desired behaviour. A lazy reduction semantics is based on {\em baby} and {\em giant} steps: whereas the first consume one resource at each time, the second comprises several baby steps; combinations of the use of resources (by permuting resources in bags) are considered. A (non-idempotent) intersection type system is proposed: normalisation and a characterisation of solvability are investigated. Unlike our work, encodings into the $\pi$-calculus are not explored in~\cite{PaganiR10}.


\bibliography{references}

\iftypes
\else
\newpage
\tableofcontents 
\newpage
\appendix

\section{Appendix to Section~\ref{s:lambda}}\label{appA}

\subsection{Diamond Property for \texorpdfstring{$\lamrfailunres$}{}}

\begin{proposition}[Diamond Property for \lamrfailunres]
\label{app:lambda}
     For all $\expr{N}$, $\expr{N}_1$, $\expr{N}_2$ in $\lamrfailunres$ s.t. $\expr{N} \red \expr{N}_1$, $\expr{N} \red \expr{N}_2$ { with } $\expr{N}_1 \neq \expr{N}_2$ { then } $\exists \expr{M}$ s.t. $\expr{N}_1 \red \expr{M}$, $\expr{N}_2 \red \expr{M}$.
\end{proposition}

\begin{proof}
    We give a short argument to convince the reader of this. Notice that an expression can only perform a choice of reduction steps when it is a nondeterministic sum of terms in which multiple terms can perform independent reductions. For simplicity sake we will only consider an expression $\expr{N}$ that consist of two terms where $ \expr{N} = N + M $. We also have that $N \red N'$ and $M \red M'$. Then we let $\expr{N}_1 = N' + M$ and $\expr{N}_2 = N + M'$ by the $\redlab{R:ECont}$ rules. Finally we prove that $\expr{M}$ exists by letting $\expr{M} = N' + M'$.
\end{proof}

\section{Appendix to Section~\ref{sec:lam_types}}\label{appB}

Here we prove subject reduction (SR) for \lamrfailunres (Theorem~\ref{t:lamrfailsr_unres}). It follows from two substitution lemmas: one for substituting a linear variable (Lemma~\ref{lem:subt_lem_failunres_lin}) and another for an unrestricted variable (Lemma~\ref{lem:subt_lem_failunres_un}). Proofs of both lemmas are standard, by structural induction; we give a complete proof of SR in  Theorem~\ref{t:app_lamrfailsr}.

\begin{lemma}[Linear Substitution Lemma for \lamrfailunres]
\label{lem:subt_lem_failunres_lin}
If $\Theta ; \Gamma ,  {x}:\sigma \wfdash M: \tau$, $\headf{M} =  {x}$, and $\Theta ; \Delta \wfdash N : \sigma$ 
then 
$\Theta ; \Gamma , \Delta \wfdash M \headlin{ N /  {x} }$.
\end{lemma}

\begin{proof}
By structural induction on $M$ with $\headf{M}= {x}$. There are three cases to be analyzed: 

\begin{enumerate}
\item $M= {x}$.

In this case, $\Theta ;  {x}:\sigma \wfdash  {x}:\sigma$ and $\Gamma=\emptyset$.  Observe that $ {x}\headlin{N/ {x}}=N$, since $\Theta ; \Delta\wfdash N:\sigma$, by hypothesis, the result follows.



    \item $M = M'\ B$.
    
    In this case, $\headf{M'\ B} = \headf{M'} =  {x}$, and one has the following derivation:
    
    \begin{prooftree}
         \AxiomC{\( \Theta ;\Gamma_1 ,   {x}:\sigma\wfdash M : (\delta^{j} , \eta ) \rightarrow \tau \quad \Theta ; \Gamma_2 \wfdash B : (\delta^{k} , \epsilon ) \)}
         \AxiomC{\( \eta \relunbag \epsilon \)}
            \LeftLabel{\redlab{F{:}app}}
        \BinaryInfC{\( \Theta ;\Gamma_1 , \Gamma_2 ,  {x}:\sigma \wfdash M\ B : \tau\)}
    \end{prooftree}

 where $\Gamma=\Gamma_1,\Gamma_2$, $\delta$ is a strict type and $j,k$ are non-negative  integers, possibly different.
 
 By IH, we get $\Theta ;\Gamma_1,\Delta\wfdash M'\headlin{N/ {x}}:(\delta^{j} , \eta ) \rightarrow \tau $, which gives the following derivation: 
    \begin{prooftree}
        \AxiomC{$\Theta ;\Gamma_1,\Delta\wfdash M'\headlin{N/ {x}}:(\delta^{j} , \eta ) \rightarrow \tau $}\
        \AxiomC{$\Theta ; \Gamma_2 \wfdash B : (\delta^{k} , \epsilon ) $}
        \AxiomC{\( \eta \relunbag \epsilon \)}
    	\LeftLabel{\redlab{F{:}app}}
        \TrinaryInfC{$\Theta ;\Gamma_1 , \Gamma_2 , \Delta \wfdash ( M'\headlin{ N /  {x} } ) B:\tau $}    
    \end{prooftree}
    Therefore, from \defref{def:linsubfail}, one has $\Theta ;\Gamma_1 , \Gamma_2 , \Delta \wfdash ( M'\headlin{ N /  {x} } ) B:\tau  $, and the result follows.
\item $M = M'\esubst{B }{ y}$.

In this case,  $\headf{M'\esubst{B }{ y}} = \headf{M'} =  {x}$, with $ {x} \not = y$, and one has the following derivation:

    \begin{prooftree}
            \AxiomC{\( \Theta , \banged{y} : \eta ; \Gamma_1 , \hat{y}: \delta^{k}, x:\sigma \wfdash M : \tau \quad \Theta ; \Gamma_2 \wfdash B : (\delta^{j} , \epsilon ) \)}
            \AxiomC{\( \eta \relunbag \epsilon \)}
        \LeftLabel{\redlab{F{:}ex \dash sub}}  
        \BinaryInfC{\( \Theta ; \Gamma_1, \Gamma_2,  x:\sigma \wfdash M' \esubst{B }{ y} : \tau \)}
    \end{prooftree}

 where $\Gamma=\Gamma_1,\Gamma_2$, $\delta$ is a strict type and $j,k$ are positive integers.
By IH, we get $\Theta , \banged{y} : \eta ; \Gamma_1 , \hat{y}: \delta^{k} , \Delta \wfdash M'\headlin{N/ {x}}:\tau$ and

\begin{prooftree}
        \AxiomC{\( \Theta , \banged{y} : \eta ; \Gamma_1 , \hat{y}: \delta^{k} , \Delta \wfdash M'\headlin{N/ {x}}:\tau \quad \Theta ; \Gamma_2 \wfdash B : (\delta^{j} , \epsilon ) \)}
        \AxiomC{\( \eta \relunbag \epsilon \)}
    \LeftLabel{\redlab{F{:}ex \dash sub}}  
    \BinaryInfC{\( \Theta ; \Gamma_1, \Gamma_2,  \Delta \wfdash M' \headlin{ N / {x} } \esubst{ B }{ y} : \tau  \)}
\end{prooftree}

\end{enumerate}

From \defref{def:linsubfail}, $M' \esubst{ B }{ y} \headlin{ N / {x} } = M' \headlin{ N /  {x}} \esubst{ B }{ y}$, therefore, $\Theta ; \Gamma, \Delta \wfdash (M'\esubst{ B }{ y})\headlin{ N / {x} }:\tau$ and  the result follows.
\end{proof}

\begin{lemma}[Unrestricted Substitution Lemma for \lamrfailunres]
\label{lem:subt_lem_failunres_un}
If $\Theta, \banged{x}: \eta ; \Gamma \wfdash M: \tau$, $\headf{M} = {x}[i]$, $\eta_i = \sigma $, and $\Theta ; \cdot \wfdash N : \sigma$
then 
$\Theta, \banged{x}: \eta ; \Gamma  \wfdash M \headlin{ N /  {x}[i] }$.
\end{lemma}

\begin{proof}
By structural induction on $M$ with $\headf{M}= {x}[i]$. There are three cases to be analyzed: 

\begin{enumerate}
\item $M= {x}[i]$.

In this case, 

    \begin{prooftree}
        \AxiomC{}
        \LeftLabel{ \redlab{F{:}var^{ \ell}}}
        \UnaryInfC{\( \Theta , \banged{x}: \eta;  {x}: \eta_i  \wfdash  {x} : \sigma\)}
        \LeftLabel{\redlab{F{:}var^!}}
        \UnaryInfC{\( \Theta, \banged{x}: \eta ; \cdot \wfdash  {x}[i] : \sigma\)}
    \end{prooftree}

 and $\Gamma=\emptyset$.  Observe that $ {x}[i]\headlin{N/ {x}[i]}=N$, since $\Theta, \banged{x}: \eta  ; \Gamma  \wfdash M \headlin{ N /  {x}[i] }$, by hypothesis, the result follows.

    \item $M = M'\ B$.
    
    In this case, $\headf{M'\ B} = \headf{M'} =  {x}[i]$, and one has the following derivation:
    
    \begin{prooftree}
         \AxiomC{\( \Theta, \banged{x}: \eta ;\Gamma_1 \wfdash M : (\delta^{j} , \epsilon ) \rightarrow \tau \quad \Theta, \banged{x}: \sigma ; \Gamma_2 \wfdash B : (\delta^{k} , \epsilon' )  \)}
         \AxiomC{\( \epsilon \relunbag \epsilon' \)}
            \LeftLabel{\redlab{F{:}app}}
        \BinaryInfC{\( \Theta, \banged{x}: \eta ;\Gamma_1 , \Gamma_2 \wfdash M\ B : \tau\)}
    \end{prooftree}
    
 where $\Gamma=\Gamma_1,\Gamma_2$, $\delta$ is a strict type and $j,k$ are non-negative  integers, possibly different.
 
 By IH, we get $\Theta, \banged{x}: \eta ;\Gamma_1 \wfdash M'\headlin{N/ {x}[i]}:(\delta^{j} , \epsilon ) \rightarrow \tau $, which gives the following derivation: 
    \begin{prooftree}
        \AxiomC{$\Theta , \banged{x}: \eta;\Gamma_1\wfdash M'\headlin{N/ {x}[i]}:(\delta^{j} , \epsilon ) \rightarrow \tau $}\
        \AxiomC{$\Theta , \banged{x}: \eta; \Gamma_2 \wfdash B : (\delta^{k} , \epsilon' ) $}
        \AxiomC{\( \epsilon \relunbag \epsilon' \)}
    	\LeftLabel{\redlab{F{:}app}}
        \TrinaryInfC{$\Theta , \banged{x}: \eta;\Gamma_1 , \Gamma_2  \wfdash ( M'\headlin{ N /  {x}[i] } ) B:\tau $}    
    \end{prooftree}
   From \defref{def:linsubfail}, one has $\Theta, \banged{x}: \eta ;\Gamma_1 , \Gamma_2  \wfdash ( M'\headlin{ N /  {x}[i] } ) B:\tau  $, and the result follows.
\item $M = M'\esubst{B }{ y}$.

In this case,  $\headf{M'\esubst{B }{ y}} = \headf{M'} =  {x}[i]$, with $x \not = y$, and one has the following derivation:

    \begin{prooftree}
            \AxiomC{\( \Theta , \banged{y} : \epsilon, x:\eta ; \Gamma_1 , \hat{y}: \delta^{k} \wfdash M : \tau \quad \Theta ,  x:\eta ; \Gamma_2 \wfdash B : (\delta^{j} , \epsilon' ) \)}
            \AxiomC{\( \epsilon \relunbag \epsilon' \)}
        \LeftLabel{\redlab{F{:}ex \dash sub}}  
        \BinaryInfC{\( \Theta,  x:\eta ; \Gamma_1, \Gamma_2 \wfdash M' \esubst{B }{ y} : \tau \)}
    \end{prooftree}
 where $\Gamma=\Gamma_1,\Gamma_2$, $\delta$ is a strict type and $j,k$ are positive integers.
By IH, we get $\Theta , \banged{y} : \epsilon ,  x:\eta ; \Gamma_1 , \hat{y}: \delta^{k} \wfdash M'\headlin{N/ {x}[i]}:\tau$ and 
\begin{prooftree}
        \AxiomC{\( \Theta , \banged{y} : \epsilon ,  x:\eta ; \Gamma_1 , \hat{y}: \delta^{k} , \wfdash M'\headlin{N/ {x}[i]}:\tau \quad \Theta ,  x:\sigma ; \Gamma_2 \wfdash B : (\delta^{j} , \epsilon' ) \)}
        \AxiomC{\( \epsilon \relunbag \epsilon' \)}
    \LeftLabel{\redlab{F{:}ex \dash sub}}  
    \BinaryInfC{\( \Theta ,  x:\eta ; \Gamma_1, \Gamma_2\wfdash M' \headlin{ N / {x}[i] } \esubst{ B }{ y} : \tau  \)}
\end{prooftree}
Then, $M' \esubst{ B }{ y} \headlin{ N / {x}[i] } = M' \headlin{ N /  {x}[i]} \esubst{ B }{ y}$, and  the result follows.
\end{enumerate}
\end{proof}

\begin{theorem}[SR in \lamrfailunres]
\label{t:app_lamrfailsr}
If $\Theta ; \Gamma \wfdash \expr{M}:\tau$ and $\expr{M} \red \expr{M}'$ then $\Theta ;\Gamma \wfdash \expr{M}' :\tau$.
\end{theorem}

\begin{proof} By structural induction on the reduction rules. We proceed by analysing the rule applied in $\expr{M}$. There are seven cases:

\begin{enumerate}

	\item {\bf Rule $\redlab{R:Beta}$.}
	
	Then $\expr{M} = (\lambda x . M)B \red M\ \esubst{B}{x}=\expr{M}'$.
 	Since $\Theta ; \Gamma\wfdash \expr{M}:\tau$, one has the  derivation:
	\begin{prooftree}
			\AxiomC{$\Theta , \banged{x} : \eta ; \Gamma' , \hat{ {x}}: \sigma^{j} \wfdash M : \tau $}
			\LeftLabel{\redlab{F{:}abs}}
            \UnaryInfC{$\Theta ; \Gamma' \wfdash \lambda x. M: (\sigma^{j} , \eta ) \rightarrow \tau $}
            \AxiomC{$\Theta ;\Delta \wfdash B: (\sigma^{k} , \epsilon ) $}
            \AxiomC{\( \eta \relunbag \epsilon \)}
			\LeftLabel{\redlab{F{:}app}}
		\TrinaryInfC{$\Theta ; \Gamma' , \Delta \wfdash (\lambda x. M) B:\tau $}
	\end{prooftree}
	for $\Gamma = \Gamma' , \Delta $. Notice that

    \begin{prooftree}
            \AxiomC{\( \Theta , \banged{x} : \eta ; \Gamma' , \hat{ {x}}: \sigma^{j} \wfdash M : \tau \quad \Theta ;\Delta \wfdash B: (\sigma^{k} , \epsilon ) \)}
            \AxiomC{\( \eta \relunbag \epsilon \)}
        \LeftLabel{\redlab{F{:}ex \dash sub}}  
        \BinaryInfC{\( \Theta ; \Gamma, \Delta \wfdash M \esubst{ B }{ x } : \tau \)}
    \end{prooftree}
    
    Therefore, $ \Theta ; \Gamma \wfdash \expr{M}':\tau$ and the result follows.

    \item {\bf  Rule $\redlab{R:Fetch^{\ell}}$.}
    
    Then $ \expr{M} = M\ \esubst{ C \bagsep U  }{x }$, where $C = {\bag{N_1}}\cdot \dots \cdot {\bag{N_k}}$ , $k\geq 1$, $ \#( {x},M) = k $ and $\headf{M} =  {x}$. The reduction is as: 
     \begin{prooftree}
    \AxiomC{$\headf{M} =  {x}$}
    \AxiomC{$C = {\bag{N_1}}\cdot \dots \cdot {\bag{N_k}} \ , \ k\geq 1 $}
    \AxiomC{$ \#( {x},M) = k $}
    \LeftLabel{\redlab{R:Fetch^{\ell}}}
    \TrinaryInfC{\(
    M\ \esubst{ C \bagsep U  }{x } \red M \headlin{ N_{1}/ {x} } \esubst{ (C \setminus N_1)\bagsep U}{ x }  + \cdots + M \headlin{ N_{k}/x } \esubst{ (C \setminus N_k)\bagsep U}{x}
    \)}
\end{prooftree}

    To simplify the proof we take $k=2$, as the case $k>2$ is similar. Therefore, $c = {\bag{N_1}}\cdot {\bag{N_2}}$ 
    
    \begin{prooftree}
            \AxiomC{\( \Theta , \banged{x} : \eta ; \Gamma' , \hat{x}: \sigma^{2} \wfdash M : \tau \)}
               \AxiomC{\(  \Theta ; \cdot\wfdash  U : \epsilon\)}
                \AxiomC{\(\Pi\)}
                \noLine
                \UnaryInfC{\( \Theta ; \Delta\wfdash {\bag{N_1}}\cdot {\bag{N_2}} : \sigma^2\)}
            \LeftLabel{\redlab{F{:}bag}}
            \BinaryInfC{\( \Theta ; \Delta \wfdash C \bagsep U : (\sigma^{2} , \epsilon ) \)}
            \AxiomC{\(  \eta \relunbag \epsilon  \)}
        \LeftLabel{\redlab{F{:}ex \dash sub}}  
        \TrinaryInfC{\( \Theta ; \Gamma', \Delta \wfdash M \esubst{ C \bagsep U }{ x } : \tau \)}
    \end{prooftree}
      with  $\Pi$ the derivation
    \begin{prooftree}
            \AxiomC{\( \Theta ; \Delta_1 \wfdash N_1 : \sigma\)}
                \AxiomC{\( \Theta ; \Delta_2 \wfdash N_2 : \sigma\)}
                \AxiomC{\(  \)}
                \LeftLabel{\redlab{F{:}\oneb^{\ell}}}
                \UnaryInfC{\( \Theta ; \dash \wfdash {\oneb} : \omega \)}
            \LeftLabel{\redlab{F{:}bag^{ \ell}}}
            \BinaryInfC{\( \Theta ; \Delta_2 \wfdash {\bag{N_2}} : \sigma\)}
        \LeftLabel{\redlab{F{:}bag^{ \ell}}}
        \BinaryInfC{\( \Theta ; \Delta \wfdash {\bag{N_1}}\cdot {\bag{N_2}} : \sigma^2\)}
    \end{prooftree}
     where $\Delta= \Delta_1,\Delta_2$ and $\Gamma = \Gamma' , \Delta $. By Lemma~\ref{lem:subt_lem_failunres_lin}, there exist derivations $\Pi_1$ of  $
    \Theta , \banged{x} : \eta ; \Gamma' , x: \sigma,  \Delta_1 \wfdash   M \headlin{ N_{1}/  {x} } : \tau $ and  $\Pi_2$ of $\Theta , \banged{x} : \eta ; \Gamma' , x: \sigma,  \Delta_2 \wfdash   M \headlin{ N_{2}/  {x} } : \tau $. Therefore, one has the following derivation where we omit the second case of the sum:

    \begin{prooftree}
    				\AxiomC{\( \Pi_1 \) }
                            \AxiomC{\(  \Theta ; \cdot\wfdash  U : \epsilon\)}
                        \AxiomC{\( \Theta ; \Delta\wfdash  {\bag{N_2}} : \sigma\)}
                    \LeftLabel{\redlab{F{:}bag}}
                    \BinaryInfC{\( \Theta ; \Delta \wfdash \bag{N_2} \bagsep U : (\sigma , \epsilon ) \)}
                \LeftLabel{\redlab{F{:}ex \dash sub}}
    			\BinaryInfC{\( \Theta ; \Gamma' ,  \Delta_1   \wfdash M \headlin{ N_{1}/x } \esubst{ {\bag{N_2}} \bagsep U}{ x }  : \tau\ \) }
    			\AxiomC{\( \vdots \) }
    	\LeftLabel{\redlab{F{:}sum}}
        \BinaryInfC{\(\Theta ; \Gamma' , \Delta  \wfdash     M \headlin{ N_{1}/x } \esubst{ {\bag{N_2}} \bagsep U }{ x } +  M \headlin{ N_{2}/x } \esubst{ {\bag{N_1}} \bagsep U }{x } : \tau\)}
    \end{prooftree}
    
    Assuming  $ \expr{M}'  =    M \headlin{ N_{1}/x } \esubst{ {\bag{N_2}} \bagsep \banged{B} }{ x } +  M \headlin{ N_{2}/x } \esubst{ {\bag{N_1}} \bagsep \banged{B} }{x }$, the result follows.

    \item {\bf Rule $\redlab{R:Fetch^!}$.}
    
    Then $ \expr{M} = M\ \esubst{ C \bagsep U }{x }$, where $U = \banged{\bag{N_1}} \concat \cdots \concat \banged{\bag{N_l}} $ and $\headf{M} =  {x}[i]$. The reduction is as:

    \begin{prooftree}
        \AxiomC{$\headf{M} =  {x}[i]$}
        \AxiomC{$ U_i = \banged{\bag{N_i}} $}
        \LeftLabel{\redlab{R:Fetch^!}}
        \BinaryInfC{\(
        M\ \esubst{ C \bagsep U  }{x } \red M \headlin{ N_i/ {x}[i] } \esubst{ C \bagsep U}{ x } 
        \)}
    \end{prooftree}
    
     By hypothesis, one has the derivation:
     \begin{prooftree}
            \AxiomC{\( \Theta , \banged{x} : \eta ; \Gamma' , \hat{x}: \sigma^{j} \wfdash M : \tau \)}
               \AxiomC{$\Pi$}
               \noLine
               \UnaryInfC{\(  \Theta ; \cdot\wfdash U : \epsilon\)}
                \AxiomC{\( \Theta ; \Delta \wfdash  {C} : \sigma^k\)}
            \LeftLabel{\redlab{F{:}bag}}
            \BinaryInfC{\( \Theta ; \Delta \wfdash C \bagsep U : (\sigma^{k} , \epsilon ) \)}
            \AxiomC{\(  \eta \relunbag \epsilon  \)}
        \LeftLabel{\redlab{F{:}ex \dash sub}}  
        \TrinaryInfC{\( \Theta ; \Gamma', \Delta \wfdash M \esubst{ C \bagsep U }{ x } : \tau \)}
    \end{prooftree}
    Where $\Pi$ has the form
    \begin{prooftree}
            \AxiomC{\( \Theta ; \cdot \wfdash N_1 : \epsilon_1\)}
            \LeftLabel{\redlab{F{:}bag^!}}
            \UnaryInfC{\( \Theta ; \cdot \wfdash \banged{\bag{N_1}}  : \epsilon_1\)}
            \AxiomC{\( \cdots \)}
            \AxiomC{\( \Theta ; \cdot \wfdash N_l : \epsilon_l \)}
            \LeftLabel{\redlab{F{:}bag^!}}
            \UnaryInfC{\( \Theta ; \cdot \wfdash \banged{\bag{N_l}} : \epsilon_l \)}
        \LeftLabel{\redlab{F{:}\concat bag^!}}
        \TrinaryInfC{\( \Theta ; \cdot  \wfdash \banged{\bag{N_1}} \concat \cdots \concat \banged{\bag{N_l}}  :\epsilon \)}
    \end{prooftree}
    where $\Gamma = \Gamma' , \Delta $. Notice that if $\epsilon_i = \delta$ and $\eta \relunbag \epsilon$ then $\eta_i = \delta$ By Lemma~\ref{lem:subt_lem_failunres_un}, there exists a derivation $\Pi_1$ of  $
    \Theta , \banged{x} : \eta ; \Gamma' , \hat{x}: \sigma^{j} \wfdash   M \headlin{ N_{i}/  {x}[i] } : \tau $. Therefore, one has the following derivation:

         \begin{prooftree}
            \AxiomC{\( \Theta , \banged{x} : \eta ; \Gamma' , \hat{x}: \sigma^{j} \wfdash M \headlin{ N_{1}/  {x}[i] } : \tau \)}
               \AxiomC{\(  \Theta ; \cdot\wfdash U : \epsilon\)}
                \AxiomC{\( \Theta ; \Delta \wfdash  {C} : \sigma^k\)}
            \LeftLabel{\redlab{F{:}bag}}
            \BinaryInfC{\( \Theta ; \Delta \wfdash C \bagsep U : (\sigma^{k} , \epsilon ) \)}
            \AxiomC{\(  \eta \relunbag \epsilon  \)}
        \LeftLabel{\redlab{F{:}ex \dash sub}}  
        \TrinaryInfC{\( \Theta ; \Gamma', \Delta \wfdash M \headlin{ N_{i}/ {x[i]} } \esubst{  C \bagsep U }{ x } : \tau \)}
    \end{prooftree}

	\item  {\bf Rule $\redlab{R:Fail^{ \ell }}$.}
	
	Then $\expr{M} =  M\ \esubst{C \bagsep U}{x } $ where $\#( {x},M) \neq \size{C}$ and we can perform the reduction:
	\begin{prooftree}
        \AxiomC{$\#( {x},M) \neq \size{C}$} 
        \AxiomC{\( \widetilde{y} = (\mlfv{M} \setminus x) \uplus \mlfv{C} \)}
        \LeftLabel{\redlab{R:Fail^{ \ell }}}
        \BinaryInfC{\(  M\ \esubst{C \bagsep U}{x } \red {}  \sum_{\perm{C}} \fail^{\widetilde{y}} \)}
    \end{prooftree}
	
	with $\expr{M}'=\sum_{\perm{B}} \fail^{\widetilde{y}}$. By hypothesis, one has the derivation:
    
    \begin{prooftree}
            \AxiomC{\( \Theta , \banged{x} : \eta ; \Gamma' , \hat{x}: \sigma^{2} \wfdash M : \tau \)}
            \AxiomC{\( \Theta ; \Delta \wfdash C \bagsep U : (\sigma^{2} , \epsilon ) \)}
            \AxiomC{\(  \eta \relunbag \epsilon  \)}
        \LeftLabel{\redlab{F{:}ex \dash sub}}  
        \TrinaryInfC{\( \Theta ; \Gamma', \Delta \wfdash M \esubst{ C \bagsep U }{ x } : \tau \)}
    \end{prooftree}
    From $\#(x,M) \neq \size{B}$  we have that $j\neq k$. Hence $\Gamma = \Gamma' , \Delta $ and we type the following:
    
    \begin{prooftree}
        \AxiomC{\( \)}
        \LeftLabel{\redlab{F{:}fail}}
        \UnaryInfC{$\Theta ; \Gamma \wfdash \fail^{\widetilde{y}} : \tau$}
        \AxiomC{\( \cdots \)}
        \AxiomC{\( \)}
        \LeftLabel{\redlab{F{:}fail}}
        \UnaryInfC{$\Theta ;\Gamma \wfdash \fail^{\widetilde{y}} : \tau$}
        \LeftLabel{\redlab{F{:}sum}}
        \TrinaryInfC{$\Theta ; \Gamma \wfdash \sum_{\perm{B}} \fail^{\widetilde{y}}: \tau$}
    \end{prooftree}

\item {\bf Rule $ \redlab{R:Fail^!}$.}

    Then $\expr{M} =  M\ \esubst{C \bagsep U}{x } $ where $\headf{M} =  {x}[i]$, $U_i = \banged{\oneb}$ and we can perform the reduction:
    \begin{prooftree}
        \AxiomC{$\headf{M} =  {x}[i]$}
        \AxiomC{$ U_i = \banged{\oneb} $}
        \AxiomC{\( \)}
        \LeftLabel{\redlab{R:Fail^!}}
        \TrinaryInfC{\(  M\ \esubst{C \bagsep U}{x } \red  M \headlin{ \fail^{\emptyset} / {x}[i] } \esubst{ C \bagsep U}{ x }\)}
    \end{prooftree}
        with $\expr{M}'=M \headlin{ \fail^{\emptyset} / {x}[i] } \esubst{ C \bagsep U}{ x }$. By hypothesis, one has the derivation:
    
    \begin{prooftree}
            \AxiomC{\( \Theta , \banged{x} : \eta ; \Gamma' , \hat{x}: \sigma^{j} \wfdash M : \tau \)}
            \AxiomC{\( \Theta ; \Delta \wfdash C \bagsep U : (\sigma^{k} , \epsilon ) \)}
            \AxiomC{\(  \eta \relunbag \epsilon  \)}
        \LeftLabel{\redlab{F{:}ex \dash sub}}  
        \TrinaryInfC{\( \Theta ; \Gamma', \Delta \wfdash M \esubst{ C \bagsep U }{ x } : \tau \)}
    \end{prooftree}
    where $\Gamma = \Gamma' , \Delta $. By Lemma~\ref{lem:subt_lem_failunres_un}, there is a derivation $\Pi_1$ of  $
    \Theta , \banged{x} : \eta ; \Gamma' , \hat{x}: \sigma^{j} \wfdash   M \headlin{ \fail^{\emptyset}/ {x[i]} } : \tau $. Therefore, one has the derivation: (the last rule applied is $\redlab{R:ex\dash sub}$)
    \begin{prooftree}
            \AxiomC{\( \Theta , \banged{x} : \eta ; \Gamma' , \hat{x}: \sigma^{j} \wfdash M \headlin{ \fail^{\emptyset}/ {x}[i] } : \tau \)}
               \AxiomC{\(  \Theta ; \cdot\wfdash U : \epsilon\)}
                \AxiomC{\( \Theta ; \Delta \wfdash  {B} : \sigma^k\)}
            \LeftLabel{\redlab{F{:}bag}}
            \BinaryInfC{\( \Theta ; \Delta \wfdash C \bagsep U : (\sigma^{k} , \epsilon ) \)}
            \AxiomC{\(  \eta \relunbag \epsilon  \)}
        \TrinaryInfC{\( \Theta ; \Gamma', \Delta \wfdash M \headlin{ \fail^{\emptyset} / {x}[i] } \esubst{  C \bagsep U }{ x } : \tau \)}
    \end{prooftree}

\item {\bf Rule $\redlab{R:Cons_1}$.}

Then $\expr{M} =   \fail^{\widetilde{x}} \ B $ where $B = C \bagsep U $ , $ C = \bag{N_1}\cdot \dots \cdot \bag{N_k} $ , $k \geq 0$ and we can perform the following reduction:
\begin{prooftree}
    \AxiomC{$\size{C} = k$} 
    \AxiomC{\( \widetilde{y} = \mlfv{C} \)}
    \LeftLabel{$\redlab{R:Cons_1}$}
    \BinaryInfC{\(\fail^{\widetilde{x}}\ C \bagsep U \red \sum_{\perm{C}} \fail^{\widetilde{x} \uplus \widetilde{y}} \)}
\end{prooftree}
where $\expr{M}'=\sum_{\perm{C}} \fail^{\widetilde{x} \uplus \widetilde{y}}$. By hypothesis, one has
    \begin{prooftree}
         \AxiomC{\( \)}
        \LeftLabel{\redlab{F{:}fail}}
        \UnaryInfC{\( \Theta ;\Gamma' \wfdash \fail^{\widetilde{x}}: (\sigma^{j} , \eta ) \rightarrow \tau \)}
         \AxiomC{\(  \Theta ;\Delta \wfdash B : (\sigma^{k} , \epsilon )  \)}
         \AxiomC{\( \eta \relunbag \epsilon \)}
            \LeftLabel{\redlab{F{:}app}}
        \TrinaryInfC{\( \Theta ; \Gamma', \Delta \wfdash \fail^{\widetilde{x}}\ B : \tau\)}
    \end{prooftree}

Hence $\Gamma = \Gamma' , \Delta $ and we may type the following:

    \begin{prooftree}
        \AxiomC{\( \)}
        \LeftLabel{\redlab{F{:}fail}}
        \UnaryInfC{$ \Theta ;\Gamma \wfdash \fail^{\widetilde{x} \uplus \widetilde{y}} : \tau$}
        \AxiomC{\( \cdots \)}
        \AxiomC{\( \)}
        \LeftLabel{\redlab{F{:}fail}}
        \UnaryInfC{$\Theta ;\Gamma \wfdash \fail^{\widetilde{x} \uplus \widetilde{y}} : \tau$}
        \LeftLabel{\redlab{F{:}sum}}
        \TrinaryInfC{$\Theta ; \Gamma \wfdash \sum_{\perm{C}} \fail^{\widetilde{x} \uplus \widetilde{y}}: \tau$}
    \end{prooftree}

\item {\bf Rule $\redlab{R:Cons_2}$.}

Then $\expr{M} =   \fail^{\widetilde{z}}\ \esubst{B}{x} $ where $B = \bag{N_1}\cdot \dots \cdot \bag{N_k} $ , $k \geq 1$ and 
one has the  reduction:

\begin{prooftree}
    \AxiomC{$ \#(z , \widetilde{x}) =  \size{C}\quad \widetilde{y} = \mlfv{C} $} 
    \AxiomC{\( \widetilde{y} = \mlfv{C} \)}
    \LeftLabel{$\redlab{R:Cons_2}$}
    \BinaryInfC{\( \fail^{\widetilde{x}}\ \esubst{C \bagsep U}{z}  \red \sum_{\perm{C}} \fail^{(\widetilde{x} \setminus z) \uplus\widetilde{y}}  \)}
\end{prooftree}

where $\expr{M}'=\sum_{\perm{B}} \fail^{(\widetilde{z} \setminus x) \uplus \widetilde{y}}$. By hypothesis, there exists a derivation:

     \begin{prooftree}
            \AxiomC{\( \dom{\Gamma' , \hat{x}:\sigma^{j}}=\widetilde{z} \)}
            \LeftLabel{\redlab{F{:}fail}}
            \UnaryInfC{\( \Theta , \banged{x} : \eta ; \Gamma' , \hat{x}: \sigma^{j} \wfdash M : \tau  \)}
            \AxiomC{\( \Theta ; \Delta \wfdash B : (\sigma^{k} , \epsilon ) \)}
            \AxiomC{\( \eta \relunbag \epsilon \)}
        \LeftLabel{\redlab{F{:}ex \dash sub}}  
        \TrinaryInfC{\( \Theta ; \Gamma', \Delta \wfdash \fail^{\widetilde{z}} \esubst{ B }{ x } : \tau \)}
    \end{prooftree}

Hence $\Gamma = \Gamma' , \Delta $ and we may type the following:

    \begin{prooftree}
        \AxiomC{\( \)}
        \LeftLabel{\redlab{F{:}fail}}
        \UnaryInfC{$\Theta ; \Gamma \wfdash \fail^{(\widetilde{z} \setminus x) \uplus\widetilde{y}} : \tau$}
        \AxiomC{$ \cdots $}
         \AxiomC{\( \)}
        \LeftLabel{\redlab{F{:}fail}}
        \UnaryInfC{$\Theta ; \Gamma \wfdash \fail^{(\widetilde{z} \setminus x) \uplus\widetilde{y}} : \tau$}
        \LeftLabel{\redlab{F{:}sum}}
        \TrinaryInfC{$\Theta ; \Gamma \wfdash \sum_{\perm{B}} \fail^{(\widetilde{z} \setminus x) \uplus\widetilde{y}} : \tau$}
    \end{prooftree}
    
    \item {\bf Rule $\redlab{R:TCont}$.}

Then $\expr{M} = C[M]$ and the reduction is as follows:

\begin{prooftree}
        \AxiomC{$   M \red  M'_{1} + \cdots +  M'_{l} $}
        \LeftLabel{\redlab{R:TCont}}
        \UnaryInfC{$ C[M] \red  C[M'_{1}] + \cdots +  C[M'_{l}] $}
\end{prooftree}

where $\expr{M}'= C[M'_{1}] + \cdots +  C[M'_{l}]$. 
The proof proceeds by analysing the context $C$:

    \begin{enumerate}
    \item $C=[\cdot]\ B$.
    
    In this case $\expr{M}=M \ B$, for some $B$, and the following derivation holds:

    \begin{prooftree}
         \AxiomC{\( \Theta ;\Gamma' \wfdash M: (\sigma^{j} , \eta ) \rightarrow \tau \)}
         \AxiomC{\(  \Theta ;\Delta \wfdash B : (\sigma^{k} , \epsilon )  \)}
         \AxiomC{\( \eta \relunbag \epsilon \)}
            \LeftLabel{\redlab{F{:}app}}
        \TrinaryInfC{\( \Theta ; \Gamma', \Delta \wfdash M\ B : \tau\)}
    \end{prooftree}

where $\Gamma = \Gamma' , \Delta $.
Since $ \Theta ;\Gamma' \wfdash M: (\sigma^{j} , \eta ) \rightarrow \tau $ and $M\red M_1'+\ldots + M_l'$, it follows by IH that $\Gamma'\wfdash M_1'+\ldots + M_l':(\sigma^{j} , \eta ) \rightarrow \tau $. By applying \redlab{F{:}sum}, one has $\Theta ;\Gamma' \wfdash M_i': (\sigma^{j} , \eta ) \rightarrow \tau $, for $i=1,\ldots, l$.  Therefore, we may type the following:
\begin{prooftree}
\AxiomC{\(  \forall i \in {1 , \ldots , l} \)}
         \AxiomC{\( \Theta ;\Gamma' \wfdash M_i': (\sigma^{j} , \eta ) \rightarrow \tau \)}
         \AxiomC{\(  \Theta ;\Delta \wfdash B : (\sigma^{k} , \epsilon )  \)}
         \AxiomC{\( \eta \relunbag \epsilon \)}
            \LeftLabel{\redlab{F{:}app}}
        \TrinaryInfC{\( \Theta ; \Gamma', \Delta \wfdash M_i'\ B : \tau\)}
 \LeftLabel{\redlab{F{:}sum}}
    \BinaryInfC{\( \Theta ; \Gamma', \Delta \wfdash (M'_{1}\ B) + \cdots +  (M'_{l} \ B) : \tau\)}
\end{prooftree}

Thus, $\Gamma\wfdash \expr{M'}: \tau$, and the result follows.

    \item $C=([\cdot])\esubst{B}{x}$.
    
    This case is similar to the previous.
\end{enumerate}
	\item {\bf Rule $ \redlab{R:ECont}.$}
	
Then $\expr{M} = D[\expr{M}'']$ where $\expr{M}'' \rightarrow \expr{M}'''$ then we can perform the following reduction:

\begin{prooftree}
        \AxiomC{$ \expr{M}''  \red \expr{M}'''  $}
        \LeftLabel{$\redlab{R:ECont}$}
        \UnaryInfC{$D[\expr{M}'']  \red D[\expr{M}''']  $}
\end{prooftree}

Hence $\expr{M}' =  D[\expr{M}'''] $.
The proof proceeds by analysing the context $D$ ($D= [\cdot] + \expr{N}$ or $D= \expr{N} + [\cdot]$), and follows easily by induction hypothesis.



\end{enumerate}
\end{proof}

\subsection{Examples}\label{wf:examples}

This section contains examples illustrating the constructions and results given in Section~\ref{sec:lam_types}.

\begin{example}\label{ex:bag_delta4}
The following is a wfderivation $\Pi_2$ for the bag concatenation $\bag{x}\bagsep \oneb^!:$
{\small 
        \begin{prooftree} 
                \AxiomC{}
                        \LeftLabel{\redlab{F{:}var^{ \ell}}}
                        \UnaryInfC{\(  \Theta' ; x : \sigma \wfdash x : \sigma\)}
                        
                        \AxiomC{\(  \)}
                        \LeftLabel{\redlab{F{:}\oneb^{\ell}}}
                        \UnaryInfC{\( \Theta' ; \dash \wfdash \oneb : \omega \)}
                    
                    \LeftLabel{\redlab{F{:}bag^{ \ell}}}
                    \BinaryInfC{\(  \Theta' ; x : \sigma \wfdash \bag{x}\cdot \oneb :\sigma\)}
                            
                    \AxiomC{\(  \)}
                    \LeftLabel{\redlab{F{:}\oneb^!}}
                    \UnaryInfC{\(   \Theta' ;\dash \wfdash  \banged{\oneb} : \sigma' \)}
                
                \LeftLabel{\redlab{F{:}bag}}
                \BinaryInfC{\( \Theta' ;  x : \sigma \wfdash (\bag{x} \bagsep \banged{\oneb} ) : (\sigma , \sigma' ) \)}
                
        \end{prooftree}
        }
\end{example}

\begin{example}[Cont.\ref{ex:bag_delta4}] The following is a well-formedness derivation (labels of the rules being applied are omitted) for term $\Delta_4=\lambda x. x[1] (\bag{x} \bagsep \banged{\oneb} )$:
     {\small    
    \begin{prooftree}
    \AxiomC{}
    \UnaryInfC{\( \Theta , \banged{x}: (\sigma^{j} , \eta ) \rightarrow \tau ; {x}: (\sigma^{j} , \eta ) \rightarrow \tau  \wfdash  {x} : (\sigma^{j} , \eta ) \rightarrow \tau \)}
    \UnaryInfC{\( \Theta , \banged{x}: (\sigma^{j} , \eta ) \rightarrow \tau ; \dash \wfdash {x}[1] : (\sigma^{j} , \eta ) \rightarrow \tau \)}
\AxiomC{\( \Pi_2 \)}
     \AxiomC{\( \eta \relunbag \sigma' \)}
      \TrinaryInfC{\( \Theta , \banged{x} : (\sigma^{j} , \eta ) \rightarrow \tau ;  x : \sigma \wfdash x[1] (\bag{x} \bagsep \banged{\oneb} ) : \tau \)}
    \UnaryInfC{\( \Theta ; \dash \wfdash \lambda x. (x[1] (\bag{x} \bagsep \banged{\oneb} )) : ( \sigma , (\sigma^{j} , \eta ) \rightarrow \tau )   \rightarrow \tau \)}
        \end{prooftree}
        }
\end{example}

\begin{example}\label{ex:bag_a}Below  we show the wf-derivation for the bag
$A=(\bag{x[1]} \cdot  \bag{x}  )  \bagsep \banged{\bag{x[2]}}$.

First, let $\Pi$ be the following derivation:

\begin{prooftree}
 \AxiomC{}
 \LeftLabel{\redlab{F{:}var^{ \ell}}}
  \UnaryInfC{\( \Theta , \banged{x}: \eta;   {x}: \sigma_3  \wfdash  {x} : \sigma_3\)}
 \LeftLabel{\redlab{F{:}var^!}}
\UnaryInfC{\( \Theta , \banged{x}: \eta; \dash \wfdash {x}[1] : \sigma_3\)}
\AxiomC{}
\LeftLabel{\redlab{F{:}var^{ \ell}}}
 \UnaryInfC{\( \Theta , \banged{x}:\eta ;  x: \sigma_3  \wfdash x : \sigma_3\)}
 \AxiomC{\(  \)}
 \LeftLabel{\redlab{F{:}\oneb^{ \ell}}}
\UnaryInfC{\( \Theta , \banged{x}:\eta ; \dash \wfdash \oneb : \omega  \)}
 \LeftLabel{\redlab{F{:}bag^{\ell}}}
\BinaryInfC{\( \Theta , \banged{x}:\eta ;  x: \sigma_3  \wfdash \bag{x} \cdot \oneb : \sigma_3 \)}
\LeftLabel{\redlab{F{:}bag^{\ell}}}
\BinaryInfC{\( \Theta , \banged{x}:\eta ;  x: \sigma_3 \wfdash \bag{x[1]} \cdot  \bag{x}  : \sigma_3^2  \)}
\end{prooftree}


    
From $\Pi$ we can obtain the well-formedness derivation $\Pi_A$ for $A$:

   \begin{prooftree}
  \AxiomC{\(\Pi\)}
  \noLine
  \UnaryInfC{\( \Theta , \banged{x}:\eta ;  x: \sigma_3 \wfdash \bag{x[1]} \cdot  \bag{x}   : \sigma_3^2\)}
 \AxiomC{}
 \LeftLabel{\redlab{F{:}var^{\ell}}}
 \UnaryInfC{\( \Theta , \banged{x}:\eta ; x:\sigma_2  \wfdash  x : \sigma_2 \)}
  \LeftLabel{\redlab{F{:}var^!}}
            \UnaryInfC{\( \Theta , \banged{x}:\eta;\dash \wfdash  x[2] : \sigma_2 \)}
            \LeftLabel{\redlab{F{:}bag^!}}
            \UnaryInfC{\( \Theta , \banged{x}:\eta  ;\dash \wfdash  \banged{\bag{x[2]}} : \sigma_2 \)}
        \LeftLabel{\redlab{F{:}bag}}
        \BinaryInfC{\( \Theta , \banged{x}:\eta ;  x: \sigma_3 \wfdash \underbrace{(\bag{x[1]} \cdot  \bag{x}  )  \bagsep \banged{\bag{x[2]}}}_{A}  : (\sigma_3^{2} , \sigma_2 ) \)}
    \end{prooftree}
  where $\eta = \sigma_3 \concat \sigma_2$.

    \end{example}

\begin{example}\label{ex:bag_b}
    Below we present the wf-derivation $\Pi_B$ of the bag $B=\bag{x} \bagsep \banged{\oneb}$:

    \begin{prooftree}
 \AxiomC{}
\LeftLabel{\redlab{F{:}var^{\ell}}}
\UnaryInfC{\( \Theta , \banged{x}:\eta ;  x: \sigma_3  \wfdash x : \sigma_3\)}
  \AxiomC{\(  \)}
\LeftLabel{\redlab{F{:}\oneb^{\ell}}}
  \UnaryInfC{\( \Theta , \banged{x}:\eta ; \dash \wfdash \oneb : \omega  \)}
   \LeftLabel{\redlab{F{:}bag^{\ell}}}
  \BinaryInfC{\( \Theta , \banged{x}:\eta ;  x: \sigma_3  \wfdash \bag{x} \cdot \oneb : \sigma_3^1\)}
\AxiomC{\(  \)}
 \LeftLabel{\redlab{F{:}\oneb^!}}
 \UnaryInfC{\(  \Theta , \banged{x}:\eta  ;\dash \wfdash  \banged{\oneb} : \sigma' \)}
 \LeftLabel{\redlab{F{:}bag}}
 \BinaryInfC{\( \Theta , \banged{x}:\eta ; x: \sigma_3 \wfdash (\bag{x} \bagsep \banged{\oneb} ) : (\sigma_3 , \sigma' ) \)}
 \end{prooftree}
    
\end{example}

\begin{example} To illustrate our well-formed rules, let $M$ be the following $\lamrfailunres$-term:  $$ M= \lambda x. (y ( \underbrace{(\bag{x[1]} \cdot  \bag{x}  )  \bagsep \banged{\bag{x[2]}})}_{A} \underbrace{(\bag{x} \bagsep \banged{\oneb} )}_{B}).$$
    To ease the notation $M$ is an abstraction $\lambda x. ((y A)\ B)$, where $A=(\bag{x[1]} \cdot  \bag{x}  )  \bagsep \banged{\bag{x[2]}}$  and $B=\bag{x} \bagsep \banged{\oneb}$.
      From the derivation $\Pi_A$ (Example~\ref{ex:bag_a}) we obtain the wf-derivation $\Pi_A'$ for the  application $y A$:
           \begin{prooftree}
         \AxiomC{}
        \LeftLabel{\redlab{F{:}var^{\ell}}}
        \UnaryInfC{\(  \Theta , \banged{x}:\eta ; \Delta \wfdash y : (\sigma_3^{k} , \eta'' ) \rightarrow ((\sigma_3^{j} , \eta' ) \rightarrow \tau) \)}
         \AxiomC{$\Pi_A$}
         \noLine
         \UnaryInfC{\(  \Theta , \banged{x}:\eta ;  x: \sigma_3 \wfdash A  : (\sigma_3^{2} , \sigma_2 )  \)}
         \AxiomC{\( \eta'' \relunbag \sigma_2 \)}
            \LeftLabel{\redlab{F{:}app}}
        \TrinaryInfC{\( \Theta , \banged{x}:\eta ; x: \sigma_3 , \Delta \wfdash y A: (\sigma_3^{j} , \eta' ) \rightarrow \tau  \)}
    \end{prooftree}
   where  $\eta=\sigma_3\bagsep \sigma_2$, for some list type $\eta'$ and integers $k,j$. From the premise $\eta''\relunbag \sigma_2$ it follows that $\eta'' =  \sigma_2 \concat \eta'''$ for an arbitrary $\eta'''$. From the derivation $\Pi_B$ (Example~\ref{ex:bag_b}) we obtain the well-formed derivation for term $M$:
        \begin{prooftree}
        \AxiomC{$\Pi_A'$}
        \noLine
        \UnaryInfC{\(\Theta , \banged{x}:\eta ;  x: \sigma_3, \Delta \wfdash y \ A: (\sigma^j_3,\eta') \to \tau \)}
        \AxiomC{$\Pi_B$}
        \noLine
        \UnaryInfC{\(\Theta , \banged{x}:\eta ; x:\sigma_3\wfdash B:(\sigma_3,\sigma')\)}
        \AxiomC{\( \eta' \relunbag \sigma'\)}
        \LeftLabel{$\redlab{F:app}$}
        \TrinaryInfC{\(\Theta , \banged{x}:\eta ;  x: \sigma_3, x: \sigma_3,\Delta\wfdash (y A  ) B  : \tau  \qquad x\notin \dom{\Delta} \)}
        \LeftLabel{\redlab{F{:}abs}}
        \UnaryInfC{\( \Theta ;  \Delta  \wfdash  \lambda x. ((y  A)B)   : (\sigma_3^{2} , \eta )   \rightarrow \tau \)}
    \end{prooftree}
    where $\Delta = y : (\sigma_3^{k} , \eta'' ) \rightarrow ((\sigma_3^{j} , \eta' ) \rightarrow \tau)$. From the premise $\eta' \relunbag \sigma'$  we obtain that $\eta' = \sigma' \concat \eta'''' $, where $\sigma'$ is an arbitrary strict type and $\eta''''$ is an arbitrary list type.
    \end{example}

\section{Appendix to Subsection~\ref{s:pi}}\label{appC}

 \begin{definition}[Structural Congruence]
 Structural congruence 
is defined as the least congruence relation on processes such that:
\[
\begin{array}{c}
\begin{array}{c@{\hspace{1.5cm}}c@{\hspace{1.5cm}}c}
    P \equiv_\alpha Q \Rightarrow P  \equiv Q
& 
P \para \zero  \equiv  P
&
P \para Q \equiv Q \para P 
\\
(\nu x)\zero  \equiv \zero
&
(P \para Q) \para R  \equiv P \para (Q \para R)
&
[x \leftrightarrow y]
 \equiv
[y \leftrightarrow x]
\end{array}
\\
\begin{array}{c@{\hspace{1.5cm}}c}
 x \not \in \fn{P} \Rightarrow ((\nu x )P) \para Q  \equiv (\nu x)(P \para Q) 
 &
 (\nu x)(\nu y)P  \equiv (\nu y)(\nu x)P
 \\
P \oplus (Q \oplus R) \equiv (P \oplus Q) \oplus R
&
P \oplus Q \equiv Q \oplus P
\\
(\nu x)(P \para (Q \oplus R))  \equiv (\nu x)(P \para Q) \oplus (\nu x)(P \para R)
&
\zero \oplus \zero  \equiv  \zero
\end{array}
\end{array}
\]
\end{definition}






\section{Appendix to Subsection~\ref{ssec:lamshar}} 
\label{app:ssec:lamshar}
We need a few auxiliary notions to formalize reduction for \lamrsharfailunres.

\begin{definition}[Head]
We amend Definition \ref{def:headfailure} for the case of terms in \lamrsharfailunres:
\[
\begin{array}{l}
\begin{array}{l@{\hspace{3cm}}l}
\headf{ {x}}  =  {x}   & \headf{{x}[i]}  = {x}[i] 
\\
\headf{M\ B}  = \headf{M} & \headf{\lambda x . (M[  {\widetilde{x}} \leftarrow  {x} ])}  = \lambda x . (M[  {\widetilde{x}} \leftarrow  {x} ])\\
\headf{M \linexsub{N / {x}}}  = \headf{M} & \headf{M \unexsub{U / \unvar{x}}}  = \headf{M}
\end{array}\\
\headf{(M[ {\widetilde{x}} \leftarrow  {x}])\esubst{ B }{ x }} = (M[ {\widetilde{x}} \leftarrow  {x}])\esubst{ B }{ x } \hspace{1cm} \headf{\fail^{\widetilde{x}}}  = \fail^{\widetilde{x}}
\\
\headf{M[ {\widetilde{x}} \leftarrow  {x}]} = 
\begin{cases}
     {x} & \text{If $\headf{M} = y \text{ and } y \in  {\widetilde{x}}$}\\
    \headf{M} & \text{Otherwise}
\end{cases}
\\
\end{array}
\]

\end{definition}
\begin{definition}[Linear Head Substitution]\label{def:headlinfail}
Given an $M$ with $\headf{M} = x$, the linear substitution of a term $N$ for the head variable $x$ of the term $M$, written $M\headlin{ N / x}$ is inductively defined as:
\small{
\begin{align*}
x \headlin{ N / x}   &= N & 
(M\ B)\headlin{ N/x }  = (M \headlin{ N/x })\ B\\
(M \unexsub{U / \unvar{y}} ) \headlin{ N/x } &= (M\headlin{ N/x })\ \unexsub{U / \unvar{y}}  & x \not = y\\
(M \linexsub{L / {y}} ) \headlin{ N/x } &= (M\headlin{ N/x })\ \linexsub{L / {y}}  & x \not = y\\
((M[ {\widetilde{y}} \leftarrow  {y}])\esubst{ B }{  { {y}} })\headlin{ N/x } &= (M[\widetilde{y} \leftarrow  {y}]\headlin{ N/x })\ \esubst{ B }{ y }  
& x \not = y \\
(M[ {\widetilde{y}} \leftarrow  {y}]) \headlin{ N/x } &=  (M\headlin{ N/x }) [ {\widetilde{y}} \leftarrow  {y}] & x \not = y
\end{align*}
}
\end{definition}

Following \defref{def:context_lamrfail}, we define contexts  for terms and expressions.
While expression contexts are as in \defref{def:context_lamrfail}; the term contexts for \lamrsharfailunres involve  explicit linear and unrestricted substitutions, rather than an explicit substitution: this is due to the reduction strategy we have chosen to adopt, as we always wish to evaluate explicit substitutions first. 
We assume that the terms that fill in the holes respect the conditions on explicit linear substitutions (i.e., variables appear in a term only once, shared variables must occur in the context), similarly for explicit unrestricted substitutions.

\begin{definition}[Evaluation Contexts]
Contexts for terms and expressions are defined by the following grammar:
{\small
\[
\begin{array}{rl}
     C[\cdot] ,  C'[\cdot] &::=([\cdot])B \mid ([\cdot])\linexsub{N /  {x}}  \mid ([\cdot])\unexsub{U / \unvar{x}} \mid ([\cdot])[ {\widetilde{x}} \leftarrow  {x}]\mid ([\cdot])[ \leftarrow  {x}]\esubst{\oneb}{ x} \\
 D[\cdot] , D'[\cdot] & ::= M + [\cdot] \mid [\cdot] + M
\end{array}
\]}
The result of replacing a hole with a \lamrsharfailunres-term $M$ in a context $C[\cdot]$, denoted with $C[M]$, has to be a term in $\lamrsharfailunres$. 
\end{definition}

This way, e.g., the hole in context $C[\cdot ]= ([\cdot])\linexsub{N/ {x}}$ cannot be filled with $y$, since  $C[y]= (y)\linexsub{N/ {x}}$ is not a well-defined term. 
Indeed, $M\linexsub{N/ {x}}$ requires that $x$ occurs exactly once within $M$.
Similarly, we cannot fill the hole with $\fail^{z}$ with $z\neq x$, since $C[\fail^{z}]= (\fail^{z})\linexsub{N/ {x}}$ is also not a well-defined term, for the same reason.

\begin{figure}[!t]
\begin{mdframed}
\small
   \centering
   \[
   \begin{array}{ll}
   \begin{array}{l}
   \llfv{ {x}}  = \{  {x} \} \\
 \llfv{{x}[i]}  = \emptyset \\
 \llfv{ {\oneb}}  = \emptyset \\
      \llfv{ {\bag{M}}}  = \llfv{M}  \\
 \llfv{\banged{\bag{M}}}  = \llfv{M} \\
 \llfv{C \bagsep U}  = \llfv{C} \\
 \llfv{ {\bag{M}}  \cdot C} = \llfv{ {M}} \cup \llfv{ C}  
   \end{array}
        &  
    \begin{array}{l}
    \llfv{M\ B}  =  \llfv{M} \cup \llfv{B} \\
\llfv{\lambda x . M[ {\widetilde{x}} \leftarrow  {x}]}  = \llfv{M[ {\widetilde{x}} \leftarrow  {x}]}\!\setminus\! \{  {x} \}\\
    \llfv{M[ {\widetilde{x}} \leftarrow  {x}] \esubst{B}{x}}  = (\llfv{M[ {\widetilde{x}} \leftarrow  {x}]}\setminus \{  {x} \}) \uplus \llfv{B}  \\
    \llfv{M \linexsub{N /  {x}}} = \llfv{ {M}} \cup \llfv{ N} \\
    \llfv{M \unexsub{U / \unvar{x}}} = \llfv{ {M}} \\
    \llfv{\expr{M}+\expr{N}}  = \llfv{\expr{M}} \cup \llfv{\expr{N}} \\
    \llfv{\fail^{ {x}_1, \cdots ,  {x}_n}}  = \{  {x}_1, \ldots ,  {x}_n \}
        \end{array}
   \end{array}
   \]
 \end{mdframed}
    \caption{Free Variables for \lamrsharfailunres.}
    \label{fig:fvarsfail}
     \vspace{-2mm}
\end{figure}

\subsection*{Operational Semantics}
As in \lamrfailunres, the reduction relation $\red$ on \lamrsharfailunres operates lazily on expressions; it is defined by the rules in \figref{fig:share-reductfailureunres}, and relies on a notion of linear free variables given in \figref{fig:fvarsfail}.

\begin{figure*}[!t]
\centering
  \begin{prooftree}
    \AxiomC{$\raisebox{17.0pt}{}$}
    \LeftLabel{\redlab{RS{:}Beta}}
    \UnaryInfC{\(  (\lambda x . (M[ {\widetilde{x}} \leftarrow  {x}])) B  \red (M[ {\widetilde{x}} \leftarrow  {x}])\esubst{ B }{ x } \)}
 \end{prooftree}

 \begin{prooftree}
     \AxiomC{$ \headf{M} =  {x}$}
     \LeftLabel{\redlab{RS{:}Fetch^{\ell}}}
     \UnaryInfC{\(  M \linexsub{N /  {x}} \red  M \headlin{ N/ {x} }  \)}
\DisplayProof\hfill
     \AxiomC{$ \headf{M} = {x}[i]$}
     \AxiomC{$ U_i = \banged{\bag{N}}$}
     \LeftLabel{\redlab{RS{:} Fetch^!}}
     \BinaryInfC{\(  M \unexsub{U / \unvar{x}} \red  M \headlin{ N /{x}[i] }\unexsub{U / \unvar{x}} \)}
 \end{prooftree}

 \begin{prooftree}
    \AxiomC{$ C = \bag{M_1}
    \cdots  \bag{M_k} $}
    \AxiomC{$ M \not= \fail^{\widetilde{y}} $}
    \LeftLabel{\redlab{RS{:}Ex \dash Sub}}
    \BinaryInfC{\( \!M[x_1 , \cdots , x_k \leftarrow  {x}]\esubst{ C \bagsep U }{ x } \red \displaystyle \sum_{C_i \in \perm{C}}M\linexsub{C_i(1)/ {x_1}} \cdots \linexsub{C_i(k)/ {x_k}} \unexsub{U / \unvar{x} }   \)}
\end{prooftree}

\begin{prooftree}
     \AxiomC{$ k \neq \size{C} $} 
     \AxiomC{$  \widetilde{y} = (\llfv{M} \setminus \{  \widetilde{x}\} ) \cup \llfv{C} $}
    \LeftLabel{\redlab{RS{:}Fail^{\ell}}}
    \BinaryInfC{\(  M[x_1 , \cdots , x_k \leftarrow  {x}]\ \esubst{C \bagsep U}{ x }  \red \displaystyle \sum_{C_i \in \perm{C}}  \fail^{\widetilde{y}} \)}
    \DisplayProof
\hfill
         \AxiomC{$\headf{M} = {x}[i]$}
    \AxiomC{$ U_i = \banged{\oneb} $}
    \noLine
    \UnaryInfC{\( \widetilde{y} = \llfv{M} \)}
    \LeftLabel{\redlab{RS{:}Fail^!}}
    \BinaryInfC{\(  M \unexsub{U / \unvar{x} } \red 
    M \headlin{ \fail^{\emptyset} /{x}[i] } \unexsub{U / \unvar{x} }  \)}
\end{prooftree}

\begin{prooftree}
    \AxiomC{\( \widetilde{y} = \llfv{C} \)}
    \LeftLabel{$\redlab{RS{:}Cons_1}$}
    \UnaryInfC{\(\fail^{\widetilde{x}}\ C \bagsep U \red \displaystyle \sum_{\perm{C}} \fail^{\widetilde{x} \uplus \widetilde{y}}  \)}
 \DisplayProof
\hfill
    \AxiomC{\(  \size{C} =   |  {\widetilde{x}} |   \)} 
    \AxiomC{\(  \widetilde{z} = \llfv{C} \)}
    \LeftLabel{$\redlab{RS{:}Cons_2}$}
    \BinaryInfC{\(  (\fail^{ {\widetilde{x}} \uplus \widetilde{y}} [ {\widetilde{x}} \leftarrow  {x}])\esubst{ C \bagsep U }{ x }  \red \displaystyle \sum_{\perm{C}} \fail^{\widetilde{y} \uplus \widetilde{z}}  \)}
\end{prooftree}

\begin{prooftree}
    \AxiomC{\( \widetilde{z} = \llfv{N} \)}
    \LeftLabel{$\redlab{RS{:}Cons_3}$}
    \UnaryInfC{\( \fail^{\widetilde{y}\cup x} \linexsub{N/  {x}} \red  \fail^{\widetilde{y} \cup \widetilde{z}}  \)}
\DisplayProof
\hfill
    \AxiomC{\(  \)}
    \LeftLabel{$\redlab{RS{:}Cons_4}$}
    \UnaryInfC{\( \fail^{\widetilde{y}} \unexsub{U / \unvar{x}}  \red  \fail^{\widetilde{y}}  \)}
\end{prooftree}

    \caption{Reduction Rules for \lamrsharfailunres (contextual rules omitted)}
    \label{fig:share-reductfailureunres}
\end{figure*}

 As expected, rule \redlab{RS:Beta} results into an explicit substitution $M[\widetilde{x}\leftarrow x]\esubst{ B }{ x }$, where  $B=C\bagsep U$ is a bag with a linear part $C$ and an unrestricted part $U$.  
 
 In the case $|\widetilde{x}|=k=\size{C}$ and  $M\neq \fail^{\widetilde{y}}$, this explicit substitution  expands into a sum of terms involving explicit linear and unrestricted substitutions $\linexsub {N /x}$ and $\unexsub{U/\unvar{x}}$, which are the ones to reduce into a  head substitution, via rule \redlab{RS{:}Ex\dash Sub}. Intuitively, rule~\redlab{RS{:}Ex \dash Sub} ``distributes'' an  explicit substitution into a sum of terms involving explicit linear substitutions; it considers all possible permutations of the elements in the bag among all shared variables. Explicit linear/unrestricted substitutions evolve either into a head substitution $\headlin{ N / x}$ (with $N \in B$), via rule $\redlab{RS:Fetch^\ell}$, or $\headlin{N/x}\unexsub{U/\unvar{x}}$ (with $U \in B$) via rule $\redlab{RS:Fetch^!}$, depending on whether the head of the term is a linear or an unrestricted variable. 
 
 In the case $|\widetilde{x}|=k\neq \size{C}$ or  $M=\fail^{\widetilde{y}}$, the term $M[\widetilde{x}\leftarrow x]\esubst{ B }{ x }$ will be a redex of either rule $\redlab{RS:Fail^\ell}$  or $\redlab{RS:Cons_2}$. The latter has a side condition $|\widetilde{x}|=\size{C}$, because we want to give priority for application of $\redlab{RS:Fail^\ell}$ when there is a mismatch of linear variables and the number of linear resources.
 Rule $\redlab{RS:Fail^!}$  applies to an unrestricted substitution $M\unexsub{U/\unvar{x}}$  when the head of $M$ is an unrestricted variable, say $x[i]$, that aims to consume the $i$-th component of the bag $U$ which is empty, i.e., $U_i=1^!$; then the term reduces to a term where all the head of $M$ is substituted by  $\fail^\emptyset$, the explicit unrestricted substitution is not consumed and continues in the resulting term. Consuming rules $\redlab{RS:Cons_1}, \redlab{RS:Cons_3}$ and $\redlab{RS:Cons_4}$ the term $\fail$ consume either a bag, or an explicit linear substitution, or an explicit unrestricted substitution, respectively.
 
Notice that the left-hand sides of the reduction rules in $\lamrsharfailunres$  do not interfere with each other. 
 Similarly to $\lamrfailunres$, reduction in \lamrsharfailunres satisfies a \emph{diamond property}. 


\begin{example}
\label{ex:id_sem_intermed}

We continue to illustrate the different behaviors of the terms below w.r.t. the reduction rules for $\lamrsharfailunres$ (\figref{fig:share-reductfailureunres}):
\begin{enumerate}
    \item The case with a linear variable $x$ in which the linear bag has size one, is close to the standard {\it meaning} of applying an identity function to a term:
    \[
    \begin{aligned}
        &(\lambda x. x_1 [x_1 \leftarrow x] ) \bag{N'}\bagsep U'\red_{\redlab{R:Beta}}  x_1 [x_1 \leftarrow x] \esubst{\bag{N'}\bagsep U'}{x} \\ 
        &\red_{\redlab{RS{:}Ex \dash Sub}}  x_1 \linexsub{ N' / {x_1}} \unexsub{ U' / \unvar{x} }  \red_{\redlab{R:Fetch^{\ell}}} x_1 \headlin{ N' / x_1 } \unexsub{ U' / \unvar{x} } =N' \unexsub{U' / \unvar{x} }
    \end{aligned}
    \]

\item The case of an abstraction of one unrestricted variable that aims to consume the first element of the unrestricted bag, which fails to contain a resource in the first component. 
    \[
    \begin{aligned}
    &(\lambda x. x[1][ \leftarrow x] ) \oneb \bagsep \oneb^! \concat U'\red_{\redlab{R:Beta}} x[1][ \leftarrow x] \esubst{\oneb \bagsep \oneb^! \concat U' }{x}\red_{\redlab{RS{:}Ex \dash Sub}}x[1]\unexsub{ \oneb^! \concat U' / \unvar{x} }  \\
    &\red_{\redlab{RS{:}Fail^!}}x[1]  \headlin{ \fail^{\emptyset} /{x}[1] } \unexsub{ \oneb^! \concat U' / \unvar{x} } =\fail^{\emptyset} \unexsub{ \oneb^! \concat U' / \unvar{x} }
    \end{aligned}
    \]

  \item The case of an abstraction of one unrestricted variable that aims to consume the $i$th component of the unrestricted bag $U'$. In the case  $C' = \oneb$ and $U'_i\neq 1^!$ the reduction is:
    \[
    \begin{aligned}
    &(\lambda x. x[i][ \leftarrow x] ) C'\concat U'\red_{\redlab{R:Beta}}x[i][ \leftarrow x]\esubst{C'\concat U'}{x}\\
    &\red_{\redlab{RS{:}Ex \dash Sub}}x[i]\unexsub{ C'\concat U' / \unvar{x} }  \\
    &\red_{\redlab{RS{:} Fetch^!}}x[i] \headlin{ N' /{x}[i] } \unexsub{ C'\concat U' / \unvar{x} }=N\unexsub{ C'\concat U' / \unvar{x} }
    \end{aligned}
    \]
    where $ U'_i=\bag{N'}^!$.
    Otherwise, $U'_i=1^!$ and the reduction relies again on the size of the linear bag $C$: if  $\#(x,x[i])=\size{C'}$ the reduction ends with an application of $\redlab{R:\fail^!}$; otherwise, it ends with an application $\redlab{R:\fail^{\ell }}$.
\end{enumerate}
    
\end{example}

\subsection{Well-formedness rules for  \texorpdfstring{$\lamrsharfailunres$}{}}

Similarly to $\lamrfailunres$ we present a set  ``well-formedness'' rules for $\lamrsharfailunres$-terms, -bags and -expressions,  based on  an intersection type system for \lamrsharfailunres, defined upon strict, multiset, list, tuple types, as introduced for \lamrfailunres and presented in \figref{app_fig:wfsh_rulesunres}. Linear contexts $\Gamma,\Delta$ and unrestricted contexts $\Theta, \Upsilon$ are the same as in $\lamrfailunres$, as well as well-formedness judgements $\Theta; \Gamma \vdash \mathbb{M}:\sigma$.

\begin{definition}[Well-formedness in $\lamrsharfailunres$]
An expression $ \expr{M}$ is well formed if there exists a  $\Theta, \Gamma$ and a   $\tau$ such that $\Theta ; \Gamma \wfdash  \expr{M} : \tau $ is entailed via the rules in \figref{app_fig:wfsh_rulesunres}.
\end{definition}

\begin{figure*}[!h]
    \centering
    
\begin{prooftree}
\AxiomC{}
\LeftLabel{\redlab{FS{:}var^{\ell}}}
\UnaryInfC{\( \Theta ;  {x}: \sigma \wfdash  {x} : \sigma\)}
\DisplayProof
\hfill
\AxiomC{\( \Theta , {x}: \eta;  {x}: \eta_i , \Delta \wfdash  {x} : \sigma\)}
\LeftLabel{\redlab{FS{:}var^!}}
\UnaryInfC{\( \Theta , {x}: \eta; \Delta \wfdash {x}[i] : \sigma\)}
  \DisplayProof
\hfill
  \AxiomC{\(  \)}
\LeftLabel{\redlab{FS{:}\oneb^{\ell}}}
\UnaryInfC{\( \Theta ; \dash \wfdash \oneb : \omega \)}
\end{prooftree}
    
\begin{prooftree}
\AxiomC{\(  \)}
\LeftLabel{\redlab{FS{:}\oneb^!}}
\UnaryInfC{\( \Theta ;  \dash  \wfdash \banged{\oneb} : \sigma \)}
\DisplayProof
\hfill
\AxiomC{\( \Theta ; \Gamma  \wfdash M : \tau\)}
\LeftLabel{ \redlab{FS\!:\!weak}}
\UnaryInfC{\( \Theta ; \Gamma ,  {x}: \omega \wfdash M[\leftarrow  {x}]: \tau \)}
\end{prooftree}

\begin{prooftree}
 \AxiomC{$ \Theta ; \Gamma \wfdash \expr{M} : \sigma$}
 \noLine
 \UnaryInfC{$\Theta ; \Gamma \wfdash \expr{N} : \sigma$}
   \LeftLabel{\redlab{FS{:}sum}}
   \UnaryInfC{$ \Theta ; \Gamma \wfdash \expr{M}+\expr{N}: \sigma$}
     \DisplayProof
  \hfill
    \AxiomC{\( \Theta , {x}:\eta ; \Gamma ,  {x}: \sigma^k \wfdash M[ {\widetilde{x}} \leftarrow  {x}] : \tau \quad  {x} \notin \dom{\Gamma} \)}
        \LeftLabel{\redlab{FS{:}abs\dash sh}}
        \UnaryInfC{\(  \Theta ; \Gamma \wfdash \lambda x . (M[ {\widetilde{x}} \leftarrow  {x}])  : (\sigma^k, \eta )  \rightarrow \tau \)}
\end{prooftree}

\begin{prooftree}
         \AxiomC{\( \Theta ;\Gamma \wfdash M : (\sigma^{j} , \eta ) \rightarrow \tau \)}
         \AxiomC{\( \eta \relunbag \epsilon \)}
         \noLine
         \UnaryInfC{\(  \Theta ;\Delta \wfdash B : (\sigma^{k} , \epsilon )  \)}
        \LeftLabel{\redlab{FS{:}app}}
        \BinaryInfC{\( \Theta ; \Gamma, \Delta \wfdash M\ B : \tau\)}
   \end{prooftree}

    \begin{prooftree}
            \AxiomC{\( \Theta ; \Gamma\wfdash C : \sigma^k\)}
            \AxiomC{\(  \Theta ;\dash \wfdash  U : \eta \)}
        \LeftLabel{\redlab{FS{:}bag}}
        \BinaryInfC{\( \Theta ; \Gamma \wfdash C \bagsep U : (\sigma^{k} , \eta  ) \)}
        \DisplayProof
        \hfill 
            \AxiomC{\( \Theta ; \dash \wfdash U : \epsilon\)}
        \AxiomC{\( \Theta ; \dash \wfdash V : \eta\)}
        \LeftLabel{\redlab{FS{:}\concat-bag^{!}}}
        \BinaryInfC{\( \Theta ; \dash  \wfdash U \concat V :\epsilon \concat \eta \)}
          \end{prooftree}

    \begin{prooftree}
\AxiomC{\( \Theta ; \dash \wfdash M : \sigma\)}
        \LeftLabel{\redlab{FS{:}bag^!}}
        \UnaryInfC{\( \Theta ; \dash  \wfdash \banged{\bag{M}}:\sigma \)}
        \DisplayProof
\hfill
  \AxiomC{\( \Theta ; \Gamma \wfdash M : \sigma\)}
        \AxiomC{\( \Theta ; \Delta \wfdash C : \sigma^k\)}
        \LeftLabel{\redlab{FS{:}bag^{\ell}}}
        \BinaryInfC{\( \Theta ; \Gamma, \Delta \wfdash \bag{M}\cdot C:\sigma^{k+1}\)}
    \end{prooftree}

    \begin{prooftree}
          \AxiomC{\({x}\notin \dom{\Gamma} \quad k \not = 0\)}
        \noLine
        \UnaryInfC{\( \Theta ;  \Gamma ,  {x}_1: \sigma, \cdots,  {x}_k: \sigma \wfdash M : \tau \)}
        \LeftLabel{ \redlab{FS{:}share}}
        \UnaryInfC{\( \Theta ;  \Gamma ,  {x}: \sigma^{k} \wfdash M [  {x}_1 , \cdots ,  {x}_k \leftarrow  {x} ]  : \tau \)}  
        \DisplayProof \hfill
        \AxiomC{\( \Theta ; \Gamma  ,  {x}:\sigma \wfdash M : \tau \)}
        \noLine
        \UnaryInfC{\( \Theta ; \Delta \wfdash N : \sigma \)}
            \LeftLabel{\redlab{FS{:}Esub^{\ell}}}
        \UnaryInfC{\( \Theta ; \Gamma, \Delta \wfdash M \linexsub{N /  {x}} : \tau \)}
         \end{prooftree}
         
\begin{prooftree}
        \AxiomC{\( \Theta , {x} : \eta; \Gamma  \wfdash M : \tau \)}
        \AxiomC{\( \Theta ; \dash \wfdash U : \epsilon \)}
        \AxiomC{\( \eta \relunbag \epsilon \)}
            \LeftLabel{\redlab{FS{:}Esub^!}}
        \TrinaryInfC{\( \Theta ; \Gamma \wfdash M \unexsub{U / \unvar{x}}  : \tau \)}
    \end{prooftree}

    \begin{prooftree}
    \AxiomC{\( \Theta ; \Delta \wfdash B : (\sigma^{k} , \epsilon ) \)}
            \noLine
            \UnaryInfC{\( \Theta , {x} : \eta ; \Gamma ,  {x}: \sigma^{j} \wfdash M[ {\widetilde{x}} \leftarrow  {x}] : \tau  \)}
            \AxiomC{\( \eta \relunbag \epsilon \)}
        \LeftLabel{\redlab{FS{:}Esub}}    
        \BinaryInfC{\( \Theta ; \Gamma, \Delta \wfdash (M[ {\widetilde{x}} \leftarrow  {x}])\esubst{ B }{ x }  : \tau \)}
   \DisplayProof
\hfill
        \AxiomC{\( \dom{\Gamma} = \widetilde{x}\)}
  \LeftLabel{\redlab{FS{:}fail}}
  \UnaryInfC{\( \Theta ; \Gamma \wfdash  \fail^{\widetilde{x}} : \tau \)}
   \end{prooftree}

    \caption{Well-Formedness Rules for \lamrsharfailunres}\label{app_fig:wfsh_rulesunres}
    \vspace{-3mm}
\end{figure*}

Well-formed rules for \lamrsharfailunres are essentially the same as the ones for $\lamrfailunres$.  Rules $\redlab{FS:abs\dash sh}$ and $\redlab{FS{:}Esub}$ are modified to  take into account  the sharing construct $[\widetilde{x}\leftarrow x]$. Rule  $\redlab{FS{:}share}$ is exclusive for $\lamrsharfailunres$ and requires, for each $i=1,\ldots, k$, the variable assignment  $x_i:\sigma$, to derive the well-formedness of $M[x_1,\ldots, x_n \leftarrow x]:\tau$ (in addition to variable assignments in $\Theta$ and $\Gamma$).

\begin{lemma}[Linear Substitution Lemma for \lamrsharfailunres]
\label{l:lamrsharfailsubsunres}
If $\Theta ; \Gamma ,  {x}:\sigma \wfdash M: \tau$, $\headf{M} =  {x}$, and $\Theta ; \Delta \wfdash N : \sigma$ 
then 
$\Gamma , \Delta \wfdash M \headlin{ N /  {x} }:\tau$.
\end{lemma}

\begin{proof}
By structural induction on $M$ with $\headf{M}=  {x}$.
There are six cases to be analyzed:
\begin{enumerate}
\item $M= {x}$

In this case, $\Theta ;  {x}:\sigma \wfdash  {x}:\sigma$ and $\Gamma=\emptyset$.  Observe that $ {x}\headlin{N/ {x}}=N$, since $\Delta\wfdash N:\sigma$, by hypothesis, the result follows.

    \item $M = M'\ B$.
    
    Then $\headf{M'\ B} = \headf{M'} =  {x}$, and the derivation is the following:
    \begin{prooftree}
        \AxiomC{$\Theta ; \Gamma_1 ,  {x}:\sigma \wfdash M': (\delta^{j} , \eta )  \rightarrow \tau$}\
        \AxiomC{$\Theta ; \Gamma_2 \wfdash B : (\delta^{k} , \epsilon ) $}
        \AxiomC{\( \eta \relunbag \epsilon \)}
    	\LeftLabel{\redlab{FS{:}app}}
        \TrinaryInfC{$\Theta ; \Gamma_1 , \Gamma_2 ,  {x}:\sigma \wfdash M'B:\tau $}    
    \end{prooftree}

    where $\Gamma=\Gamma_1,\Gamma_2$, and  $j,k$ are non-negative integers, possibly different.  Since $\Delta \vdash N : \sigma$, by IH, the result holds for $M'$, that is,
    $$\Gamma_1 , \Delta \wfdash M'\headlin{ N /  {x} }: (\delta^{j} , \eta )  \rightarrow \tau$$
    which gives the  derivation:


    \begin{prooftree}
        \AxiomC{$\Theta ; \Gamma_1 , \Delta \wfdash M'\headlin{ N /  {x} }: (\delta^{j} , \eta )  \rightarrow \tau$}\
        \AxiomC{$\Theta ; \Gamma_2 \wfdash B : (\delta^{k} , \epsilon ) $}
        \AxiomC{\( \eta \relunbag \epsilon \)}
    	\LeftLabel{\redlab{FS{:}app}}
        \TrinaryInfC{$\Theta ; \Gamma_1 , \Gamma_2 , \Delta \wfdash ( M'\headlin{ N /  {x} } ) B:\tau $}    
    \end{prooftree}
   
      From \defref{def:headlinfail},   $(M'B) \headlin{ N /  {x} } = ( M'\headlin{ N /  {x} } ) B$, 
      and the result follows.
    
    \item $M = M'[ {\widetilde{y}} \leftarrow  {y}] $.
    
    Then $ \headf{M'[ {\widetilde{y}} \leftarrow  {y}]} = \headf{M'}= {x}$, for  $y\neq x$. Therefore, 
    \begin{prooftree}
        \AxiomC{\(\Theta ; \Gamma_1 ,  {y}_1: \delta, \cdots,  {y}_k: \delta ,  {x}: \sigma \wfdash M' : \tau \quad  {y}\notin \Gamma_1 \quad k \not = 0\)}
        \LeftLabel{ \redlab{FS{:}share}}
        \UnaryInfC{\( \Theta ; \Gamma_1 ,  {y}: \delta^k,  {x}: \sigma \wfdash M'[ {y}_1 , \cdots ,  {y}_k \leftarrow x] : \tau \)}
    \end{prooftree}
    where $\Gamma=\Gamma_1 ,  {y}: \delta^k$. 
    By IH, the result follows for $M'$, that is, 
    $$\Theta ; \Gamma_1 ,  {y}_1: \delta, \cdots,  {y}_k: \delta ,\Delta \wfdash M'\headlin{N/ {x}} : \tau $$
    
    and we have the derivation:
    
    \begin{prooftree}
        \AxiomC{\( \Theta ; \Gamma_1 ,  {y}_1: \delta, \cdots,  {y}_k: \delta , \Delta \wfdash  M' \headlin{ N /  {x}} : \tau \quad  {y}\notin \Gamma_1 \quad k \not = 0\)}
        \LeftLabel{ \redlab{FS{:}share} }
        \UnaryInfC{\( \Theta ; \Gamma_1 ,  {y}: \delta^k, \Delta \wfdash M' \headlin{ N /  {x}} [ {\widetilde{y}} \leftarrow  {y}] : \tau \)}
    \end{prooftree}
    From \defref{def:headlinfail}  $M'[ {\widetilde{y}} \leftarrow  {y}] \headlin{ N /  {x} } = M' \headlin{ N /  {x}} [ {\widetilde{y}} \leftarrow  {y}]$, 
    and the result follows.

    \item $M = M'[ \leftarrow  {y}] $.
    
    Then $ \headf{M'[ \leftarrow  {y}]} = \headf{M'}= {x}$ with  $x \not  = y $, 
    \begin{prooftree}
        \AxiomC{\( \Theta ; \Gamma  ,  {x}: \sigma  \wfdash M : \tau\)}
        \LeftLabel{ \redlab{FS{:}weak} }
        \UnaryInfC{\( \Theta ;  \Gamma  ,  {y}: \omega,  {x}: \sigma  \wfdash M[\leftarrow  {y}]: \tau \)}
    \end{prooftree}
     and $M'[ \leftarrow  {y}] \headlin{ N /  {x} } = M' \headlin{ N /  {x}} [ \leftarrow  {y}]$. Then by the induction hypothesis:
    \begin{prooftree}
        \AxiomC{\( \Theta ;  \Gamma , \Delta  \wfdash M \headlin{ N /  {x}}: \tau\)}
        \LeftLabel{ \redlab{FS{:}weak}}
        \UnaryInfC{\( \Theta ;  \Gamma  ,  {y}: \omega, \Delta \wfdash M\headlin{ N /  {x}}[\leftarrow  {y}]: \tau \)}
    \end{prooftree}

    \item If $M =  M' \linexsub {M'' / {y}} $ then $\headf{M' \linexsub {M'' / {y}}} = \headf{M'} = x \not = y$, 
    
    \begin{prooftree}
        \AxiomC{\( \Theta ;\Gamma  ,  {y}:\delta , x: \sigma \wfdash M : \tau \)}
        \AxiomC{\( \Theta ;\Delta \wfdash M'' : \delta \)}
        \LeftLabel{ \redlab{FS{:}ex \dash sub^{\ell}} }
        \BinaryInfC{\( \Theta ;\Gamma_1, \Gamma_2 , x: \sigma \wfdash M' \linexsub {M'' / {y}} : \tau \)}
    \end{prooftree}
     and $M' \linexsub {M'' / {y}}  \headlin{ N / x } = M'  \headlin{N / x } \linexsub {M'' / {y}}$. Then by the induction hypothesis:
    
    \begin{prooftree}
        \AxiomC{\( \Theta ;\Gamma  ,  {y}:\delta , \Delta  \wfdash M'  \headlin{N / x } : \tau \)}
        \AxiomC{\( \Theta ;\Delta \wfdash M'' : \delta \)}
        \LeftLabel{ \redlab{FS{:}ex \dash sub^{\ell}} }
        \BinaryInfC{\( \Theta ;\Gamma_1, \Gamma_2 , \Delta  \wfdash M'  \headlin{N / x } \linexsub {M'' / {y}} : \tau \)}
    \end{prooftree}

    \item If $M =  M' \unexsub{U / \unvar{y}} $ then $\headf{M' \unexsub{U / \unvar{y}}} = \headf{M'} = x $, and the proofs is similar to the case above.
    
    
    
    \end{enumerate}
\end{proof}

\begin{lemma}[Unrestricted Substitution Lemma for \lamrsharfailunres]
\label{lem:subt_lem_sharefailunres_un}
If $\Theta, \banged{x}: \eta ; \Gamma \wfdash M: \tau$, $\headf{M} = {x}[i]$, $\eta_i = \sigma $, and $\Theta ; \cdot \wfdash N : \sigma$
then 
$\Theta, \banged{x}: \eta ; \Gamma  \wfdash M \headlin{ N / {x}[i] }$.
\end{lemma}

\begin{proof}
By structural induction on $M$ with $\headf{M}= {x}[i]$. There are three cases to be analyzed: 

\begin{enumerate}
\item $M= {x}[i]$.

In this case, 
    \begin{prooftree}
        \AxiomC{}
        \LeftLabel{ \redlab{F{:}var^{ \ell}}}
        \UnaryInfC{\( \Theta , \banged{x}: \eta;  {x}: \eta_i  \wfdash  {x} : \sigma\)}
        \LeftLabel{\redlab{F{:}var^!}}
        \UnaryInfC{\( \Theta, \banged{x}: \eta ; \cdot \wfdash {x}[i] : \sigma\)}
    \end{prooftree}
 and $\Gamma=\emptyset$.  Observe that ${x}[i]\headlin{N/{x}[i]}=N$, since $\Theta, \banged{x}: \eta  ; \Gamma  \wfdash M \headlin{ N / {x}[i] }$, by hypothesis, the result follows.

    \item $M = M'\ B$.
    
    In this case, $\headf{M'\ B} = \headf{M'} =  {x}[i]$, and one has the following derivation:
    
    \begin{prooftree}
         \AxiomC{\( \Theta, \banged{x}: \eta ;\Gamma_1 \wfdash M : (\delta^{j} , \epsilon ) \rightarrow \tau \quad \Theta, \banged{x}: \sigma ; \Gamma_2 \wfdash B : (\delta^{k} , \epsilon' )  \)}
         \AxiomC{\( \epsilon \relunbag \epsilon' \)}
            \LeftLabel{\redlab{F{:}app}}
        \BinaryInfC{\( \Theta, \banged{x}: \eta ;\Gamma_1 , \Gamma_2 \wfdash M\ B : \tau\)}
    \end{prooftree}
    
 where $\Gamma=\Gamma_1,\Gamma_2$, $\delta$ is a strict type and $j,k$ are non-negative  integers, possibly different.
 
 By IH, we get $\Theta, \banged{x}: \eta ;\Gamma_1 \wfdash M'\headlin{N/ {x[i]}}:(\delta^{j} , \epsilon ) \rightarrow \tau $, which gives the derivation: 
    \begin{prooftree}
        \AxiomC{$\Theta , \banged{x}: \eta;\Gamma_1\wfdash M'\headlin{N/ {x}[i]}:(\delta^{j} , \epsilon ) \rightarrow \tau $}\
        \AxiomC{$\Theta , \banged{x}: \eta; \Gamma_2 \wfdash B : (\delta^{k} , \epsilon' ) $}
        \AxiomC{\( \epsilon \relunbag \epsilon' \)}
    	\LeftLabel{\redlab{F{:}app}}
        \TrinaryInfC{$\Theta , \banged{x}: \eta;\Gamma_1 , \Gamma_2  \wfdash ( M'\headlin{ N / {x}[i] } ) B:\tau $}    
    \end{prooftree}
From \defref{def:linsubfail}, $M' \esubst{ B }{ y} \headlin{ N / {x}[i] } = M' \headlin{ N / {x}[i] } \esubst{ B }{ y}$,
and  the result follows.

    \item $M = M'[ {\widetilde{y}} \leftarrow  {y}] $.
    
    Then $ \headf{M'[ {\widetilde{y}} \leftarrow  {y}]} = \headf{M'}= {x}[i]$, for  $y\neq x$. Therefore, 
    \begin{prooftree}
        \AxiomC{\(\Theta , \banged{x}: \eta; \Gamma_1 ,  {y}_1: \delta, \cdots,  {y}_k: \delta  \wfdash M' : \tau \quad  {y}\notin \Gamma_1 \quad k \not = 0\)}
        \LeftLabel{ \redlab{FS{:}share}}
        \UnaryInfC{\( \Theta , \banged{x}: \eta; \Gamma_1 ,  {y}: \delta^k \wfdash M'[ {y}_1 , \cdots ,  {y}_k \leftarrow y] : \tau \)}
    \end{prooftree}
    where $\Gamma=\Gamma_1 ,  {y}: \delta^k$. 
    By IH, the result follows for $M'$, that is, 
    $$\Theta, \banged{x}: \eta ; \Gamma_1 ,  {y}_1: \delta, \cdots,  {y}_k: \delta  \wfdash M'\headlin{N/ {x}[i] } : \tau $$
    
    and we have the derivation:
    
    \begin{prooftree}
        \AxiomC{\( \Theta , \banged{x}: \eta; \Gamma_1 ,  {y}_1: \delta, \cdots,  {y}_k: \delta  \wfdash  M' \headlin{ N / {x}[i] } : \tau \quad  {y}\notin \Gamma_1 \quad k \not = 0\)}
        \LeftLabel{ \redlab{FS{:}share} }
        \UnaryInfC{\( \Theta , \banged{x}: \eta; \Gamma_1 ,  {y}: \delta^k \wfdash M' \headlin{ N / {x}[i]} [ {\widetilde{y}} \leftarrow  {y}] : \tau \)}
    \end{prooftree}
    From \defref{def:headlinfail}  $M'[ {\widetilde{y}} \leftarrow  {y}] \headlin{ N / {x}[i] } = M' \headlin{ N / {x}[i]} [ {\widetilde{y}} \leftarrow  {y}]$, 
    and the result follows.

\item $M = M'[ \leftarrow  {y}] $.
    
    Then $ \headf{M'[ \leftarrow  {y}]} = \headf{M'}= {x}[i]$ with  $x \not  = y $, 
    \begin{prooftree}
        \AxiomC{\( \Theta , \banged{x}: \eta; \Gamma   \wfdash M : \tau\)}
        \LeftLabel{ \redlab{FS{:}weak} }
        \UnaryInfC{\( \Theta , \banged{x}: \eta;  \Gamma  ,  {y}: \omega \wfdash M[\leftarrow  {y}]: \tau \)}
    \end{prooftree}
     and $M'[ \leftarrow  {y}] \headlin{ N / {x}[i] } = M' \headlin{ N / {x}[i] } [ \leftarrow  {y}]$. Then by the induction hypothesis:
    \begin{prooftree}
        \AxiomC{\( \Theta , \banged{x}: \eta;  \Gamma   \wfdash M \headlin{ N / {x}[i]  }: \tau\)}
        \LeftLabel{ \redlab{FS{:}weak}}
        \UnaryInfC{\( \Theta , \banged{x}: \eta;  \Gamma  ,  {y}: \omega \wfdash M\headlin{ N / {x}[i] }[\leftarrow  {y}]: \tau \)}
    \end{prooftree}

    \item $M =  M' \linexsub {M'' / {y}} $. 
    
    Then $\headf{M' \linexsub {M'' / {y}}} = \headf{M'} = {x}[i]$ with $x \not = y$, 
    
    \begin{prooftree}
        \AxiomC{\( \Theta , \banged{x}: \eta;\Gamma  ,  {y}:\delta \wfdash M : \tau \)}
        \AxiomC{\( \Theta, \banged{x}: \eta ;\Delta \wfdash M'' : \delta \)}
        \LeftLabel{ \redlab{FS{:}ex \dash sub^{\ell}} }
        \BinaryInfC{\( \Theta, \banged{x}: \eta ;\Gamma, \Delta  \wfdash M' \linexsub {M'' / {y}} : \tau \)}
    \end{prooftree}
     and $M' \linexsub {M'' / {y}}  \headlin{ N / {x}[i]  } = M'  \headlin{N / {x}[i]  } \linexsub {M'' / {y}}$. Then by the induction hypothesis:
    
    \begin{prooftree}
        \AxiomC{\( \Theta , \banged{x}: \eta;\Gamma  ,  {y}:\delta \wfdash M'  \headlin{N /  {x}[i] } : \tau \)}
        \AxiomC{\( \Theta , \banged{x}: \eta;\Delta \wfdash M'' : \delta \)}
        \LeftLabel{ \redlab{FS{:}ex \dash sub^{\ell}} }
        \BinaryInfC{\( \Theta ;\Gamma , \Delta  \wfdash M'  \headlin{N /  {x}[i] } \linexsub {M'' / {y}} : \tau \)}
    \end{prooftree}

    \item $M =  M' \unexsub{U / \unvar{y}}$.  
    
    Then $\headf{M' \unexsub {U / \unvar{y}}} = \headf{M'} = {x}[i] $, 
    
    \begin{prooftree}
        \AxiomC{\( \Theta , \banged{x}: \eta, \banged{y} : \epsilon; \Gamma  \wfdash M : \tau \quad  \Theta , \banged{x}: \eta; \dash \wfdash U : \epsilon \)}
            \LeftLabel{\redlab{FS{:}ex \dash sub^!}}
        \UnaryInfC{\( \Theta , \banged{x}: \eta; \Gamma  \wfdash M \unexsub{U / \unvar{y}}  : \tau \)}
    \end{prooftree}
    
     and $M' \unexsub{U /\unvar{y}}  \headlin{ N /  {x}[i] } = M'  \headlin{N /  {x}[i] } \unexsub{U /\unvar{y}}$. Then by the induction hypothesis:
    
    \begin{prooftree}
        \AxiomC{\( \Theta , \banged{x}: \eta, \banged{y} : \epsilon; \Gamma  \wfdash M'  \headlin{N / {x}[i] } : \tau \quad  \Theta , \banged{x}: \eta; \dash \wfdash U : \eta \)}
            \LeftLabel{\redlab{FS{:}ex \dash sub^!}}
        \UnaryInfC{\( \Theta , \banged{x}: \eta ; \Gamma  \wfdash M'  \headlin{N /  {x}[i] } \unexsub {U / \unvar{y}}  : \tau \)}
    \end{prooftree}
\end{enumerate}

\end{proof}

\begin{theorem}[SR in \lamrsharfailunres]
\label{t:app_lamrsharfailsrunres}
If $\Theta ; \Gamma \wfdash \expr{M}:\tau$ and $\expr{M} \red \expr{M}'$ then $\Theta ; \Gamma \wfdash \expr{M}' :\tau$.
\end{theorem}

\begin{proof} By structural induction on the reduction rule from \figref{fig:share-reductfailureunres} applied in $\expr{M}\red \expr{N}$.

\begin{enumerate}

	\item Rule $\redlab{RS{:}Beta}$.
	
	Then $\expr{M} = (\lambda x. M[ {\widetilde{x}} \leftarrow  {x}]) B $  and the reduction is:
	  \begin{prooftree}
        \AxiomC{}
        \LeftLabel{\redlab{RS{:}Beta}}
        \UnaryInfC{\((\lambda x. M[ {\widetilde{x}} \leftarrow  {x}]) B \red M[ {\widetilde{x}} \leftarrow  {x}]\ \esubst{ B }{ x }\)}
     \end{prooftree}

 	where $ \expr{M}'  =  M[ {\widetilde{x}} \leftarrow  {x}]\ \esubst{ B }{ x }$. Since $\Theta ; \Gamma\wfdash \expr{M}:\tau$ we get the following derivation:
	\begin{prooftree}
			\AxiomC{$\Theta , \banged{x} : \eta; \Gamma' ,  {x}_1:\sigma , \cdots ,  {x}_j:\sigma  \wfdash  M: \tau $}
			\LeftLabel{ \redlab{FS{:}share} }
			\UnaryInfC{$\Theta , \banged{x} : \eta;  \Gamma' ,   {x}:\sigma^{j}  \wfdash  M[ {\widetilde{x}} \leftarrow  {x}]: \tau $}
			\LeftLabel{ \redlab{FS{:}abs \dash sh} }
            \UnaryInfC{$\Theta ; \Gamma' \wfdash \lambda x. M[ {\widetilde{x}} \leftarrow  {x}]: (\sigma^{j} , \eta ) \rightarrow \tau $}
            
            \AxiomC{$\Theta ;\Delta \wfdash B: (\sigma^{k} , \epsilon ) $}
            
			\AxiomC{\( \eta \relunbag \epsilon \)}
			\LeftLabel{ \redlab{FS{:}app} }
		\TrinaryInfC{$ \Theta ;\Gamma' , \Delta \wfdash (\lambda x. M[ {\widetilde{x}} \leftarrow  {x}]) B:\tau $}
	\end{prooftree}

	for $\Gamma = \Gamma' , \Delta $ and $x\notin \dom{\Gamma'}$. 
	Notice that: 
    \begin{prooftree}
            \AxiomC{$\Theta , \banged{x} : \eta; \Gamma' ,  {x}_1:\sigma , \cdots ,  {x}_j:\sigma  \wfdash  M: \tau $}
			\LeftLabel{ \redlab{FS{:}share} }
			\UnaryInfC{$\Theta , \banged{x} : \eta;  \Gamma' ,   {x}:\sigma^{j}  \wfdash  M[ {\widetilde{x}} \leftarrow  {x}]: \tau $}
			
            \AxiomC{$\Theta ;\Delta \wfdash B:(\sigma^{k} , \epsilon )  $}
            
            \AxiomC{\( \eta \relunbag \epsilon \)}
            \LeftLabel{ \redlab{FS{:}ex \dash sub} }
        \TrinaryInfC{$ \Theta ;\Gamma' , \Delta \wfdash M[ {\widetilde{x}} \leftarrow  {x}]\ \esubst{ B }{ x }:\tau $}
    \end{prooftree}

    Therefore $ \Theta ; \Gamma',\Delta\wfdash\expr{M}' :\tau$ and the result follows.

    \item Rule $ \redlab{RS{:}Ex \dash Sub}.$
    
    Then $ \expr{M} =  M[ {x}_1, \cdots ,  {x}_k \leftarrow  {x}]\ \esubst{ C \bagsep U }{ x }$ where $C=  \bag{N_1}\cdot \dots \cdot \bag{N_k} $. The reduction is:
    
     \begin{prooftree}
        \AxiomC{$ C = \bag{M_1}
        \cdots  \bag{M_k} $}
        \AxiomC{$ M \not= \fail^{\widetilde{y}} $}
        \LeftLabel{\redlab{RS{:}Ex \dash Sub}}
        \BinaryInfC{\( \!M[ {x}_1, \!\cdots\! ,  {x}_k \leftarrow  {x}]\esubst{ C \bagsep U }{ x } \red \displaystyle \sum_{C_i \in \perm{C}}M\linexsub{C_i(1)/ {x_1}} \cdots \linexsub{C_i(k)/ {x_k}} \unexsub{U / \unvar{x} }   \)}
    \end{prooftree}
    
    and $\expr{M'}= \sum_{C_i \in \perm{C}}M\linexsub{C_i(1)/ {x_1}} \cdots \linexsub{C_i(k)/ {x_k}} \unexsub{U / \unvar{x} }$
    To simplify the proof we take $k=2$, as the case $k>2$ is similar. Therefore,
    \begin{itemize}
        \item $C=\bag{N_1}\cdot \bag{N_2}$; and
        \item $\perm{C}=\{\bag{N_1}\cdot \bag{N_2}, \bag{N_2}\cdot \bag{N_1}\}$.
    \end{itemize} 
    Since $\Theta ; \Gamma\wfdash \expr{M}:\tau$ we get a derivation: (we omit the labels \redlab{FS:ex\dash sub} and \redlab{FS{:}share})
    \begin{prooftree}
            \AxiomC{\( \Theta, \banged{x} : \eta  ;  \Gamma' ,  {x}_1: \sigma,  {x}_2: \sigma \wfdash M : \tau \quad  {x}\notin \dom{\Gamma} \quad k \not = 0\)}
            \UnaryInfC{\(  \Theta , \banged{x} : \eta ; \Gamma' ,  {x}: \sigma^{2} \wfdash M[ {\widetilde{x}} \leftarrow  {x}] : \tau  \)} 
            
            \AxiomC{\( \Theta ; \Delta \wfdash B : (\sigma^{k} , \epsilon ) \)}
            \AxiomC{\( \eta \relunbag \epsilon \)}
        \TrinaryInfC{\( \Theta ; \Gamma', \Delta \wfdash (M[ {\widetilde{x}} \leftarrow  {x}])\esubst{ B }{ x }  : \tau \)}
    \end{prooftree}
    where $\Gamma = \Gamma' , \Delta $. Consider the wf derivation for $\Pi_{1,2}$: (we omit the labels \redlab{FS:ex\dash sub^!} and \redlab{FS:ex\dash sub^{\ell}})
        {\small
        \begin{prooftree}
                    \AxiomC{\( \Theta, \banged{x} : \eta  ;  \Gamma' ,  {x}_1: \sigma,  {x}_2: \sigma \wfdash M : \tau\)}
                    \AxiomC{\( \Theta ; \Delta_1 \wfdash N_1 : \sigma \)}
                \BinaryInfC{\( \Theta, \banged{x} : \eta  ;  \Gamma' ,  {x}_2: \sigma , \Delta_1 \wfdash M\linexsub{N_1/ {x_1}} : \tau \)}
                \AxiomC{\( \Theta ; \Delta_2 \wfdash N_2 : \sigma \)}
            \BinaryInfC{\( \Theta, \banged{x} : \eta  ;  \Gamma' , \Delta \wfdash M\linexsub{N_1/ {x_1}} \linexsub{N_2/ {x_2}} : \tau \)}
            \AxiomC{\( \Theta ; \dash \wfdash U : \epsilon \)}
            \AxiomC{\( \eta \relunbag \epsilon \)}
            \TrinaryInfC{\( \Theta ; \Gamma' , \Delta \wfdash M\linexsub{N_1/ {x_1}} \linexsub{N_2/ {x_2}} \unexsub{U / \unvar{x} }  : \tau \)}
        \end{prooftree}
        }

    Similarly, we can obtain a derivation $\Pi_{2,1}$ of $ \Theta ; \Gamma' , \Delta \wfdash M\linexsub{N_2/ {x_1}} \linexsub{N_1/ {x_2}} \unexsub{U / \unvar{x} }  : \tau$.   Finally, applying \redlab{FS{:}sum}:
     \begin{prooftree}
        \AxiomC{\( \Pi_{1,2} \)}
        \AxiomC{\( \Pi_{2,1} \)}
        \LeftLabel{ \redlab{FS{:}sum} }
        \BinaryInfC{$ \Theta ; \Gamma' , \Delta  \wfdash M\linexsub{N_1/ {x_1}} \linexsub{N_2/ {x_2}} \unexsub{U / \unvar{x} } + M\linexsub{N_2/ {x_1}} \linexsub{N_1/ {x_2}} \unexsub{U / \unvar{x} }   : \tau $}
    \end{prooftree}
    
    and the result follows.

     \item Rule $ \redlab{RS{:}Fetch^{\ell}}$.
    
    Then $ \expr{M} =  M \linexsub{N /  {x}}  $ where  $\headf{M} =  {x}$. The reduction is:
    
        \begin{prooftree}
             \AxiomC{$ \headf{M} =  {x}$}
             \LeftLabel{\redlab{RS{:}Fetch^{\ell}}}
             \UnaryInfC{\(  M \linexsub{N /  {x}} \red  M \headlin{ N/ {x} }  \)}
         \end{prooftree}
          and $\expr{M'}= M \linexsub{N /  {x}} $.     Since $\Theta ; \Gamma\wfdash \expr{M}:\tau$ we get the following derivation:
        \begin{prooftree}
        \AxiomC{\( \Theta ; \Gamma' ,  {x}:\sigma \wfdash M : \tau \quad  \Theta ; \Delta \wfdash N : \sigma \)}
            \LeftLabel{\redlab{FS{:}ex \dash sub^{\ell}}}
        \UnaryInfC{\( \Theta ; \Gamma', \Delta \wfdash M \linexsub{N /  {x}} : \tau \)}
    \end{prooftree}
      where $\Gamma = \Gamma' , \Delta $. By Lemma~\ref{l:lamrsharfailsubsunres}, we obtain the derivation $ \Theta ; \Gamma' , \Delta \wfdash   M \headlin{ N/ {x} } : \tau $. 

    \item Rule $ \redlab{RS{:} Fetch^!}$.
    
    Then $ \expr{M} =  M \unexsub{U /  \unvar{x}}  $ where  $\headf{M} = {x}[i]$. The reduction is:
    
         \begin{prooftree}
             \AxiomC{$ \headf{M} = {x}[i]$}
             \AxiomC{$ U_i = \banged{\bag{N}}$}
             \LeftLabel{}
             \BinaryInfC{\(  M \unexsub{U /  \unvar{x}} \red  M \headlin{ N / {x}[i] }\unexsub{U /  \unvar{x}} \)}
         \end{prooftree}
    
    and $\expr{M'}= M \unexsub{U /  \unvar{x}} $.     Since $\Theta ; \Gamma \wfdash \expr{M}:\tau$ we get the following derivation:

    \begin{prooftree}
        \AxiomC{\( \Theta , \banged{x} : \eta; \Gamma  \wfdash M : \tau \)}
        \AxiomC{\( \Theta ; \dash \wfdash U : \epsilon \)}
        \AxiomC{\( \eta \relunbag \epsilon \)}
            \LeftLabel{\redlab{FS{:}ex \dash sub^!}}
        \TrinaryInfC{\( \Theta ; \Gamma \wfdash M \unexsub{U /  \unvar{x}}  : \tau \)}
    \end{prooftree}
    
    By Lemma~\ref{lem:subt_lem_sharefailunres_un}, we obtain the derivation $ \Theta ; \Gamma \wfdash   M \headlin{ N /{x}[i] }\unexsub{U /  \unvar{x}} : \tau $. 

\item Rule $\redlab{RS{:}TCont}$.

Then $\expr{M} = C[M]$ and the reduction is as follows:

\begin{prooftree}
        \AxiomC{$   M \red M'_{1} + \cdots +  M'_{k} $}
        \LeftLabel{\redlab{RS{:}TCont}}
        \UnaryInfC{$ C[M] \red  C[M'_{1}] + \cdots +  C[M'_{k}] $}
\end{prooftree}
with $\expr{M'}= C[M] \red  C[M'_{1}] + \cdots +  C[M'_{k}] $. 
The proof proceeds by analysing the context $C$.
There are four cases:

\begin{enumerate}
    \item $C=[\cdot]\ B$.
    
    In this case $\expr{M}=M \ B$, for some $B$. Since $\Gamma\vdash \expr{M}:\tau$ one has a derivation:
\begin{prooftree}
        \AxiomC{\( \Theta ;\Gamma' \wfdash M : (\sigma^{j} , \eta ) \rightarrow \tau \)}
         \AxiomC{\(  \Theta ;\Delta \wfdash B : (\sigma^{k} , \epsilon )  \)}
         \AxiomC{\( \eta \relunbag \epsilon \)}
        \LeftLabel{\redlab{FS{:}app}}
        \TrinaryInfC{\( \Theta ; \Gamma', \Delta \wfdash M\ B : \tau\)}
\end{prooftree}

where $\Gamma = \Gamma' , \Delta $. From  $\Gamma'\wfdash M:\sigma^j\rightarrow\tau$ and the reduction $M \red M'_{1} + \cdots +  M'_{k} $, one has by IH that  $\Gamma'\wfdash M_1'+\ldots, M_k':\sigma^j\rightarrow\tau$, which entails $\Gamma'\wfdash M_i':\sigma^j\rightarrow\tau$, for $i=1,\ldots, k$, via rule \redlab{FS{:}sum}. Finally, we may type the following:
 
\begin{prooftree}

            \AxiomC{\(  \forall i \in {1 , \cdots , l} \)}

			 \AxiomC{\( \Theta ;\Gamma' \wfdash M'_{i} : (\sigma^{j} , \eta ) \rightarrow \tau \)}
            \AxiomC{\(  \Theta ;\Delta \wfdash B : (\sigma^{k} , \epsilon )  \)}
            \AxiomC{\( \eta \relunbag \epsilon \)}
            \LeftLabel{\redlab{FS{:}app}}
            \TrinaryInfC{\( \Theta ; \Gamma', \Delta \wfdash M'_{i}\ B : \tau\)}
			
        \LeftLabel{ \redlab{FS{:}sum} }
    \BinaryInfC{\( \Gamma', \Delta \wfdash (M'_{1}\ B) + \cdots +  (M'_{l} \ B) : \tau\)}
\end{prooftree}

Since $ \expr{M}'  =   (C[M'_{1}]) + \cdots +  (C[M'_{l}]) = M_1'B+\ldots+M_k'B$, the result follows.
\item  Cases $C=[\cdot]\linexsub{N/x} $ and $C=[\cdot][\widetilde{x} \leftarrow x]$ are similar to the previous.
\item Other cases proceed similarly.

\end{enumerate}

	\item Rule $ \redlab{RS{:}ECont}$. 
	
	This case is analogous to the previous.
	




\item Rule $ \redlab{RS{:}Fail^{\ell}}.$

Then $\expr{M} =   M[\widetilde{x} \leftarrow  {x}]\ \esubst{C \bagsep U}{ x } $ where $C = \bag{N_1}\cdot \dots \cdot \bag{N_l}  $ and  the reduction is:
\begin{prooftree}
     \AxiomC{$ k \neq \size{C} $} 
     \AxiomC{$  \widetilde{y} = (\llfv{M} \setminus \{  \widetilde{x}\} ) \cup \llfv{C} $}
    \LeftLabel{\redlab{RS{:}Fail^{\ell}}}
    \BinaryInfC{\(  M[x_1 , \cdots , x_k \leftarrow  {x}]\ \esubst{C \bagsep U}{ x }  \red \displaystyle \sum_{C_i \in \perm{C}}  \fail^{\widetilde{y}} \)}
\end{prooftree}

where $\expr{M'}=\sum_{C_i \in \perm{C}}  \fail^{\widetilde{y}}$. Since $\Theta, x: \eta ; \Gamma' , x_1:\sigma,\ldots, x_k:\sigma \wfdash \expr{M}$, one has a derivation:
\begin{prooftree}
            \AxiomC{\( \Theta, x: \eta ; \Gamma' , x_1:\sigma,\ldots, x_k:\sigma \wfdash M: \tau \)}
            \LeftLabel{ \redlab{FS{:}ex \dash sub} }    
            \UnaryInfC{\( \Theta, x: \eta ;\Gamma' , x:\sigma^{k} \wfdash M[x_1, \cdots , x_k \leftarrow x] : \tau \)}
            \AxiomC{\(\Theta ; \Delta \wfdash C \bagsep U : (\sigma^{j} , \epsilon ) \)}
            \AxiomC{\( \eta \relunbag \epsilon \)}
        \LeftLabel{ \redlab{FS{:}ex \dash sub} }    
        \TrinaryInfC{\(\Theta ; \Gamma', \Delta \wfdash M[x_1 , \cdots , x_k] \leftarrow  {x}]\ \esubst{C \bagsep U}{ x }  : \tau \)}
    \end{prooftree}

where $\Gamma = \Gamma' , \Delta $. We may type the following:
    \begin{prooftree}
        \AxiomC{\( \)}
        \LeftLabel{ \redlab{FS{:}fail}}
        \UnaryInfC{\(\Theta ; \Gamma' , \Delta \wfdash  \fail^{\widetilde{y}} : \tau  \)}
    \end{prooftree}
since $\Gamma',\Delta$ contain assignments on the free variables in $M$ and $B$.
Therefore, $\Theta ;\Gamma\wfdash \fail^{\widetilde{y}}:\tau$, by applying \redlab{FS{:}sum}, it follows that $\Theta ;\Gamma\wfdash \sum_{B_i\in \perm{B}}\fail^{\widetilde{y}}:\tau$,as required.

\item Rule $ \redlab{RS{:}Fail^!}$.

Then $ M \unexsub{U / \unvar{x} } $ where $\headf{M} = {x}[i]$ and $B = U_i = \banged{\oneb} $ and  the reduction is:

\begin{prooftree}
     \AxiomC{$\headf{M} = {x}[i]$}
    \AxiomC{$ U_i = \banged{\oneb} $}
    \AxiomC{\( \widetilde{y} = \llfv{M} \)}
    \LeftLabel{\redlab{RS{:}Fail^!}}
    \TrinaryInfC{\(  M \unexsub{U / \unvar{x} } \red 
    M \headlin{ \fail^{\emptyset} /{x}[i] } \unexsub{U / \unvar{x} }  \)}
\end{prooftree}

with $\expr{M}'=M \headlin{ \fail^{\emptyset} /{x}[i] } \unexsub{U / \unvar{x} }$. By hypothesis, one has the derivation:
    
\begin{prooftree}
        \AxiomC{\( \Theta , {x} : \eta; \Gamma  \wfdash M : \tau \)}
        \AxiomC{\( \Theta ; \dash \wfdash U : \epsilon \)}
        \AxiomC{\( \eta \relunbag \epsilon \)}
            \LeftLabel{\redlab{FS{:}Esub^!}}
        \TrinaryInfC{\( \Theta ; \Gamma \wfdash M \unexsub{U / \unvar{x}}  : \tau \)}
\end{prooftree}
By Lemma~\ref{lem:subt_lem_sharefailunres_un}, there exists a derivation $\Pi_1$ of  $
    \Theta , \banged{x} : \eta ; \Gamma'  \wfdash   M \headlin{ \fail^{\emptyset} /{x}[i] }  : \tau $. Thus,
    
\begin{prooftree}
        \AxiomC{\( \Theta , \banged{x} : \eta ; \Gamma \wfdash   M \headlin{ \fail^{\emptyset} /{x}[i] }  : \tau  \)}
        \AxiomC{\( \Theta ; \dash \wfdash U : \epsilon \)}
        \AxiomC{\( \eta \relunbag \epsilon \)}
            \LeftLabel{\redlab{FS{:}Esub^!}}
        \TrinaryInfC{\( \Theta ; \Gamma \wfdash  M \headlin{ \fail^{\emptyset} /{x}[i] } \unexsub{U / \unvar{x}}  : \tau \)}
\end{prooftree}

\item Rule $\redlab{RS{:}Cons_1}$.

Then $\expr{M} =   \fail^{\widetilde{x}}\ B $ where $B = \bag{N_1}\cdot \dots \cdot \bag{N_k} $  and  the  reduction is:

\begin{prooftree}
    \AxiomC{\( \widetilde{y} = \llfv{C} \)}
    \LeftLabel{$\redlab{RS{:}Cons_1}$}
    \UnaryInfC{\(\fail^{\widetilde{x}}\ C \bagsep U \red \displaystyle \sum_{\perm{C}} \fail^{\widetilde{x} \uplus \widetilde{y}}  \)}
\end{prooftree}
and $ \expr{M}'  =  \sum_{\perm{B}} \fail^{\widetilde{x} \cup \widetilde{y}} $.
Since $\Gamma\wfdash \expr{M}:\tau$, one has the derivation: 

    \begin{prooftree}
         \AxiomC{\( \)}
        \LeftLabel{\redlab{F{:}fail}}
        \UnaryInfC{\( \Theta ;\Gamma' \wfdash \fail^{\widetilde{x}}: (\sigma^{j} , \eta ) \rightarrow \tau \)}
         \AxiomC{\(  \Theta ;\Delta \wfdash B : (\sigma^{k} , \epsilon )  \)}
         \AxiomC{\( \eta \relunbag \epsilon \)}
            \LeftLabel{\redlab{F{:}app}}
        \TrinaryInfC{\( \Theta ; \Gamma', \Delta \wfdash \fail^{\widetilde{x}}\ B : \tau\)}
    \end{prooftree}

Hence $\Gamma = \Gamma' , \Delta $ and we may type the following:

    \begin{prooftree}
        \AxiomC{\( \)}
        \LeftLabel{\redlab{F{:}fail}}
        \UnaryInfC{$ \Theta ;\Gamma \wfdash \fail^{\widetilde{x} \uplus \widetilde{y}} : \tau$}
        \AxiomC{\( \cdots \)}
        \AxiomC{\( \)}
        \LeftLabel{\redlab{F{:}fail}}
        \UnaryInfC{$\Theta ;\Gamma \wfdash \fail^{\widetilde{x} \uplus \widetilde{y}} : \tau$}
        \LeftLabel{\redlab{F{:}sum}}
        \TrinaryInfC{$\Theta ; \Gamma \wfdash \sum_{\perm{C}} \fail^{\widetilde{x} \uplus \widetilde{y}}: \tau$}
    \end{prooftree}
    The proof for the cases of $\redlab{RS{:}Cons_2}$, $\redlab{RS{:}Cons_3}$ and $\redlab{RS{:}Cons_4}$ proceed similarly
    \end{enumerate}
\end{proof}

\section{Appendix to Subsection~\ref{ssec:first_enc}}\label{app:encodingone}

\subsection{Encodability Criteria}
\label{ss:criteria}
We follow the criteria in~\cite{DBLP:journals/iandc/Gorla10}, a widely studied abstract framework for establishing the \emph{quality} of encodings.
A \emph{language} $\mathcal{L}$ is a pair: a set of terms and a reduction semantics $\red$ on terms (with reflexive, transitive closure denoted $\tred$).
A correct encoding translates terms of a source language $\mathcal{L}_1= (\mathcal{M}, \red_1)$ into terms of a target language  $\mathcal{L}_2(\mathcal{P}, \red_2)$ by respecting certain criteria. 
The criteria in~\cite{DBLP:journals/iandc/Gorla10} concern \emph{untyped} languages; because we treat \emph{typed} languages,   we follow~\cite{DBLP:journals/iandc/KouzapasPY19} in requiring that encodings satisfy the following criteria:
\begin{enumerate}
\item {\bf Type preservation:} For every well-typed $M$, it holds that $\encod{M}{}$ is well-typed.

    \item {\bf Operational Completeness:} For every ${M}, {M}'$ such that ${M} \tred_1 {M}'$, it holds that $\encod{{M}}{} \tred_2 \approx_2 \encod{{M}'}{}$.
    
    \item {\bf Operational Soundness:} For every $M$ and $P$ such that $\encod{M}{} \tred_2 P$, there exists an $M'$ such that $M \red^*_1 M'$ and $P \tred_2 \approx_2 \encod{M'}{}  $.
    
    \item {\bf Success Sensitiveness:} For every ${M}$, it holds that $M \checkmark_1$ if and only if $\encod{M}{} \checkmark_2$, where $\checkmark_1$ and $\checkmark_2$ denote a success predicate in $\mathcal{M}$ and $\mathcal{P}$, respectively. 
    
\end{enumerate}

In addition to these semantic criteria, we shall also consider \emph{compositionality}:  a composite source term is encoded as the combination of the encodings of its sub-terms. 
Success sensitiveness complements completeness and soundness, giving information about observable behaviors. 
The so-called success predicates $\checkmark_1$ and $\checkmark_2$ serve as a minimal notion of \emph{observables}; the criterion then says that observability of success of a source term implies observability of success in the corresponding target term, and vice-versa.

\subsection{Correctness of  \texorpdfstring{$\recencodopenf{\cdot }$}{}}
The correctness of the encoding from $\recencodopenf{\cdot}$ from $\lamrfailunres$ to $\lamrsharfailunres$ relies on an encoding on contexts (\defref{d:enclamcontfailunres}),  auxiliary propositions (Propositions~\ref{prop:linhed_encfail} and  \ref{prop:wf_linsubunres}) for well-formedness preservation (Theorema~\ref{prop:preservencintolamrfailunres}), operational soundness (Theorem~\ref{l:app_completenessone}) and completeness (Theorem~\ref{l:soundnessoneunres}), and success sensitivity (Theorem~\ref{proof:app_successsensce}).

\begin{definition}[Encoding on Contexts] We define an encoding $\recencod{\cdot }$ on contexts:
\label{d:enclamcontfailunres}
\begin{align*}
\recencod{\Theta} &= \Theta
&
\recencod{\emptyset} = \emptyset
\\
\recencod{ x: \tau , \Gamma }  &=   x:\tau  , \recencod{ \Gamma } &  (x \not \in \dom{\Gamma}) \\
\recencod{ x: \tau , \cdots , x: \tau, \Gamma } &= x : \tau \wedge \cdots \wedge \tau , \recencod{ \Gamma } &  (x \not \in \dom{\Gamma})
    \end{align*}
\end{definition}

\begin{proposition}\label{prop:linhed_encfail}
Let $M, N$ be terms. We have:
 \begin{enumerate}
 \item $ \recencodf{M\headlin{N/x}}=\recencodf{M}\headlin{\recencodf{N}/x}$.
 \item $ \recencodf{M\linsub{\widetilde{x}}{x}}=\recencodf{M}\linsub{\widetilde{x}}{x}$, where $\widetilde{x}=x_1,\ldots, x_k$ is sequence of pairwise distinct fresh variables.
 \end{enumerate}
\end{proposition}

\begin{proof}
By induction of the structure of $M$.
\end{proof}

\begin{proposition}[Well-formedness Preservation under Linear Substitutions in \lamrfailunres]
\label{prop:wf_linsubunres}
Let ${M} \in \lamrfail$.
If $\Theta ; \Gamma, x:\sigma \wfdash {M} : \tau$
and $\Theta ; \Gamma \wfdash x_i : \sigma$ then 
 $\Theta ; \Gamma, x_i:\sigma \wfdash {M}\linsub{x_i}{x} : \tau$.
 \end{proposition}

\begin{proof}
 Standard, by induction on the well-formedness derivation rules in \figref{fig:wf_rules_unres}.
 \end{proof}

\begin{proposition}[Well-formedness preservation for $\recencodf{-}$]
\label{prop:preservencintolamrfailunres}
Let $B$ and  $\expr{M}$  be a bag and a expression in $\lamrfailunres$, respectively. 
\begin{enumerate}
\item
    If $\Theta ; \Gamma \wfdash B:(\sigma^{k} , \eta  )$ 
and $\dom{\Gamma} = \mlfv{B}$ 
then $ \recencod{\Theta} ; \recencod{\Gamma}\wfdash \recencodf{B}:(\sigma^{k} , \eta  )$ and $\forall \  x:\pi \in \Gamma, \ \pi = \tau $ for some $\tau$.

    \item 
    If $\Theta ; \Gamma \wfdash \expr{M}:\sigma$ 
and $\dom{\Gamma} = \mlfv{\expr{M}}$ 
then $ \recencod{\Theta} ; \recencod{\Gamma}\wfdash \recencodf{\expr{M}}:\sigma$ and $\forall \  x:\pi \in \Gamma, \ \pi = \tau $ for some $\tau$.

\end{enumerate}
\end{proposition}

\begin{theorem}[Well-formedness Preservation for $\recencodopenf{-}$]
\label{thm:preservencintolamrfail}
Let $B$ and  $\expr{M}$  be a bag and an expression in $\lamrfail$, respectively. 
\begin{enumerate}
\item
    If $\Theta ; \Gamma \wfdash B:(\sigma^{k} , \eta  )$ 
and $\dom{\Gamma} = \lfv{B}$ 
then $\recencod{ \Theta} ;\recencod{\Gamma}\wfdash \recencodopenf{B}:(\sigma^{k} , \eta  )$.

    \item 
    If $\Theta ;\Gamma \wfdash \expr{M}:\sigma$ 
and $\dom{\Gamma} = \lfv{\expr{M}}$ 
then $\recencod{\Theta} ;\recencod{\Gamma}\wfdash \recencodopenf{\expr{M}}:\sigma$.

\end{enumerate}
\end{theorem}

\begin{proof}
By mutual induction on the typing derivations $\Theta; \Gamma\wfdash B:(\sigma^{k} , \eta  )$ and $\Theta; \Gamma\wfdash \expr{M}:\sigma$, exploiting Proposition~\ref{prop:preservencintolamrfailunres}. The analysis for bags Part 1. follows directly from the IHs and will be omitted. 
As for Part 2. there are two main cases to consider:
\begin{enumerate}
    \item $\expr{M} = M$. 
    
Without loss of generality, assume $\lfv{M} = \{x,y\}$. Then, 
\(\Theta; \hat{x}:\sigma_1^j, \hat{y}:\sigma_2^k \wfdash M : \tau \)
where  $\#(x,M)=j$ and 
$\#(y,M)=k$, for some positive integers $j$ and $k$.

After $j+k$ applications of Proposition~\ref{prop:wf_linsubunres}
we obtain:
\begin{equation*}\label{eq:thmpres2fail}
\Theta; x_1:\sigma_1, \cdots, x_j:\sigma_1, y_1:\sigma_2, \cdots, y_k:\sigma_2 \wfdash M\linsub{\widetilde{x}}{x}\linsub{\widetilde{y}}{y} : \tau    
\end{equation*}
 where  $\widetilde{x}=x_{1},\cdots, x_{j}$
   and $\widetilde{y}=y_{1},\cdots, y_{k}$. 
From Proposition~\ref{prop:preservencintolamrfailunres}
one has
\begin{equation*}\label{eq:thmpres3fail}
\recencod{\Theta}; \recencod{x_1:\sigma_1, \cdots, x_j:\sigma_1, y_1:\sigma_2, \cdots, y_k:\sigma_2} \wfdash \recencodf{M\linsub{\widetilde{x}}{x}\linsub{\widetilde{y}}{y}} : \tau  
\end{equation*}

Since $\recencod{x_1:\sigma_1, \cdots, x_j:\sigma_1, y_1:\sigma_2, \cdots, y_k:\sigma_2}= x_1:\sigma_1, \cdots, x_j:\sigma_1, y_1:\sigma_2, \cdots, y_k:\sigma_2$ and $\recencod{\Theta} = \Theta $, 
we have the following derivation:
\begin{prooftree}
    \AxiomC{$\Theta; {x_1:\sigma_1, \cdots, x_j:\sigma_1, y_1:\sigma_2, \cdots, y_k:\sigma_2} \wfdash \recencodf{M\linsub{\widetilde{x}}{x}\linsub{\widetilde{y}}{y}} : \tau$}
    \LeftLabel{$\redlab{FS:share}$}
    \UnaryInfC{$\Theta; x:\sigma_1^j, y_1:\sigma_2, \cdots, y_k:\sigma_2 \wfdash \recencodf{M\linsub{\widetilde{x}}{x}\linsub{\widetilde{y}}{y}}[\widetilde{x}\leftarrow x] : \tau$}
\LeftLabel{$\redlab{FS:share}$}
    \UnaryInfC{$\Theta; x:\sigma_1^j, y:\sigma_2^k \wfdash \recencodf{M\linsub{\widetilde{x}}{x}\linsub{\widetilde{y}}{y}}[\widetilde{x}\leftarrow x][\widetilde{y}\leftarrow y] : \tau$}
    \end{prooftree}
    
    By expanding \defref{def:enctolamrsharfailunres}, we have 
$$
 \recencodopenf{M} = 
\recencodf{M\linsub{\widetilde{x}}{x}\linsub{\widetilde{y}}{y}}[\widetilde{x}\leftarrow x][\widetilde{y}\leftarrow y] 
$$

    which completes the proof for this case. 
    
    \item $\expr{M} = M_1 + \cdots + M_n$:
    
    This case proceeds easily by IH, using Rule~$\redlab{FS:sum}$.
    \end{enumerate}
\end{proof}

\begin{theorem}[Operational Completeness]
\label{l:app_completenessone}
Let $\expr{M}, \expr{N}$ be well-formed $\lamrfail$ expressions. 
Suppose $\expr{N}\red_{\redlab{R}} \expr{M}$.
\begin{enumerate}
\item If $\redlab{R} =  \redlab{R:Beta}$  then $ \recencodopenf{\expr{N}}  \red^{\leq 2}\recencodopenf{\expr{M}}$;

\item If $\redlab{R} = \redlab{R:Fetch}$   then   $ \recencodopenf{\expr{N}}  \red^+ \recencodopenf{\expr{M}'}$, for some $ \expr{M}$. 
\item If $\redlab{R} \neq \redlab{R:Beta}$    and $\redlab{R} \neq \redlab{R:Fetch}$ then   $ \recencodopenf{\expr{N}}  \red \recencodopenf{\expr{M}}$.
\end{enumerate}
\end{theorem}

\begin{proof}
We proceed by induction on the the rule from \figref{fig:reductions_lamrfailunres} applied to infer $\expr{N}\red \expr{M}$, distinguishing the three cases: (below $\widetilde{[x_1\leftarrow x_n]}$ abbreviates  
            $[\widetilde{x_1}\leftarrow x_1]\cdots [\widetilde{x_n}\leftarrow x_n]$).

\begin{enumerate}

    \item The rule applied is $\redlab{R}=\redlab{R:Beta}$. 
    
    In this case, 
        $\expr{N}= (\lambda x. M') B$, where  $B = C \bagsep U$, the reduction is 
    \begin{prooftree}
        \AxiomC{}
        \LeftLabel{\redlab{R:Beta}}
        \UnaryInfC{\((\lambda x. M) B \red M\ \esubst{B}{x}\)}
    \end{prooftree}

          and $\expr{M}= M'\esubst{B}{x}$. Below we assume $\llfv{\expr{N}}=\{x_1,\ldots, x_k\}$ and $\widetilde{x_i}=x_{i_1},\ldots, x_{i_{j_i}}$, where $j_i= \#(x_i, N)$, for $1\leq i\leq k$.
        On the one hand, we have:
        \begin{equation}\label{eq:beta1fail}
            \begin{aligned}
            \recencodopenf{\expr{N}}&= \recencodopenf{(\lambda x. M')B}
            =  \recencodf{((\lambda x. M')B) \langle{ {\widetilde{x_1}}/ {x_1}}\rangle\cdots \langle{ {\widetilde{x_k}}/ {x_k}}\rangle}\widetilde{[x_1\leftarrow x_k]}
            \\
            &=  \recencodf{(\lambda x. M^{''})B'}\widetilde{[x_1\leftarrow x_k]}
            =  (\recencodf{\lambda x. M^{''}}\recencodf{B'})\widetilde{[x_1\leftarrow x_k]} \\
            &=  ((\lambda x.\recencodf{ M^{''}\langle{ {\widetilde{y}}/ {x}}\rangle}[ {\widetilde{y}}\leftarrow  {x}])\recencodf{B'})\widetilde{[x_1\leftarrow x_k]} \\
            &\red_{\redlab{RS:Beta}} (\recencodf{ M^{''} \langle{ {\widetilde{y}}/ {x}} \rangle} [ {\widetilde{y}} \leftarrow  {x}] \esubst{\recencodf{B'}}{x}) \widetilde{[x_1\leftarrow x_k]}=\expr{L}
            \end{aligned}
        \end{equation}
      
        On the other hand, we have:
        \begin{equation}\label{eq:beta2fail}
            \begin{aligned}
               \recencodopenf{\expr{M}}&=\recencodopenf{M'\esubst{B}{x}}=\recencodf{M'\esubst{B}{x}\langle{ {\widetilde{x_1}}/ {x_1}}\rangle\cdots \langle{ {\widetilde{x_k}}/ {x_k}}\rangle} \widetilde{[x_1\leftarrow x_n]}\\
               &=\recencodf{M^{''}\esubst{B'}{x}} \widetilde{[x_1\leftarrow x_k]}
            \end{aligned}
        \end{equation}

        We need to analyze two sub-cases: either $\#( {x},M) = \size{C} $ or $\#( {x},M) = k \geq 0$ and our first sub-case is not met.
        \begin{enumerate}
            \item If $\#( {x},M) = \size{C}  $ then we can reduce $\expr{L}$ as: (via {\redlab{RS:Ex-sub}})
            \begin{equation*}
                \begin{aligned}
            \expr{L}\red & \sum_{C_i\in \perm{\recencodf{C}}}\recencodf{M^{''}\linsub{ {\widetilde{y}}}{ {x}}}\linexsub{C_i(1)/ {y_1}}\cdots \linexsub{C_i(n)/ {y_n}} \unexsub{U/ \unvar{x}} \widetilde{[x_1\leftarrow x_k]}         =\recencodopenf{\expr{M}}
                \end{aligned}
            \end{equation*}
            
            From \eqref{eq:beta1fail} and \eqref{eq:beta2fail} and $ {\widetilde{y}}= {y_1}\ldots  {y_n}$, one has the result.

            \item Otherwise,  $\#( {x},M) = n \geq 0$.
            


            Expanding  the encoding in \eqref{eq:beta2fail} :
            \begin{align*}
                \recencodopenf{M}&= \recencodf{M^{''}\esubst{B'}{x}} \widetilde{[x_1\leftarrow x_k]} 
                = (\recencodf{ M^{''} \langle{ {\widetilde{y}}/ {x}} \rangle} [ {\widetilde{y}} \leftarrow  {x}] \esubst{\recencodf{B'}}{x}) \widetilde{[x_1\leftarrow x_k]}
            \end{align*}
           Therefore  $\recencodopenf{M} =\expr{L}$ and $\recencodopenf{\expr{N}}\red \recencodopenf{\expr{M}}$.
        
        \end{enumerate}

    \item The rule applied is $\redlab{R}=\redlab{R:Fetch^{\ell}}$. 
    
    Then $\expr{N}=M\esubst{ C \bagsep U  }{x } $ and  the reduction is 

    \begin{prooftree}
        \AxiomC{$\headf{M} =  {x}$}
        \AxiomC{$C = {\bag{N_1}}\cdot \dots \cdot {\bag{N_k}} \ , \ k\geq 1 $}
        \AxiomC{$ \#( {x},M) = k $}
        \LeftLabel{\redlab{R:Fetch^{\ell}}}
        \TrinaryInfC{\(
        M\esubst{ C \bagsep U  }{x } \red M \headlin{ N_{1}/ {x} } \esubst{ (C \setminus N_1)\bagsep U}{ x }  + \cdots + M \headlin{ N_{k}/ {x} } \esubst{ (C \setminus N_k)\bagsep U}{x}
        \)}
    \end{prooftree}

    with $\expr{M}= M \headlin{ N_{1}/ {x} } \esubst{ (C \setminus N_1)\bagsep U}{ x }  + \cdots + M \headlin{ N_{k}/ {x} } \esubst{ (C \setminus N_k)\bagsep U}{x}$.

    Below we assume $\lfv{\expr{N}}=\{x_1,\ldots, x_k\}$ and $\widetilde{x_i}=x_{i_1},\ldots, x_{i_{j_i}}$, where $j_i= \#(x_i, N)$, for $1\leq i\leq k$.
    On the one hand, we have: (last rule is {\redlab{RS{:}Fetch^{\ell}}})
        \begin{equation*}
            \begin{aligned}
            \recencodopenf{\expr{N}}&= \recencodopenf{M\esubst{ C \bagsep U  }{x}}=  \recencodf{M\esubst{ C \bagsep U  }{x } \langle{ {\widetilde{x_1}}/ {x_1}}\rangle\cdots \langle{ {\widetilde{x_k}}/ {x_k}}\rangle}\widetilde{[x_1\leftarrow x_k]}
            \\
            &=  \recencodf{ M'\esubst{ C' \bagsep U  }{x } }\widetilde{[x_1\leftarrow x_k]}\\
            &=   \sum_{C_i \in \perm{\recencodf{ C'  } }}  (\recencodf{ M' \langle  {\widetilde{y}}/  {x}  \rangle } \linexsub{C_i(1)/ {y}_1} \cdots \linexsub{C_i(k)/ {y}_k}\unexsub{U/\unvar{x}} )\widetilde{[x_1\leftarrow x_k]} \\
            &=  \sum_{C_i \in \perm{\recencodf{ C'  } }}  (\recencodf{ M'' } \linexsub{C_i(1)/ {y}_1} \cdots \linexsub{C_i(k)/ {y}_k}\unexsub{U/\unvar{x}} )\widetilde{[x_1\leftarrow x_k]} \\
            &\red \sum_{C_i \in \perm{\recencodf{ C'  } }}  (\recencodf{ M''\headlin{ C_i(1) /  {y}_1 } } \linexsub{C_i(2)/ {y}_2} \cdots \linexsub{C_i(k)/ {y}_k}\unexsub{U/\unvar{x}} )\widetilde{[x_1\leftarrow x_k]} \\
             &=\expr{L}
            \end{aligned}
        \end{equation*}
        We assume for simplicity that $  \headf{M''} =  {y}_1$
        On the other hand, we have:
        \begin{equation*}
            \begin{aligned}
               \recencodopenf{\expr{M}}&=\recencodopenf{M \headlin{ N_{1}/ {x} } \esubst{ (C \setminus N_1)\bagsep U}{ x }  + \cdots + M \headlin{ N_{k}/ {x} } \esubst{ (C \setminus N_k)\bagsep U}{x}}\\
               &=\sum_{C_i \in \perm{\recencodf{ C'  } }}  (\recencodf{ M''\headlin{ C_i(1) /  {y}_1 } } \linexsub{C_i(2)/ {y}_2} \cdots \linexsub{C_i(k)/ {y}_k}\unexsub{U/\unvar{x}} )\widetilde{[x_1\leftarrow x_k]} \\
               &=\expr{L}
            \end{aligned}
        \end{equation*}
        From these developments from $\recencodopenf{\mathbb{N}}$ and $\recencodopenf{\mathbb{M}}$ ,  and $ {\widetilde{y}}= {y_1}\ldots  {y_n}$, one has the result.

    \item The rule applied is $\redlab{R}=\redlab{R:Fetch^!}$. 
    
    Then $\expr{N}=M\esubst{ C \bagsep U  }{x } $ and  the reduction is 

    \begin{prooftree}
        \AxiomC{$\headf{M} = {x}[i] \quad U_i = \banged{\bag{N}} $}
        \LeftLabel{\redlab{R:Fetch^!}}
        \UnaryInfC{\(
        M\ \esubst{ C \bagsep U  }{x } \red M \headlin{ N/ {x}[i] } \esubst{ C \bagsep U}{ x } 
        \)}
    \end{prooftree}   

    with $\expr{M}= M \headlin{ N/ {x}[i] } \esubst{ C \bagsep U}{ x }$.
    Below we assume $\lfv{\expr{N}}=\{x_1,\ldots, x_k\}$ and $\widetilde{x_i}=x_{i_1},\ldots, x_{i_{j_i}}$, where $j_i= \#(x_i, N)$, for $1\leq i\leq k$.

    On the one hand, we have: (the last rule is {\redlab{RS:Fetch^!}})
        {\small
        \begin{equation}\label{eq:unfetch1fail}
            \begin{aligned}
            \recencodopenf{\expr{N}}&= \recencodopenf{M\esubst{ C \bagsep U  }{x}}=  \recencodf{M\esubst{ C \bagsep U  }{x } \langle{ {\widetilde{x_1}}/ {x_1}}\rangle\cdots \langle{ {\widetilde{x_k}}/ {x_k}}\rangle}\widetilde{[x_1\leftarrow x_k]}
            \\
            &=  \recencodf{ M'\esubst{ C' \bagsep U  }{x } }\widetilde{[x_1\leftarrow x_k]}]\\
            &=   \sum_{C_i \in \perm{\recencodf{ C'  } }}  (\recencodf{ M' \langle  {\widetilde{y}}/  {x}  \rangle } \linexsub{C_i(1)/ {y}_1} \cdots \linexsub{C_i(k)/ {y}_k}\unexsub{U/ \unvar{x}} )\widetilde{[x_1\leftarrow x_k]} \\
            &=  \sum_{C_i \in \perm{\recencodf{ C'  } }}  (\recencodf{ M'' } \linexsub{C_i(1)/ {y}_1} \cdots \linexsub{C_i(k)/ {y}_k}\unexsub{U/ \unvar{x}} )\widetilde{[x_1\leftarrow x_k]}
            \\ &\red \sum_{C_i \in \perm{\recencodf{ C'  } }}  (\recencodf{ M''\headlin{ N / {x}[i]  } } \linexsub{C_i(2)/ {y}_2} \cdots \linexsub{C_i(k)/ {y}_k}\unexsub{U/ \unvar{x}} )\widetilde{[x_1\leftarrow x_k]}  =\expr{L}
            \end{aligned}
        \end{equation}}
        On the other hand, assuming for simplicity that $  \headf{M''} = {x}[i]$ and $U_i = N$, we have 
        {\small
        \begin{equation}\label{eq:unfetch2fail}
            \begin{aligned}
            \recencodopenf{\expr{M}}&=\recencodopenf{M \headlin{ N/ {x}[i] } \esubst{ C \bagsep U}{ x } }\\
               &=\sum_{C_i \in \perm{\recencodf{ C'  } }}  (\recencodf{ M''\headlin{ N / {x}[i]  } } \linexsub{C_i(2)/ {y}_2} \cdots \linexsub{C_i(k)/ {y}_k}\unexsub{U/\unvar{x}} )\widetilde{[x_1\leftarrow x_k]}=\expr{L}
            \end{aligned}
        \end{equation}}
        From \eqref{eq:unfetch1fail} and \eqref{eq:unfetch2fail}, one has the result.

        \item The rule applied is $ \redlab{R}\neq \redlab{R:Beta}$ and $ \redlab{R}\neq \redlab{R:Fetch}$. There are two possible cases:
            \begin{enumerate}
                 
                \item  $\redlab{R}=\redlab{R:Fail^{\ell}}$
                
                Then $\expr{N}=M\esubst{C \bagsep U}{x }$ and the reduction is 
                
                \begin{prooftree}    
                        \AxiomC{$\#( {x},M) \neq \size{C} $} 
                         \AxiomC{$\widetilde{z} = (\mlfv{M} \!\setminus x) \uplus \mlfv{C} $}
                        \LeftLabel{\redlab{R:Fail^{\ell}}}
                        \BinaryInfC{\(  M\esubst{C \bagsep U}{x } \red \sum_{\perm{C}} \fail^{\widetilde{z}} \)}
                    \end{prooftree}
                where $\expr{M}=  \sum_{\perm{C}} \fail^{\widetilde{y}}$. Below assume $\lfv{\expr{N}}=\{x_1,\ldots, x_n\}$.
                
                    On the one hand, we have:
     \begin{equation*}\label{eq:fail1fail}
                    \begin{aligned}
                    \recencodopenf{\expr{N}} &= \recencodopenf{M\esubst{C \bagsep U}{x }}= \recencodf{M\esubst{C \bagsep U}{x }\langle  {\widetilde{x_1}}/ {x_1}\rangle\cdots \langle  {\widetilde{x_n}}/ {x_n}\rangle } \widetilde{[x_1\leftarrow x_n]} \\
                     &= \recencodf{M'\esubst{C' \bagsep U}{x } } \widetilde{[x_1\leftarrow x_n]} \\
                    &=  \recencodf{M'\langle y_1, \cdots ,  {y_k} /  {x}  \rangle} [ {y_1}, \cdots ,  {y_k} \leftarrow  {x}] \esubst{C' \bagsep U}{x }\widetilde{[x_1\leftarrow x_n]}\\
                        & \red_{\redlab{RS:Fail^{\ell}}} \sum_{\perm{C}} \fail^{\widetilde{y},  {\widetilde{x_1}}, \cdots ,  {\widetilde{x_n}}}\widetilde{[x_1\leftarrow x_n]}    =\expr{L}\\
                    \end{aligned}
                \end{equation*}
                
                On the other hand, we have:
            \begin{equation*}\label{eq:fail2fail}
                    \begin{aligned}
                    \recencodopenf{\expr{M}}  &= \sum_{\perm{C}} \recencodopenf{\fail^{\widetilde{z}}}
                    = \sum_{\perm{C}} \fail^{\widetilde{y},  {\widetilde{x_1}}, \cdots ,  {\widetilde{x_n}}}\widetilde{[x_1\leftarrow x_n]}    =\expr{L}
                    \end{aligned}
                \end{equation*}
                
                Therefore, $\recencodopenf{\expr{N}} \red \recencodopenf{\expr{M}} $ and the result follows.

                \item  $\redlab{R}=\redlab{R:Fail^!}$
                
                Then $\expr{N}=M\esubst{C \bagsep U}{x }$ and the reduction is 
                
                \begin{prooftree} 
                    \AxiomC{$\#( {x},M) = \size{C} $}
                    \AxiomC{$U_i = \banged{\oneb}$}
                    \AxiomC{$\headf{M} = {x}[i]  $}
                    \LeftLabel{\redlab{R:Fail^!}}
                    \TrinaryInfC{\(  M \esubst{C \bagsep U}{x } \red  M \headlin{ \fail^{\emptyset} / {x}[i] } \esubst{ C \bagsep U}{ x }\)}
                \end{prooftree}
                
                where $\expr{M}=  M \headlin{ \fail^{\emptyset} / {x}[i] } \esubst{ C \bagsep U}{ x } $. 
                

                Below we assume $\lfv{\expr{N}}=\{x_1,\ldots, x_k\}$ and $\widetilde{x_i}=x_{i_1},\ldots, x_{i_{j_i}}$, where $j_i= \#(x_i, N)$, for $1\leq i\leq k$.
            
                On the one hand, we have: (the last rule applied was {\redlab{RS:Fail^!}} )
                   {\small  \begin{equation*}\label{eq:unfetch1failp2}
                        \begin{aligned}
                        \recencodopenf{\expr{N}}&= \recencodopenf{M\esubst{ C \bagsep U  }{x}}=  \recencodf{M\esubst{ C \bagsep U  }{x } \langle{ {\widetilde{x_1}}/ {x_1}}\rangle\cdots \langle{ {\widetilde{x_k}}/ {x_k}}\rangle}\widetilde{[x_1\leftarrow x_k]}
                        \\
                        &=  \recencodf{ M'\esubst{ C' \bagsep U  }{x } }\widetilde{[x_1\leftarrow x_k]}\\
                        &=   \sum_{C_i \in \perm{\recencodf{ C'  } }}  (\recencodf{ M' \langle  {\widetilde{y}}/  {x}  \rangle } \linexsub{C_i(1)/ {y}_1} \cdots \linexsub{C_i(k)/ {y}_k}\unexsub{U/\unvar{x} } )\widetilde{[x_1\leftarrow x_k]}\\
                        &=  \sum_{C_i \in \perm{\recencodf{ C'  } }}  (\recencodf{ M'' } \linexsub{C_i(1)/ {y}_1} \cdots \linexsub{C_i(k)/ {y}_k}\unexsub{U/ \unvar{x} } )\widetilde{[x_1\leftarrow x_k]} \\
                        &\red \sum_{C_i \in \perm{\recencodf{ C'  } }}  (\recencodf{ M''\headlin{  \fail^{\emptyset} / {x}[i]  } } \linexsub{C_i(2)/ {y}_2} \cdots \linexsub{C_i(k)/ {y}_k}\unexsub{U/\unvar{x}} )\widetilde{[x_1\leftarrow x_k]}\\
                        &=\expr{L}
                        \end{aligned}
                    \end{equation*}
                    }
                    We assume for simplicity that $  \headf{M''} = {x}[i]$.
                    On the other hand, we have:
                   {\small  \begin{equation*}\label{eq:unfetch2failp2}
                        \begin{aligned}
                           \recencodopenf{\expr{M}}&=\recencodopenf{M \headlin{ \fail^{\emptyset}/ {x}[i] } \esubst{ C \bagsep U}{ x } }\\
                           &=\sum_{C_i \in \perm{\recencodf{ C'  } }}  (\recencodf{ M''\headlin{ \fail^{\emptyset} / {x}[i]  } } \linexsub{C_i(2)/ {y}_2} \cdots \linexsub{C_i(k)/ {y}_k}\unexsub{U/\unvar{x}} )\widetilde{[x_1\leftarrow x_k]}  \\
                           &=\expr{L}
                        \end{aligned}
                    \end{equation*}
            }
                 From the $\recencodopenf{M}$ and  $\recencodopenf{N}$ above one has the result.

            \item  $\redlab{R}= \redlab{R:Cons_1}$.
            
            Then $\expr{N}=(\fail^{\widetilde{z}})\ C \bagsep U $ and the reduction is

            \begin{prooftree}
                    \AxiomC{$\widetilde{z} = \mlfv{C} $}
                \LeftLabel{$\redlab{R:Cons_1}$}
                \UnaryInfC{\(  (\fail^{\widetilde{y}})\ C \bagsep U \red {}  \sum_{\perm{C}} \fail^{\widetilde{y} \uplus \widetilde{z}} \)}
            \end{prooftree}
            

           and $\expr{M}'= \sum_{\perm{B}} \fail^{\widetilde{y} \uplus \widetilde{z}}$. Below we assume $\lfv{\expr{N}}=\{x_1,\ldots, x_n\}$.
           
                On the one hand, we have:  
            
            \begin{equation*}\label{eq:consume1fail}
                \begin{aligned}
                \recencodopenf{N} &= \recencodopenf{\fail^{\widetilde{y}}\ B}= \recencodf{ \fail^{\widetilde{y}}\  C \bagsep U \langle \widetilde{x_1}/x_1\rangle\cdots \langle \widetilde{x_n}/x_n\rangle } \widetilde{[x_1\leftarrow x_n]}\\
                 &= \recencodf{ \fail^{\widetilde{y'}}\ \ C' \bagsep U  } \widetilde{[x_1\leftarrow x_n]}
                = \recencodf{ \fail^{\widetilde{y'}}} \ \recencodf{ C' \bagsep U } \widetilde{[x_1\leftarrow x_n]}\\
                &= \fail^{\widetilde{y'}} \ \recencodf{ C' \bagsep U } \widetilde{[x_1\leftarrow x_n]}
                 \red_{\redlab{RS:Cons_1}} \sum_{\perm{B}} \fail^{\widetilde{y'} \cup \widetilde{z'}}  \widetilde{[x_1\leftarrow x_n]}=\expr{L}\\
                \end{aligned}
            \end{equation*}
            Where $\widetilde{y'} \cup \widetilde{z'} = \widetilde{x_1} , \cdots , \widetilde{x_n}$. On the other hand, we have:
            
            \begin{equation*}\label{eq:consume2fail}
                \begin{aligned}
                \recencodopenf{M}  
                = \sum_{\perm{B}} \recencodf{\fail^{\widetilde{y'} \uplus \widetilde{z'}}}\widetilde{[x_1\leftarrow x_n]}
                = \sum_{\perm{B}} \fail^{\widetilde{y'} \cup \widetilde{z'}}\widetilde{[x_1\leftarrow x_n]}=\expr{L}
                \end{aligned}
            \end{equation*}
            
            Therefore, $\recencodopenf{\expr{N}}\red \expr{L}= \recencodopenf{ \expr{M}}$, and the result follows.

            \item $\redlab{R}= \redlab{R:Cons_2}$
            
            Then $\expr{N}= \fail^{\widetilde{y}}\ \esubst{C \bagsep U}{z} $ and the reduction is 
            
            \begin{prooftree}
                \AxiomC{$ \#(z , \widetilde{y}) =  \size{C}\quad \widetilde{z} = \mlfv{C} $}
                \LeftLabel{$\redlab{R:Cons_2}$}
                \UnaryInfC{$\fail^{\widetilde{y}}\ \esubst{C \bagsep U}{z}  \red {}\displaystyle\sum_{\perm{C}} \fail^{(\widetilde{x} \setminus z) \uplus\widetilde{z}}$}
            \end{prooftree}
            
      and $\expr{M}=\sum_{\perm{C}} \fail^{(\widetilde{y} \setminus x) \uplus\widetilde{z}}$. 
            Below we assume $\lfv{\expr{N}}=\{x_1,\ldots, x_n\}$.
            
            
            On the one hand, we have:
            
            \begin{equation}\label{eq:consume3fail}
                \begin{aligned}
                \recencodopenf{\expr{N}} &= \recencodopenf{\fail^{\widetilde{y}}\ \esubst{C \bagsep U}{z}}
                = \recencodf{ \fail^{\widetilde{y}}\ \esubst{C \bagsep U}{z} \langle \widetilde{x_1}/x_1\rangle\cdots \langle \widetilde{x_n}/x_n\rangle } \widetilde{[x_1\leftarrow x_n]}\\
                &=  \sum_{C_i \in \perm{\recencodf{ C'  } }}\recencodf{ \fail^{\widetilde{y'}} \langle  {\widetilde{y}}/  {x}  \rangle } \linexsub{C_i(1)/ {y}_1} \cdots \linexsub{C_i(k)/ {y}_k}\unexsub{U/\unvar{x}} \widetilde{[x_1\leftarrow x_n]}\\
                & \red^*_{\redlab{RS{:}Cons_3}} \sum_{C_i \in \perm{\recencodf{ C'  } }}\recencodf{ \fail^{(\widetilde{y'} \setminus  {\widetilde{y}} ) \uplus\widetilde{z}}} \unexsub{U/\unvar{x}} \widetilde{[x_1\leftarrow x_n]}\\
                & \red^*_{\redlab{RS{:}Cons_4}} \sum_{C_i \in \perm{\recencodf{ C'  } }}\recencodf{ \fail^{(\widetilde{y'} \setminus  {\widetilde{y}} ) \uplus\widetilde{z}}} \widetilde{[x_1\leftarrow x_n]}\\
                \end{aligned}
            \end{equation}
            
            As $\widetilde{y}$ consists of free variables, we have that in $\fail^{\widetilde{y}}\ \esubst{C \bagsep U}{x} \langle \widetilde{x_1}/x_1\rangle\cdots \langle \widetilde{x_n}/x_n\rangle$ the substitutions also occur on $ \widetilde{y}$ resulting in a new $\widetilde{y'}$ where all $x_i$'s are replaced with their fresh components in $\widetilde{x_i}$. Similarly $\widetilde{y''}$ is $\widetilde{y'}$ with each $x$ replaced with a fresh $y_i$. On the other hand, we have:
            
            \begin{equation}\label{eq:consume4fail}
                \begin{aligned}
                \recencodopenf{M} = \recencodopenf{\sum_{\perm{C}} \fail^{(\widetilde{y} \setminus x) \uplus\widetilde{z}}}
                &= \sum_{C_i \in \perm{\recencodf{ C'  } }}\recencodf{ \fail^{(\widetilde{y'} \setminus  {\widetilde{y}} ) \uplus\widetilde{z}}} \widetilde{[x_1\leftarrow x_n]}
                \end{aligned}
            \end{equation}
            
            The reductions in \eqref{eq:consume3fail} and \eqref{eq:consume4fail} lead to identical expressions.
        \end{enumerate}
    \end{enumerate}
    
     As before, the reduction via rule $ \redlab{R} $ could occur inside a context (cf. Rules $\redlab{R:TCont}$ and $\redlab{R:ECont}$). We consider only the case when the contextual rule used is $\redlab{R:TCont}$. We have $\expr{N} = C[N]$. When we have $C[N] \red_{\redlab{R}} C[M] $ such that $N \red_{\redlab{R}} M$ we need to show that $\recencodopenf{ C[N]} \red^j \recencodopenf{ C[M] }$for some $j$ dependent on ${\redlab{R}}$. Firstly let us assume ${\redlab{R}} = \redlab{R:Cons_2}$  then we take $j = 1$. Let us take $C[\cdot]$ to be $[\cdot]B$ and $\lfv{NB} = \{ x_1, \cdots , x_k  \}$ then 
     \[
        \begin{aligned}
            \recencodopenf{ N B}  & = \recencodf{NB\linsub{\widetilde{x_{1}}}{x_1}\cdots \linsub{\widetilde{x_k}}{x_k}}\widetilde{[x_1\leftarrow x_k]} = \recencodf{N' B'}\widetilde{[x_1\leftarrow x_k]}  = \recencodf{N'}\recencodf{ B'}\widetilde{[x_1\leftarrow x_k]} \\
        \end{aligned}
     \]
    We take $N'B'= NB\linsub{\widetilde{x_{1}}}{x_1}\cdots \linsub{\widetilde{x_k}}{x_k}$, we have by the IH that $ \recencodf{N}\red \recencodf{M}$ and hence we can deduce that $\recencodf{N'}\red \recencodf{M'}$ where $M'B'= MB\linsub{\widetilde{x_{1}}}{x_1}\cdots \linsub{\widetilde{x_k}}{x_k}$. Finally we have
    \(            \recencodf{N'}\recencodf{ B'}\widetilde{[x_1\leftarrow x_k]} \red \recencodf{M'}\recencodf{ B'}\widetilde{[x_1\leftarrow x_k]}
        \)
     and hence $ \recencodopenf{C[N]} \red \recencodopenf{C[M]} $.
\end{proof}

\begin{theorem}[Operational Soundness]
\label{l:soundnessoneunres}
Let $\expr{N}$ be a well-formed $\lamrfailunres$ expression. 
Suppose $ \recencodopenf{\expr{N}}  \red \expr{L}$. Then, there exists $ \expr{N}' $ such that $ \expr{N}  \red_{\redlab{R}} \expr{N}'$ and 

\begin{enumerate}
 \item If $\redlab{R} = \redlab{R:Beta}$ then $\expr{ L } \red^{\leq 1} \recencodopenf{\expr{N}'}$;

    \item If $\redlab{R} \neq \redlab{R:Beta}$ then $\expr{ L } \red^*  \recencodopenf{\expr{N}''}$, for $ \expr{N}''$ such that  $\expr{N}' \pequiv \expr{N}''$.
\end{enumerate}
\end{theorem}

\begin{proof}
By induction on the structure of $\expr{N}$:

\begin{enumerate}

    \item Cases $\expr{N} =  {x}$, $\expr{N} =  x[i]$, $\fail^{\widetilde{y}}$ and $\expr{N} =  \lambda x. N$, are trivial, since no reductions can be performed.
    
    
    



    \item $\expr{N} = N B$:
    
    Suppose  $\llfv{NB} = \{  {x}_1, \cdots ,  {x}_n\}$. Then,

    \begin{equation}\label{eq:app_npfail}
    \begin{aligned}
        \recencodopenf{\expr{N}}=\recencodopenf{NB}  = \recencodf{NB\langle  {\widetilde{x_1}} /  {x}_1 \rangle \cdots \langle  {\widetilde{x_n}} /  {x}_n \rangle} \widetilde{[x_1\leftarrow x_n]}
         &= \recencodf{N' B'} \widetilde{[x_1\leftarrow x_n]}\\
        & = \recencodf{N'} \recencodf{B'} \widetilde{[x_1\leftarrow x_n]}
    \end{aligned}
    \end{equation}
    where $ {\widetilde{x_i}}= {x}_{i1},\ldots,  {x}_{ij_i}$, for $1\leq i \leq n$.
    By the reduction rules in \figref{fig:share-reductfailureunres} there are three possible reductions starting in $\expr{N}$:
    \begin{enumerate}
        \item $\recencodf{N'}\recencodf{B'}\widetilde{[x_1\leftarrow x_n]}$ reduces via a $\redlab{RS:Beta}$.
        
        In this case  $N=\lambda x. N_1$, and the encoding in (\ref{eq:app_npfail}) gives $N'= N\langle  {\widetilde{x_1}} /  {x}_1 \rangle \cdots \langle  {\widetilde{x_n}} /  {x}_n \rangle$, which  implies $N' =\lambda x. N_1^{'}$ and the following holds:
        \begin{equation*}
        \begin{aligned}
            \recencodf{N'}=\recencodf{(\lambda x. N'_1)} &= (\lambda x. \recencodf{N'_1 \langle  {\widetilde{y}} /  {x} \rangle}  {[\widetilde{y}} \leftarrow  {x}])
             = (\lambda x. \recencodf{N^{''}} [ {\widetilde{y}} \leftarrow  {x}])
        \end{aligned}
        \end{equation*}
        Thus, we have the following $\redlab{RS:Beta}$ reduction from \eqref{eq:app_npfail}:
        \begin{equation}\label{eq:sound.appfail}
            \begin{aligned}
                \recencodopenf{\expr{N}} &= \recencodf{N'} \recencodf{B'}\widetilde{[x_1\leftarrow x_n]}
                =(\lambda x. \recencodf{N''} [ {\widetilde{y}} \leftarrow  {x}] \recencodf{B'}) \widetilde{[x_1\leftarrow x_n]}\\
                &\red_{\redlab{RS:Beta}}  \recencodf{N^{''}} [ {\widetilde{y}} \leftarrow  {x}] \esubst{\recencodf{B'}}{x}  \widetilde{[x_1\leftarrow x_n]} =\expr{L}
            \end{aligned}
        \end{equation}

        Notice that the expression $\expr{N}$ can perform the following $\redlab{R:Beta}$ reduction:
        \[\expr{N}=(\lambda x. N_1) B\red_{\redlab{R:Beta}} N_1 \esubst{B}{x} \]

        Assuming $\expr{N'}=N_1 \esubst{B}{x}$ and we take $B = C \bagsep U$, there are two cases:
        
        \begin{enumerate}
            
            \item$\#( {x},M) = \size{C} = k$.

            On the one hand, 
            
            \begin{equation*}\label{eq:sound_appn1fail}
            \begin{aligned}
                \recencodopenf{\expr{N'}}&=\recencodopenf{N_1 \esubst{B}{x}}
                = \recencodf{N_1 \esubst{B}{x}\langle  {\widetilde{x_1}} /  {x}_1 \rangle \cdots \langle  {\widetilde{x_n}} /  {x}_n \rangle} \widetilde{[x_1\leftarrow x_n]}\\
                & = \recencodf{N_1' \esubst{B'}{x}}\widetilde{[x_1\leftarrow x_n]}\\
                & = \sum_{C_i \in \perm{\recencodf{ C' }}}\recencodf{ N_1' \langle  {y}_1 , \cdots ,  {y}_k / x  \rangle } \linexsub{C_i(1)/ {y}_1} \cdots \linexsub{C_i(k)/ {y}_k}\unexsub{U /\unvar{x}} \widetilde{[x_1\leftarrow x_n]}\\
                & = \sum_{C_i \in \perm{\recencodf{ C' }}}\recencodf{ N_1''} \linexsub{C_i(1)/ {y}_1} \cdots \linexsub{C_i(k)/ {y}_k}\unexsub{U /\unvar{x}} \widetilde{[x_1\leftarrow x_n]}\\
            \end{aligned}
            \end{equation*}
            On the other hand, via application of rule \redlab{RS:Ex\dash Sub}
            \begin{equation*}\label{eq:sound_appn2fail}
            \begin{aligned}
                \expr{L} &= \recencodf{N''} [ {\widetilde{y}} \leftarrow  {x}] \esubst{\recencodf{B'}}{x}  \widetilde{[x_1\leftarrow x_n]} \\
                &\red \sum_{C_i \in \perm{\recencodf{ C }}}\recencodf{ N_1''} \linexsub{C_i(1)/ {y}_1} \cdots \linexsub{C_i(k)/ {y}_k}\unexsub{U /\unvar{x}} \widetilde{[x_1\leftarrow x_n]} = \recencodopenf{\expr{N}'}
            \end{aligned}
            \end{equation*}
            
         and the result follows.
            
            \item Otherwise $\#( {x},N_1)\neq \size{C}$.

            In this case, 
            \begin{equation*}\label{eq:sound_appn3fail}
            \begin{aligned}
                \recencodopenf{\expr{N'}}&=\recencodopenf{N_1 \esubst{B}{x}}= \recencodf{N_1 \esubst{B}{x}\langle  {\widetilde{x_1}} /  {x}_1 \rangle \cdots \langle  {\widetilde{x_n}} /  {x}_n \rangle}\widetilde{[x_1\leftarrow x_n]}\\
                & = \recencodf{N_1' \esubst{B'}{x}}\widetilde{[x_1\leftarrow x_n]}
                 =  \recencodf{N^{''}} [\widetilde{y} \leftarrow x] \esubst{\recencodf{B'}}{x} \widetilde{[x_1\leftarrow x_n]} = \expr{L} 
            \end{aligned}
            \end{equation*}
            From (\ref{eq:sound.appfail}):  $\recencodopenf{\expr{N}}\red \expr{L}=\recencodopenf{\expr{N'}}$ and the result follows.

        \end{enumerate}
        
        \item $\recencodf{N'}\recencodf{B'}\widetilde{[x_1\leftarrow x_n]}$ reduces via a $\redlab{RS: Cons_1}$.

        In this case,  $N=\fail^{\widetilde{y}}$, and the encoding in (\ref{eq:app_npfail}) gives $N'= N\linsub{ {\widetilde{x_1}}}{ {x}_1}\ldots \linsub{ {\widetilde{x_n}}}{ {x}_n}$, which implies $N'
        =\fail^{\widetilde{y'}} $, we let $B = C \bagsep U$ and the following:
        \begin{equation*}\label{eq:sound.consumfail}
            \begin{aligned}
                \recencodopenf{\expr{N}} &= \recencodf{N'} \recencodf{B'}\widetilde{[x_1\leftarrow x_n]}= \recencodf{\fail^{\widetilde{y'}}} \recencodf{B'}\widetilde{[x_1\leftarrow x_n]}\\
                &= \fail^{\widetilde{y'}} \recencodf{B'}\widetilde{[x_1\leftarrow x_n]} \red \sum_{\perm{C}} \fail^{\widetilde{y'} \uplus \widetilde{z}}\widetilde{[x_1\leftarrow x_n]}, \text{ where } \widetilde{z} = \llfv{C'}.
            \end{aligned}
        \end{equation*}
         The expression $\expr{N}$ can perform the  reduction:
        
        \begin{equation*}\label{eq:sound.2consumfail}
            \expr{N}=\fail^{\widetilde{y}} \  B\red_{\redlab{R: Cons1}} \sum_{\perm{C}} \fail^{\widetilde{y}\uplus \widetilde{z}}, \text{  where } \widetilde{z} = \mlfv{C}
        \end{equation*}
        
        Thus, $\expr{L}=\recencodopenf{\expr{N'}}$ and so the result follows.
        
        \item Suppose that $\recencodf{N'} \red \recencodf{N''}$.  
        This case follows from the induction hypothesis.
    \end{enumerate}

    \item  $\expr{N} = N \esubst{B}{x}$:
    
    Suppose  $\llfv{N \esubst{B}{x}} = \{  {x}_1, \cdots ,  {x}_k\}$. Then,
    
    \begin{equation}\label{eq:sound_expsubfail}
    \begin{aligned}
    \recencodopenf{\expr{N}}= \recencodopenf{N \esubst{B}{x}}
         &= \recencodf{N \esubst{B}{x}\langle  {\widetilde{x_1}} /  {x}_1 \rangle \cdots \langle  {\widetilde{x_k}} /  {x}_k \rangle} \widetilde{[x_1\leftarrow x_k]}\\
         &= \recencodf{N' \esubst{B'}{x}} \widetilde{[x_1\leftarrow x_k]}
        \end{aligned}
    \end{equation}
    
    Let us consider the two possibilities of the encoding where we take $B = C \bagsep U$:
    
    \begin{enumerate}
        
        \item Where $ \#( {x},M) = \size{B} = k $
        
        Then we continue equation (\ref{eq:sound_expsubfail}) as follows
        
        \begin{equation}\label{eq:sound_expsub2fail}
            \begin{aligned}
                \recencodopenf{\expr{N}} &= \recencodf{N' \esubst{B'}{x}} [\widetilde{[x_1\leftarrow x_k]}\\
                &=  \sum_{C_i \in \perm{\recencodf{ C' }}}\recencodf{ N' \langle  {y}_1 , \cdots ,  {y}_n /  {x}  \rangle } \linexsub{C_i(1)/ {y}_1} \cdots \linexsub{C_i(n)/ {y}_n} \unexsub{U/\unvar{x}}\widetilde{[x_1\leftarrow x_k]} \\
                &=  \sum_{C_i \in \perm{\recencodf{ C' }}}\recencodf{ N'' } \linexsub{C_i(1)/ {y}_1} \cdots \linexsub{C_i(n)/ {y}_n} \unexsub{U/\unvar{x}} \widetilde{[x_1\leftarrow x_k]} \\
            \end{aligned}
        \end{equation}
        
        There are five possible reductions that can take place, these being $\redlab{RS{:}Fetch^{\ell}}$, $\redlab{RS{:} Fetch^!}$, $\redlab{RS{:}Fail^!}$ , $\redlab{RS:Cons_3}$ and when we apply the \redlab{RS:Cont} rules
        
        \begin{enumerate}
            
            \item Suppose that $\headf{N''} =  {y}_1$ and for simplicity we assume $C'$ has only one element $N_1$ then from (\ref{eq:sound_expsub2fail}) and buy letting $C' = \bag{N'_1}$ we have
                    
                    \begin{equation*}
                        \begin{aligned}
                            \recencodopenf{\expr{N}} 
                            & = \recencodf{ N'' } \linexsub{\recencod{N_1'}/ {y}_1} \unexsub{U/\unvar{x}} \widetilde{[x_1\leftarrow x_k]}\\
                            &\red \recencodf{N^{''}}\headlin{\recencod{N_1'}/ {y}_1}\unexsub{U/\unvar{x}}\widetilde{[x_1\leftarrow x_k]} = \expr{L}
                        \end{aligned}
                    \end{equation*}
Also,
\(\expr{N} 
= N\esubst{\bag{N_1} \bagsep U}{x}                            \red N\headlin{N_1/x}\esubst{ \oneb \bagsep U}{x} = \expr{N}'.\)
Then  $\expr{L}'=\recencodopenf{\expr{N'}}$ and the result follows.

        \item Suppose that $\headf{N''} = {x}[i]$ and then from (\ref{eq:sound_expsub2fail}) we have
                    
\begin{equation*}\label{eq:sound_expsub3failp2}
\begin{aligned}
\recencodopenf{\expr{N}} 
& = \recencodf{ N'' } \linexsub{C_i(1)/ {y}_1} \cdots \linexsub{C_i(n)/ {y}_n} \unexsub{U/\unvar{x}}\widetilde{[x_1\leftarrow x_k]}\\
&\red \recencodf{N^{''}}\headlin{ U_{i} / {x}[i] } \linexsub{C_i(1)/ {y}_1} \cdots \linexsub{C_i(n)/ {y}_n}\unexsub{U/\unvar{x}}\widetilde{[x_1\leftarrow x_k]} = \expr{L}
  \end{aligned}
  \end{equation*}
We also have that
                            \(\expr{N} 
                             = N\esubst{C \bagsep U}{x}
                            \red N\headlin{U_{ind}/\banged{x}}\esubst{ C \bagsep U}{x} = \expr{N}'.\)
                            
                    
                    Then,  $\expr{L}'=\recencodopenf{\expr{N'}}$ and so the result follows.

            \item Suppose that $N'' = \fail^{\widetilde{z'}}$ proceed similarly then from (\ref{eq:sound_expsub2fail})
            
                    \begin{equation*}\label{eq:sound_expsub99fail}
                        \begin{aligned}
                            \recencodopenf{\expr{N}} 
                            & = \sum_{C_i \in \perm{\recencodf{ C' }}}\fail^{\widetilde{z'}} \linexsub{C_i(1)/ {y}_1} \cdots \linexsub{C_i(n)/ {y}_n}\unexsub{U/\unvar{x}}\widetilde{[x_1\leftarrow x_k]}\\
                            &\red^* \sum_{C_i \in \perm{\recencodf{ C' }}}\fail^{(\widetilde{z'} \setminus  {y}_1, \cdots ,  {y}_n) \uplus\widetilde{y}}\unexsub{U/\unvar{x}} \widetilde{[x_1\leftarrow x_k]} \\
                            &\red^* \sum_{C_i \in \perm{\recencodf{ C' }}}\fail^{(\widetilde{z'} \setminus  {y}_1, \cdots ,  {y}_n) \uplus\widetilde{y}} \widetilde{[x_1\leftarrow x_k]} 
                            = \expr{L}'
                        \end{aligned}
                    \end{equation*}
                    
         where $\widetilde{y} = \llfv{C_i(1)} \uplus \cdots \uplus \llfv{C_i(n)}$.       We also have that
                    
                    \begin{equation*}\label{eq:sound_expsub11fail}
                        \begin{aligned}
                            \expr{N} 
                            & = \fail^{\widetilde{z}} \esubst{B}{x} 
                             \red \fail^{(\widetilde{z} \setminus x) \uplus\widetilde{y}} 
                            = \expr{N}',  \text{ where }  \widetilde{y} = \mfv{B}. 
                        \end{aligned}
                    \end{equation*}

                Then, $\expr{L}'=\recencodopenf{\expr{N'}}$ and so the result follows.
            
            \item  Suppose that $N'' \red N'''$.
                This case follows by the induction hypothesis
            
        \end{enumerate}

        \item Otherwise, we continue from equation (\ref{eq:sound_expsubfail}), where $\#(x,M) \not= k$, as follows
        
            \begin{equation*}
            \begin{aligned}
                \recencodopenf{\expr{N}} &= \recencodf{N' \esubst{B'}{x}}\widetilde{[x_1\leftarrow x_k]} \\
                &=  \recencodf{N'\langle  {y}_1. \cdots ,  {y}_k / x  \rangle} [ {y}_1. \cdots ,  {y}_k \leftarrow  {x}] \esubst{ \recencodf{B'} }{ x }\widetilde{[x_1\leftarrow x_k]} \\
                &=  \recencodf{N''} [ {y}_1. \cdots ,  {y}_k \leftarrow  {x}] \esubst{ \recencodf{B'} }{ x }\widetilde{[x_1\leftarrow x_k]}\\
            \end{aligned}
            \end{equation*}
            
            We can perform the reduction
                
            \begin{equation*}\label{eq:sound_expsubotherwise2}
            \begin{aligned}
                \recencodopenf{\expr{N}} &= \recencodf{N''} [ {y}_1. \cdots ,  {y}_k \leftarrow  {x}] \esubst{ \recencodf{B'} }{ x } \widetilde{[x_1\leftarrow x_k]} \\
                &\red \sum_{C_i \in \perm{C}}  \fail^{\widetilde{z'}}\widetilde{[x_1\leftarrow x_k]}, \text{ where } \widetilde{z'} = \lfv{N''} \uplus \lfv{C'}
                = \expr{L}'
            \end{aligned}
            \end{equation*}
        We also have that
            \begin{equation*}\label{eq:sound_expsubotherwise3}
                \begin{aligned}
                    \expr{N} 
                    & = N \esubst{C}{x}  \red \sum_{\perm{C}} \fail^{\widetilde{z}}  = \expr{N}' , \text{ where } \widetilde{z} = \mlfv{M} \uplus \mlfv{C}.
                \end{aligned}
            \end{equation*}
        Then, $\expr{L}'=\recencodopenf{\expr{N'}}$ and so the result follows.

    \end{enumerate}

        

    \item $\expr{N} = \expr{N}_1 + \expr{N}_2$: \\ Then this case holds by the induction hypothesis.
\end{enumerate}
\end{proof}

\subsection{Success Sensitiveness of  \texorpdfstring{$\recencodopenf{\cdot}$}{}}

We now consider success sensitiveness, a property that complements (and relies on) operational completeness and soundness. For the purposes of the proof, we consider the extension of $\lamrfailunres$ and $\lamrsharfailunres$ with dedicated constructs and predicates that specify success. 

\begin{definition}
We extend the syntax of terms for $\lamrfailunres$ and $\lamrsharfailunres$ with the same $\checkmark$ construct. 
In both cases, we assume $\checkmark$ is well formed. 
Also, we also define $\headf{\checkmark} = \checkmark$ and $\recencodf{\checkmark} = \checkmark$
\end{definition}

An expression $\expr{M}$ has success, denoted
\succp{\expr{M}}{\checkmark}, when there is a sequence of reductions from \expr{M} that leads to an expression that includes a summand that contains an occurrence of $\checkmark$ in head position.

\begin{definition}[Success in \lamrfailunres and \lamrsharfailunres]
\label{def:app_Suc3unres}
In $\lamrfailunres$ and $\lamrsharfailunres $, we define 
\succp{\expr{M}}{\checkmark} if and only if
there exist  $M_1 , \cdots , M_k$ such that 
$\expr{M} \red^*  M_1 + \cdots + M_k$ and
$\headf{M_j'} = \checkmark$, for some  $j \in \{1, \ldots, k\}$ and term $M_j'$ such that $M_j\pequiv  M_j'$.
\end{definition}

\begin{notation}

    We use the notation $\headsum{\expr{M}}$ to be that $\forall M_i, M_j \in \expr{M}$ we have that $\head{M_i} = \head{M_j}$ hence we say that $\headsum{\expr{M}} = \headf{M_i}$ for some $M_i \in \expr{M}$

\end{notation}

\begin{proposition}[Preservation of Head term]
\label{Prop:checkpres}
The head of a term is preserved when applying the encoding $\recencodf{\dash}$. That is to say:
\[ \forall M \in \lamrfailunres \quad \headf{M} = \checkmark \iff \headsum{\recencodopenf{M}} = \checkmark \]
\end{proposition}

\begin{proof}

By induction on the structure of $M$. We only need to consider terms of the following form.

\begin{enumerate}

    \item When $ M = \checkmark $ the case is immediate.
    
    \item When $ M = NB $ with $\lfv{NB} = \{x_1,\cdots,x_k\}$ and  $\#(x_i,M)=j_i$ we have that: 

        \[ 
            \begin{aligned}
                \headsum{\recencodopenf{NB}} &= \headsum{\recencodf{NB\linsub{\widetilde{x_{1}}}{x_1}\cdots \linsub{\widetilde{x_k}}{x_k}}[\widetilde{x_1}\leftarrow x_1]\cdots [\widetilde{x_k}\leftarrow x_k]}\\
                &= \headsum{\recencodf{NB}}
                 = \headsum{\recencodf{N}} 
            \end{aligned}
        \]
         and $\head{NB}= \head{N} $, by the IH we have $\head{N} = \checkmark \iff \headsum{\recencodf{N}} = \checkmark$.
         
    \item When $M = N \esubst{C \bagsep U}{x}$, we must have that $\#(x,M) = \size{C }$ for the head of this term to be $\checkmark$. Let $\lfv{N \esubst{C \bagsep U}{x}} = \{x_1,\cdots,x_k\}$ and  $\#(x_i,M)=j_i$. We have that: 

        \[
            \begin{aligned}
                \headsum{\recencodopenf{N \esubst{C \bagsep U}{x}}} &= \headsum{\recencodf{N \esubst{C \bagsep U}{x}\linsub{\widetilde{x_{1}}}{x_1}\cdots \linsub{\widetilde{x_k}}{x_k}}\widetilde{[x_1\leftarrow x_k]}}\\
                &= \headsum{\recencodf{N \esubst{C \bagsep U}{x}}}\\
                &= \headsum{\sum_{C_i \in \perm{\recencodf{ C }}}\recencodf{ N \langle \widetilde{x} / x  \rangle } \linexsub{C_i(1)/x_1} \cdots \linexsub{C_i(k)/x_k}\unexsub{ U/ \unvar{x}}}\\
                &= \headsum{\recencodf{ N \langle \widetilde{x} / x  \rangle } \linexsub{C_i(1)/x_1} \cdots \linexsub{C_i(k)/x_k}\unexsub{ U/ \unvar{x}}}\\
                &= \headsum{\recencodf{ N \langle \widetilde{x} / x  \rangle } } 
            \end{aligned}
        \]
        
        and $\head{N \esubst{B}{x}} = \head{N}$, by the IH  $\head{N} = \checkmark \iff \headsum{\recencodf{N}} = \checkmark$.
\end{enumerate}
\end{proof}

\begin{theorem}[Success Sensitivity]
\label{proof:app_successsensce}
Let  \expr{M} be a well-formed expression.
We have
$\expr{M} \Downarrow_{\checkmark}$ if and only if $\recencodopenf{\expr{M}} \Downarrow_{\checkmark}$.
\end{theorem}

\begin{proof}
By induction on the structure of expressions $\lamrfailunres$ and $\lamrsharfailunres$.

\begin{enumerate}
    
    \item Suppose that  $\expr{M} \Downarrow_{\checkmark} $. We will prove that $\recencodopenf{\expr{M}} \Downarrow_{\checkmark}$.

    By operational completeness (\thmref{l:app_completenessone}) we have that if $\expr{M}\red_{\redlab{R}} \expr{M'}$ then

    \begin{enumerate}
        \item If $\redlab{R} =  \redlab{R:Beta}$  then $ \recencodopenf{\expr{M}}  \red^{\leq 2}\recencodopenf{\expr{M}'}$;

        \item If $\redlab{R} =\redlab{R:Fetch}$   then   $ \recencodopenf{\expr{M}}  \red^+ \recencodopenf{\expr{M}''}$, for some $ \expr{M}''$ such that  $\expr{M}' \pequiv \expr{M}''$. 
        \item If $\redlab{R} \neq  \redlab{R:Beta}$ and $\redlab{R}\neq \redlab{R:Fetch}$  then $ \recencodopenf{\expr{M}}  \red\recencodopenf{\expr{M}'}$;
    \end{enumerate}
    Notice that  neither our  reduction rules  (in \defref{fig:share-reductfailureunres}), or our congruence $\pequiv$ (in \figref{def:rsPrecongruencefailure}),  or  our encoding ($\recencodopenf{\checkmark }=\checkmark$)  create or destroy a $\checkmark$ occurring in the head of term. By Proposition \ref{Prop:checkpres} the encoding preserves the head of a term being $\checkmark$. The encoding acts homomorphically over sums, therefore, if a $\checkmark$ appears as the head of a term in a sum, it will stay in the encoded sum. We can iterate the operational completeness lemma and obtain the result.

    \item Suppose that $\recencodopenf{\expr{M}} \Downarrow_{\checkmark}$. We will prove that $ \expr{M} \Downarrow_{\checkmark}$. 

    By operational soundness (\thmref{l:soundnessoneunres}) we have that if $ \recencodopenf{\expr{M}}  \red \expr{L}$ then there exist $ \expr{M}' $ such that $ \expr{M}  \red_{\redlab{R}} \expr{M}'$ and 

    \begin{enumerate}
        \item If $\redlab{R} = \redlab{R:Beta}$ then $\expr{ L } \red^{\leq 1} \recencodopenf{\expr{M}'}$;

    \item If $\redlab{R} \neq \redlab{R:Beta}$ then $\expr{ L } \red^*  \recencodopenf{\expr{M}''}$, for $ \expr{M}''$ such that  $\expr{M}' \pequiv \expr{M}''$.
    \end{enumerate}
    
   Since $\recencodopenf{\expr{M}}\red^* M_1+\ldots+M_k$, and $\headf{M_j'}=\checkmark$, for some $j$ and $M_j'$, s.t. $M_j\pequiv M_j'$. 
   
   Notice that if $\recencodopenf{\expr{M}}$ is itself a term headed with $\checkmark$, say $\headf{\recencodopenf{\expr{M}}}=\checkmark$, then $\expr{M}$ is itself headed with $\checkmark$, from Proposition \ref{Prop:checkpres}.
   In the case $\recencodopenf{\expr{M}}= M_1+\ldots+M_k$, $k\geq 2$, and $\checkmark$ occurs in the head of an $M_j$, the reasoning is similar.  $\expr{M}$ has one of the forms:
   \begin{enumerate}
       \item   $\expr{M}= N_1$, then $N_1$ must contain the subterm $ M\esubst{C \bagsep U}{x}$ and $\size{C}=\#(x,M)$.
       
       The encoding of $\expr{M}$ is
       
       $\recencodopenf{M\esubst{C \bagsep U}{x}}=\sum_{C_i \in \perm{\recencodf{ C }}}\recencodf{M\linsub{\widetilde{x}}{x}}\linexsub{C_i(1)/x_i}\ldots \linexsub{C_i(k)/x_i}\unexsub{U / \unvar{x} }$.
       We can apply Proposition~\ref{Prop:checkpres} and the result follows.
       
       \item $\expr{M}=N_1+\ldots+N_l$ for $l \geq 2$.
       
       This reasoning is similar and uses the fact  that the encoding distributes homomorphically over sums.
   \end{enumerate}
   
   In the case where $\recencodopenf{\expr{M}}\red^+ M_1+\ldots+M_k$, and $\headf{M_j'}=\checkmark$, for some $j$ and $M_j'$, such that $M_j\pequiv M_j'$, the reasoning is similar to the previous, since our reduction rules do not introduce/eliminate $\checkmark$ occurring in the head of terms. 
\end{enumerate}
\end{proof}


\section{Appendix to Subsection~\ref{ssec:second_enc}}\label{app:encodingtwo}

\subsection{Type Preservation}

\begin{lemma}\label{prop:app_auxunres}
 $ \piencodf{\sigma^{j}}_{(\tau_1, m)} = \piencodf{\sigma^{k}}_{(\tau_2, n)}$ and $ \piencodf{(\sigma^{j} , \eta)}_{(\tau_1, m)} = \piencodf{(\sigma^{k}, \eta)}_{(\tau_2, n)}$ hold, provided that  $\tau_1,\tau_2,n$ and $m$ are as follows:

        \begin{enumerate}
        \item If $j > k$ then take $\tau_1 $ to be an arbitrary type, $m = 0$,  take $\tau_2 $ to be $\sigma$ and $n = j-k$.
        
        \item If $j < k$ then take $\tau_1 $ to be $\sigma$, $m = k-j$,  take $\tau_2 $ to be an arbitrary type and $n = 0$. 
        
        \item Otherwise, if $j = k$ then take $m = n = 0$. In this case, $\tau_1. \tau_2 $ are unimportant.
    \end{enumerate}
    
\end{lemma}

\begin{proof}
We shall prove the case of $(1)$ for the first equality, and the case for the second equality and of $(2)$ are analogous. The case of $(3)$ follows  by the encoding on types in \defref{def:enc_sestypfailunres}. 

Hence take $j,k,\tau_1,\tau_2, m,n$ satisfying the conditions in (1): $j > k$, $\tau_1 $ to be an arbitrary type,  $m = 0$,  $\tau_2 =\sigma$ and $n = j-k$. 
    We want to show that $ \piencodf{\sigma^{j}}_{(\tau_1, 0)} = \piencodf{\sigma^{k}}_{(\sigma, n)} $. In fact, 
    \[
        \begin{aligned}
            \piencodf{\sigma^{k}}_{(\sigma, n)} &= \oplus(( \with \onef) \ampy ( \oplus  \with (( \oplus \piencodf{\sigma} ) \otimes (\piencodf{\sigma^{k-1}}_{(\sigma, n)}))))\\
            \piencodf{\sigma^{k-1}}_{(\sigma, n)} &= \oplus(( \with \onef) \ampy ( \oplus  \with (( \oplus \piencodf{\sigma} ) \otimes (\piencodf{\sigma^{k-2}}_{(\sigma, n)}))))\\
            \vdots\\
            \piencodf{\sigma^{1}}_{(\sigma, n)} &= \oplus(( \with \onef) \ampy ( \oplus  \with (( \oplus \piencodf{\sigma} ) \otimes (\piencodf{\omega}_{(\sigma, n)}))))
        \end{aligned}
    \]
    and
    \[
        \begin{aligned}
            \piencodf{\sigma^{j}}_{(\tau_1, 0)} &= \oplus(( \with \onef) \ampy ( \oplus  \with (( \oplus \piencodf{\sigma} ) \otimes (\piencodf{\sigma^{j-1}}_{(\tau_1, 0)}))))\\
            \piencodf{\sigma^{j-1}}_{(\tau_1, 0)} &= \oplus(( \with \onef) \ampy ( \oplus  \with (( \oplus \piencodf{\sigma} ) \otimes (\piencodf{\sigma^{j-2}}_{(\tau_1, 0)}))))\\
            \vdots\\
            \piencodf{\sigma^{j-k + 1}}_{(\tau_1, 0)} &= \oplus(( \with \onef) \ampy ( \oplus  \with (( \oplus \piencodf{\sigma} ) \otimes (\piencodf{\sigma^{j-k}}_{(\tau_1, 0)}))))
        \end{aligned}
    \]
    Notice that $n = j-k$, hence we wish to show that $ \piencodf{\sigma^{n}}_{(\tau_1, 0)} = \piencodf{\omega}_{(\sigma, n)} $.  Finally,
    
    \[
        \begin{aligned}
            \piencodf{\omega}_{(\sigma, n)} & = \oplus(( \with \onef) \ampy ( \oplus  \with (( \oplus \piencodf{\sigma} ) \otimes (\piencodf{\omega}_{(\sigma, n-1)})))) \\
            \piencodf{\omega}_{(\sigma, n-1)} & = \oplus(( \with \onef) \ampy ( \oplus  \with (( \oplus \piencodf{\sigma} ) \otimes (\piencodf{\omega}_{(\sigma, n-2)})))) \\
            \vdots\\
            \piencodf{\omega}_{(\sigma, 1)} & = \oplus(( \with \onef) \ampy ( \oplus  \with (( \oplus \piencodf{\sigma} ) \otimes (\piencodf{\omega}_{(\sigma, 0)})))) \\
            \piencodf{\omega}_{(\sigma, 0)} &= \oplus(( \with \onef) \ampy ( \oplus  \with \onef ) \\
        \end{aligned}
    \]
    and 
    \[
        \begin{aligned}
            \piencodf{\sigma^{n}}_{(\tau_1, 0)} &= \oplus(( \with \onef) \ampy ( \oplus  \with (( \oplus \piencodf{\sigma} ) \otimes (\piencodf{\sigma^{n-1}}_{(\tau_1, 0)}))))\\
            \piencodf{\sigma^{n-1}}_{(\tau_1, 0)} &= \oplus(( \with \onef) \ampy ( \oplus  \with (( \oplus \piencodf{\sigma} ) \otimes (\piencodf{\sigma^{n-2}}_{(\tau_1, 0)}))))\\
            \vdots\\
            \piencodf{\sigma^{1}}_{(\tau_1, 0)} &= \oplus(( \with \onef) \ampy ( \oplus  \with (( \oplus \piencodf{\sigma} ) \otimes (\piencodf{\omega}_{(\tau_1, 0)})))) \\
            \piencodf{\omega}_{(\tau_1, 0)} &= \oplus(( \with \onef) \ampy ( \oplus  \with \onef ) \\
        \end{aligned}
    \]
\end{proof}


        
        
    

\begin{lemma}\label{lem:relunbag-typeunres}

If 
$ \eta \relunbag \epsilon $ Then
\begin{enumerate}
    \item If $\piencodf{M}_u\vdash \piencodf{\Gamma} ; \piencodf{\Theta} , \banged{x} : \piencodf{\eta}$
    then 
    $\piencodf{M}_u\vdash \piencodf{\Gamma} ; \piencodf{\Theta}, \banged{x} : \piencodf{\epsilon}$.
    
    \item If $\piencodf{M}_u\vdash \piencodf{\Gamma}, u:\piencodf{(\sigma^{j} , \eta ) \rightarrow \tau} ; \piencodf{\Theta}$
    then 
    $\piencodf{M}_u\vdash \piencodf{\Gamma}, u:\piencodf{(\sigma^{j} , \epsilon ) \rightarrow \tau} ; \piencodf{\Theta}$.

\end{enumerate}

\end{lemma}

\begin{proof}

\begin{enumerate}
    
    \item We consider the first case where if $\piencodf{M}_u\vdash \piencodf{\Gamma} ; \piencodf{\Theta} , \banged{x} : \piencodf{\eta}$
    then 
    $\piencodf{M}_u\vdash \piencodf{\Gamma} ; \piencodf{\Theta}, \banged{x} : \piencodf{\epsilon}$ and by \defref{def:enc_sestypfailunres},  $\piencodf{ \eta } = \&_{\eta_i \in \eta} \{ \mathtt{l}_i ; \piencodf{\eta_i} \}$. We now proceed by induction on the structure of $M$:
    \begin{enumerate}
        
        \item $M =   {x}$.
        \label{proof:relunbag-no}
       
        By \figref{fig:encfailunres},  $\piencodf{ {x}}_u =  {x}.\overline{\some} ; [ {x} \leftrightarrow u] $. We have the following derivation:
                
        \begin{prooftree}
            \AxiomC{}
            \LeftLabel{$\redlab{ (Tid)}$}
            \UnaryInfC{$ [ {x} \leftrightarrow u] \vdash  {x}:  \overline{A}  , u :  A; \piencodf{\Theta} , \banged{x} : \piencodf{\eta}$}
            \LeftLabel{$\redlab{T\with^{x}_{d})}$}
            \UnaryInfC{$  {x}.\overline{\some} ; [ {x} \leftrightarrow u] \vdash  {x}: \with  \overline{A} , u :  A ; \piencodf{\Theta} , \banged{x} : \piencodf{\eta} $}
        \end{prooftree}
        
        For some type A. Notice the derivation is independent of $\banged{x} : \piencodf{\eta}$ , hence holds when $ M =   {x}$. Note that we do not consider $M =  {y}$ where $y \not = x$, this is due to the case being trivial due to the typing of $y$ being independent on $x$.

        \item $ M =  {x}[ind]$.
        \label{proof:relunbag-yes}

            By \figref{fig:encfailunres},  $\piencodf{ {x}[ind]}_u = \outsev{\banged{x}}{{x_i}}. {x}_i.l_{ind}; [{x_i} \leftrightarrow u] $. We have the following derivation:
            
            \begin{prooftree}
                    \AxiomC{}
                    \LeftLabel{$\redlab{ (Tid)}$}
                    \UnaryInfC{$ [{x_i} \leftrightarrow u]  \vdash  u :  \piencodf{ \tau },  x_i:  \overline{\piencodf{\eta_{ind} }}  ; \banged{x}:\&_{\eta_i \in \eta} \{ \mathtt{l}_i ; \piencodf{\eta_i} \} , \piencodf{\Theta} $}
                \LeftLabel{\redlab{T\oplus_i}}
                \UnaryInfC{$  {x}_i.l_{ind}; [{x_i} \leftrightarrow u] \vdash  u :  \piencodf{ \tau }, {x}_i :  \oplus_{\eta_i \in \eta} \{ \mathtt{l}_i ; \dual{\piencodf{\eta_i}}  \} ; \banged{x}:\oplus_{\eta_i \in \eta} \{ \mathtt{l}_i ; \dual{\piencodf{\eta_i}} \} , \piencodf{\Theta}  $}
                \LeftLabel{\redlab{Tcopy}}
                \UnaryInfC{$ \outsev{\banged{x}}{{x_i}}. {x}_i.l_{ind}; [{x_i} \leftrightarrow u] \vdash  u :  \piencodf{ \tau }; \banged{x}:\oplus_{\eta_i \in \eta} \{ \mathtt{l}_i ; \dual{\piencodf{\eta_i}} \} , \piencodf{\Theta} $}
            \end{prooftree}
            
            On the other hand we have derivation:
            
            \begin{prooftree}
                    \AxiomC{}
                    \LeftLabel{$\redlab{ (Tid)}$}
                    \UnaryInfC{$ [{x_i} \leftrightarrow u]  \vdash  u :  \piencodf{ \tau },  x_i:  \overline{\piencodf{\epsilon_{ind} }}  ; \banged{x}:\&_{\epsilon_i \in \epsilon} \{ \mathtt{l}_i ; \piencodf{\epsilon_i} \} , \piencodf{\Theta} $}
                \LeftLabel{\redlab{T\oplus_i}}
                \UnaryInfC{$  {x}_i.l_{ind}; [{x_i} \leftrightarrow u] \vdash  u :  \piencodf{ \tau }, {x}_i :  \oplus_{\epsilon_i \in \epsilon} \{ \mathtt{l}_i ; \dual{\piencodf{\epsilon_i}}  \} ; \banged{x}:\oplus_{\epsilon_i \in \epsilon} \{ \mathtt{l}_i ; \dual{\piencodf{\epsilon_i}} \} , \piencodf{\Theta}  $}
                \LeftLabel{\redlab{Tcopy}}
                \UnaryInfC{$ \outsev{\banged{x}}{{x_i}}. {x}_i.l_{ind}; [{x_i} \leftrightarrow u] \vdash  u :  \piencodf{ \tau }; \banged{x}:\oplus_{\epsilon_i \in \epsilon} \{ \mathtt{l}_i ; \dual{\piencodf{\epsilon_i}} \} , \piencodf{\Theta} $}
            \end{prooftree}
            
            By $ \eta \relunbag \epsilon $ we have that $ \epsilon_{ind} = \eta_{ind} $. Similarly for the case of $ M =  {y}[ind]$ with $y \not =  x$ we use the argument that the typing of $y$ is independent on $x$.
        
        \item  $M = M' [  {\widetilde{y}} \leftarrow  {y} ] $.
        
        If $y = x$ the case proceeds similarly to (\ref{proof:relunbag-no}) otherwise we proceed by induction on $M'$.
            
        \item $ M =  \lambda x . (M'[ {\widetilde{x}} \leftarrow  {x}])$.
 
            From 
            \defref{def:enc_lamrsharpifailunres} it follows that
            
            $ \piencodf{\lambda x.M'[ {\widetilde{x}} \leftarrow x]}_u = u.\overline{\some}; u(x). x.\overline{\some}; x(\linvar{x}). x(\banged{x}). x. \close ; \piencodf{M'[ {\widetilde{x}}\leftarrow  {x}]}_u$.
            
            We give the final derivation in parts. The first part we name $\Pi_1$ derived by:
            
            \begin{prooftree}
                \AxiomC{$\piencodf{M'[ {\widetilde{x}} \leftarrow  {x}]}_u \vdash  u:\piencodf{\tau} , \piencodf{\Gamma'} , \linvar{x}: \overline{\piencodf{\sigma^k}_{(\sigma, i)}}  ; \piencodf{\Theta} , \banged{x}: \overline{\piencodf{\eta}}$}
                \LeftLabel{\redlab{T\bot}}
                \UnaryInfC{$ x. \close ; \piencodf{M'[ {\widetilde{x}} \leftarrow  {x}]}_u \vdash x{:}\bot, u:\piencodf{\tau} , \piencodf{\Gamma'} , \linvar{x}: \overline{\piencodf{\sigma^k}_{(\sigma, i)}}  ; \piencodf{\Theta} , \banged{x}: \overline{\piencodf{\eta}}$}
                \LeftLabel{\redlab{T?}}
                \UnaryInfC{$x. \close ; \piencodf{M'[ {\widetilde{x}} \leftarrow  {x}]}_u \vdash x{:}\bot, u:\piencodf{\tau} , \piencodf{\Gamma'} , \linvar{x}: \overline{\piencodf{\sigma^k}_{(\sigma, i)}} , \banged{x}: ? \overline{\piencodf{\eta}} ; \piencodf{\Theta} $}
                \LeftLabel{\redlab{T\ampy}}
                \UnaryInfC{$x(\banged{x}). x. \close ; \piencodf{M'[ {\widetilde{x}} \leftarrow  {x}]}_u \vdash x: (? \overline{\piencodf{\eta}}) \ampy (\bot) , u:\piencodf{\tau} , \piencodf{\Gamma'} , \linvar{x}: \overline{\piencodf{\sigma^k}_{(\sigma, i)}} ; \piencodf{\Theta} $}
                \LeftLabel{\redlab{T\ampy}}
                \UnaryInfC{$  x(\linvar{x}). x(\banged{x}). x. \close ; \piencodf{M'[ {\widetilde{x}} \leftarrow  {x}]}_u \vdash  x: \overline{\piencodf{\sigma^k}_{(\sigma, i)}} \ampy ((? \overline{\piencodf{\eta}}) \ampy (\bot)) , u:\piencodf{\tau} , \piencodf{\Gamma'} ; \piencodf{\Theta} $}
            \end{prooftree}
            
            We take $ P = x(\linvar{x}). x(\banged{x}). x. \close ; \piencodf{M'[ {\widetilde{x}} \leftarrow  {x}]}_u $ and continue the derivation:

            \begin{prooftree}
                \AxiomC{$ \Pi_1 $}
                \noLine
                \UnaryInfC{$ \vdots $}
                \noLine
                \UnaryInfC{$  P \vdash  x: \overline{\piencodf{\sigma^k}_{(\sigma, i)}} \ampy ((? \overline{\piencodf{\eta}}) \ampy (\bot)) , u:\piencodf{\tau} , \piencodf{\Gamma'} ; \piencodf{\Theta} $}
                \LeftLabel{\redlab{T\with_d^x}}
                \UnaryInfC{$ x.\overline{\some}; P \vdash x :\with(  \overline{\piencodf{\sigma^k}_{(\sigma, i)}} \ampy ((? \overline{\piencodf{\eta}}) \ampy (\bot))) , u:\piencodf{\tau} , \piencodf{\Gamma'} ; \piencodf{\Theta} $}
                \LeftLabel{\redlab{T\ampy}}
                \UnaryInfC{$u(x). x.\overline{\some}; P \vdash u: \with(  \overline{\piencodf{\sigma^k}_{(\sigma, i)}} \ampy ((? \overline{\piencodf{\eta}}) \ampy (\bot))) \ampy \piencodf{\tau} , \piencodf{\Gamma'} ; \piencodf{\Theta}  $}
                \LeftLabel{\redlab{T\with_d^x}}
                \UnaryInfC{$u.\overline{\some}; u(x). x.\overline{\some}; P \vdash u :\with (\with(  \overline{\piencodf{\sigma^k}_{(\sigma, i)}} \ampy ((? \overline{\piencodf{\eta}}) \ampy (\bot))) \ampy \piencodf{\tau}) , \piencodf{\Gamma'} ; \piencodf{\Theta} $}
            \end{prooftree}
            
            By Definition \ref{def:enc_sestypfailunres} we have that $ \piencodf{(\sigma^{k} , \eta )   \rightarrow \tau} = \with (\with(  \overline{\piencodf{\sigma^k}_{(\sigma, i)}} \ampy ((? \overline{\piencodf{\eta}}) \ampy (\bot))) \ampy \piencodf{\tau}) $. In this case we must have that the variable names for $ x $ from our hypothesis and $x$ from $ M $ must be distinct.

        \item $ M =  (M'\ B) $, or $ M =   (M[ {\widetilde{x}} \leftarrow  {x}])\esubst{ B }{ x }$, or  $ M =  M' \unexsub{U / \unvar{x}} $.
        
            The proof  follows similarly to that of (\ref{proof:relunbag-yes}).
       
        \item $ M =  M' \linexsub{N /  {x}} $
        
            Case follows by that of (\ref{proof:relunbag-no}) and applying induction hypothesis on $\piencodf{M'}_u$.

        
        
            

        \item When $ M =  \fail^{\widetilde{x}}$
            Case follows by that of (\ref{proof:relunbag-no}).

    \end{enumerate}

    \item If $\piencodf{M}_u\vdash \piencodf{\Gamma}, u:\piencodf{(\sigma^{j} , \eta ) \rightarrow \tau} ; \piencodf{\Theta}$
    then 
    $\piencodf{M}_u\vdash \piencodf{\Gamma}, u:\piencodf{(\sigma^{j} , \epsilon ) \rightarrow \tau} ; \piencodf{\Theta}$ follows from previous case along a similar argument.
\end{enumerate}

\end{proof}

\begin{theorem}[Type Preservation for $\piencodf{\cdot}_u$]
\label{t:preservationtwounres}
Let $B$ and $\expr{M}$ be a bag and an expression in $\lamrsharfailunres$, respectively.
\begin{enumerate}
\item If $\Theta ; \Gamma \wfdash B : (\sigma^{k} , \eta )$
then 
$\piencodf{B}_u \wfdash  \piencodf{\Gamma}, u : \piencodf{(\sigma^{k} , \eta )}_{(\sigma, i)} ; \piencodf{\Theta}$.

\item If $\Theta ; \Gamma \wfdash \expr{M} : \tau$
then 
$\piencodf{\expr{M}}_u \wfdash  \piencodf{\Gamma}, u :\piencodf{\tau} ; \piencodf{\Theta}$.
\end{enumerate}
\end{theorem}

\begin{proof}
The proof is by mutual induction on the typing derivation of $B$ and $\expr{M}$, with an analysis for the last rule applied.
Recall that the encoding of types ($\piencodf{-}$) has been given in 
\defref{def:enc_sestypfailunres}.
    
    \begin{enumerate}
        \item  We have the following derivation where we take $ B =  C \bagsep U$:
        
            \begin{prooftree}
                \AxiomC{\( \Theta ; \Gamma\wfdash C : \sigma^k\)}
                \AxiomC{\(  \Theta ;\cdot \wfdash  U : \eta \)}
            \LeftLabel{\redlab{FS{:}bag}}
            \BinaryInfC{\( \Theta ; \Gamma \wfdash C \bagsep U : (\sigma^{k} , \eta  ) \)}
            \end{prooftree}
        
        Our encoding gives: 
$ \piencodf{C \bagsep U}_u = x.\some_{\llfv{C}} ; \outact{x}{\linvar{x}} .( \piencodf{ C }_{\linvar{x}} \para \outact{x}{\banged{x}} .( !\banged{x} (x_i). \piencodf{ U }_{x_i} \para x.\overline{\close} ) )$.
        In addition, the encoding of $(\sigma^{k} , \eta  )$ is:
        $$\piencodf{ (\sigma^{k} , \eta  )  }_{(\sigma, i)} = \oplus( (\piencodf{\sigma^{k} }_{(\sigma, i)}) \otimes ((!\piencodf{\eta}) \otimes (\onef))  )\quad \text{(for some  $i \geq 0$ and  strict type $\sigma$)}$$
            
        And one can build the following type derivation (rules from \figref{fig:trulespi}):

        \begin{prooftree}
                \AxiomC{$ \piencodf{ C }_{\linvar{x}} \vdash \piencodf{\Gamma}, \linvar{x}:\piencodf{\sigma^{k} }_{(\sigma, i)} ; \piencodf{\Theta}$}
                
                    \AxiomC{$ \piencodf{ U }_{x_i} \vdash x_i: \piencodf{\eta}; \piencodf{\Theta}$}
                    \LeftLabel{\redlab{T!}}
                    \UnaryInfC{$!\banged{x} (x_i). \piencodf{ U }_{x_i} \vdash \banged{x}: !\piencodf{\eta} ; \piencodf{\Theta}$}
                    
                    \AxiomC{\mbox{\ }}
                    \LeftLabel{\redlab{T\onef}}
                    \UnaryInfC{$x.\dual{\close} \vdash x: \onef ; \piencodf{\Theta}$}
                \LeftLabel{\redlab{T\otimes}}
                \BinaryInfC{$ \outact{x}{\banged{x}} .( !\banged{x}. (x_i). \piencodf{ U }_{x_i} \para x.\overline{\close}  \vdash x: (!\piencodf{\eta}) \otimes (\onef) ; \piencodf{\Theta}$}
            \LeftLabel{\redlab{T\otimes}}
            \BinaryInfC{$\outact{x}{\linvar{x}} .( \piencodf{ C }_{\linvar{x}} \para \outact{x}{\banged{x}} .( !\banged{x}. (x_i). \piencodf{ U }_{x_i} \para x.\overline{\close} ) ) \vdash  \piencodf{\Gamma}, x:(\piencodf{\sigma^{k} }_{(\sigma, i)}) \otimes ((!\piencodf{\eta}) \otimes (\onef)) ; \piencodf{\Theta} $}
            \LeftLabel{\redlab{T\oplus^x_{\widetilde{w}}}}
            \UnaryInfC{$x.\some_{\llfv{C}} ; \outact{x}{\linvar{x}} .( \piencodf{ C }_{\linvar{x}} \para \outact{x}{\banged{x}} .( !\banged{x}. (x_i). \piencodf{ U }_{x_i} \para x.\overline{\close} ) ) \vdash \piencodf{\Gamma}, x{:}\piencodf{ (\sigma^{k} , \eta  )  }_{(\sigma, i)})$}
        \end{prooftree}
        
        Hence true provided both $ \piencodf{ C }_{\linvar{x}} \vdash \piencodf{\Gamma}, \linvar{x}:\piencodf{\sigma^{k} }_{(\sigma, i)} ; \piencodf{\Theta}$ and $ \piencodf{ U }_{x_i} \vdash x_i: \piencodf{\eta}; \piencodf{\Theta}$ hold.
        
        Let us consider the two cases:
        
        \begin{enumerate}
        
            \item For $ \piencodf{ C }_{\linvar{x}} \vdash \piencodf{\Gamma}, \linvar{x}:\piencodf{\sigma^{k} }_{(\sigma, i)} ; \piencodf{\Theta}$ to hold we must consider two cases on the shape of $C$:
            
            \begin{enumerate}
                \item When $C = \oneb$ we may type bags with the $\redlab{FS{:}\oneb^{\ell}}$ rule.
                
                That is, 
                
                    \begin{prooftree}
                        \AxiomC{\(  \)}
                        \LeftLabel{\redlab{FS{:}\oneb^{\ell}}}
                        \UnaryInfC{\( \Theta ; \dash \wfdash \oneb : \omega \)}
                    \end{prooftree}
                    
                Our encoding gives:         $ \piencodf{\oneb}_{\linvar{x}} = \linvar{x}.\some_{\emptyset} ; \linvar{x}(y_n). ( y_n.\overline{\some};y_n . \overline{\close} \para \linvar{x}.\some_{\emptyset} ; \linvar{x}. \overline{\none}). $
                
                and  the encoding of $\omega$ can be either:
                \begin{enumerate}
                \item  $\piencodf{\omega}_{(\sigma,0)} =  \overline{\with(( \oplus \bot )\otimes ( \with \oplus \bot ))}$; or
                \item $\piencodf{\omega}_{(\sigma, i)} =  \overline{   \with(( \oplus \bot) \otimes ( \with  \oplus (( \with  \overline{\piencodf{ \sigma }} )  \ampy (\overline{\piencodf{\omega}_{(\sigma, i - 1)}})))) }$
                \end{enumerate}

                And one can build the following type derivation (rules from \figref{fig:trulespi}):

                \begin{prooftree}
                        \AxiomC{\mbox{\ }}
                        \LeftLabel{\redlab{T\onef}}
                        \UnaryInfC{$y_n . \overline{\close} \vdash y_n: \onef; \piencodf{\Theta}$}
                        \LeftLabel{\redlab{T\with_d^x}}
                        \UnaryInfC{$y_n.\overline{\some};y_n . \overline{\close} \vdash  y_n :\with \onef; \piencodf{\Theta}$}

                        \AxiomC{}
                        \LeftLabel{\redlab{T\with^x}}
                        \UnaryInfC{$\linvar{x}.\dual{\none} \vdash \linvar{x} :\with A; \piencodf{\Theta}$}
                        \LeftLabel{\redlab{T\oplus^x_{\widetilde{w}}}}
                        \UnaryInfC{$\linvar{x}.\some_{\emptyset} ; \linvar{x}. \overline{\none} \vdash  \linvar{x}{:}\oplus \with A; \piencodf{\Theta}$}
                    
                    \LeftLabel{\redlab{T\mid}}
                    \BinaryInfC{$( y_n.\overline{\some};y_n . \overline{\close} \para \linvar{x}.\some_{\emptyset} ; \linvar{x}. \overline{\none}) \vdash y_n :\with \onef, \linvar{x}{:}\oplus \with A; \piencodf{\Theta}$}
                    \LeftLabel{\redlab{T\ampy}}
                    \UnaryInfC{$\linvar{x}(y_n). ( y_n.\overline{\some};y_n . \overline{\close} \para \linvar{x}.\some_{\emptyset} ; \linvar{x}. \overline{\none}) \vdash  \linvar{x}: (\with \onef) \ampy (\oplus \with A) ; \piencodf{\Theta}$}
                    \LeftLabel{\redlab{T\oplus^x_{\widetilde{w}}}}
                    \UnaryInfC{$\linvar{x}.\some_{\emptyset} ; \linvar{x}(y_n). ( y_n.\overline{\some};y_n . \overline{\close} \para \linvar{x}.\some_{\emptyset} ; \linvar{x}. \overline{\none}) \vdash  \linvar{x}{:}\oplus ((\with \onef) \ampy (\oplus \with A)); \piencodf{\Theta}$}
                \end{prooftree}
                
                Since $A$ is arbitrary,  we can take $A=\oneb$ for $\piencodf{\omega}_{(\sigma,0)} $ and  $A= \overline{(( \with  \overline{\piencodf{ \sigma }} )  \ampy (\overline{\piencodf{\omega}_{(\sigma, i - 1)}}))}$  for $\piencodf{\omega}_{(\sigma,i)} $, in both cases, the result follows.

                \item When $C = \bag{M}  \cdot C'$ we may type bags with the $\redlab{FS{:}bag^{\ell}}$ rule. 
                
                \begin{prooftree}
                    \AxiomC{\( \Theta ; \Gamma' \wfdash M : \sigma\)}
                    \AxiomC{\( \Theta ; \Delta \wfdash C' : \sigma^k\)}
                    \LeftLabel{\redlab{FS{:}bag^{\ell}}}
                    \BinaryInfC{\( \Theta ; \Gamma', \Delta \wfdash \bag{M}  \cdot C':\sigma^{k}\)}
                \end{prooftree}
                
                Where $ \Gamma = \Gamma' , \Delta$.
                To simplify the proof, we will consider $k=3$.

                By IH we have
                \begin{align*}
                    \piencodf{M}_{x_i} & \vdash \piencodf{\Gamma'}, x_i: \piencodf{\sigma}; \piencodf{\Theta}
                    \qquad 
                    \piencodf{C'}_{\linvar{x}} & \vdash \piencodf{\Delta}, \linvar{x}: \piencodf{\sigma\wedge \sigma}_{(\tau, j)}; \piencodf{\Theta}
                \end{align*}

                By \defref{def:enc_lamrsharpifailunres},
                \begin{equation}
                \begin{aligned}
                    \piencodf{\bag{M} \cdot C'}_{\linvar{x}} =&  \linvar{x}.\some_{\llfv{\bag{M} \cdot C} } ; \linvar{x}(y_i). \linvar{x}.\some_{y_i, \llfv{\bag{M} \cdot C}};\linvar{x}.\overline{\some} ; \outact{\linvar{x}}{x_i}.
                   \\
                   & (x_i.\some_{\llfv{M}} ; \piencodf{M}_{x_i} \para \piencodf{C'}_{\linvar{x}} \para y_i. \overline{\none})
                    \end{aligned}
                \end{equation}

                Let $\Pi_1$ be the derivation:

                \begin{prooftree}
                            \AxiomC{$\piencodf{M}_{x_i} \;{ \vdash} \piencodf{\Gamma'}, x_i: \piencodf{\sigma}; \piencodf{\Theta} $}
                            \LeftLabel{\redlab{T\oplus^x_{\widetilde{w}}}}
                            \UnaryInfC{$x_i.\some_{\llfv{M}} ; \piencodf{M}_{x_i} \vdash \piencodf{\Gamma'} ,x_i: \oplus \piencodf{\sigma} ; \piencodf{\Theta}$}
                            
                            \AxiomC{}
                            \LeftLabel{\redlab{T\with^x}}
                            \UnaryInfC{$ y_i. \overline{\none} \vdash y_i :\with \onef; \piencodf{\Theta}$}
                            
                        \LeftLabel{\redlab{T\para}}
                        \BinaryInfC{$\underbrace{x_i.\some_{\llfv{M}} ; \piencodf{M}_{x_i} \para y_i. \overline{\none}}_{P_1} \vdash \piencodf{\Gamma'} ,x_i: \oplus \piencodf{\sigma}, y_i :\with \onef ; \piencodf{\Theta}$}
                \end{prooftree}
                
                Let $ P_1 = (x_i.\some_{\llfv{M}} ; \piencodf{M}_{x_i} \para y_i. \overline{\none})$, in the the derivation $\Pi_2$ below:

                \begin{prooftree}
                        \AxiomC{$ \Pi_1$} 
                        
                        \AxiomC{$ \piencodf{C'}_{\linvar{x}}  \vdash  \piencodf{\Delta}, \linvar{x}: \piencodf{\sigma\wedge \sigma}_{(\tau, j)}; \piencodf{\Theta} $}
                        
                        \LeftLabel{\redlab{T\otimes}}
                    \BinaryInfC{$ \outact{\linvar{x}}{x_i}. (P_1 \para \piencodf{C'}_{\linvar{x}}) \vdash  \piencodf{\Gamma'}  ,  \piencodf{\Delta}, y_i :\with \onef, \linvar{x}: (\oplus \piencodf{\sigma})  \otimes (\piencodf{\sigma\wedge \sigma}_{(\tau, j)}) ; \piencodf{\Theta}$}
                    \LeftLabel{\redlab{T\with_d^x}}
                    \UnaryInfC{$\underbrace{\linvar{x}.\overline{\some} ; \outact{\linvar{x}}{x_i}. (P_1 \para \piencodf{C'}_{\linvar{x}}  )}_{P_2} \vdash \piencodf{\Gamma'}  ,  \piencodf{\Delta}, y_i :\with \onef, \linvar{x}: \with (( \oplus \piencodf{\sigma} ) \otimes (\piencodf{\sigma\wedge \sigma}_{(\tau, j)})) ; \piencodf{\Theta} $}
                \end{prooftree}
                
                Let $P_2 = (\linvar{x}.\overline{\some} ; \outact{\linvar{x}}{x_i}. (P_1 \para \piencodf{A}_{\linvar{x}} ))$ in the derivation below:

                \begin{prooftree}
                    \AxiomC{$ \Pi_2$} 
                    \noLine
                    \UnaryInfC{$\vdots$}
                    \noLine
                    \UnaryInfC{$P_2\vdash \piencodf{\Gamma}, y_i :\with \onef, \linvar{x}: \with (( \oplus \piencodf{\sigma} ) \otimes (\piencodf{\sigma\wedge \sigma}_{(\tau, j)})) ; \piencodf{\Theta} $}
                    \LeftLabel{\redlab{T\oplus^x_{\widetilde{w}}}}
                    \UnaryInfC{$\linvar{x}.\some_{y_i, \llfv{\bag{M} \cdot C'}};P_2  \vdash \piencodf{\Gamma}, y_i :\with \onef, \linvar{x}:\oplus  \with (( \oplus \piencodf{\sigma} ) \otimes (\piencodf{\sigma\wedge \sigma}_{(\tau, j)})); \piencodf{\Theta}$}
                    \LeftLabel{\redlab{T\ampy}}
                    \UnaryInfC{$\linvar{x}(y_i). \linvar{x}.\some_{y_i, \llfv{\bag{M} \cdot C'}};P_2  \vdash \piencodf{\Gamma}, \linvar{x}: ( \with \onef) \ampy ( \oplus  \with (( \oplus \piencodf{\sigma} ) \otimes (\piencodf{\sigma\wedge \sigma}_{(\tau, j)}))); \piencodf{\Theta} $}
                    \LeftLabel{\redlab{T\oplus^x_{\widetilde{w}}}}
                    \UnaryInfC{$\piencodf{\bag{M}\cdot C'}_{\linvar{x}} \vdash   \piencodf{\Gamma}, \linvar{x}: \oplus(( \with \onef) \ampy ( \oplus  \with (( \oplus \piencodf{\sigma} ) \otimes (\piencodf{\sigma\wedge \sigma}_{(\tau, j)})))); \piencodf{\Theta} $}
                \end{prooftree}
                
                From Definitions~\ref{def:duality} (duality) and \ref{def:enc_sestypfailunres}, we infer:
                \begin{equation*}
                    \begin{aligned}
                         \oplus(( \with \onef) \ampy ( \oplus  \with (( \oplus \piencodf{\sigma} ) \otimes (\piencodf{\sigma\wedge \sigma}_{(\tau, j)})))) &=\piencodf{\sigma\wedge \sigma \wedge \sigma}_{(\tau, j)}
                    \end{aligned}
                \end{equation*}
                Therefore, $\piencodf{\bag{M}\cdot C'}_{\linvar{x}} \vdash \piencodf{\Gamma}, \linvar{x}: \piencodf{\sigma\wedge \sigma \wedge \sigma}_{(\tau, j)} $ and the result follows.

            \end{enumerate}
            
             \item For $ \piencodf{ U }_{x_i} \vdash x_i: \piencodf{\eta}; \piencodf{\Theta}$ we consider $U$ to be a binary concatenation of 2 components, one being an empty unrestricted bag and the other being $\banged{\bag{M}}$. Hence we take $U = \banged{\oneb} \concat \banged{\bag{M}}$ with $\eta =  \sigma_1 \concat \sigma_2 $, $\piencodf{\eta_i} = \& \{ \mathtt{l}_1 ; \piencodf{ \sigma_1} , \mathtt{l}_2 ; \piencodf{\sigma_2} \}$ by \defref{def:enc_sestypfailunres} and finally by \defref{def:enc_lamrsharpifailunres} we have $\piencodf{ U }_{x_i} =  {x_i.\mathtt{case} \{ \mathtt{l}_{1} : \piencodf{\banged{\oneb}}_{x_i} , \mathtt{l}_{2} : \piencodf{\banged{\bag{M}}}_{x_i} \}}$, $\piencodf{\banged{\oneb}}_{x_i} =  x_i.\overline{\none} $ and $\piencodf{\banged{\bag{M}}}_{x_i} =  \piencodf{M}_{x_i} $, we can conclude $ \piencodf{ U }_{x_i} =  {x_i.\mathtt{case} \{ \mathtt{l}_{1} : x_i.\overline{\none} , \mathtt{l}_{2} : \piencodf{M}_{x_i} \}} $.

             Hence we have:
             \begin{prooftree}
                    \AxiomC{\(  \)}
                    \LeftLabel{\redlab{FS{:}bag^!}}
                    \UnaryInfC{\( \Theta ;  \dash  \wfdash \banged{\oneb} : \sigma_1 \)}
        
                    \AxiomC{\( \Theta ; \cdot \wfdash M : \sigma_2\)}
                    \LeftLabel{\redlab{FS{:}bag^!}}
                    \UnaryInfC{\( \Theta ; \cdot  \wfdash \banged{\bag{M}}:\sigma_2 \)}
                \LeftLabel{\redlab{FS{:}\concat-bag^{!}}}
                \BinaryInfC{\( \Theta ; \cdot  \wfdash \banged{\oneb} \concat \banged{\bag{M}} : \sigma_1 \concat \sigma_2 \)}
            \end{prooftree}
            
            By the induction hypothesis we have that $ \Theta ; \cdot \wfdash M : \sigma $ implies $ \piencodf{M}_{x_i} \wfdash x_i : \piencodf{\sigma} ;  \piencodf{\Theta}$

             \begin{prooftree}
                    \AxiomC{}
                    \LeftLabel{\redlab{T\with^x}}
                    \UnaryInfC{$ x_i.\dual{\none} \vdash x_i:  \piencodf{ \sigma_1}  ; \piencodf{\Theta} $}
                
                    \AxiomC{$ \piencodf{M}_{x_i} \vdash  x_i:  \piencodf{\sigma_2} ; \piencodf{\Theta} $}
                \LeftLabel{\redlab{T\with}}
                \BinaryInfC{$ {x_i.\mathtt{case} \{ \mathtt{l}_{1} : x_i.\none , \mathtt{l}_{2} : \piencodf{M}_{x_i} \}} \vdash  x_i: \& \{ \mathtt{l}_1 ; \piencodf{ \sigma_1} , \mathtt{l}_2 ; \piencodf{\sigma_2} \} ; \piencodf{\Theta} $}
            \end{prooftree}
            
            Therefore, ${x_i.\mathtt{case} \{ \mathtt{l}_{1} : x_i.\none , \mathtt{l}_{2} : \piencodf{M}_{x_i} \}} \vdash  x_i: \& \{ \mathtt{l}_1 ; \piencodf{ \sigma_1} , \mathtt{l}_2 ; \piencodf{\sigma_2} \} ; \piencodf{\Theta}  $ and the result follows.

        \end{enumerate}

        \item  The proof of type preservation for expressions, relies on the analysis of twelve cases:

        \begin{enumerate}
            
            \item {\bf Rule \redlab{FS{:}var^{\ell}}:}
            Then we have the following derivation:
            
            \begin{prooftree}
                \AxiomC{}
                \LeftLabel{\redlab{FS{:}var^{\ell}}}
                \UnaryInfC{\( \Theta ;  {x}: \tau \wfdash  {x} : \tau\)}
            \end{prooftree}
            
            By \defref{def:enc_sestypfailunres},  $\piencodf{ {x}:\tau}=  {x}:\with \overline{\piencodf{\tau }}$, and by \figref{fig:encfailunres},  $\piencodf{ {x}}_u =  {x}.\overline{\some} ; [ {x} \leftrightarrow u] $. The thesis holds thanks to the following derivation:
            
                \begin{prooftree}
                    \AxiomC{}
                    \LeftLabel{$\redlab{ (Tid)}$}
                    \UnaryInfC{$ [ {x} \leftrightarrow u] \vdash  {x}:  \overline{\piencodf{\tau }}  , u :  \piencodf{ \tau } ; \piencodf{\Theta} $}
                    \LeftLabel{$\redlab{T\with^{x}_{d})}$}
                    \UnaryInfC{$  {x}.\overline{\some} ; [ {x} \leftrightarrow u] \vdash  {x}: \with  \overline{\piencodf{ \tau }} , u :  \piencodf{ \tau } ; \piencodf{\Theta} $}
                \end{prooftree}

            \item {\bf Rule \redlab{FS{:}var^!}:}
            Then we have the following derivation provided $\eta_{ind} = \tau $:
            
            \begin{prooftree}
                \AxiomC{}
                \LeftLabel{\redlab{FS{:}var^{\ell}}}
                \UnaryInfC{\( \Theta , \banged{x}: \eta;  {x}: \eta_{ind}  \wfdash  {x} : \tau\)}
                \LeftLabel{\redlab{FS{:}var^!}}
                \UnaryInfC{\( \Theta , \banged{x}: \eta; \dash \wfdash  {x}[ind] : \tau\)}
            \end{prooftree}
            
            By \defref{def:enc_sestypfailunres},  $\piencodf{\Theta , \banged{x}: \eta}= \piencodf{\Theta} , \banged{x}: \dual{ \&_{\eta_i \in \eta} \{ \mathtt{l}_i ; \piencodf{\eta_i} \} }$, and by \figref{fig:encfailunres},  $\piencodf{ {x}[ind]}_u = \outsev{\banged{x}}{{x_i}}. {x}_i.l_{ind}; [{x_i} \leftrightarrow u] $. The thesis holds thanks to the following derivation:
            
            \begin{prooftree}
                    \AxiomC{}
                    \LeftLabel{$\redlab{ (Tid)}$}
                    \UnaryInfC{$ [{x_i} \leftrightarrow u]  \vdash  u :  \piencodf{ \tau },  x_i:  \overline{\piencodf{\eta_{ind} }}  ; \banged{x}:\&_{\eta_i \in \eta} \{ \mathtt{l}_i ; \piencodf{\eta_i} \} , \piencodf{\Theta} $}
                \LeftLabel{\redlab{T\oplus_i}}
                \UnaryInfC{$  {x}_i.l_{ind}; [{x_i} \leftrightarrow u] \vdash  u :  \piencodf{ \tau }, {x}_i :  \oplus_{\eta_i \in \eta} \{ \mathtt{l}_i ; \dual{\piencodf{\eta_i}}  \} ; \banged{x}:\oplus_{\eta_i \in \eta} \{ \mathtt{l}_i ; \dual{\piencodf{\eta_i}} \} , \piencodf{\Theta}  $}
                \LeftLabel{\redlab{Tcopy}}
                \UnaryInfC{$ \outsev{\banged{x}}{{x_i}}. {x}_i.l_{ind}; [{x_i} \leftrightarrow u] \vdash  u :  \piencodf{ \tau }; \banged{x}:\oplus_{\eta_i \in \eta} \{ \mathtt{l}_i ; \dual{\piencodf{\eta_i}} \} , \piencodf{\Theta} $}
            \end{prooftree}

            \item {\bf Rule \redlab{FS\!:\!weak}:}
                        Then we have the following derivation:
            
            \begin{prooftree}
                \AxiomC{\( \Theta ; \Gamma  \wfdash M : \tau\)}
                \LeftLabel{ \redlab{FS\!:\!weak}}
                \UnaryInfC{\( \Theta ; \Gamma ,  {x}: \omega \wfdash M[\leftarrow  {x}]: \tau \)}
            \end{prooftree}
            
            By \defref{def:enc_sestypfailunres},  $\piencodf{\Gamma ,  {x}: \omega}= \piencodf{\Gamma}, \linvar{x}: \overline{\piencodf{\omega }_{(\sigma, i_1)}}$, and by \figref{fig:encfailunres}, 
            
            $\piencodf{M[  \leftarrow  {x}]}_u = \linvar{x}. \overline{\some}. \outact{\linvar{x}}{y_i} . ( y_i . \some_{u,\llfv{M}} ;y_{i}.\close; \piencodf{M}_u \para \linvar{x}. \overline{\none}) $. 
            
             By IH, we have $\piencodf{M}_u\vdash  \piencodf{\Gamma }, u:\piencodf{\tau} ; \piencodf{\Theta }$. 
            The thesis holds thanks to the following derivation:
                        \begin{prooftree}
                    \AxiomC{$\piencodf{M}_u\vdash  \piencodf{\Gamma }, u:\piencodf{\tau} ; \piencodf{\Theta} $}
                    \LeftLabel{\redlab{T\bot}}
                    \UnaryInfC{$y_{i}.\close; \piencodf{M}_u \vdash y_i{:}\bot, \piencodf{\Gamma }, u:\piencodf{\tau} ; \piencodf{\Theta}$}
                    \LeftLabel{\redlab{T\oplus^x_{\widetilde{w}}}}
                    \UnaryInfC{$y_i . \some_{u,\llfv{M}} ;y_{i}.\close; \piencodf{M}_u \vdash  y_i{:}\oplus \bot , \piencodf{\Gamma }, u:\piencodf{\tau} ; \piencodf{\Theta}$}
                    
                    \AxiomC{}
                    \LeftLabel{\redlab{T\with^x}}
                    \UnaryInfC{$\linvar{x}. \overline{\none} \vdash \linvar{x} :\with A$}
                \LeftLabel{\redlab{T\otimes}}
                \BinaryInfC{$\outact{\linvar{x}}{y_i} . ( y_i . \some_{u,\llfv{M}} ;y_{i}.\close; \piencodf{M}_u \para \linvar{x}. \overline{\none}) \vdash  \linvar{x}: (\oplus \bot) \otimes (\with A) , \piencodf{\Gamma }, u:\piencodf{\tau} ; \piencodf{\Theta}$}
                \LeftLabel{\redlab{T\with_d^x}}
                \UnaryInfC{$\piencodf{M[  \leftarrow  {x}]}_u \vdash  \linvar{x} :\with ((\oplus \bot) \otimes (\with A))  , \piencodf{\Gamma }, u:\piencodf{\tau} ; \piencodf{\Theta} $}
             \end{prooftree}
            
             Since $A$ is arbitrary,  we can take $A=\oneb$ for $\piencodf{\omega}_{(\sigma,0)} $ and  $A= \overline{(( \with  \overline{\piencodf{ \sigma }} )  \ampy (\overline{\piencodf{\omega}_{(\sigma, i - 1)}}))}$  for $\piencodf{\omega}_{(\sigma,i)} $ where $i > 0$, in both cases, the result follows.

            \item {\bf Rule $\redlab{FS:abs \dash sh}$:}
        
            Then $\expr{M} = \lambda x . (M[ {\widetilde{x}} \leftarrow  {x}])$, and the derivation is:
            
            \begin{prooftree}
                \AxiomC{\( \Theta , \banged{x}:\eta ; \Gamma ,  {x}: \sigma^k \wfdash M[ {\widetilde{x}} \leftarrow  {x}] : \tau \quad  {x} \notin \dom{\Gamma} \)}
                \LeftLabel{\redlab{FS{:}abs\dash sh}}
                \UnaryInfC{\( \Theta ; \Gamma \wfdash \lambda x . (M[ {\widetilde{x}} \leftarrow  {x}])  : (\sigma^k, \eta )  \rightarrow \tau \)}
            \end{prooftree}
    
            By IH, we have $\piencodf{M[ {\widetilde{x}} \leftarrow  {x}]}_u \vdash  u:\piencodf{\tau} , \piencodf{\Gamma} , \linvar{x}: \overline{\piencodf{\sigma^k}_{(\sigma, i)}}  ; \piencodf{\Theta} , \banged{x}: \overline{\piencodf{\eta}} $, 
            From 
            \defref{def:enc_lamrsharpifailunres}, it follows 
            $ \piencodf{\lambda x.M[ {\widetilde{x}} \leftarrow x]}_u = u.\overline{\some}; u(x). x.\overline{\some}; x(\linvar{x}). x(\banged{x}). x. \close ; \piencodf{M[ {\widetilde{x}} \leftarrow  {x}]}_u $
            
            We give the final derivation in parts. The first part we name $\Pi_1$ derived by:
            \begin{prooftree}
                \AxiomC{$\piencodf{M[ {\widetilde{x}} \leftarrow  {x}]}_u \vdash  u:\piencodf{\tau} , \piencodf{\Gamma} , \linvar{x}: \overline{\piencodf{\sigma^k}_{(\sigma, i)}}  ; \piencodf{\Theta} , \banged{x}: \overline{\piencodf{\eta}}$}
                \LeftLabel{\redlab{T\bot}}
                \UnaryInfC{$ x. \close ; \piencodf{M[ {\widetilde{x}} \leftarrow  {x}]}_u \vdash x{:}\bot, u:\piencodf{\tau} , \piencodf{\Gamma} , \linvar{x}: \overline{\piencodf{\sigma^k}_{(\sigma, i)}}  ; \piencodf{\Theta} , \banged{x}: \overline{\piencodf{\eta}}$}
                \LeftLabel{\redlab{T?}}
                \UnaryInfC{$x. \close ; \piencodf{M[ {\widetilde{x}} \leftarrow  {x}]}_u \vdash x{:}\bot, u:\piencodf{\tau} , \piencodf{\Gamma} , \linvar{x}: \overline{\piencodf{\sigma^k}_{(\sigma, i)}} , \banged{x}: ? \overline{\piencodf{\eta}} ; \piencodf{\Theta} $}
                \LeftLabel{\redlab{T\ampy}}
                \UnaryInfC{$x(\banged{x}). x. \close ; \piencodf{M[ {\widetilde{x}} \leftarrow  {x}]}_u \vdash x: (? \overline{\piencodf{\eta}}) \ampy (\bot) , u:\piencodf{\tau} , \piencodf{\Gamma} , \linvar{x}: \overline{\piencodf{\sigma^k}_{(\sigma, i)}} ; \piencodf{\Theta} $}
                \LeftLabel{\redlab{T\ampy}}
                \UnaryInfC{$  x(\linvar{x}). x(\banged{x}). x. \close ; \piencodf{M[ {\widetilde{x}} \leftarrow  {x}]}_u \vdash  x: \overline{\piencodf{\sigma^k}_{(\sigma, i)}} \ampy ((? \overline{\piencodf{\eta}}) \ampy (\bot)) , u:\piencodf{\tau} , \piencodf{\Gamma} ; \piencodf{\Theta} $}
            \end{prooftree}
            We take $ P = x(\linvar{x}). x(\banged{x}). x. \close ; \piencodf{M[ {\widetilde{x}} \leftarrow  {x}]}_u $ and continue the derivation:

            \begin{prooftree}
                \AxiomC{$ \Pi_1 $}
                \noLine
                \UnaryInfC{$ \vdots $}
                \noLine
                \UnaryInfC{$  P \vdash  x: \overline{\piencodf{\sigma^k}_{(\sigma, i)}} \ampy ((? \overline{\piencodf{\eta}}) \ampy (\bot)) , u:\piencodf{\tau} , \piencodf{\Gamma} ; \piencodf{\Theta} $}
                \LeftLabel{\redlab{T\with_d^x}}
                \UnaryInfC{$ x.\overline{\some}; P \vdash x :\with(  \overline{\piencodf{\sigma^k}_{(\sigma, i)}} \ampy ((? \overline{\piencodf{\eta}}) \ampy (\bot))) , u:\piencodf{\tau} , \piencodf{\Gamma} ; \piencodf{\Theta} $}
                \LeftLabel{\redlab{T\ampy}}
                \UnaryInfC{$u(x). x.\overline{\some}; P \vdash u: \with(  \overline{\piencodf{\sigma^k}_{(\sigma, i)}} \ampy ((? \overline{\piencodf{\eta}}) \ampy (\bot))) \ampy \piencodf{\tau} , \piencodf{\Gamma} ; \piencodf{\Theta}  $}
                \LeftLabel{\redlab{T\with_d^x}}
                \UnaryInfC{$u.\overline{\some}; u(x). x.\overline{\some}; P \vdash u :\with (\with(  \overline{\piencodf{\sigma^k}_{(\sigma, i)}} \ampy ((? \overline{\piencodf{\eta}}) \ampy (\bot))) \ampy \piencodf{\tau}) , \piencodf{\Gamma} ; \piencodf{\Theta} $}
            \end{prooftree}
            
            By Definition \ref{def:enc_sestypfailunres} we have that $ \piencodf{(\sigma^{k} , \eta )   \rightarrow \tau} = \with (\with(  \overline{\piencodf{\sigma^k}_{(\sigma, i)}} \ampy ((? \overline{\piencodf{\eta}}) \ampy (\bot))) \ampy \piencodf{\tau}) $. Hence the case holds by $ \piencodf{\lambda x.M[\linvar{\widetilde{x}} \leftarrow x]}_u \vdash u:\piencodf{(\sigma^{k} , \eta )   \rightarrow \tau} , \piencodf{\Gamma} ; \piencodf{\Theta} $.

            \item {\bf Rule $\redlab{FS:app}$:}
            Then $\expr{M} = M\ B$, where $ B = C \bagsep U $ and the derivation is:
            
            \begin{prooftree}
                \AxiomC{\( \Theta ;\Gamma \wfdash M : (\sigma^{j} , \eta ) \rightarrow \tau \)}
                 \AxiomC{\(  \Theta ;\Delta \wfdash B : (\sigma^{k} , \epsilon )  \)}
                 \AxiomC{\( \eta \relunbag \epsilon \)}
                \LeftLabel{\redlab{FS{:}app}}
                \TrinaryInfC{\( \Theta ; \Gamma, \Delta \wfdash M\ B : \tau\)}
            \end{prooftree}

        
            By IH, we have both 
            
            \begin{itemize}
                \item  $\piencodf{M}_u\vdash \piencodf{\Gamma}, u:\piencodf{(\sigma^{j} , \eta ) \rightarrow \tau} ; \piencodf{\Theta} $
                \item $\piencodf{M}_u\vdash \piencodf{\Gamma}, u:\piencodf{(\sigma^{j} , \epsilon ) \rightarrow \tau} ; \piencodf{\Theta}$,   by Lemma \ref{lem:relunbag-typeunres}
                \item $\piencodf{B}_u\vdash \piencodf{\Delta}, u:\overline{\piencodf{(\sigma^{k} , \epsilon ) }_{(\tau_2, n)}}  ; \piencodf{\Theta} $, for some $\tau_2$ and some $n$.
            \end{itemize}
            
            Therefore, from the fact that $\expr{M}$ is well-formed and Definitions~\ref{def:enc_lamrsharpifailunres} and \ref{def:enc_sestypfailunres}, we  have:
            
            \begin{itemize}
                \item $\displaystyle{\piencodf{M (C \bagsep U)}_u = \bigoplus_{C_i \in \perm{C}} (\nu v)(\piencodf{M}_v \para v.\some_{u , \llfv{C}} ; \outact{v}{x} . ([v \leftrightarrow u] \para \piencodf{C_i \bagsep U}_x ) )} $;
                \item $\piencodf{(\sigma^{j} , \eta ) \rightarrow \tau}= \oplus( (\piencodf{\sigma^{k} }_{(\tau_1, m)}) \otimes ((!\piencodf{\eta})  $, for some $\tau_1$ and some $m$.
            \end{itemize}
            
            Also, since $\piencodf{B}_u\vdash \piencodf{\Delta}, u:\piencodf{(\sigma^{k} , \epsilon )}_{(\tau_2, n)}$, we have the following derivation $\Pi_i$:
         
          \begin{prooftree}
                    \AxiomC{$\piencodf{C_i \bagsep U}_x\vdash \piencodf{\Delta}, x:\piencodf{(\sigma^{k} , \epsilon )}_{(\tau_2, n)}  ; \piencodf{\Theta} $ }
                        
                    \AxiomC{\(\)}                                   
                    \LeftLabel{$ \redlab{Tid}$}
                    \UnaryInfC{$ [v \leftrightarrow u]                   \vdash v:  \overline{\piencodf{ \tau }} , u: \piencodf{ \tau }$}
                \LeftLabel{$\redlab{T \otimes}$}
                \BinaryInfC{$\outact{v}{x} . ([v \leftrightarrow u] \para \piencodf{C_i \bagsep U}_x ) \vdash \piencodf{ \Delta }, v:\piencodf{(\sigma^{k} , \epsilon )}_{(\tau_2, n)} \otimes \overline{ \piencodf{ \tau }} , u:\piencodf{ \tau }  ; \piencodf{\Theta} $}
                \LeftLabel{$\redlab{T\oplus^{v}_{w}}$}
                \UnaryInfC{$  v.\some_{u , \llfv{C}} ; \outact{v}{x} . ([v \leftrightarrow u] \para \piencodf{C_i \bagsep U}_x )  \vdash\piencodf{ \Delta }, v:\oplus (\piencodf{(\sigma^{k} , \epsilon )}_{(\tau_2, n)} \otimes  \overline{\piencodf{ \tau }}), u:\piencodf{ \tau } ; \piencodf{\Theta} $} 
            \end{prooftree}

            Notice that 
          \(
              \oplus (\piencodf{(\sigma^{k} , \epsilon )}_{(\tau_2, n)} \otimes  \overline{\piencodf{ \tau }}) =  \overline{\piencodf{(\sigma^{k} , \epsilon ) \rightarrow \tau}}.
            \)
            Therefore, by one application of $\redlab{Tcut}$ we obtain the derivations $\nabla_i$, for each $C_i \in \perm{C}$:
    
            \begin{prooftree}
            \AxiomC{\( \piencodf{M}_v\vdash \piencodf{\Gamma}, v:\piencodf{(\sigma^{j} , \epsilon ) \rightarrow \tau} ; \piencodf{\Theta} \)}
            \AxiomC{$\Pi_i$}
            \LeftLabel{\( \redlab{Tcut} \)}    
            \BinaryInfC{$ (\nu v)( \piencodf{ M}_v \para v.\some_{u , \llfv{C}} ; \outact{v}{x} . ([v \leftrightarrow u] \para \piencodf{B_i}_x ) ) \vdash \piencodf{ \Gamma } ,\piencodf{ \Delta } , u: \piencodf{ \tau }  ; \piencodf{\Theta} $}
            \end{prooftree}
            
            In order to apply \redlab{Tcut}, we must have that $\piencodf{\sigma^{j}}_{(\tau_1, m)} = \piencodf{\sigma^{k}}_{(\tau_2, n)}$, therefore, the choice of $\tau_1,\tau_2,n$ and $m$, will consider the different possibilities for $j$ and $k$, as in Proposition~\ref{prop:app_auxunres}.
            We can then conclude that $\piencodf{M B}_u \vdash \piencodf{ \Gamma}, \piencodf{ \Delta }, u:\piencodf{ \tau } ; \piencodf{\Theta} $:
            
            \begin{prooftree}
            \AxiomC{For each $C_i \in \perm{C} \qquad  \nabla_i$}
            \LeftLabel{$\redlab{T\with}$}
            \UnaryInfC{$\displaystyle{\bigoplus_{C_i \in \perm{C}} (\nu v)( \piencodf{ M}_v \para v.\some_{u , \lfv{B}} ; \outact{v}{x} . ([v \leftrightarrow u] \para \piencodf{B_i}_x ) ) \vdash \piencodf{ \Gamma } ,\piencodf{ \Delta } , u: \piencodf{ \tau }} ; \piencodf{\Theta} $}
            \end{prooftree}
            
            and the result follows.

            \item {\bf Rule $\redlab{FS:share}$:}
            Then $\expr{M} = M [  {x}_1, \dots  {x}_k \leftarrow x ]$ and the derivation is:
            \begin{prooftree}
                \AxiomC{\( \Theta ; \Delta ,  {x}_1: \sigma, \cdots,  {x}_k: \sigma \wfdash M : \tau \quad  {x} \notin \Delta \quad k \not = 0\)}
                \LeftLabel{ \redlab{FS:share}}
                \UnaryInfC{\( \Theta ; \Delta ,  {x}: \sigma_{k} \wfdash M[ {x}_1 , \cdots ,  {x}_k \leftarrow  {x}] : \tau \)}
            \end{prooftree}
    
            To simplify the proof we will consider $k=1$ (the case in which $k>1$ follows similarly). 
    
            By IH, we have $\piencodf{M}_u\vdash  \piencodf{\Delta ,  {x}_1:\sigma }, u:\piencodf{\tau} ; \piencodf{\Theta}$. 
            From 
            Definitions~\ref{def:enc_lamrsharpifailunres} and \ref{def:enc_sestypfailunres}, it follows 
            
            \begin{itemize}
             \item $\piencodf{ \Delta ,  {x}_1: \sigma} = \piencodf{\Delta}, \linvar{x}_1:\with\overline{\piencodf{\sigma}} $.
            \item 
            $
            \piencodf{M[ {x}_1, \leftarrow  {x}]}_u =
               \begin{array}[t]{l}
                  \linvar{x}.\overline{\some}. \outact{\linvar{x}}{y_1}. (y_1 . \some_{\emptyset} ;y_{1}.\close;0
                    \para \linvar{x}.\overline{\some};\linvar{x}.\some_{u, (\llfv{M} \setminus  {x}_1 )};
                  \\
                  \linvar{x}( {x}_1) .  \linvar{x}. \overline{\some}. \outact{\linvar{x}}{y_2} . ( y_2 . \some_{u,\llfv{M}} ;y_{2}.\close; \piencodf{M}_u \para \linvar{x}. \overline{\none}) )
               \end{array}
                $
           
            \end{itemize}

            We shall split the expression into two parts:
            \[
            \begin{aligned}
                N_1 &= \linvar{x}. \overline{\some}. \outact{\linvar{x}}{y_2} . ( y_2 . \some_{u,\llfv{M}} ;y_{2}.\close; \piencodf{M}_u \para \linvar{x}. \overline{\none}) \\
                N_2 &= \linvar{x}.\overline{\some}. \outact{\linvar{x}}{y_1}. (y_1 . \some_{\emptyset} ;y_{1}.\close;0 \para \linvar{x}.\overline{\some};\linvar{x}.\some_{u, (\llfv{M} \setminus  {x}_1 )};\linvar{x}( {x}_1) .N_1)
            \end{aligned}
            \]

            and we obtain the  derivation for term $N_1$ as follows where we omit $; \piencodf{\Theta}$:
            \begin{prooftree}
                    \AxiomC{$\piencodf{M}_u \vdash  \piencodf{\Delta ,  {x}_1:\sigma }, u:\piencodf{\tau} $}
                    \LeftLabel{\redlab{T\bot}}
                    \UnaryInfC{$y_{2}.\close; \piencodf{M}_u \vdash \piencodf{\Delta ,  {x}_1:\sigma }, u:\piencodf{\tau}, y_{2}{:}\bot $}
                    \LeftLabel{ \redlab{ T\oplus^x_{ \widetilde{w}}}}
                    \UnaryInfC{$ y_2 . \some_{u,\llfv{M}} ;y_{2}.\close; \piencodf{M}_u \vdash \piencodf{\Delta ,  {x}_1:\sigma }, u:\piencodf{\tau}, y_{2}{:}\oplus \bot $}
                    
                    \AxiomC{}
                    \LeftLabel{\redlab{T\with^x}}
                    \UnaryInfC{$\linvar{x}.\dual{\none} \vdash \linvar{x} :\with A $}
                \LeftLabel{\redlab{T\otimes}}
                \BinaryInfC{$  \outact{\linvar{x}}{y_2} . ( y_2 . \some_{u,\llfv{M}} ;y_{2}.\close; \piencodf{M}_u \para \linvar{x}. \overline{\none}) \vdash \piencodf{\Delta ,  {x}_1:\sigma }, u:\piencodf{\tau} , \linvar{x}: ( \oplus \bot )\otimes ( \with A ) $}
                \LeftLabel{\redlab{T\with_d^x}}
                \UnaryInfC{$\underbrace{ \linvar{x}. \overline{\some}. \outact{\linvar{x}}{y_2} . ( y_2 . \some_{u,\llfv{M}} ;y_{2}.\close; \piencodf{M}_u \para \linvar{x}. \overline{\none}) }_{N_1} \vdash \piencodf{\Delta ,  {x}_1:\sigma } , u:\piencodf{\tau} , \linvar{x}: \overline{\piencodf{\omega}_{(\sigma, i)}} $}
            \end{prooftree}
            
            Notice that the last rule applied \redlab{T\with_d^x} assigns $x: \with ((\oplus \bot) \otimes (\with A))$. Again, since $A$ is arbitrary, we can take $A= \oplus (( \with  \overline{\piencodf{ \sigma }} )  \ampy (\overline{\piencodf{\omega}_{(\sigma, i - 1)}}))$, obtaining $x:\overline{\piencodf{\omega}_{(\sigma,i)}}$.
            
            In order to obtain a type derivation for $N_2$, consider the derivation $\Pi_1$:
            \begin{prooftree}
                    \AxiomC{$N_1 \vdash \piencodf{\Delta},  {x}_1:\with\overline{\piencodf{\sigma}} , u:\piencodf{\tau}, \linvar{x}: \overline{\piencodf{\omega}_{(\sigma, i)}} $}
                    \LeftLabel{\redlab{T\ampy}}
                    \UnaryInfC{$\linvar{x}( {x}_1) .N_1   \vdash \piencodf{\Delta} ,  u:\piencodf{\tau}, \linvar{x}: ( \with\overline{\piencodf{\sigma}} ) \ampy (\overline{\piencodf{\omega}_{(\sigma, i)}}) $}
                    \LeftLabel{\redlab {T\oplus^x_{ \widetilde{w} }}}
                    \UnaryInfC{$ \linvar{x}.\some_{u, (\llfv{M} \setminus  {x}_1 )}; \linvar{x}( {x}_1) .N_1  \vdash \piencodf{\Delta} ,  u:\piencodf{\tau}, \linvar{x} {:}\oplus (( \with \overline{\piencodf{\sigma}} ) \ampy (\overline{\piencodf{\omega}_{(\sigma, i)}}))$}
                    \LeftLabel{\redlab{T\with_d^x}}
                    \UnaryInfC{$ \linvar{x}.\overline{\some};\linvar{x}.\some_{u, (\llfv{M} \setminus  {x}_1 )};\linvar{x}( {x}_1) .N_1  \vdash \piencodf{\Delta},  u:\piencodf{\tau} , \linvar{x} :\with \oplus (( \with\overline{\piencodf{\sigma}} ) \ampy ( \overline{\piencodf{\omega}_{(\sigma, i)}} ))$}
            \end{prooftree}
            
            We take $ P_1 = \linvar{x}.\overline{\some};\linvar{x}.\some_{u, (\llfv{M} \setminus  {x}_1 )};\linvar{x}( {x}_1) .N_1 $ and $\Gamma_1 =   \piencodf{\Delta},  u:\piencodf{\tau} $ and continue the derivation of $ N_2 $

            \begin{prooftree}
                    \AxiomC{}
                    \LeftLabel{\redlab{T\cdot}}
                    \UnaryInfC{$0\vdash \dash ; \piencodf{\Theta} $}
                    \LeftLabel{\redlab{T\bot}}
                    \UnaryInfC{$ y_{1}.\close;0 \vdash y_{1} : \bot ; \piencodf{\Theta} $}
                    \LeftLabel{\redlab{T\oplus^x_{\widetilde{w}}}}
                    \UnaryInfC{$ y_1 . \some_{\emptyset} ;y_{1}.\close;0 \vdash  y_1{:}\oplus \bot ; \piencodf{\Theta} $}
                    
                    \AxiomC{$\Pi_1 $}
                    \noLine
                    \UnaryInfC{$\vdots$}
                    \noLine
                    \UnaryInfC{$P_1\vdash \Gamma_1, \linvar{x} :\with \oplus (( \with\overline{ \piencodf{\sigma}} ) \ampy ( \overline{\piencodf{\omega}_{(\sigma, i)}} )) ; \piencodf{\Theta} $}
                \LeftLabel{\redlab{T\otimes}}
                \BinaryInfC{$\outact{\linvar{x}}{y_1}. (y_1 . \some_{\emptyset} ;y_{1}.\close;0 \para P_1) \vdash \Gamma_1 ,  \linvar{x} : (\oplus \bot)\otimes (\with \oplus (( \with\overline{\piencodf{\sigma}} ) \ampy ( \overline{\piencodf{\omega}_{(\sigma, i)}} )) ) ; \piencodf{\Theta} $}
                \LeftLabel{\redlab{T\with_d^x}}
                \UnaryInfC{$\underbrace{\linvar{x}.\overline{\some}. \outact{\linvar{x}}{y_1}. (y_1 . \some_{\emptyset} ;y_{1}.\close;0 \para P_1)}_{N_2} \vdash \Gamma_1 , \linvar{x} : \overline{ \piencodf{\sigma \wedge \omega}_{(\sigma, i)}}  ; \piencodf{\Theta}$}
            \end{prooftree}
            
            Hence the theorem holds for this case.

            \item {\bf Rule $\redlab{FS:ex \dash sub}$:}
            Then $\expr{M} = (M[ {\widetilde{x}} \leftarrow  {x}])\esubst{ B }{ x }$ and
            
            \begin{prooftree}
                    \AxiomC{\( \Theta , \banged{x} : \eta ; \Gamma ,  {x}: \sigma^{j} \wfdash M[ {\widetilde{x}} \leftarrow  {x}] : \tau  \)}
                    \AxiomC{\( \Theta ; \Delta \wfdash B : (\sigma^{k} , \epsilon ) \)}
                    \AxiomC{\( \eta \relunbag \epsilon \)}
                \LeftLabel{\redlab{FS{:}ex \dash sub}}    
                \TrinaryInfC{\( \Theta ; \Gamma, \Delta \wfdash (M[ {\widetilde{x}} \leftarrow  {x}])\esubst{ B }{ x }  : \tau \)}
            \end{prooftree}
        
            By Proposition~\ref{prop:app_auxunres} and IH we have:
            $$
            \begin{array}{rl} 
            \piencodf{ M[x_1, \cdots , x_k \leftarrow x]}_u&\vdash \piencodf{\Gamma}, \linvar{x}: \overline{ \piencodf{ \sigma^j }_{(\tau, n)}} , u:\piencodf{\tau} ; \piencodf{\Theta} , \banged{x} : \dual{\piencodf{\eta}} \\
            \piencodf{ M[x_1, \cdots , x_k \leftarrow x]}_u&\vdash \piencodf{\Gamma}, \linvar{x}: \overline{ \piencodf{ \sigma^j }_{(\tau, n)}} , u:\piencodf{\tau} ; \piencodf{\Theta} , \banged{x} : \dual{\piencodf{\epsilon}},\text{ by Lemma \ref{lem:relunbag-typeunres} }\\
            \piencodf{B}_x&\vdash \piencodf{\Delta}, x:\piencodf{ (\sigma^{k} , \epsilon ) }_{(\tau, m)}  ; \piencodf{\Theta}
            \end{array}
            $$
        
            From \defref{def:enc_lamrsharpifailunres}, we have 
            \begin{equation*}
               \piencodf{ M[ {\widetilde{x}} \leftarrow  {x}]\ \esubst{ B }{ x }}_u = \bigoplus_{C_i \in \perm{C}} (\nu x)( x.\overline{\some}; x(\linvar{x}). x(\banged{x}). x. \close ;\piencodf{ M[ {\widetilde{x}} \leftarrow  {x}]}_u \para \piencodf{ C_i \bagsep U}_x )  
            \end{equation*}

            Therefore, for each $B_i \in \perm{B} $, we obtain the following derivation $\Pi_i$:
            \begin{prooftree}
                \AxiomC{$ \piencodf{ M[ {\widetilde{x}} \leftarrow  {x}]}_u \vdash \piencodf{\Gamma}, \linvar{x}: \overline{ \piencodf{ \sigma^j }_{(\tau, n)}} , u:\piencodf{\tau} ; \piencodf{\Theta} , \banged{x} : \dual{\piencodf{\epsilon}} $}
                \LeftLabel{\redlab{T\bot}}
                \UnaryInfC{$ x.\close ;\piencodf{ M[ {\widetilde{x}} \leftarrow  {x}]}_u \vdash x{:}\bot, \piencodf{\Gamma}, \linvar{x}: \overline{ \piencodf{ \sigma^j }_{(\tau, n)}} , u:\piencodf{\tau} ; \piencodf{\Theta} , \banged{x} : \dual{\piencodf{\epsilon}}$}
                \LeftLabel{\redlab{T?}}
                \UnaryInfC{$ x.\close ;\piencodf{ M[ {\widetilde{x}} \leftarrow  {x}]}_u \vdash x{:}\bot, \piencodf{\Gamma}, \linvar{x}: \overline{ \piencodf{ \sigma^j }_{(\tau, n)}} , u:\piencodf{\tau} , \banged{x} :! \dual{\piencodf{\epsilon}}; \piencodf{\Theta} $}
                \LeftLabel{\redlab{T\ampy}}
                \UnaryInfC{$x(\banged{x}). x.\close ;\piencodf{ M[ {\widetilde{x}} \leftarrow  {x}]}_u \vdash x:  ( ! \dual{\piencodf{\epsilon}} ) \ampy \bot , \piencodf{\Gamma}, \linvar{x}: \overline{ \piencodf{ \sigma^j }_{(\tau, n)}} , u:\piencodf{\tau} ; \piencodf{\Theta} $}
                \LeftLabel{\redlab{T\ampy}}
                \UnaryInfC{$x(\linvar{x}). x(\banged{x}). x. \close ;\piencodf{ M[ {\widetilde{x}} \leftarrow  {x}]}_u \vdash x: \overline{ \piencodf{ \sigma^j }_{(\tau, n)}} \ampy ( ( ! \dual{\piencodf{\epsilon}} ) \ampy \bot) , \piencodf{\Gamma}, u:\piencodf{\tau} ; \piencodf{\Theta} $}
                \LeftLabel{\redlab{T\with_d^x}}
                \UnaryInfC{$x.\overline{\some}; x(\linvar{x}). x(\banged{x}). x. \close ;\piencodf{ M[ {\widetilde{x}} \leftarrow  {x}]}_u \vdash x : \overline{\piencodf{ (\sigma^{j} , \epsilon  )  }_{(\tau, n)}}  , \piencodf{\Gamma}, u:\piencodf{\tau} ; \piencodf{\Theta}  $}
             \end{prooftree}
            
            We take $ P_1 = x.\overline{\some}; x(\linvar{x}). x(\banged{x}). x. \close ;\piencodf{ M[ {\widetilde{x}} \leftarrow  {x}]}_u$ and continue the derivation of $ \Pi_i $
            \begin{prooftree}
                \AxiomC{$ P_1 \vdash  x : \overline{\piencodf{ (\sigma^{j} , \epsilon  )  }_{(\tau, n)}}  , \piencodf{\Gamma}, u:\piencodf{\tau} ; \piencodf{\Theta}  $}
                \AxiomC{$  \piencodf{ C_i \bagsep U}_x \vdash \piencodf{\Delta}, x:\piencodf{ (\sigma^{k} , \epsilon ) }_{(\tau, m)}  ; \piencodf{\Theta} $}
            \LeftLabel{$\redlab{Tcut}$}
            \BinaryInfC{$ (\nu x)( P_1 \para \piencodf{ C_i \bagsep U}_x )  \vdash \piencodf{ \Gamma} , \piencodf{ \Delta } , u: \piencodf{ \tau }  ; \piencodf{\Theta}  $}               
            \end{prooftree}
            
            We must have that $\piencodf{\sigma^{j}}_{(\tau, m)} = \piencodf{\sigma^{k}}_{(\tau, n)}$ which by our restrictions allows.
            Therefore, from $\Pi_i$ and multiple applications of $\redlab{T\with}$ it follows that
            
            \begin{prooftree}
                        \AxiomC{$\forall \bigoplus_{C_i \in \perm{C}} \hspace{1cm} \Pi_i$}
                        \LeftLabel{$\redlab{T\with}$}
            \UnaryInfC{$ \bigoplus_{C_i \in \perm{C}}  (\nu x)( P_1 \para \piencodf{ C_i \bagsep U}_x )  \vdash \piencodf{ \Gamma} , \piencodf{ \Delta } , u: \piencodf{ \tau }  ; \piencodf{\Theta} $}
            \end{prooftree}
            that is, $\piencodf{M[  {x}_1 \leftarrow  {x}]\ \esubst{ B }{ x }}\vdash \piencodf{\Gamma, \Delta}, u:\piencodf{\tau}  ; \piencodf{\Theta}$ and the result follows.

            \item {\bf Rule $\redlab{FS{:}ex \dash sub^{\ell}}$: } 
            Then $\expr{M} =  M \linexsub{N /  {x}}$ and
            \begin{prooftree}
                \AxiomC{\( \Theta ; \Gamma  ,  {x}:\sigma \wfdash M : \tau \quad  \Theta ; \Delta \wfdash N : \sigma \)}
                    \LeftLabel{\redlab{FS{:}ex \dash sub^{\ell}}}
                \UnaryInfC{\( \Theta ; \Gamma, \Delta \wfdash M \linexsub{N /  {x}} : \tau \)}
            \end{prooftree}
        
            By IH we have both 
            $$\piencodf{N}_{ {x}}\vdash \piencodf{\Delta},  {x}: \piencodf{\sigma} ; \piencodf{\Theta}$$ 
            $$ \piencodf{M}_u\vdash \piencodf{\Gamma},  {x}: \with\overline{\piencodf{\sigma}}, u:\piencodf{\tau} ; \piencodf{\Theta}$$
        
            From Definition \ref{def:enc_lamrsharpifailunres}, $\piencodf{M \linexsub{N /  {x}} }_u= (\nu  {x}) ( \piencodf{ M }_u \para    {x}.\some_{\llfv{N}};\piencodf{ N }_{ {x}} ) $ and 
            \begin{prooftree}
                    \AxiomC{\( \piencodf{ M }_u \vdash   \piencodf{\Gamma},  {x}: \with\overline{\piencodf{\sigma}}, u:\piencodf{\tau} ; \piencodf{\Theta}\)}
                    
                    \AxiomC{\( \piencodf{ N }_x  \vdash  \piencodf{\Delta},  {x}: \piencodf{\sigma} ; \piencodf{\Theta} \)}
                    \LeftLabel{$\redlab{T\oplus^x}$}
                    \UnaryInfC{\(  {x}.\some_{\llfv{N}};\piencodf{ N }_x \vdash  \piencodf{ \Delta } ,  {x}: : \oplus \piencodf{\sigma}\)}
                \LeftLabel{$\redlab{TCut}$}
                \BinaryInfC{\(  (\nu  {x}) ( \piencodf{ M }_u \para    {x}.\some_{\llfv{N}};\piencodf{ N }_{ {x}} ) \vdash  \piencodf{ \Gamma} , \piencodf{ \Delta } , u : \piencodf{ \tau }  \)}
            \end{prooftree}
    
            Observe that for the application of rule $\redlab{TCut}$ we used the fact that $\overline{\oplus\piencodf{\sigma}}=\with \overline{\piencodf{\sigma}}$. Therefore, $\piencodf{M \linexsub{N /  {x}} }_u\vdash \piencodf{ \Gamma} , \piencodf{ \Delta } , u : \piencodf{ \tau } $ and the result follows.

            \item {\bf Rule $\redlab{FS{:}ex \dash sub^!}$: }
            Then $\expr{M} =  M \unexsub{U / \unvar{x}}$ and
            \begin{prooftree}
                \AxiomC{\( \Theta , \banged{x} : \eta; \Gamma \wfdash M : \tau \quad  \Theta ; \dash \wfdash U : \eta \)}
                    \LeftLabel{\redlab{FS{:}ex \dash sub^!}}
                \UnaryInfC{\( \Theta ; \Gamma \wfdash M \unexsub{U / \unvar{x}}  : \tau \)}
            \end{prooftree}
        
            By IH we have both 
            $$
            \begin{array}{rl}
            \piencodf{U}_{x_i}&\vdash x_i : \piencodf{\eta}  ; \piencodf{\Theta}\\ 
             \piencodf{M}_u&\vdash \piencodf{\Gamma} , u:\piencodf{\tau} ; \banged{x}: \overline{\piencodf{\eta}} , \piencodf{\Theta}
             \end{array}
             $$
        
            From Definition \ref{def:enc_lamrsharpifailunres}, $ \piencodf{ M \unexsub{U / \unvar{x}}  }_u   =   (\nu \banged{x}) ( \piencodf{ M }_u \para   !\banged{x}. (x_i).\piencodf{ U }_{x_i} ) $ and 
            
            \begin{prooftree}
                    \AxiomC{$\piencodf{M}_u\vdash \piencodf{\Gamma} , u:\piencodf{\tau}; \banged{x}: \overline{\piencodf{\eta}} , \piencodf{\Theta}$}
                    \LeftLabel{\redlab{T?}}
                    \UnaryInfC{$\piencodf{M}_u \vdash \piencodf{\Gamma} , u:\piencodf{\tau}, \banged{x}: ? \overline{\piencodf{\eta}} ; \piencodf{\Theta}$}
                    
                    \AxiomC{$ \piencodf{ U }_{x_i} \vdash x_i : \piencodf{\eta} ; \piencodf{\Theta} $}
                    \LeftLabel{\redlab{T!}}
                    \UnaryInfC{$!\banged{x}. (x_i).\piencodf{ U }_{x_i} \vdash \banged{x}: !\piencodf{\eta} ; \piencodf{\Theta}  $}
                \LeftLabel{$\redlab{TCut}$}
                \BinaryInfC{\(   (\nu \banged{x}) ( \piencodf{ M }_u \para   !\banged{x}. (x_i).\piencodf{ U }_{x_i} ) \vdash \piencodf{\Gamma} , u:\piencodf{\tau} ; \piencodf{\Theta} \)}
            \end{prooftree}
    
            Observe that for the application of rule $\redlab{TCut}$ we used the fact that $\overline{ !\piencodf{\eta} }= ? \overline{\piencodf{\eta}} $. Therefore, $\piencodf{M \unexsub{U / \unvar{x}} }_u \vdash \piencodf{\Gamma} , u:\piencodf{\tau} ; \piencodf{\Theta} $ and the result follows.

            \item {\bf Rule $\redlab{FS:fail}$:}
            Then $\expr{M} = \fail^{\widetilde{x}}$ where $ \widetilde{x} = x_1, \cdots , x_n$ and
            
                \begin{prooftree}
                    \AxiomC{\( \dom{\Gamma} = \widetilde{x}\)}
                    \LeftLabel{\redlab{FS{:}fail}}
                    \UnaryInfC{\( \Theta ; \Gamma \wfdash  \fail^{\widetilde{x}} : \tau \)}
                \end{prooftree}
            
            From Definition \ref{def:enc_lamrsharpifailunres}, $\piencodf{\fail^{x_1, \cdots , x_n} }_u= u.\overline{\none} \para x_1.\overline{\none} \para \cdots \para x_k.\overline{\none} $ and

            \begin{prooftree}
                    \AxiomC{}
                    \LeftLabel{\redlab{T\with^u}}
                    \UnaryInfC{$u.\overline{\none} \vdash u : \piencodf{ \tau } ; \piencodf{\Theta} $}

                        \AxiomC{}
                        \LeftLabel{\redlab{T\with^{x_1}}}
                        \UnaryInfC{$x_1.\overline{\none} \vdash_1 : \with \overline{\piencodf{\sigma_1}} ; \piencodf{\Theta} $}
                        
                        \AxiomC{}
                        \LeftLabel{\redlab{T\with^{x_n}}}
                        \UnaryInfC{$x_n.\overline{\none} \vdash x_n : \with \overline{\piencodf{\sigma_n}} ; \piencodf{\Theta} $}
                        \UnaryInfC{$\vdots$}
                    \BinaryInfC{$x_1.\overline{\none} \para \cdots \para x_k.\overline{\none} \vdash  x_1 : \with \overline{\piencodf{\sigma_1}}, \cdots  ,x_n : \with \overline{\piencodf{\sigma_n}} ; \piencodf{\Theta} $}
                \LeftLabel{\redlab{T\para}}
                \BinaryInfC{$u.\overline{\none} \para x_1.\overline{\none} \para \cdots \para x_k.\overline{\none} \vdash x_1 : \with \overline{\piencodf{\sigma_1}}, \cdots  ,x_n : \with \overline{\piencodf{\sigma_n}}, u : \piencodf{ \tau } ; \piencodf{\Theta} $}
            \end{prooftree}
            
            Thus, $\piencodf{\fail^{x_1, \cdots , x_n} }_u\vdash  x_1 : \with \overline{\piencodf{\sigma_1}}, \cdots  ,x_n : \with \overline{\piencodf{\sigma_n}}, u : \piencodf{ \tau } ; \piencodf{\Theta} $ and the result follows.

            \item Rule $\redlab{FS:sum}$: 
            This case follows easily by IH.
        \end{enumerate}
    \end{enumerate}
\end{proof}

\subsection{Operational Correspondence: Completeness and Soundness}

\begin{proposition}
\label{prop:correctformfailunres}
Let  $N$ be a well-formed linearly closed $\lamrsharfailunres$-term with $\headf{N} = x$ ($x$ denoting either linear or unrestricted occurrence of $x$) such that $\llfv{N} = \emptyset$ and $N$ does not fail, that is, there is no $Q\in \lamrsharfailunres$ for which there is a  reduction $N  \red_{\redlab{RS:Fail}} Q$.
Then,
$$\piencodf{ N }_{u} \red^* \bigoplus_{i \in I}(\nu \widetilde{y})(\piencodf{ x }_{n} \para P_i) $$ for some index set $I$, names $\widetilde{y}$ and $n$, and processes $P_i$.
\end{proposition}

\begin{proof}
By induction on the structure of $N$.
\begin{enumerate}

    \item $N =  {x}$ or $N=x[j]$:
    
    These cases are trivial, and follow taking 
    $ I = \emptyset$ and $ \widetilde{y} = \emptyset$.
    
    

    \item $N = (M\ B)$:
    
    Then $\headf{M\ B} = \headf{M} = x$ then 
    \[ \piencodf{N}_u = \piencodf{M\ B}_u  = \bigoplus_{B_i \in \perm{B}} (\nu v)(\piencodf{M}_v \para v.\some_{u, \llfv{B}} ; \outact{v}{x} . ([v \leftrightarrow u] \para \piencodf{B^x_i} ) ) \]
   and the proof follows by induction on $\piencodf{M}_u$.
    
    \item  $N = (M[\widetilde{y} \leftarrow y])\esubst{ C \bagsep U }{ y }$:
    
    Then $\headf{(M[\widetilde{y} \leftarrow y])\esubst{ C \bagsep U }{ y }} = \headf{(M[\widetilde{y} \leftarrow y])} = x$. As $N  \red_{\redlab{R}} $ where $\redlab{R} \not = \redlab{RS:Fail} $ we must have that $\size{\widetilde{y}} = \size{C}$. Thus,

    \[
    \begin{aligned}
       \piencodf{N}_u &= \piencodf{(M[\widetilde{y} \leftarrow y])\esubst{ C \bagsep U }{ y }}_u  \\
       & =  \bigoplus_{C_i \in \perm{C}} (\nu y)( y.\overline{\some}; y(\linvar{y}). y(\banged{y}). y. \close ;\piencodf{ M[\widetilde{y} \leftarrow  {y}]}_u \para \piencodf{ C_i \bagsep U}_y )
       \\
       & =  \bigoplus_{C_i \in \perm{C}} (\nu y)( y.\overline{\some}; y(\linvar{y}). y(\banged{y}). y. \close ;\piencodf{ M[\widetilde{y} \leftarrow  {y}]}_u \para \\
       & \hspace{2cm} y.\some_{\llfv{C}} ; \outact{y}{\linvar{y}} .( \piencodf{ C_i }_{\linvar{y}} \para \outact{y}{\banged{y}} .( !\banged{y}. (y_i). \piencodf{ U }_{y_i} \para y.\overline{\close} ) ) )
       \\[4pt]
       & \red^*  \bigoplus_{C_i \in \perm{C}} (\nu  \linvar{y}, \banged{y})(  \piencodf{ M[\widetilde{y} \leftarrow \linvar{y}]}_u \para  \piencodf{ C_i }_{\linvar{y}} \para  !\banged{y}. (y_i). \piencodf{ U }_{y_i}  )
       \\[4pt]
       &=\bigoplus_{C_i \in \perm{C}} (\nu \linvar{y}, \banged{y})( \linvar{y}.\overline{\some}. \outact{\linvar{y}}{z_1}. (z_1 . \some_{\emptyset} ;z_{1}.\close;0 \para \linvar{y}.\overline{\some};\\
         &\hspace{1cm}\linvar{y}. \some_{u, (\llfv{M} \setminus  {y}_1 , \cdots ,  {y}_n )};\linvar{y}( {y}_1) . \cdots\linvar{y}.\overline{\some}. \outact{\linvar{y}}{z_n} . (z_n . \some_{\emptyset} ; z_{n}.\close;0 \\
        &\hspace{1cm}\para \linvar{y}.\overline{\some};\linvar{y}.\some_{u,(\llfv{M} \setminus  {y}_n )}; \linvar{y}( {y}_n).\linvar{y}.\overline{\some}; \outact{\linvar{y}}{z_{n+1}}. ( z_{n+1} . \some_{u, \llfv{M}} ; \\%
        &\hspace{1cm}z_{n+1}.\close; \piencodf{M}_u \para \linvar{y}.\overline{\none} ) ) \cdots ) \para \linvar{y}.\some_{\llfv{C}} ; \linvar{y}(z_1). \linvar{y}.\some_{z_1,\llfv{C}};\linvar{y}.\overline{\some} ; \\ 
              &\hspace{1cm}  \outact{\linvar{y}}{ {y}_1}. ( {y}_1.\some_{\llfv{C_i(1)}} ;\piencodf{C_i(1)}_{ {y}_1} \para  \cdots \linvar{y}.\some_{\llfv{C_i(n)}} ; \linvar{y}(z_n). \linvar{y}.\some_{z_n,\llfv{C_i(n)}};
               \\
              &\hspace{1cm} \linvar{y}.\overline{\some} ;\outact{\linvar{y}}{ {y}_n}. ( {y}_n.\some_{\llfv{C_i(n)}} ; \piencodf{C_i(n)}_{ {y}_n} \para\linvar{y}.\some_{\emptyset} ; \linvar{y}(z_{n+1}). ( z_{n+1}.\overline{\some}; \\
              &\hspace{1cm}z_{n+1} . \overline{\close} \para \linvar{y}.\some_{\emptyset} ; \linvar{y}. \overline{\none})  \para z_1. \overline{\none}) \para\cdots\para z_n. \overline{\none}) \para  !\banged{y}. (y_i). \piencodf{ U }_{y_i} )
              \\
       &\red^* \bigoplus_{C_i \in \perm{C}} (\nu  \widetilde{y}, \banged{y})(  \piencodf{M}_u \para  {y}_1.\some_{\llfv{C_i(1)}} ;\piencodf{C_i(1)}_{ {y}_1} \para   \cdots \para   {y}_n.\some_{\llfv{C_i(n)}} ; \\
       & \hspace{1cm}\piencodf{C_i(n)}_{ {y}_n} \para  !\banged{y}. (y_i). \piencodf{ U }_{y_i}  )\\
    \end{aligned}
    \]
    and the result follows by induction on $\piencodf{ M }_u $.
    
    \item $N = M \linexsub {N' / {y}}$ and $N = M \unexsub {u /\unvar{y}}$:
    

    
    These cases  follow easily by induction  on $\piencodf{M}_u$.
    
\end{enumerate}
\end{proof}

\subsubsection{Completeness}

Here again, because of the diamond property (Proposition \ref{app:lambda}), it suffices to consider a completeness result based on a single reduction step in $\lamrsharfailunres$:

\begin{notation}
    We use the notation $\llfv{M}.\overline{\none}$ and $\widetilde{x}.\overline{\none}$ where $\llfv{M}$ or $\widetilde{x}$ are equal to $ x_1 , \cdots , x_k$ to describe a process of the form $x_1.\overline{\none} \para \cdots \para x_k.\overline{\none} $
\end{notation}

\begin{theorem}[Well Formed Operational Completeness]
\label{l:app_completenesstwounres}
Let $\expr{N} $ and $ \expr{M} $ be well-formed, linearly closed $\lamrsharfailunres $ expressions. If $ \expr{N}\red \expr{M}$ then there exists $Q$ such that $\piencodf{\expr{N}}_u  \red^* Q \equiv \piencodf{\expr{M}}_u$.
\end{theorem}

\begin{proof}

By induction on the reduction rule applied to infer $\expr{N}\red \expr{M}$.  
We have ten cases.

    \begin{enumerate}
        \item  {\bf Case $\redlab{RS:Beta}$: }
              
        Then  $ \expr{N}= (\lambda x . (M[ {\widetilde{x}} \leftarrow  {x}])) B  \red (M[ {\widetilde{x}} \leftarrow  {x}])\esubst{ B }{ x }  = \expr{M}$ , where $B = C \bagsep U$. Notice that
        \begin{equation*}\label{eq:compl_lsbeta1failunres}
        \begin{aligned}
        \piencodf{\expr{N}}_u =&  \bigoplus_{C_i \in \perm{C}} (\nu v)(\piencodf{\lambda x . (M[ {\widetilde{x}} \leftarrow  {x}])}_v \para v.\some_{u , \llfv{C}} ; \outact{v}{x} . ([v \leftrightarrow u] \para \piencodf{C_i \bagsep U}_x ) )\\
        =&  
            \bigoplus_{C_i \in \perm{C}} (\nu v) 
                 (v.\overline{\some}; v(x). x.\overline{\some}; x(\linvar{x}). x(\banged{x}). x. \close ; \piencodf{M[\widetilde{x} \leftarrow x]}_v \\
                 &\para v.\some_{u,\llfv{C}} ; \outact{v}{x} . ( \piencodf{C_i \bagsep U}_x \para [v \leftrightarrow u] ) )\\
        \red &  \bigoplus_{C_i \in \perm{C}} (\nu v)( v(x). x.\overline{\some}; x(\linvar{x}). x(\banged{x}). x. \close ; \piencodf{M[\widetilde{x} \leftarrow x]}_v \para \outact{v}{x} . ( \piencodf{ C_i \bagsep U}_x \\
        & \hspace{1cm}\para [v \leftrightarrow u] ) )\\
    \red &  \bigoplus_{C_i \in \perm{C}} (\nu v, x)( x.\overline{\some}; x(\linvar{x}). x(\banged{x}). x. \close ; \piencodf{M[\widetilde{x} \leftarrow x]}_v \para  \piencodf{ C_i \bagsep U}_x \para [v \leftrightarrow u] )\\ 
    \red&  \bigoplus_{C_i \in \perm{C}} (\nu x)( x.\overline{\some}; x(\linvar{x}). x(\banged{x}). x. \close ; \piencodf{M[\widetilde{x} \leftarrow x]}_v \para  \piencodf{ C_i \bagsep U}_x )=\piencodf{\expr{M}}_u \\
        \end{aligned}
        \end{equation*}
and the result follows.

        
        \item {\bf Case $ \redlab{RS:Ex \dash Sub}$:}
        
        Then $ \expr{N} =M[ {x}_1, \!\cdots\! ,  {x}_k \leftarrow  {x}]\esubst{ C \bagsep U }{ x }$, with $C = \bag{M_1}
            \cdots  \bag{M_k}$, $k\geq 0$ and $M \not= \fail^{\widetilde{y}}$.

        The reduction is $$\expr{N} = M[ {x}_1, \!\cdots\! ,  {x}_k \leftarrow  {x}]\esubst{ C \bagsep U }{ x } \red \sum_{C_i \in \perm{C}}M\linexsub{C_i(1)/ {x_1}} \cdots \linexsub{C_i(k)/ {x_k}} \unexsub{U /\unvar{x} } = \expr{M}.$$
        
        We detail the encodings of $\piencodf{\expr{N}}_u$ and $\piencodf{\expr{M}}_u$. To simplify the proof, we will consider $k=1$ (the case in which $k> 1$ is follows analogously, similarly the case of $k=0$ is contained within the proof of $k=1$). 
        
        On the one hand, we have:
        \begin{equation*}\label{eq:compl_lsbeta3failunres}
        \begin{aligned}
        \piencodf{\expr{N}}_u &= \piencodf{M[ {x}_1 \leftarrow  {x}]\esubst{ C \bagsep U }{ x }}_u\\
        &= \bigoplus_{C_i \in \perm{C}} (\nu x)( x.\overline{\some}; x(\linvar{x}). x(\banged{x}). x. \close ;\piencodf{ M[ {x}_1 \leftarrow  {x}]}_u \para \piencodf{ C_i \bagsep U}_x ) \\
        &= \bigoplus_{C_i \in \perm{C}} (\nu x)( x.\overline{\some}; x(\linvar{x}). x(\banged{x}). x. \close ; \piencodf{ M[ {x}_1 \leftarrow  {x}]}_u \para  x.\some_{\llfv{C}} ;\outact{x}{\linvar{x}} .\\
        &\qquad ( \piencodf{ C_i }_{\linvar{x}} \para \outact{x}{\banged{x}} .( !\banged{x}. (x_i). \piencodf{ U }_{x_i} \para x.\overline{\close}   ) ) ) \qquad (:= P_{\mathbb{N}}) \\[4pt]
        \end{aligned}
        \end{equation*}
        
      Note that
        \begin{equation*}
        \begin{aligned}
        P_{\mathbb{N}} \red^*& \bigoplus_{C_i \in \perm{C}} (\nu \linvar{x} , \banged{x} )( \piencodf{ M[ {x}_1 \leftarrow  {x}]}_u \para  \piencodf{ C_i }_{\linvar{x}} \para  !\banged{x}. (x_i). \piencodf{ U }_{x_i}    ) \\
        =&  \bigoplus_{C_i \in \perm{C}} (\nu \linvar{x} , \banged{x} )( \linvar{x}.\overline{\some}. \outact{\linvar{x}}{y_1}. (y_1 . \some_{\emptyset} ;y_{1}.\close;0 \para \linvar{x}.\overline{\some};\linvar{x}.\some_{u, (\llfv{M} \setminus  {x}_1 )};\\
        &\linvar{x}( {x}_1) . \linvar{x}.\overline{\some}; \outact{\linvar{x}}{y_{2}}.( y_{2} .   \some_{u,\llfv{M}} ;y_{2}.\close; \piencodf{M}_u \para \linvar{x}.\overline{\none} ) ) \para\linvar{x}.\some_{\llfv{B_i(1)}} ;\\
         &  \linvar{x}(y_1). \linvar{x}.\some_{y_1,\llfv{C_i(1)}};\linvar{x}.\overline{\some} ;   \outact{\linvar{x}}{ {x}_1}.( {x}_1.\some_{\llfv{C_i(1)}} ; \piencodf{C_i(1)}_{ {x}_1}  \para y_1. \overline{\none} \para \linvar{x}.\\
                &\some_{\emptyset} ;\linvar{x}(y_2). ( y_2.\overline{\some};y_2 . \overline{\close} \para \linvar{x}.\some_{\emptyset} ; \linvar{x}.\overline{\none}) ) \para !\banged{x}. (x_i). \piencodf{ U }_{x_i})
        \\
        & \red^* \bigoplus_{C_i \in \perm{C}} (\nu  {x}_1 , \banged{x} )(\piencodf{M}_u \para  {x}_1.\some_{\llfv{C_i(1)}} ; \piencodf{C_i(1)}_{ {x}_1} \para !\banged{x}. (x_i). \piencodf{ U }_{x_i} )= \piencodf{\expr{M}}_u 
        \\
        \end{aligned}
        \end{equation*}
        
and the result follows.
        
        \item {\bf Case $\redlab{RS{:}Fetch^{\ell}}$:}
        
        Then we have
        $\expr{N} = M \linexsub{N /  {x}}$ with $\headf{M} =  {x}$ and $\expr{N} \red   M \headlin{ N/ {x} } = \expr{M}$.
   Note that
            \begin{equation*}\label{eq:compl_lsbeta5failunres}
            \begin{aligned}
            \piencodf{\expr{N}}_u &= \piencodf{M \linexsub{N /  {x}}}_u\\
            &= (\nu  {x}) ( \piencodf{ M }_u \para   {x}.\some_{\llfv{N}};\piencodf{ N }_{ {x}}  ) \\
            &\red^* (\nu  {x}) ( \bigoplus_{i \in I}(\nu \widetilde{y})(\piencodf{  {x} }_{j} \para P_i) \para    {x}.\some_{\llfv{N}};\piencodf{ N }_{ {x}}  ) \qquad (*)   
            \\
            &= (\nu  {x}) ( \bigoplus_{i \in I}(\nu \widetilde{y})(\piencodf{  {x} }_{j} \para P_i)  \para    {x}.\some;\piencodf{ N }_{ {x}}  )  \\
            & \red (\nu  {x}) ( \bigoplus_{i \in I}(\nu \widetilde{y})([ {x} \leftrightarrow j ] \para P_i) \para   \piencodf{ N }_{ {x}}  )\\ &\red \bigoplus_{i \in I}(\nu \widetilde{y})(P_i \para   \piencodf{ N }_j  )    = \piencodf{\expr{M}}_u
            \end{aligned}
            \end{equation*}
       where the reductions denoted by $(*)$ are inferred via Proposition~\ref{prop:correctformfailunres}, and the result follows.
        

        \item {\bf Case $ \redlab{RS{:} Fetch^!}$: }
        
        Then,
        $\expr{N} = M \unexsub{U / \unvar{x}}$ with $\headf{M} = \banged{x}[k]$, $U_i = \banged{\bag{N}}$ and $\expr{N} \red  M \headlin{ N /\banged{x} }\unexsub{U / \unvar{x}} = \expr{M}$.
        Note that 
            \begin{equation*}\label{eq:compl_lsbeta5failunresunres}
            \begin{aligned}
            \piencodf{\expr{N}}_u &= \piencodf{M \unexsub{U / \unvar{x}}}_u
            = (\nu \banged{x}) ( \piencodf{ M }_u \para   !\banged{x}. (x_k).\piencodf{ U }_{x_k}  ) \\
            & \red^*(\nu \banged{x}) ( \bigoplus_{i \in I}(\nu \widetilde{y})(\piencodf{ \banged{x}[k] }_{j} \para P_i) \para   !\banged{x}. (x_k).\piencodf{ U }_{x_k}  ) \qquad (*)
            \\
            &  = (\nu \banged{x}) ( \bigoplus_{i \in I}(\nu \widetilde{y})(\outsev{\banged{x}}{{x_k}}. {x}_k.l_{i}; [{x_k} \leftrightarrow j] \para P_i) \para   !\banged{x}. (x_k).\piencodf{ U }_{x_k}  ) \qquad (*)
            \\
            & \red (\nu \banged{x}) ( \bigoplus_{i \in I}(\nu \widetilde{y})(  (\nu x_k)( x_k.l_{i}; [x_k \leftrightarrow j] \para \piencodf{ U }_{x_k}) \para P_i) \para   !\banged{x}. (x_k).\piencodf{ U }_{x_k}  )
            \\
            & = (\nu \banged{x}) ( \bigoplus_{i \in I}(\nu \widetilde{y})(  (\nu x_k)( x_k.l_{i}; [x_k \leftrightarrow j] \para x_k. case( i.\piencodf{U_i}_{x} )) \para P_i) \para   !\banged{x}. (x_k).\piencodf{ U }_{x_k}  )
            \\
            & \red (\nu \banged{x}) ( \bigoplus_{i \in I}(\nu \widetilde{y})(  \piencodf{\banged{\bag{N}}}_{j} ) \para P_i) \para   !\banged{x}. (x_k).\piencodf{ U }_{x_k}  )\\
            &= (\nu \banged{x}) ( \bigoplus_{i \in I}(\nu \widetilde{y})(  \piencodf{N}_{j} ) \para P_i) \para   !\banged{x}. (x_k).\piencodf{ U }_{x_k}  ) = \piencodf{\expr{M}}_u
            \end{aligned}
            \end{equation*}
       where the reductions denoted by $(*)$ are inferred via Proposition~\ref{prop:correctformfailunres}.
        

        \item {\bf Cases $\redlab{RS:TCont}$ and $\redlab{RS:ECont}$:}
         
         These cases follow by IH.
         

        \item {\bf Case $\redlab{RS{:}Fail^{\ell}}$:}
        
        Then, 
        $\expr{N} = M[ {x}_1, \!\cdots\! ,  {x}_k \leftarrow  {x}]\ \esubst{C \bagsep U}{ x } $ with $k \neq \size{C}$ and
        
        $\expr{N} \red  \sum_{C_i \in \perm{C}}  \fail^{\widetilde{y}} = \expr{M}$, where $\widetilde{y} = (\llfv{M} \setminus \{   {x}_1, \cdots ,  {x}_k \} ) \cup \llfv{C}$. 
        
        Let $ \size{C} = l$ and we assume that $k > l$ (proceed similarly for $k > l$). Hence $k = l + m$ for some $m \geq 1$, and 
        
        \begin{equation*}\label{eq:compl_fail1-failunres}
        \begin{aligned}
            \piencodf{N}_u =& \piencodf{M[ {x}_1, \!\cdots\! ,  {x}_k \leftarrow  {x}]\ \esubst{C \bagsep U}{ x } }_u\\
             =&  \bigoplus_{C_i \in \perm{C}} (\nu x)( x.\overline{\some}; x(\linvar{x}). x(\banged{x}). x. \close ;\piencodf{ M[\widetilde{x} \leftarrow  {x}]}_u \para \\
            &  x.\some_{\llfv{C}} ; \outact{x}{\linvar{x}} .( \piencodf{ C }_{\linvar{x}} \para \outact{x}{\banged{x}} .( !\banged{x}. (x_i). \piencodf{ U }_{x_i} \para x.\overline{\close} ) ) )
            \\
             \red^* &  \bigoplus_{C_i \in \perm{C}} (\nu \linvar{x} , \banged{x})(  \piencodf{ M[\widetilde{x} \leftarrow  {x}]}_u \para \piencodf{ C }_{\linvar{x}} \para !\banged{x}. (x_i). \piencodf{ U }_{x_i}  )\\
              =& \bigoplus_{C_i \in \perm{C}} (\nu \linvar{x} , \banged{x})( \linvar{x}.\overline{\some}. \outact{\linvar{x}}{y_1}. (y_1 . \some_{\emptyset} ;y_{1}.\close;0 \para \linvar{x}.\overline{\some};\\
              & \linvar{x}.\some_{u,(\llfv{M} \setminus  \widetilde{x} )}; \linvar{x}( {x}_1) . \cdots  \linvar{x}.\overline{\some}. \outact{\linvar{x}}{y_k} . (y_k . \some_{\emptyset} ; y_{k}.\close;0 \para \linvar{x}.\overline{\some}; \\
               &\linvar{x}.\some_{u,(\llfv{M} \setminus   {x}_k )};\linvar{x}( {x}_k) . \linvar{x}.\overline{\some}; \outact{\linvar{x}}{y_{k+1}}. ( y_{k+1} . \some_{u,\llfv{M} } ;y_{k+1}.\close; \piencodf{M}_u  \\
               & \para \linvar{x}.\overline{\none} )) \cdots ) \para\linvar{x}.\some_{\llfv{C}} ; \linvar{x}(y_1). \linvar{x}.\some_{y_1,\llfv{C}};\linvar{x}.\overline{\some} ; \outact{\linvar{x}}{ {x}_1}. ( {x}_1.\some_{\llfv{C_i(1)}} ; \\
               & \piencodf{C_i(1)}_{ {x}_1} \para y_1. \overline{\none} \para \cdots \linvar{x}.\some_{\llfv{C_i(l)}} ; \linvar{x}(y_l). \linvar{x}.\some_{y_l ,\llfv{C_i(l)}};\linvar{x}.\overline{\some} ; \outact{\linvar{x}}{ {x}_l}.  \\
               & ( {x}_l.\some_{\llfv{C_i(l)}} ;\piencodf{C_i(l)}_{ {x}_l} \para y_l. \overline{\none} \para \linvar{x}.\some_{\emptyset} ; \linvar{x}(y_{l+1}). ( y_{l+1}.\overline{\some};y_{l+1} . \overline{\close} \\
               &\para \linvar{x}.\some_{\emptyset} ;\linvar{x}. \overline{\none}) ) )\para !\banged{x}. (x_i). \piencodf{ U }_{x_i} ) \qquad 
            (:= P_\mathbb{N})\\
 \red^* & \bigoplus_{C_i \in \perm{C}} (\nu \linvar{x} , \banged{x}, y_1,  {x}_1, \cdots  y_l,  {x}_l)( y_1 . \some_{\emptyset} ;y_{1}.\close;0 \para \cdots \para y_l . \some_{\emptyset} ;y_{l}.\close;0 \\
    &  {x}.\overline{\some}. \outact{\linvar{x}}{y_{l+1}} . (y_{l+1} . \some_{\emptyset} ; y_{l+1}.\close;0 \para \linvar{x}.\overline{\some};\linvar{x}.\some_{u,(\llfv{M} \setminus  {x}_{l+1} , \cdots ,  {x}_k )};\\
                 & \linvar{x}( {x}_{l+1}) . \cdots \linvar{x}.\overline{\some}. \outact{\linvar{x}}{y_k} . (y_k . \some_{\emptyset} ; y_{k}.\close;0 \para \linvar{x}.\overline{\some};\linvar{x}.\some_{u,(\llfv{M} \setminus   {x}_k )}; \\
                & \linvar{x}( {x}_k) .\linvar{x}.\overline{\some}; \outact{\linvar{x}}{y_{k+1}}. ( y_{k+1} . \some_{u,\llfv{M} } ;y_{k+1}.\close; \piencodf{M}_u \para \linvar{x}.\overline{\none} )) \cdots ) \para \\
             &  {x}_1.\some_{\llfv{C_i(1)}} ; \piencodf{C_i(1)}_{ {x}_1} \para \cdots \para   {x}_l.\some_{\llfv{C_i(l)}} ; \piencodf{C_i(l)}_{ {x}_l} \para  y_1. \overline{\none} \para \cdots \para y_l. \overline{\none}\\
              & \linvar{x}.\some_{\emptyset} ; \linvar{x}(y_{l+1}). ( y_{l+1}.\overline{\some};y_{l+1} . \overline{\close} \para \linvar{x}.\some_{\emptyset} ; \linvar{x}. \overline{\none}) \para !\banged{x}. (x_i). \piencodf{ U }_{x_i}) 
                    \\
 \red & \bigoplus_{C_i \in \perm{C}} (\nu  \banged{x},   {x}_1, \cdots  {x}_l)(  u . \overline{\none} \para  {x}_1 . \overline{\none} \para  \cdots \para  {x}_{l} . \overline{\none} \para (\llfv{M} \setminus  {x}_{1} , \cdots ,  {x}_k ).  \\
  & \overline{\none}  \para  {x}_1.\some_{\llfv{C_i(1)}} ; \piencodf{C_i(1)}_{ {x}_1} \para \cdots \para   {x}_l.\some_{\llfv{C_i(l)}} ; \piencodf{C_i(l)}_{ {x}_l} \para !\banged{x}. (x_i). \piencodf{ U }_{x_i} ) 
            \\
 \red^* & \bigoplus_{C_i \in \perm{C}} (\nu  \banged{x}) (  u . \overline{\none} \para (\llfv{M} \setminus  {x}_{1} , \cdots ,  {x}_k ) . \overline{\none}  \para !\banged{x}. (x_i). \piencodf{ U }_{x_i} ) \\
 &\equiv \bigoplus_{C_i \in \perm{C}}   u . \overline{\none} \para (\llfv{M} \setminus  {x}_{1} , \cdots ,  {x}_k ) . \overline{\none} =  \piencodf{\expr{M}}_u
        \end{aligned}
 \end{equation*}
        
    and  the result follows.

       \item {\bf Case $\redlab{RS{:}Fail^!}$:}
       
        Then,
        $\expr{N} = M \unexsub{U /\unvar{x}}  $ with $\headf{M} =  {x}[i]$, $U_i = \banged{\oneb} $ and
        $\expr{N} \red   M \headlin{ \fail^{\emptyset} /\banged{x} } \unexsub{U /\unvar{x} } $, where $\widetilde{y} = \llfv{M}  $. 
        Notice that 
        \begin{equation*}\label{eq:compl_unfail-failunres}
        \begin{aligned}
            \piencodf{N}_u &= \piencodf{ M \unexsub{U /\unvar{x}} }_u  =  (\nu \banged{x}) ( \piencodf{ M }_u \para   !\banged{x}. (x_i).\piencodf{ U }_{x_i}  )  \\
            & \red^*(\nu \banged{x}) ( \bigoplus_{i \in I}(\nu \widetilde{y})(\piencodf{  {x}[i] }_{j} \para P_i) \para   !\banged{x}. (x_k).\piencodf{ U }_{x_k}  ) \qquad (*)
            \\
            &  = (\nu \banged{x}) ( \bigoplus_{i \in I}(\nu \widetilde{y})(\outsev{\banged{x}}{{x_k}}. {x}_k.l_{i}; [{x_k} \leftrightarrow j] \para P_i) \para   !\banged{x}. (x_k).\piencodf{ U }_{x_k}  ) \qquad (*)
            \\
            & \red (\nu \banged{x}) ( \bigoplus_{i \in I}(\nu \widetilde{y})(  (\nu x_k)( x_k.l_{i}; [x_k \leftrightarrow j] \para \piencodf{ U }_{x_k}) \para P_i) \para   !\banged{x}. (x_k).\piencodf{ U }_{x_k}  )
            \\
            & = (\nu \banged{x}) ( \bigoplus_{i \in I}(\nu \widetilde{y})(  (\nu x_k)( x_k.l_{i}; [x_k \leftrightarrow j] \para \choice{x_k}{U_i}{U}{i}{\piencodf{U_i}_{x}}) \para P_i) \para   !\banged{x}. (x_k).\piencodf{ U }_{x_k}  )
            \\
            & \red (\nu \banged{x}) ( \bigoplus_{i \in I}(\nu \widetilde{y})(  \piencodf{\banged{\oneb}}_{j}  \para P_i) \para   !\banged{x}. (x_k).\piencodf{ U }_{x_k}  )\\
          & = (\nu \banged{x}) ( \bigoplus_{i \in I}(\nu \widetilde{y}) ( j.\none \para P_i) \para   !\banged{x}. (x_k).\piencodf{ U }_{x_k}  ) = \piencodf{\expr{M}}_u \\
        \end{aligned}
        \end{equation*}
        and the result follows.
        
        

        \item {\bf Case $\redlab{RS:Cons_1}$:}
        
        Then,
        $\expr{N} = \fail^{\widetilde{x}}\ C \bagsep U$ and $\expr{N} \red \sum_{\perm{C}} \fail^{\widetilde{x} \uplus \widetilde{y}}  = \expr{M}$ where $ \widetilde{y} = \llfv{C}$. 
        Notice that 
        \begin{equation*}\label{eq:compl_cons1-failunres}
        \begin{aligned}
            \piencodf{N}_u &= \piencodf{ \fail^{\widetilde{x}}\ C \bagsep U }_u\\
            &= \bigoplus_{C_i \in \perm{C}} (\nu v)(\piencodf{\fail^{\widetilde{x}}}_v \para v.\some_{u , \llfv{C}} ; \outact{v}{x} . ([v \leftrightarrow u] \para \piencodf{C_i \bagsep U}_x ) ) \\
            &= \bigoplus_{C_i \in \perm{C}} (\nu v)( v . \overline{\none} \para \widetilde{x}. \overline{\none} \para v.\some_{u , \llfv{C}} ; \outact{v}{x} . ([v \leftrightarrow u] \para \piencodf{C_i \bagsep U}_x ) ) \\
            & \red \bigoplus_{C_i \in \perm{C}} u . \overline{\none} \para \widetilde{x}. \overline{\none} \para \widetilde{y}. \overline{\none}= \bigoplus_{C_i \in \perm{C}} u . \overline{\none} \para \widetilde{x}. \overline{\none} \para \widetilde{y}. \overline{\none} = \piencodf{\expr{M}}_u 
        \end{aligned}
        \end{equation*}
        and the result follows.
        

        \item {\bf Cases $\redlab{RS:Cons_2}$, $\redlab{RS:Cons_3}$ and \redlab{RS:Cons_4}:}
        
        These cases follow by IH similarly to the previous.

        
        
        
        

        
    \end{enumerate}
\end{proof}

\subsubsection{Soundness}
\begin{theorem}[Well Formed Weak Operational Soundness]
\label{l:app_soundnesstwounres}
Let $\expr{N}$ be a 
well-formed, linearly closed  $ \lamrsharfailunres$ expression. 
If $ \piencodf{\expr{N}}_u \red^* Q$
then there exist $Q'$  and $\expr{N}' $ such that 
$Q \red^* Q'$, $\expr{N}  \red_{\pequiv}^* \expr{N}'$ 
and 
$\piencodf{\expr{N}'}_u \equiv Q'$.
\end{theorem}

\begin{proof}
By induction on the structure of $\expr{N} $ and then induction on the number of reductions of $\piencodf{\expr{N}} \red^* Q$.

\begin{enumerate}
    \item {\bf Base case:} $\expr{N} =  {x}$, $\expr{N} =  {x}[j]$, $\expr{N} = \fail^{\emptyset}$ and $\expr{N} = \lambda x . (M[ {\widetilde{x}} \leftarrow  {x}])$.
.
    
    No reductions can take place, and the result follows trivially.
    $Q =  \piencodf{\expr{N}}_u \red^0 \piencodf{\expr{N}}_u = Q'$ and $ {x} \red^0  {x} = \expr{N}'$.
    
    
    
    
    

    
    \item $\expr{N} =  M (C \bagsep U) $.

        Then, 
        $ \piencodf{M (C \bagsep U)}_u = \bigoplus_{C_i \in \perm{C}} (\nu v)(\piencodf{M}_v \para v.\some_{u , \llfv{C}} ; \outact{v}{x} . ([v \leftrightarrow u] \para \piencodf{C_i \bagsep U}_x ) )$, and we are able to perform the  reductions from $\piencodf{M (C \bagsep U)}_u$. 

        We now proceed by induction on $k$, with  $\piencodf{\expr{N}}_u \red^k Q$. There are two main cases:
        \begin{enumerate}
            \item When $k = 0$ the thesis follows easily:
            
            We have 
    $Q =  \piencodf{M (C \bagsep U)}_u \red^0 \piencodf{M (C \bagsep U)}_u = Q'$ and $M (C \bagsep U) \red^0 M (C \bagsep U) = \expr{N}'$.
    
            \item The interesting case is when $k \geq 1$.
            
            Then, for some process $R$ and $n, m$ such that $k = n+m$, we have the following:
            \[
            \begin{aligned}
               \piencodf{\expr{N}}_u & =  \bigoplus_{C_i \in \perm{C}} (\nu v)(\piencodf{M}_v \para v.\some_{u , \llfv{C}} ; \outact{v}{x} . ([v \leftrightarrow u] \para \piencodf{C_i \bagsep U}_x ) )\\
               & \red^m  \bigoplus_{C_i \in \perm{C}} (\nu v)(R \para v.\some_{u , \llfv{C}} ; \outact{v}{x} . ([v \leftrightarrow u] \para \piencodf{C_i \bagsep U}_x ) ) \red^n  Q\\
            \end{aligned}
            \]
            Thus, the first $m \geq 0$ reduction steps are  internal to $\piencodf{ M}_v$; type preservation in \spi ensures that, if they occur,  these reductions  do not discard the possibility of synchronizing with $v.\some$. Then, the first of the $n \geq 0$ reduction steps towards $Q$ is a synchronization between $R$ and $v.\some_{u, \llfv{C}}$.
            
            We consider two sub-cases, depending on the values of  $m$ and $n$:
            \begin{enumerate}
                \item $m = 0$ and $n \geq 1$:
                
Then $R = \piencodf{\expr{M}}_v$ as $\piencodf{\expr{M}}_v \red^0 \piencodf{\expr{M}}_v$. 
 Notice that there are two possibilities of having an unguarded:
                    
 

\begin{enumerate}
\item $M =  (\lambda x . (M'[ {\widetilde{x}} \leftarrow  {x}])) \linexsub{N_1 / y_1} \cdots \linexsub{N_p / y_p} \unexsub{U_1 / \unvar{z_1}} \cdots \unexsub{U_q / \unvar{z_q}}   \quad (p, q \geq 0)$

   \[
   \begin{aligned}
   \piencodf{M}_v &= \piencodf{ (\lambda x . (M'[ {\widetilde{x}} \leftarrow  {x}])) \linexsub{N_1 / y_1} \cdots \linexsub{N_p / y_p} \unexsub{U_1 / \unvar{z_1}} \cdots \unexsub{U_q / \unvar{z_q}} }_v \\
      &= (\nu y_1, \cdots , y_p , \banged{z}_1, \cdots ,\banged{z}_q) ( \piencodf{\lambda x . (M'[ {\widetilde{x}} \leftarrow  {x}])}_v \para y_1.\some_{\llfv{N_1}};\piencodf{ N_1 }_{y_1} \para \cdots \\
    & \hspace{.5cm}  \para y_p.\some_{\llfv{N_p}};\piencodf{ N_p }_{y_p} \para   !\banged{z}_1. (z_1).\piencodf{ U }_{z_1} \para \cdots  \para   !\banged{z}_q. (z_q).\piencodf{ U }_{z_q}          )\\
  &= (\nu \widetilde{y} ,\widetilde{z} ) ( \piencodf{\lambda x . (M'[ {\widetilde{x}} \leftarrow  {x}])}_v \para Q'' )\\
  &= (\nu \widetilde{y},\widetilde{z}) ( v.\overline{\some}; v(x). x.\overline{\some}; x(\linvar{x}). x(\banged{x}). x. \close ; \piencodf{M'[ {\widetilde{x}} \leftarrow  {x}]}_v \para Q'' )
    \end{aligned}
    \]
  \noindent where $\widetilde{y} = y_1 , \cdots , y_p$. $\widetilde{z} = \banged{z}_1, \cdots ,\banged{z}_q$ and
  
  $$
  \begin{aligned}
      Q''& = y_1.\some_{\llfv{N_1}};\piencodf{ N_1 }_{y_1} \para \cdots \para y_p.\some_{\llfv{N_p}};\piencodf{ N_p }_{y_p} \para   !\banged{z}_1. (z_1).\piencodf{ U }_{z_1} \para \cdots \\
      &\para   !\banged{z}_q. (z_q).\piencodf{ U }_{z_q}.
      \end{aligned}$$
With this shape for $M$, we then have the following:
 \[
 \begin{aligned}
 \piencodf{\expr{N}}_u & = \piencodf{(M\ B)}_u\\
 &= \bigoplus_{C_i \in \perm{C}} (\nu v)(\piencodf{M}_v \para v.\some_{u , \llfv{C}} ; \outact{v}{x} . ([v \leftrightarrow u] \para \piencodf{C_i \bagsep U}_x ) )\\
 & \red \bigoplus_{C_i \in \perm{C}} (\nu v , \widetilde{y},\widetilde{z})(   v(x). x.\overline{\some}; x(\linvar{x}). x(\banged{x}). x. \close ; \piencodf{M'[ {\widetilde{x}} \leftarrow  {x}]}_v \\
 & \hspace{.5cm} \para Q''  \para \outact{v}{x} . ([v \leftrightarrow u] \para \piencodf{C_i \bagsep U}_x ) ) & = Q_1 \\
& \red \bigoplus_{C_i \in \perm{C}} (\nu v , \widetilde{y},\widetilde{z}, x)( x.\overline{\some}; x(\linvar{x}). x(\banged{x}). x. \close ; \piencodf{M'[ {\widetilde{x}} \leftarrow  {x}]}_v  \\
 & \hspace{.5cm}\para Q''  \para  [v \leftrightarrow u] \para \piencodf{C_i \bagsep U}_x ) & = Q_2 \\
 & \red \bigoplus_{C_i \in \perm{C}} (\nu  \widetilde{y},\widetilde{z}, x)( x.\overline{\some}; x(\linvar{x}). x(\banged{x}). x. \close ; \piencodf{M'[ {\widetilde{x}} \leftarrow  {x}]}_u \para Q''\\
 & \hspace{.5cm}\para  \piencodf{C_i \bagsep U}_x ) & = Q_3 \\
\end{aligned}
     \]
We also have that 
                    \[
                    \begin{aligned}
                        \expr{N} &=(\lambda x . (M'[ {\widetilde{x}} \leftarrow  {x}])) \linexsub{N_1 / y_1} \cdots \linexsub{N_p / y_p} \unexsub{U_1 / \unvar{z_1}} \cdots \unexsub{U_q / \unvar{z_q}} (C \bagsep U) \\
                        &\pequiv (\lambda x . (M'[ {\widetilde{x}} \leftarrow  {x}]) (C \bagsep U)) \linexsub{N_1 / y_1} \cdots \linexsub{N_p / y_p} \unexsub{U_1 / \unvar{z_1}} \cdots \unexsub{U_q / \unvar{z_q}} \\
                      & \red   M'[ {\widetilde{x}} \leftarrow  {x}] \esubst{(C \bagsep U)}{x} \linexsub{N_1 / y_1} \cdots \linexsub{N_p / y_p} \unexsub{U_1 / \unvar{z_1}} \cdots \unexsub{U_q / \unvar{z_q}} = \expr{M}
                    \end{aligned}
                    \]
  Furthermore, we have:
  \[
                     \begin{aligned}
                          &\piencodf{\expr{M}}_u = \piencodf{M'[ {\widetilde{x}} \leftarrow  {x}] \esubst{(C \bagsep U)}{x} \linexsub{N_1 / y_1} \cdots \linexsub{N_p / y_p} \unexsub{U_1 / \unvar{z_1}} \cdots \unexsub{U_q / \unvar{z_q}}}_u \\
                          & = \bigoplus_{C_i \in \perm{C}} (\nu  \widetilde{y},\widetilde{z}, x)( x.\overline{\some}; x(\linvar{x}). x(\banged{x}). x. \close ; \piencodf{M'[ {\widetilde{x}} \leftarrow  {x}]}_u   \para  \piencodf{C_i \bagsep U}_x \para Q'' ) 
                    \end{aligned}
                    \]
                     
        We consider different possibilities for $n \geq 1$; in all  the cases, the result follows.                 
                         \smallskip

 \noindent  {\bf When $n = 1$:}
      We have $Q = Q_1$, $ \piencodf{\expr{N}}_u \red^1 Q_1$.
          We also have that 
          \begin{itemize}
          \item  $Q_1 \red^2 Q_3 = Q'$ , 
        \item $\expr{N} \red^1 M'[\widetilde{x} \leftarrow x]) \esubst{B}{x} = \expr{N}'$
        \item and $\piencodf{M'[\widetilde{x} \leftarrow x]) \esubst{B}{x}}_u = Q_3$.
        \end{itemize}
        
                                \smallskip

 \noindent  {\bf When $n = 2$:} the analysis is similar.


\noindent {\bf When $n \geq 3$:}
 We have $ \piencodf{\expr{N}}_u \red^3 Q_3 \red^l Q$, for $l \geq 0$. We also know that $\expr{N} \red \expr{M}$, $Q_3 = \piencodf{\expr{M}}_u$. By the IH, there exist $ Q' , \expr{N}'$ such that $Q \red^i Q'$, $\expr{M} \red_{\pequiv}^j \expr{N}'$ and $\piencodf{\expr{N}'}_u = Q'$ . Finally, $\piencodf{\expr{N}}_u \red^3 Q_3 \red^l Q \red^i Q'$ and $\expr{N} \rightarrow \expr{M}  \red_{\pequiv}^j \expr{N}'$.
                        
\item $M = \fail^{\widetilde{z}}$. 

Then,                     \(
                        \begin{aligned}
                            \piencodf{M}_v &= \piencodf{\fail^{\widetilde{z}}}_v = v.\overline{\none} \para \widetilde{z}.\overline{\none}.
                        \end{aligned}
                    \)
                    With this shape for $M$, we have:
                    
                    \[
                    \begin{aligned}
                        \piencodf{\expr{N}}_u & = \piencodf{(M\ (C \bagsep U))}_u\\ & = \bigoplus_{C_i \in \perm{C}} (\nu v)(\piencodf{M}_v \para v.\some_{u , \llfv{C}} ; \outact{v}{x} . ([v \leftrightarrow u] \para \piencodf{C_i \bagsep U}_x ) )\\
                        & = \bigoplus_{C_i \in \perm{C}} (\nu v)(v.\overline{\none} \para \widetilde{z}.\overline{\none} \para v.\some_{u , \llfv{C}} ; \outact{v}{x} . ([v \leftrightarrow u] \para \piencodf{C_i \bagsep U}_x ) )\\
                        & \red \bigoplus_{B_i \in \perm{B}}   u.\overline{\none} \para \widetilde{z}.\overline{\none}  \para \llfv{C_i}.\overline{\none} \\
                        \end{aligned}
                    \]

                    \end{enumerate}
                    
                    We also have that 
                    \(  \expr{N} = \fail^{\widetilde{x}}\ C \bagsep U \red  \sum_{\perm{C}} \fail^{\widetilde{x} \uplus \llfv{C}}  = \expr{M}.  \)
                    Furthermore,
                    \[
                     \begin{aligned}
                          \piencodf{\expr{M}}_u &= \piencodf{\sum_{\perm{C}} \fail^{\widetilde{z} \uplus \llfv{C}  } }_u 
                          = \bigoplus_{\perm{C}}\piencodf{ \fail^{\widetilde{z} \uplus \llfv{C} }}_u\\
                          &= \bigoplus_{\perm{C}}    u.\overline{\none} \para \widetilde{z}.\overline{\none}  \para  \llfv{C}.\overline{\none}.
                    \end{aligned}
                    \]

  \item When $m \geq 1$ and $ n \geq 0$, we distinguish two cases:
                   
 \begin{enumerate}
\item When $n = 0$:
                            
Then, $ \bigoplus_{C_i \in \perm{C}} (\nu v)(R \para v.\some_{u , \llfv{C}} ; \outact{v}{x} . ([v \leftrightarrow u] \para \piencodf{C_i \bagsep U}_x ) ) =  Q $ and $\piencodf{M}_u \red^m R$ where $m \geq 1$. Then by the IH there exist $R'$  and $\expr{M}' $ such that $R \red^i R'$, $M \red_{\pequiv}^j \expr{M}'$, and $\piencodf{\expr{M}'}_u = R'$.  Hence we have that 
        
    \[ 
    \begin{aligned}
  \piencodf{\expr{N}}_u & =  \bigoplus_{C_i \in \perm{C}} (\nu v)(\piencodf{M}_v \para v.\some_{u , \llfv{C}} ; \outact{v}{x} . ([v \leftrightarrow u] \para \piencodf{C_i \bagsep U}_x ) )\\
                            & \red^m  \bigoplus_{C_i \in \perm{C}} (\nu v)(R \para v.\some_{u , \llfv{C}} ; \outact{v}{x} . ([v \leftrightarrow u] \para \piencodf{C_i \bagsep U}_x ) )  = Q
                            \end{aligned}
                             \]
                            We also know that
                            \[ 
                            \begin{aligned}
                              Q & \red^i  \bigoplus_{C_i \in \perm{C}} (\nu v)(R' \para v.\some_{u , \llfv{C}} ; \outact{v}{x} . ([v \leftrightarrow u] \para \piencodf{C_i \bagsep U}_x ) ) = Q'\\
                            \end{aligned}
                            \]
                            and so the \lamrsharfailunres term can reduce as follows: $\expr{N} = (M\ ( C \bagsep U )) \red_{\pequiv}^j M'\ ( C \bagsep U ) = \expr{N}'$ and  $\piencodf{\expr{N}'}_u = Q'$.

                        \item When $n \geq 1$:
                        
                            Then  $R$ has an occurrence of an unguarded $v.\overline{\some}$ or $v.\overline{\none}$, hence it is of the form 
                            $ \piencodf{(\lambda x . (M'[ {\widetilde{x}} \leftarrow  {x}])) \linexsub{N_1 / y_1} \cdots \linexsub{N_p / y_p} \unexsub{U_1 / \unvar{z_1}} \cdots \unexsub{U_q / \unvar{z_q}} }_v $ or $ \piencodf{\fail^{\widetilde{x}}}_v. $ 
              This case follows by IH.
                    \end{enumerate}

            \end{enumerate}

        \end{enumerate}

        This concludes the analysis for the case $\expr{N} = (M \, ( C \bagsep U ))$.
        
        \item $\expr{N} = M[ {\widetilde{x}} \leftarrow  {x}]$.

    The sharing variable $ {x}$ is not free and the result follows by vacuity.
        
        \item $\expr{N} = M[ {\widetilde{x}} \leftarrow  {x}] \esubst{ C \bagsep U }{ x}$. Then we have
            
            \[
                \begin{aligned}
                    \piencodf{\expr{N}}_u &=\piencodf{ M[ {\widetilde{x}} \leftarrow  {x}] \esubst{ C \bagsep U }{ x} }_u\\
                    &= \bigoplus_{C_i \in \perm{C}} (\nu x)( x.\overline{\some}; x(\linvar{x}). x(\banged{x}). x. \close ;\piencodf{ M[ {\widetilde{x}} \leftarrow  {x}]}_u \para \piencodf{ C_i \bagsep U}_x )
                \end{aligned}
            \]

            Let us consider three cases.

            \begin{enumerate}
                \item When $\size{ {\widetilde{x}}} = \size{C}$.
                    Then let us consider the shape of the bag $ C$.
                    
  \begin{enumerate}
  \item When $C = \oneb$.
  
  We have the following
 \[
 \begin{aligned}
 \piencodf{\expr{N}}_u  &=  (\nu x)( x.\overline{\some}; x(\linvar{x}). x(\banged{x}). x. \close ;\piencodf{ M[ \leftarrow  {x}]}_u \para \piencodf{ \oneb \bagsep U}_x )\\
 &=  (\nu x)( x.\overline{\some}; x(\linvar{x}). x(\banged{x}). x. \close ;\piencodf{ M[ \leftarrow  {x}]}_u \para x.\some_{\llfv{C}} ; \outact{x}{\linvar{x}} . \\
 & \hspace{1cm}( \piencodf{ \oneb }_{\linvar{x}} \para\outact{x}{\banged{x}} .( !\banged{x}. (x_i). \piencodf{ U }_{x_i} \para x.\overline{\close} ) ) )\\
&\red (\nu x)(  x(\linvar{x}). x(\banged{x}). x. \close ;\piencodf{ M[ \leftarrow  {x}]}_u \para
 \outact{x}{\linvar{x}} .( \piencodf{ \oneb }_{\linvar{x}} \para \outact{x}{\banged{x}} . \\
 &\hspace{1cm} ( !\banged{x}. (x_i). \piencodf{ U }_{x_i}\para x.\overline{\close} ) ) ) & = Q_1  
                              \\
  &\red (\nu x,\linvar{x})(  x(\banged{x}). x. \close ;\piencodf{ M[ \leftarrow  {x}]}_u \para \piencodf{ \oneb }_{\linvar{x}} \para\outact{x}{\banged{x}} .\\
  & \hspace{1cm}( !\banged{x}. (x_i). \piencodf{ U }_{x_i} \para x.\overline{\close} ) ) & = Q_2  
                              \\
  &\red (\nu x,\linvar{x}, \banged{x})(  x. \close ;\piencodf{ M[ \leftarrow  {x}]}_u \para \piencodf{ \oneb }_{\linvar{x}} \para !\banged{x}. (x_i). \piencodf{ U }_{x_i} \para x.\overline{\close} ) & = Q_3   \\
 &\red (\nu \linvar{x}, \banged{x})(  \piencodf{ M[ \leftarrow  {x}]}_u \para \piencodf{ \oneb }_{\linvar{x}} \para !\banged{x}. (x_i). \piencodf{ U }_{x_i} ) & = Q_4\\
    & = (\nu \linvar{x}, \banged{x})( \linvar{x}. \overline{\some}. \outact{\linvar{x}}{y_i} . ( y_i . \some_{u,\llfv{M}} ;y_{i}.\close; \piencodf{M}_u \para \linvar{x}. \overline{\none}) \para \\
    & \qquad \linvar{x}.\some_{\emptyset} ; \linvar{x}(y_n). ( y_n.\overline{\some};y_n . \overline{\close} \para \linvar{x}.\some_{\emptyset} ; \linvar{x}. \overline{\none})  \para !\banged{x}. (x_i). \piencodf{ U }_{x_i} )
                            \\
      & \red (\nu \linvar{x}, \banged{x})(  \outact{\linvar{x}}{y_i} . ( y_i . \some_{u,\llfv{M}} ;y_{i}.\close; \piencodf{M}_u \para \linvar{x}. \overline{\none}) \para \\
                              & \qquad \linvar{x}(y_n). ( y_n.\overline{\some};y_n . \overline{\close} \para \linvar{x}.\some_{\emptyset} ; \linvar{x}. \overline{\none})  \para !\banged{x}. (x_i). \piencodf{ U }_{x_i} )  & = Q_5
                            \\
                            & \red (\nu \linvar{x}, \banged{x} , y_i)(   y_i . \some_{u,\llfv{M}} ;y_{i}.\close; \piencodf{M}_u \para \linvar{x}. \overline{\none} \para  y_i.\overline{\some};y_i . \overline{\close}  \\
 & \hspace{1cm}\para \linvar{x}.\some_{\emptyset} ;\linvar{x}. \overline{\none}  \para !\banged{x}. (x_i). \piencodf{ U }_{x_i} )  & = Q_6
                            \\
                            & \red (\nu \linvar{x}, \banged{x} , y_i)(  y_{i}.\close; \piencodf{M}_u \para \linvar{x}. \overline{\none} \para  y_i . \overline{\close} \para \linvar{x}.\some_{\emptyset} ; \linvar{x}. \overline{\none}\\
                            & \hspace{
                            1cm }\para !\banged{x}. (x_i). \piencodf{ U }_{x_i} )  & = Q_7
                            \\
                            & \red (\nu \linvar{x}, \banged{x} )(  \piencodf{M}_u \para \linvar{x}. \overline{\none} \para  \linvar{x}.\some_{\emptyset} ; \linvar{x}. \overline{\none}  \para !\banged{x}. (x_i). \piencodf{ U }_{x_i} )  & = Q_8
                            \\
                            & \red (\nu \banged{x})(  \piencodf{M}_u \para !\banged{x}. (x_i). \piencodf{ U }_{x_i} )  
                            =  \piencodf{M \unexsub{U / \unvar{x}}}_u
                            & = Q_9
                            \end{aligned}
                            \]
                            Notice how $Q_8$ has a choice however the $\linvar{x}$ name can be closed at any time so for simplicity we only perform communication across this name once all other names have completed their reductions.
                            
                        Now we proceed by induction on the number of reductions $\piencodf{\expr{N}}_u \red^k Q$.
                            
                            \begin{enumerate}
                                
                                \item When $k = 0$, the result follows trivially. Just take $\mathbb{N}=\mathbb{N}'$ and $\piencodf{\expr{N}}_u=Q=Q'$.
                                

                                \item When $k = 1$.
                                
                                    We have $Q = Q_1$, $ \piencodf{\expr{N}}_u \red^1 Q_1$
                                    We also have that $Q_1 \red^8 Q_9 = Q'$ , $\expr{N} \red M \unexsub{U / \unvar{x}} = M$ and $\piencodf{ M }_u = Q_9$
                                    
                                \item When $2 \leq  k \leq 8$.
                                
                                    Proceeds similarly to the previous case
                                
                                \item When $k \geq 9$.
                                
      We have $ \piencodf{\expr{N}}_u \red^9 Q_9 \red^l Q$, for $l \geq 0$. Since $Q_9 = \piencodf{ M }_u$ we apply the induction hypothesis we have that  there exist $ Q' , \expr{N}' \ s.t. \ Q \red^i Q' ,  M \red_{\pequiv}^j \expr{N}'$ and $\piencodf{\expr{N}'}_u = Q'$.                                    Then,  $ \piencodf{\expr{N}}_u \red^5 Q_5 \red^l Q \red^i Q'$ and by the contextual reduction rule it follows that $\expr{N} = (M[ \leftarrow x])\esubst{ 1 }{ x } \red_{\pequiv}^j  \expr{N}' $ and the case holds.

\end{enumerate}     
                        
\item When $C = \bag{N_1} \cdot \cdots \cdot \bag{N_l}$, for $l \geq 1$.
                    Then,
                    
 \[
   \begin{aligned}
   \piencodf{\expr{N}}_u &=\piencodf{ M[ {\widetilde{x}} \leftarrow  {x}] \esubst{ C \bagsep U }{x} }_u\\
   &= \bigoplus_{C_i \in \perm{C}} (\nu x)( x.\overline{\some}; x(\linvar{x}). x(\banged{x}). x. \close ;\piencodf{ M[ {\widetilde{x}} \leftarrow  {x}]}_u \para \piencodf{ C_i \bagsep U}_x ) \\
  &\red ^{4} \bigoplus_{C_i \in \perm{C}}(\nu \linvar{x}, \banged{x})(  \piencodf{ M[ {\widetilde{x}} \leftarrow  {x}]}_u \para \piencodf{ C_i }_{\linvar{x}} \para !\banged{x}. (x_i). \piencodf{ U }_{x_i} )\\
 &= 
  \bigoplus_{C_i \in \perm{C}} (\nu \linvar{x}, \banged{x})( \linvar{x}.\overline{\some}. \outact{\linvar{x}}{y_1}. (y_1 . \some_{\emptyset} ;y_{1}.\close;0 \para \linvar{x}.\overline{\some};\\
  &\qquad \linvar{x}.\some_{u, (\llfv{M} \setminus  {x}_1 , \cdots ,  {x}_l )}; 
     \linvar{x}( {x}_1) . \cdots \linvar{x}.\overline{\some}. \outact{\linvar{x}}{y_l} . (y_l . \some_{\emptyset} ; y_{l}.\close;0 \\
     &\qquad\para \linvar{x}.\overline{\some};\linvar{x}.\some_{u,(\llfv{M} \setminus  {x}_l )};\linvar{x}( {x}_l) . \linvar{x}.\overline{\some}; \outact{\linvar{x}}{y_{l+1}}. ( y_{l+1} . \some_{u,\llfv{M}} ; \\
    &\qquad y_{l+1}.\close;  \piencodf{M}_u \para \linvar{x}.\overline{\none} )) \cdots ) \para  \linvar{x}.\some_{\llfv{C}} ; \linvar{x}(y_1). \linvar{x}.\some_{y_1, \llfv{C} };\\
   & \qquad \linvar{x}.\overline{\some} ; \outact{\linvar{x}}{ {x}_1}. ( {x}_1.\some_{\llfv{C_i(1)}} ; \piencodf{C_{i}(1)}_{ {x}_1}  \para y_1. \overline{\none}\para \cdots  \linvar{x}.\some_{\llfv{C_{i}(l)}} ; \\
   &\qquad \linvar{x}(y_l). \linvar{x}. \some_{y_l, \llfv{C_{i}(l)}} ;\linvar{x}.\overline{\some} ; \outact{\linvar{x}}{ {x}_l}. ( {x}_l.\some_{\llfv{C_{i}(l)}} ; \piencodf{C_{i}(l)}_{ {x}_l}  \\
    & \qquad  \para y_l. \overline{\none}\para\linvar{x}.\some_{\emptyset} ; \linvar{x}(y_{l+1}). ( y_{l+1}.\overline{\some};y_{l+1} . \overline{\close} \para \linvar{x}.\some_{\emptyset} ; \linvar{x}. \overline{\none})
                                  )
                                  )
 \\
 &\qquad \para !\banged{x}. (x_i). \piencodf{ U }_{x_i} )\\ 
 & \red ^{5l}
  \bigoplus_{C_i \in \perm{C}} (\nu \linvar{x}, \banged{x} ,  {x}_1,y_1, \cdots ,  {x}_l,y_1)( y_1 . \some_{\emptyset} ;y_{1}.\close;0 \para  \cdots  y_l . \some_{\emptyset} ;\\
 &\hspace{1cm}  y_{l}.\close;0 \para  \linvar{x}.\overline{\some}; \outact{\linvar{x}}{y_{l+1}}. ( y_{l+1} . \some_{u,\llfv{M}} ;y_{l+1}.\close; \piencodf{M}_u \para \linvar{x}.\overline{\none} )
 \para \\
    & \hspace{1cm}  {x}_1.\some_{\llfv{C_{i}(1)}} ; \piencodf{C_{i}(1)}_{ {x}_1}  \para y_1. \overline{\none} \para \cdots    {x}_l.\some_{\llfv{C_{i}(l)}} ; \piencodf{C_{i}(l)}_{ {x}_l}  \para y_l. \overline{\none}\para\\
     & \hspace{1cm}\linvar{x}.\some_{\emptyset} ; \linvar{x}(y_{l+1}). ( y_{l+1}.\overline{\some};y_{l+1} . \overline{\close} \para \linvar{x}.\some_{\emptyset} ; \linvar{x}. \overline{\none})
  \para !\banged{x}. (x_i). \piencodf{ U }_{x_i} )
                               \\
     & \red ^{5}
  \bigoplus_{C_i \in \perm{C}} (\nu \banged{x} ,  {x}_1,y_1, \cdots ,  {x}_l,y_1)(  y_1 . \some_{\emptyset} ;y_{1}.\close;0 \para  \cdots  y_l . \some_{\emptyset} ; y_{l}.\close;0  \\
   & \qquad   \para \piencodf{M}_u \para {x}_1.\some_{\llfv{C_{i}(1)}} ; \piencodf{C_{i}(1)}_{ {x}_1}  \para y_1. \overline{\none} \para \cdots    {x}_l.\some_{\llfv{C_{i}(l)}} ; \piencodf{C_{i}(l)}_{ {x}_l} \\
   & \qquad \para y_l. \overline{\none} \para !\banged{x}. (x_i). \piencodf{ U }_{x_i})
                           \\
  & \red ^{l}  \bigoplus_{C_i \in \perm{C}} (\nu \banged{x} ,  {x}_1 \cdots ,  {x}_l)( \piencodf{M}_u \para  {x}_1.\some_{\llfv{C_{i}(1)}} ; \piencodf{C_{i}(1)}_{ {x}_1}  \para  \cdots  \\
  & \hspace{1cm}\para   {x}_l.\some_{\llfv{C_{i}(l)}} ;\piencodf{C_{i}(l)}_{ {x}_l} \para  !\banged{x}. (x_i). \piencodf{ U }_{x_i} )
                                 \\
  & = \piencodf{\sum_{C_i \in \perm{C}}M\linexsub{C_i(1)/ {x_1}} \cdots \linexsub{C_i(l)/ {x_l}} \unexsub{U /\unvar{x} }}_{u}= Q_{6l + 9}\\
  \end{aligned}
           \]

                            The proof follows by induction on the number of reductions $\piencodf{\expr{N}}_u \red^k Q$.
                            
\begin{enumerate}
\item When $k = 0$, the result follows trivially. Just take $\mathbb{N}=\mathbb{N}'$ and $\piencodf{\expr{N}}_u=Q=Q'$. 
                                
 \item When $1 \leq k \leq 6l + 9$.

 Let $Q_k$ such that $ \piencodf{\expr{N}}_u \red^k Q_k$.
                                    We also have that $Q_k \red^{6l + 9 - k} Q_{6l + 9} = Q'$ ,
                                    
                                    $\expr{N} \red^1 \sum_{C_i \in \perm{C}}M\linexsub{C_i(1)/ {x_1}} \cdots \linexsub{C_i(l)/ {x_l}} \unexsub{U /\unvar{x} } = \expr{N}'$ and 
                                    
                                    $\piencodf{\sum_{C_i \in \perm{C}}M\linexsub{C_i(1)/ {x_1}} \cdots \linexsub{C_i(l)/ {x_l}} \unexsub{U /\unvar{x} }}_u = Q_{6l + 9}$.
                                
\item When $k > 6l + 9$.

Then,  $ \piencodf{\expr{N}}_u \red^{6l + 9} Q_{6l + 9} \red^n Q$ for $n \geq 1$. Also, 

\(
\begin{aligned} 
&\expr{N} \red^1 \sum_{C_i \in \perm{C}}M\linexsub{C_i(1)/ {x_1}} \cdots \linexsub{C_i(l)/ {x_l}} \unexsub{U /\unvar{x} } \text { and } \\
&    Q_{6l + 9} = \piencodf{\sum_{C_i \in \perm{C}}M\linexsub{C_i(1)/ {x_1}} \cdots \linexsub{C_i(l)/ {x_l}} \unexsub{U /\unvar{x} }}_u.
    \end{aligned}
\)

By the induction hypothesis, there exist $ Q'$ and $\expr{N}'$ such that  
$\ Q \red^i Q'$, 

$\sum_{C_i \in \perm{C}}M\linexsub{C_i(1)/ {x_1}} \cdots \linexsub{C_i(l)/ {x_l}} \unexsub{U /\unvar{x} } \red_{\pequiv}^j \expr{N}'$ and $\piencodf{\expr{N}'}_u = Q'$. 

Finally, $\piencodf{\expr{N}}_u \red^{6l + 9} Q_{6l + 9} \red^n Q \red^i Q'$ and $$ \expr{N} \rightarrow \sum_{C_i \in \perm{C}}M\linexsub{C_i(1)/ {x_1}} \cdots \linexsub{C_i(l)/ {x_l}} \unexsub{U /\unvar{x} }  \red_{\pequiv}^j \expr{N}'. $$

                            \end{enumerate}

                    \end{enumerate}

                \item When $\size{\widetilde{x}} > \size{C}$.
                
                    Then we have 
                    $\expr{N} = M[ {x}_1, \cdots ,  {x}_k \leftarrow  {x}]\ \esubst{ C \bagsep U }{x}$ with $C = \bag{N_1}  \cdots  \bag{N_l} \quad k > l$. $\expr{N} \red  \sum_{C_i \in \perm{C}}  \fail^{\widetilde{z}} = \expr{M}$ and $ \widetilde{z} =  (\llfv{M} \setminus \{   {x}_1, \cdots ,  {x}_k \} ) \cup \llfv{C} $. On the one hand, we have:
                    Hence $k = l + m$ for some $m \geq 1$
                
                    \[
                    \begin{aligned}
                        \piencodf{N}_u &= \piencodf{M[ {x}_1, \cdots ,  {x}_k \leftarrow  {x}]\ \esubst{ C \bagsep U }{x}}_u \\
                        & = \bigoplus_{C_i \in \perm{C}} (\nu x)( x.\overline{\some}; x(\linvar{x}). x(\banged{x}). x. \close ;\piencodf{ M[ {x}_1, \cdots ,  {x}_k \leftarrow  {x}]}_u \para \piencodf{ C_i \bagsep U}_x ) \\
                              &\red ^{4} \bigoplus_{C_i \in \perm{C}}(\nu \linvar{x}, \banged{x})(  \piencodf{M[ {x}_1, \cdots ,  {x}_k \leftarrow  {x}]}_u \para \piencodf{ C_i }_{\linvar{x}} \para !\banged{x}. (x_i). \piencodf{ U }_{x_i} )\\
                              &=  \bigoplus_{C_i \in \perm{C}} (\nu \linvar{x}, \banged{x})(  \linvar{x}.\overline{\some}. \outact{\linvar{x}}{y_1}. (y_1 . \some_{\emptyset} ;y_{1}.\close;0 \para \linvar{x}.\overline{\some};\linvar{x}.\some_{u, (\llfv{M} \setminus  {x}_1 , \cdots ,  {x}_k )}; \\
  &\linvar{x}( {x}_1) . \cdots\linvar{x}.\overline{\some}. \outact{\linvar{x}}{y_k} . (y_k . \some_{\emptyset} ; y_{k}.\close;0 \para \linvar{x}.\overline{\some};\linvar{x}.\some_{u,(\llfv{M} \setminus  {x}_k )};\linvar{x}( {x}_k) . \\
 &\linvar{x}.\overline{\some}; \outact{\linvar{x}}{y_{k+1}}. ( y_{k+1} . \some_{u,\llfv{M}} ;y_{k+1}.\close; \piencodf{M}_u \para \linvar{x}.\overline{\none} )) \cdots ) \para \\
 & \linvar{x}.\some_{\llfv{C}} ; \linvar{x}(y_1). \linvar{x}.\some_{y_1, \llfv{C}} ;\linvar{x}.\overline{\some} ; \outact{\linvar{x}}{ {x}_1}. ( {x}_1.\some_{\llfv{C_i(1)}} ; \piencodf{C_i(1)}_{ {x}_1} \para y_1. \overline{\none} \para \cdots \\
 & \linvar{x}.\some_{\llfv{C_i(l)}} ; \linvar{x}(y_l). \linvar{x}.\some_{y_l, \llfv{C_i(l)} };\linvar{x}.\overline{\some} ; \outact{\linvar{x}}{ {x}_l}. ( {x}_l.\some_{\llfv{C_i(l)}} ; \piencodf{C_i(l)}_{ {x}_l} \para y_l. \overline{\none} \para \\
  & \linvar{x}.\some_{\emptyset} ; \linvar{x}(y_{l+1}). ( y_{l+1}.\overline{\some};y_{l+1} . \overline{\close} \para \linvar{x}.\some_{\emptyset} ; \linvar{x}. \overline{\none}) ) 
                               ) \para !\banged{x}. (x_i). \piencodf{ U }_{x_i})  \\
&\red^{5l} \bigoplus_{C_i \in \perm{C}} (\nu \linvar{x}, \banged{x} , y_1,  {x}_1, \cdots  y_l,  {x}_l)(   y_1 . \some_{\emptyset} ;y_{1}.\close;0 \para \cdots \para y_l . \some_{\emptyset} ;y_{l}.\close;0 \\
 &\linvar{x}.\overline{\some}. \outact{\linvar{x}}{y_{l+1}} . (y_{l+1} . \some_{\emptyset} ; y_{l+1}.\close;0 \para \linvar{x}.\overline{\some};\linvar{x}.\some_{u,(\llfv{M} \setminus  {x}_{l+1} , \cdots ,  {x}_k )};\linvar{x}( {x}_{l+1}) . \cdots \\
&\linvar{x}.\overline{\some}. \outact{\linvar{x}}{y_k} . (y_k . \some_{\emptyset} ; y_{k}.\close;0 \para \linvar{x}.\overline{\some};\linvar{x}.\some_{u,(\llfv{M} \setminus  {x}_k )};\linvar{x}( {x}_k) . \\
 &\linvar{x}.\overline{\some}; \outact{\linvar{x}}{y_{k+1}}. ( y_{k+1} . \some_{u,\llfv{M}} ;y_{k+1}.\close; \piencodf{M}_u \para \linvar{x}.\overline{\none} )) \cdots ) \para \\
 &    {x}_1.\some_{\llfv{C_i(1)}} ; \piencodf{C_i(1)}_{ {x}_1} \para \cdots \para   {x}_l.\some_{\llfv{C_i(l)}} ; \piencodf{C_i(l)}_{ {x}_l} \para  y_1. \overline{\none} \para \cdots \para y_l. \overline{\none}\\
& \linvar{x}.\some_{\emptyset} ; \linvar{x}(y_{l+1}). ( y_{l+1}.\overline{\some};y_{l+1} . \overline{\close} \para \linvar{x}.\some_{\emptyset} ; \linvar{x}. \overline{\none}) \para !\banged{x}. (x_i). \piencodf{ U }_{x_i}) \qquad (:= P_{\mathbb{N}})
                        \end{aligned}
                        \]
                        \[
                        \begin{aligned}
    \qquad  P_{\mathbb{N}}              
 & \red^{l+ 5} 
 \bigoplus_{C_i \in \perm{C}} (\nu \linvar{x}, \banged{x} ,  {x}_1, \cdots,  {x}_l) ( \linvar{x}.\some_{u,(\llfv{M} \setminus  {x}_{l+1} , \cdots ,  {x}_k )};\linvar{x}( {x}_{l+1}) . \cdots \\ 
 &\linvar{x}.\overline{\some}. \outact{\linvar{x}}{y_k} . (y_k . \some_{\emptyset} ; y_{k}.\close;0 \para \linvar{x}.\overline{\some};\linvar{x}.\some_{u,(\llfv{M} \setminus  {x}_k )};\linvar{x}( {x}_k) . \\
 &\linvar{x}.\overline{\some}; \outact{\linvar{x}}{y_{k+1}}. ( y_{k+1} . \some_{u,\llfv{M}} ;y_{k+1}.\close; \piencodf{M}_u \para \linvar{x}.\overline{\none} ) )  \para \\
 &    {x}_1.\some_{\llfv{C_i(1)}} ; \piencodf{C_i(1)}_{ {x}_1} \para \cdots \para   {x}_l.\some_{\llfv{C_i(l)}} ; \piencodf{C_i(l)}_{ {x}_l} \para   \linvar{x}. \overline{\none} \para !\banged{x}. (x_i). \piencodf{ U }_{x_i} ) \\
  & \red
  \bigoplus_{C_i \in \perm{C}} (\nu \banged{x} ,  {x}_1, \cdots,  {x}_l)(  u . \overline{\none} \para  {x}_1 . \overline{\none} \para  \cdots \para  {x}_{l} . \overline{\none} \para (\llfv{M} \setminus \{   {x}_1, \cdots ,  {x}_k \} ). \overline{\none} \para  \\
  &  {x}_1.\some_{\llfv{C_i(1)}} ; \piencodf{C_i(1)}_{ {x}_1} \para \cdots \para   {x}_l.\some_{\llfv{C_i(l)}} ; \piencodf{C_i(l)}_{ {x}_l}  ) \\ 
  & \red^{l}   \bigoplus_{C_i \in \perm{C}}  u . \overline{\none} \para (\llfv{M} \setminus \{  x_1, \cdots , x_k \} ). \overline{\none} \para \llfv{C}. \overline{\none} \\
 &= \piencodf{\sum_{C_i \in \perm{C}}  \fail^{\widetilde{z}}}_u = Q_{7l + 10}  \\
 \end{aligned}
 \]
The rest of the proof is by induction on the number of reductions $\piencodf{\expr{N}}_u \red^j Q$.
                            
                            \begin{enumerate}
                                \item When $j = 0$, the result follows trivially. Just take $\mathbb{N}=\mathbb{N}'$ and $\piencodf{\expr{N}}_u=Q=Q'$.
  \item When $1 \leq j \leq 7l + 10$.
                                    
Let $Q_j$ be such that $ \piencodf{\expr{N}}_u \red^j Q_j$.
By the steps above one has 

\(\begin{aligned}
  &Q_j \red^{7l + 10 - j} Q_{7l + 6} = Q',\\ &\expr{N} \red^1 \sum_{C_i \in \perm{C}}  \fail^{\widetilde{z}} = \expr{N}';\text{ and} \piencodf{\sum_{C_i \in \perm{C}}  \fail^{\widetilde{z}}}_u = Q_{7l + 10}.
\end{aligned}
\)
\item When $j > 7l + 10$.

In this case, we have 
$$ \piencodf{\expr{N}}_u \red^{7l + 10} Q_{7l + 10} \red^n Q,$$ for $n \geq 1$. 
We also know that 
$\expr{N} \red^1 \sum_{C_i \in \perm{C}}  \fail^{\widetilde{z}}$. However no further reductions can be performed.

                            \end{enumerate}

                \item When $\size{\widetilde{x}} < \size{C}$, the proof  proceeds similarly to the previous case.
                    
            \end{enumerate}

        \item  $\expr{N} =  M \linexsub{N' /  {x}}$. 
    
           In this case, 
            \(                \piencodf{M \linexsub{N' /  {x}}}_u =  (\nu  {x}) ( \piencodf{ M }_u \para  {x}.\some_{\llfv{N'}};\piencodf{ N' }_{ {x}} ).
            \)
            Therefore,
            \[
            \begin{aligned}
               \piencodf{\expr{N}}_u & =  (\nu  {x}) ( \piencodf{ M }_u \para  {x}.\some_{\llfv{N'}};\piencodf{ N' }_{ {x}} ) \red^m  (\nu  {x}) ( R \para  {x}.\some_{\llfv{N'}};\piencodf{ N' }_{ {x}}  ) \red^n  Q,\\
            \end{aligned}
            \]
            
            for some process $R$. Where $\red^n$ is a reduction that  initially synchronizes with $ {x}.\some_{\llfv{N'}}$ when $n \geq 1$, $n + m = k \geq 1$. Type preservation in \spi ensures reducing $\piencodf{ M}_v \red^m$ does not consume possible synchronizations with $ {x}.\some$, if they occur. Let us consider the the possible sizes of both $m$ and $n$.
            
            \begin{enumerate}
                \item For $m = 0$ and $n \geq 1$.
                
                    We have that $R = \piencodf{M}_u$ as $\piencodf{M}_u \red^0 \piencodf{M}_u$. 
                    
                    Notice that there are two possibilities of having an unguarded $x.\overline{\some}$ or $x.\overline{\none}$ without internal reductions: 
                    
                    \begin{enumerate}
                        \item $M = \fail^{ {x}, \widetilde{y}}$.
  \[
  \begin{aligned}
  \piencodf{\expr{N}}_u & =   (\nu  {x}) ( \piencodf{ M }_u \para  {x}.\some_{\llfv{N'}};\piencodf{ N' }_{ {x}} )\\
  &=   (\nu  {x}) ( \piencodf{ \fail^{ {x}, \widetilde{y}}}_u \para  {x}.\some_{\llfv{N'}};\piencodf{ N' }_{ {x}} ) \\
    & =   (\nu  {x}) ( u.\overline{\none} \para  {x}.\overline{\none} \para  \widetilde{y}.\overline{\none} \para  {x}.\some_{\llfv{N'}};\piencodf{ N' }_{ {x}} )\\
    &\red u.\overline{\none} \para \widetilde{y}.\overline{\none} \para \llfv{N'}.\overline{\none} \\
  \end{aligned}
    \]
  Notice that no further reductions can be performed.
  Thus,                          
 $$ \piencodf{\expr{N}}_u \red u.\overline{\none} \para \widetilde{y}.\overline{\none} \para \lfv{N'}.\overline{\none}  = Q'.$$
We also have that  $\expr{N} \red \fail^{ \widetilde{y} \cup \llfv{N'} } = \expr{N}'$ and $\piencodf{ \fail^{\widetilde{y} \cup \llfv{N'}} }_u = Q'$.

                        \item $\headf{M} =  {x}$.
                    
                            By the diamond property we will be reducing each non-deterministic choice of a process simultaneously.
                            Then we have the following
 \[
 \begin{aligned}
 \piencodf{\expr{N}}_u  & =  (\nu  {x}) ( \bigoplus_{i \in I}(\nu \widetilde{y})(\piencodf{  {x} }_{j} \para P_i) \para    {x}.\some_{\llfv{N'}};\piencodf{ N' }_{ {x}}  ) \\
 & =  (\nu  {x}) ( \bigoplus_{i \in I}(\nu \widetilde{y})(  {x}.\overline{\some} ; [ {x} \leftrightarrow j]  \para P_i) \para    {x}.\some_{\llfv{N'}};\piencodf{ N' }_{ {x}}  ) \\
 & \red   (\nu  {x}) ( \bigoplus_{i \in I}(\nu \widetilde{y})(  [ {x} \leftrightarrow j]  \para P_i) \para   \piencodf{ N' }_{ {x}}  )& = Q_1 \\
 & \red \bigoplus_{i \in I}(\nu \widetilde{y})(  \piencodf{ N' }_{j} \para P_i ) & = Q_2\\
 \end{aligned}
 \]
                            
In addition,
\(
\expr{N} =M \linexsub {N' / {x}}\red  M \headlin{ N' / {x} } = \expr{M}\).
Finally,
\[
\begin{aligned}                               \piencodf{\expr{M}}_u= \piencodf{M \headlin{ N'/  {x} }}_u &= \bigoplus_{i \in I}(\nu \widetilde{y})(  \piencodf{ N' }_{j} \para P_i ) & = Q_2. 
\end{aligned}
\]
                            
\begin{enumerate}
\item When $n = 1$:

Then, $Q = Q_1$ and  $ \piencodf{\expr{N}}_u \red^1 Q_1$. Also,

$Q_1 \red^1 Q_2 = Q'$, $\expr{N} \red^1 M \headlin{ N'/ {x}} = \expr{N}'$ and $\piencodf{M \headlin{ N'/ {x}}}_u = Q_2$.
\item When $n \geq 2$:

Then  $ \piencodf{\expr{N}}_u \red^2 Q_2 \red^l Q$, for $l \geq 0$.  Also, 
$\expr{N} \rightarrow \expr{M}$, $Q_2 = \piencodf{\expr{M}}_u$. By the induction hypothesis, there exist $ Q'$ and $\expr{N}'$ such that $ Q \red^i Q'$, $\expr{M} \red_{\pequiv}^j \expr{N}'$ and $\piencodf{\expr{N}'}_u = Q'$. Finally, $\piencodf{\expr{N}}_u \red^2 Q_2 \red^l Q \red^i Q'$ and $\expr{N} \rightarrow \expr{M}  \red_{\pequiv}^j \expr{N}'$.
                                
                            \end{enumerate}

                    \end{enumerate}
 \item  For $m \geq 1$ and $ n \geq 0$.
                    
            \begin{enumerate}
            \item When $n = 0$.
            
               Then $  (\nu  {x}) ( R \para  {x}.\some_{\llfv{N'}};\piencodf{ N' }_{ {x}} )  = Q$ and $\piencodf{M}_u \red^m R$ where $m \geq 1$. By the IH there exist $R'$  and $\expr{M}' $ such that $R \red^i R'$, $M \red_{\pequiv}^j \expr{M}'$ and $\piencodf{\expr{M}'}_u = R'$. Thus,
              \[ 
               \begin{aligned}
                   \piencodf{\expr{N}}_u & = (\nu  {x}) ( \piencodf{M}_u  \para  {x}.\some_{\llfv{N'}};\piencodf{ N' }_{ {x}}  ) \red^m  (\nu  {x}) ( R \para  {x}.\some_{\llfv{N'}};\piencodf{ N' }_{ {x}} )  = Q
                \end{aligned}
                 \]
                Also,
                \(
               Q  \red^i  (\nu  {x}) ( R' \para  {x}.\some_{\llfv{N'}};\piencodf{ N' }_{ {x}} )  = Q',
                \)
                and the term can reduce as follows: $\expr{N} = M \linexsub {N' / {x}} \red_{\pequiv}^j \sum_{M_i' \in \expr{M}'} M_i' \linexsub {N' / {x}} = \expr{N}'$ and  $\piencodf{\expr{N}'}_u = Q'$

            \item When $n \geq 1$.
                Then  $R$ has an occurrence of an unguarded $x.\overline{\some}$ or $x.\overline{\none}$, this case follows by IH.
                        
                    \end{enumerate}
            \end{enumerate}

            \item  $\expr{N} =  M \unexsub{U / \unvar{x}}$. 
    
            In this case, 
            \(
            \begin{aligned}
                \piencodf{M \unexsub{U / \unvar{x}}}_u &=   (\nu \banged{x}) ( \piencodf{ M }_u \para   !\banged{x}. (x_i).\piencodf{ U }_{x_i} ).
            \end{aligned}
            \)
            Then, 
            \[
            \begin{aligned}
               \piencodf{\expr{N}}_u & =   (\nu \banged{x}) ( \piencodf{ M }_u \para   !\banged{x}. (x_i).\piencodf{ U }_{x_i} )  \red^m  (\nu \banged{x}) ( R \para   !\banged{x}. (x_i).\piencodf{ U }_{x_i} ) \red^n  Q.
            \end{aligned}
            \]
            for some process $R$. Where $\red^n$ is a reduction initially synchronises with $!\banged{x}. (x_i)$ when $n \geq 1$, $n + m = k \geq 1$. Type preservation in \spi ensures reducing $\piencodf{ M}_v \red^m$ doesn't consume possible synchronisations with $!\banged{x}. (x_i)$ if they occur. Let us consider the the possible sizes of both $m$ and $n$.
            
            \begin{enumerate}
                \item For $m = 0$ and $n \geq 1$.
                
                   In this case,  $R = \piencodf{M}_u$ as $\piencodf{M}_u \red^0 \piencodf{M}_u$. 
                    
                    Notice that the only possibility of having an unguarded $ \outsev{\banged{x}}{{x_i}}$ without internal reductions is when   $\headf{M} =  {x}[ind].$
                           By the diamond property  we will be reducing each non-deterministic choice of a process simultaneously.
                            Then we have the following:
                            
                            \[
                            \begin{aligned}
                            \piencodf{\expr{N}}_u  & =  (\nu \banged{x}) ( \bigoplus_{i \in I}(\nu \widetilde{y})(\piencodf{  {x}[ind] }_{j} \para P_i) \para   !\banged{x}. (x_i).\piencodf{ U }_{x_i} ) \\
                            & =  (\nu \banged{x}) ( \bigoplus_{i \in I}(\nu \widetilde{y})( \outsev{\banged{x}}{{x_i}}. {x}_i.l_{ind}; [{x_i} \leftrightarrow j]  \para P_i) \para   !\banged{x}. (x_i).\piencodf{ U }_{x_i}  ) \\
                            & \red  (\nu \banged{x}) ( \bigoplus_{i \in I}(\nu \widetilde{y})( (\nu x_i) ({x}_i.l_{ind}; [{x_i} \leftrightarrow j] \para \piencodf{ U }_{x_i} ) \para P_i) \para   !\banged{x}. (x_i).\piencodf{ U }_{x_i}  ) &= Q_1 \\
                            & =  (\nu \banged{x}) ( \bigoplus_{i \in I}(\nu \widetilde{y})( (\nu x_i) ({x}_i.l_{ind}; [{x_i} \leftrightarrow j] \para x_i. case( ind.\piencodf{U_{ind}}_{x_i} ) ) \para P_i) \\
                            &\para   !\banged{x}. (x_i).\piencodf{ U }_{x_i}  ) \\
                            & \red  (\nu \banged{x}) ( \bigoplus_{i \in I}(\nu \widetilde{y})( (\nu x_i) ( [{x_i} \leftrightarrow j] \para \piencodf{U_{ind}}_{x_i} ) \para P_i) \para   !\banged{x}. (x_i).\piencodf{ U }_{x_i}  ) &= Q_2\\
                            & \red  (\nu \banged{x}) ( \bigoplus_{i \in I}(\nu \widetilde{y})( \piencodf{U_{ind}}_{j}  \para P_i) \para   !\banged{x}. (x_i).\piencodf{ U }_{x_i}  ) &= Q_3 \\
                            \end{aligned}
                            \]

                We consider the two cases of the form of $U_{ind}$ and show that the choice of $U_{ind}$ is inconsequential
            
                \begin{itemize}
                \item When $ U_i = \banged{\bag{N}}$:
                
                In this case, 
                \(
                \begin{aligned}
                \expr{N} &=M \unexsub{U / \unvar{x}}\red M \headlin{ N /\banged{x} }\unexsub{U / \unvar{x}} = \expr{M}.
                \end{aligned}
                \)
                 and 
                \[
                 \begin{aligned}
                 \piencodf{\expr{M}}_u= \piencodf{M \headlin{ N /\banged{x} }\unexsub{U / \unvar{x}}}_u &=  (\nu \banged{x}) ( \bigoplus_{i \in I}(\nu \widetilde{y})( \piencodf{\bag{N}}_{j}  \para P_i) \para   !\banged{x}. (x_i).\piencodf{ U }_{x_i}  ) & = Q_3 
                                            \end{aligned}
                 \]
                
                \item When $ U_i = \banged{\oneb} $:
                  
                  In this case,
                        \(
                        \begin{aligned}
                            \expr{N} &=M \unexsub{U / \unvar{x}} \red M \headlin{ \fail^{\emptyset} /\banged{x} } \unexsub{U /\unvar{x} } = \expr{M}.
                        \end{aligned}
                        \)
                        
                        Notice that $\piencodf{\banged{\oneb}}_{j} =  j.\none$ and that $\piencodf{\fail^{\emptyset}}_j = j.\overline{\none}$. In addition,
 
                            \[
                            \begin{aligned}
                               \piencodf{\expr{M}}_u&= \piencodf{M \headlin{ \fail^{\emptyset} /\banged{x} } \unexsub{U /\unvar{x} }}_u\\
                               &=  (\nu \banged{x}) ( \bigoplus_{i \in I}(\nu \widetilde{y})( \piencodf{\fail^{\emptyset}}_{j}  \para P_i) \para   !\banged{x}. (x_i).\piencodf{ U }_{x_i}  ) \\
                               & = (\nu \banged{x}) ( \bigoplus_{i \in I}(\nu \widetilde{y})( \piencodf{\banged{\oneb}}_{j}  \para P_i) \para   !\banged{x}. (x_i).\piencodf{ U }_{x_i}  ) & = Q_3 
                            \end{aligned}
                            \]
                \end{itemize}
                
                Both choices give an $\expr{M}$ that are equivalent to $Q_3$.

    \begin{enumerate}
    \item When $n \leq 2$.
    
   In this case, $Q = Q_n$ and  $ \piencodf{\expr{N}}_u \red^n Q_n$.
                                 
Also, $Q_n \red^{3-n} Q_3 = Q'$, $\expr{N} \red^1 \expr{M} = \expr{N}'$ and $\piencodf{\expr{M} }_u = Q_2$.
                                 
     \item When $n \geq 3$.
     
     We have $ \piencodf{\expr{N}}_u \red^3 Q_3 \red^l Q$ for $l \geq 0$. We also know that $\expr{N} \rightarrow \expr{M}$, $Q_3 = \piencodf{\expr{M}}_u$. By the IH, there exist $ Q$ and $\expr{N}'$ such that $Q \red^i Q'$, $\expr{M} \red_{\pequiv}^j \expr{N}'$ and $\piencodf{\expr{N}'}_u = Q'$. Finally, $\piencodf{\expr{N}}_u \red^2 Q_3 \red^l Q \red^i Q'$ and $\expr{N} \rightarrow \expr{M}  \red_{\pequiv}^j \expr{N}' $.

                    \end{enumerate}
 \item For $m \geq 1$ and $ n \geq 0$.
                    
            \begin{enumerate}
            \item When $n = 0$.
            
               Then $   (\nu \banged{x}) ( R \para   !\banged{x}. (x_i).\piencodf{ U }_{x_i} )  = Q$ and $\piencodf{M}_u \red^m R$ where $m \geq 1$. By the IH there exist $R'$  and $\expr{M}' $ such that $R \red^i R'$, $M \red_{\pequiv}^j \expr{M}'$ and $\piencodf{\expr{M}'}_u = R'$. 
              Hence,
              \[
               \begin{aligned}
                   \piencodf{\expr{N}}_u & =  (\nu \banged{x}) ( \piencodf{ M }_u \para   !\banged{x}. (x_i).\piencodf{ U }_{x_i} ) \red^m   (\nu \banged{x}) ( R \para   !\banged{x}. (x_i).\piencodf{ U }_{x_i} )  = Q.
                \end{aligned}
                 \]
                In addition,
                \(
                   Q  \red^i   (\nu \banged{x}) ( R' \para   !\banged{x}. (x_i).\piencodf{ U }_{x_i} ) = Q\), and the term can reduce as follows: $ \expr{N} = M \unexsub{U / \unvar{x}} \red_{\pequiv}^j \sum_{M_i' \in \expr{M}'} M_i' \unexsub{U / \unvar{x}} = \expr{N}'$ and  $\piencodf{\expr{N}'}_u = Q'$.

            \item When $n \geq 1$.
            
            Then $R$ has an occurrence of an unguarded $\outsev{\banged{x}}{{x_i}}$, and the case follows by IH.
\end{enumerate}
            \end{enumerate}
               \end{enumerate}
\end{proof}

\subsection{Success Sensitiveness of $\piencodf{\cdot}_u$}


We say that a process occurs \emph{guarded} when it occurs behind a prefix (input, output, closing of channels, servers, server request, choice an selection and non-deterministic session behaviour). Formally,

\begin{definition}
 A process $P\in \spi$ is {\em guarded} if  $\alpha.P$ , $ \alpha;P$ or $ \choice{x}{i}{I}{i}{P} $, where $ \alpha = \overline{x}(y), x(y), x.\overline{\close}, x.\close,$ $ x.\overline{\some}, x.\some_{(w_1, \cdots, w_n)} , \case{x}{i} , !x(y) ,  \outsev{x}{y}$. 
We say it occurs \emph{unguarded} if it is not guarded for any prefix.
\end{definition}

\begin{proposition}[Preservation of Success]
\label{Prop:checkprespiunres}
For all $M\in \lamrsharfailunres$, the following hold:
\begin{enumerate}
    \item $ \headf{M} = \checkmark \implies \piencodf{M} = P \para \checkmark \oplus Q $
    \item $ \piencodf{M}_u =  P \para \checkmark \oplus Q \implies \headf{M} = \checkmark$
\end{enumerate}

\end{proposition}

\begin{proof}

Proof of both cases by induction on the structure of $M$. 

\begin{enumerate} 
\item We only need to consider terms of the following form:

    \begin{enumerate}

        \item  $ M = \checkmark $:
        
        This case is immediate.
        
        \item $M = N\ (C \bagsep U)$:
        
        Then, $\headf{N \ (C \bagsep U)} = \headf{N}$. If $\headf{N} = \checkmark$, then  $$ \piencodf{M (C \bagsep U)}_u = \bigoplus_{C_i \in \perm{C}} (\nu v)(\piencodf{M}_v \para v.\some_{u , \llfv{C}} ; \outact{v}{x} . ([v \leftrightarrow u] \para \piencodf{C_i \bagsep U}_x ) ).$$
        By the IH,  $\checkmark$ is unguarded in $\piencodf{N}_u$.

        \item $M = M' \linexsub {N /x}$
        
        Then we have that $\headf{M' \linexsub{N /  {x}}}  = \headf{M'} = \checkmark$. Then $\piencodf{ M' \linexsub{N /  {x}}  }_u   =   (\nu  {x}) ( \piencodf{ M' }_u \para    {x}.\some_{\llfv{N}};\piencodf{ N }_{ {x}}  )$ and by the IH $\checkmark$ is unguarded in $\piencodf{M'}_u$.
        
        \item $M = M' \unexsub{U / \unvar{x}}$
        
        Then we have that $\headf{M' \unexsub{U / \unvar{x}}}  = \headf{M'} = \checkmark$. Then $\piencodf{ M' \unexsub{U / \unvar{x}}  }_u   =   (\nu \banged{x}) ( \piencodf{ M' }_u \para   !\banged{x}. (x_i).\piencodf{ U }_{x_i} ) $ and by the IH  $\checkmark$ is unguarded in $\piencodf{M'}_u$.

    \end{enumerate}

   \item We only need to consider terms of the following form:    
         
    \begin{enumerate}
    
        \item {\bf Case $M = \checkmark$:}
        
        Then, 
        $\piencodf{\checkmark}_u = \checkmark$ 
        which is an unguarded occurrence of $\checkmark$ and that $\headf{\checkmark} = \checkmark$.
        
        \item {\bf Case $M = N (C \bagsep U)$:}
        
        Then, $\piencodf{N (C \bagsep U)}_u = \bigoplus_{C_i \in \perm{C}} (\nu v)(\piencodf{N}_v \para v.\some_{u , \llfv{C}} ; \outact{v}{x} . ([v \leftrightarrow u] \para \piencodf{C_i \bagsep U}_x ) )$. The only occurrence of an unguarded $\checkmark$ can occur is within $\piencodf{N}_v$. By the IH, $\headf{N} = \checkmark$ and finally $\headf{N \ B} = \headf{N}$.

        \item {\bf Case $M = M' \linexsub{N /  {x}}$:} 
        
        Then,  $\piencodf{ M' \linexsub{N /  {x}}  }_u   =   (\nu  {x}) ( \piencodf{ M' }_u \para    {x}.\some_{\llfv{N}};\piencodf{ N }_{ {x}}  ) $, an unguarded occurrence of $\checkmark$ can only occur within $\piencodf{ M' }_u $. By the IH,  $\headf{M'} = \checkmark$ and hence $\headf{ M' \linexsub{N /  {x}}}  = \headf{M'}$.
        
        \item {\bf Case $M = M' \unexsub{U / \unvar{x}}$:}
        This case is analogous to the previous.
    \end{enumerate}  
\end{enumerate}
\end{proof}

\begin{theorem}[Success Sensitivity]
\label{proof:successsenscetwounres}
The encoding $\piencodf{-}_u:\lamrsharfailunres \rightarrow \spi$ is success sensitive on well formed linearly closed expression if
for any expression we have $\succp{\expr{M}}{\checkmark}$ iff $\succp{\piencodf{\expr{M}}_u}{\checkmark}$.
\end{theorem}

\begin{proof}
We proceed with the proof in two parts.

\begin{enumerate}
    
    \item Suppose that  $\expr{M} \Downarrow_{\checkmark} $. We will prove that $\piencodf{\expr{M}} \Downarrow_{\checkmark}$.

    By \defref{def:app_Suc3unres}, there exists  $ \expr{M}' = M_1 , \cdots , M_k$ such that $\expr{M} \red^* \expr{M}'$ and
    $\headf{M_j'} = \checkmark$, for some  $j \in \{1, \ldots, k\}$ and term $M_j'$ such that $M_j\pequiv  M_j'$.
    By completeness, there exists $ Q$ such that $\piencodf{\expr{M}}_u  \red^* Q = \piencodf{\expr{M'}}_u$.
    
    We wish to show that there exists $Q'$such that $Q \red^* Q'$ and $Q'$ has an unguarded occurrence of $\checkmark$.
    
    From $Q = \piencodf{\expr{M}'}_u$ and due to compositionality and the homomorphic preservation of non-determinism we have that
    \(
        \begin{aligned}
            Q &= \piencodf{M_1}_u \oplus \cdots \oplus \piencodf{M_k}_u\\
        \end{aligned}
    \).
    
    By Proposition \ref{Prop:checkprespiunres} (1) we have that $\headf{M_j} = \checkmark \implies \piencodf{M_j}_u =  P \para \checkmark \oplus Q'$. Hence $Q$ reduces to a process that has an unguarded occurence of $\checkmark$.

    \item Suppose that $\piencodf{\expr{M}}_u \Downarrow_{\checkmark}$. We will prove that $ \expr{M} \Downarrow_{\checkmark}$.

    By operational soundness (Lemma~\ref{l:app_soundnesstwounres}) we have that if $ \piencodf{\expr{N}}_u \red^* Q$
    then there exist $Q'$  and $\expr{N}' $ such that 
    $Q \red^* Q'$, $\expr{N}  \red_{\pequiv}^* \expr{N}'$ 
    and 
    $\piencodf{\expr{N}'}_u = Q'$.
    
   Since $\piencodf{\expr{M}}_u \red^* P_1 \oplus \ldots \oplus P_k$, and $P_j'= P_j'' \para \checkmark$, for some $j$ and $P_j'$, such that $P_j \equiv P_j'$. 
   
   Notice that if $\piencodf{\expr{M}}_u$ is itself a term with unguarded $\checkmark$, say $\piencodf{\expr{M}}_u=P \para \checkmark$, then $\expr{M}$ is itself headed with $\checkmark$, from Proposition \ref{Prop:checkprespiunres} (2).
   
   In the case $\piencodf{\expr{M}}_u= P_1 \oplus \ldots \oplus P_k$, $k\geq 2$, and $\checkmark$ occurs unguarded in an $P_j$, The encoding acts homomorphically over sums and the reasoning is similar. We have that $P_j = P_j' \para \checkmark$ we apply Proposition \ref{Prop:checkprespiunres} (2).

\end{enumerate}
\end{proof}

\fi

\end{document}